\documentclass[12pt]{article}
\usepackage[a4paper,margin=1in]{geometry}
\usepackage{setspace}
\usepackage{etex}
\usepackage{setspace}
\usepackage{comment}
\usepackage[pdftex]{graphicx}
\usepackage{rotating}
\usepackage{amsmath}
\usepackage{amsthm}
\usepackage{amssymb}
\usepackage{graphicx}
\usepackage{tikz}
\usepackage{subcaption}
\usepackage{rotating}
\usepackage{mathtools}
\usepackage{bbm}
\usepackage{amsfonts}
\usepackage{setspace}
\usepackage[round]{natbib}
\usepackage{color}
\usepackage{booktabs}
\usepackage{caption}
\newtheorem{theorem}{Theorem}[section]
\newtheorem{lemma}{Lemma}[section]
\newtheorem{definition}{Definition}[section]
\usepackage{amssymb}
\usepackage[linesnumbered, ruled]{algorithm2e}

\SetKwRepeat{Do}{do}{while}%

\usepackage{amsmath}               
  {
      \theoremstyle{plain}
      \newtheorem{assumption}{Assumption}
  }

\usepackage{amsmath}               
{
  \theoremstyle{remark}
  
}

\usepackage{tikz}
\usepackage{pgf}
\usepackage{graphicx}
\usepackage{pstricks}
\usepackage{framed}
\usepackage{varwidth}
\usepackage[normalem]{ulem}
\usetikzlibrary{shapes}
\usetikzlibrary{arrows,decorations.pathmorphing,backgrounds,fit,
positioning,shapes.symbols,chains}

\usepackage{longtable}
\usepackage{booktabs}
\usepackage{array}
\usepackage{multirow}
\usepackage{wrapfig}
\usepackage{float}
\usepackage{colortbl}
\usepackage{pdflscape}
\usepackage{tabu}
\usepackage{threeparttable}
\usepackage{threeparttablex}
\usepackage[normalem]{ulem}
\usepackage{makecell}
\usepackage{xcolor}
\usepackage{threeparttable}
\usepackage{comment}

\newcommand{\vect}[1]{\mbox{\boldmath $#1$}}

\newcommand{\vA}{\vect{A}}

\newcommand{\bbE}{\mathbb{E}}

\newcommand{\OMIT}[1]{\relax}

\def\text{{\rm}}

 \newcommand{\bma}[1]{\mbox{\boldmath $#1$}}

 \newcommand{\appendixhead}%
{\textbf{\huge Appendix}
\vspace{0.1in}}

 \newcommand{\bA}{ {\mathbf{A}} }
 \newcommand{\ba}{ {\mathbf{a}} }

 \newcommand{\bc}{ {\bma{c}} }
 \newcommand{\bD}{ {\bma{D}} }
 \newcommand{\bd}{ {\bma{d}} }

 \newcommand{\bH}{ {\bma{H}} }
 \newcommand{\bh}{ {\bma{h}} }

 \newcommand{\bK}{ {\bma{K}} }
 \newcommand{\bk}{ {\bma{k}} }

 \newcommand{\bq}{ {\bma{q}} }

 \newcommand{\bT}{ {\bma{T}} }

 \newcommand{\bV}{ {\bma{V}} }
 
 \newcommand{\bW}{ {\bma{W}} }
 
 \newcommand{\bX}{ {\bma{X}} }
 \newcommand{\bx}{ {\bma{x}} }
 \newcommand{\bY}{ {\bma{Y}} }
 
 \newcommand{\bZ}{ {\bma{Z}} }

\newcommand{\bbG}{ {\mathbb{G}} }
\newcommand{\bbT}{ {\mathbb{T}} }
\newcommand{\bbI}{ {\mathbb{I}} }

\newcommand{\epochfield}{\mathcal{F}_{j-1}}

\newcommand{\epochgroup}{\mathcal{E}_{j-1}}

\newcommand{\bhistory}{ \bh^{n_J}}
\newcommand{\bHistory}{ \mathcal{H}^{n_J}}
\newcommand{\branchasymp}{\mathscr{I}}
\newcommand{\branchingvalue}{V^{n_J}}
\newcommand{\noepochsamplesize}{\kappa}

 \newcommand{\bmu}{ {\bma{\mu}} }

 \newcommand{\bbeta}{ {\pmb{\beta}} }
 
 \newcommand{\btheta}{ {\pmb{\theta}} }
 \newcommand{\bTheta}{ {\pmb{\Theta}} }

 \newcommand{\bgamma}{ {\pmb{\gamma}} }

 \newcommand{\balpha}{ {\pmb{\alpha}} }

 \newcommand{\bbP}{\mathbb{P}}
 \newcommand{\bpi}{\pmb{\pi}}
 \newcommand{\bPi}{\pmb{\Pi}}
\newcommand{\bphi}{\pmb{\phi}}

\usepackage{mathrsfs}
\usepackage{arydshln}

 \usepackage[shortlabels]{enumitem}

\newtheorem{thm}{Theorem}[section]

\theoremstyle{definition}

\usepackage{xr}
\makeatletter
\newcommand*{\addFileDependency}[1]{
\typeout{(#1)}
\@addtofilelist{#1}
\IfFileExists{#1}{}{\typeout{No file #1.}}
}\makeatother

\title{Reinforcement Learning for Respondent-Driven Sampling}
\author{
 Justin Weltz \thanks{Department of Statistical Science, Duke University, justin.weltz@duke.edu}
 \and
 Angela Yoon \thanks{Department of Neurosurgery and Neurology, Duke University, angela.yoon@duke.edu}
 \and
 Yichi Zhang \thanks{Department of Computer Science and Statistics, University of Rhode Island, yichizhang@uri.edu}
 \and
 Alexander Volfovsky \thanks{Department of Statistical Science, Duke University, alexander.volfovsky@duke.edu} 
 \and
 Eric Laber \thanks{Department of Statistical Science, Duke University, eric.laber@duke.edu} 
}
\date{\vspace{-0.5cm}}

\begin{document}

\maketitle

\begin{abstract} 
Respondent-driven sampling (RDS) is widely used to study hidden or hard-to-reach populations 
by incentivizing study participants to recruit their social connections.  
The success and efficiency of RDS can depend critically on the nature of the 
incentives, including their number, value, call to action, etc. Standard
RDS uses an incentive structure that is set {\em a priori} and held
fixed throughout the study. Thus, it does not make use of accumulating information
on which incentives are effective and for whom. 
We propose a reinforcement learning (RL) based adaptive RDS study design in which 
the incentives are tailored over time to maximize cumulative utility during the study.  
We show that these designs are more efficient, cost-effective, and can generate 
new insights into the social structure of hidden populations.  
In addition, we develop methods for valid post-study inference which are non-trivial
due to the adaptive sampling induced by RL as well as the complex dependencies among
subjects due to latent (unobserved) social network structure.  We provide asymptotic
regret bounds and illustrate its finite sample behavior through a suite of 
simulation experiments.  
\end{abstract}

Social network data are at the forefront of healthcare research. 
It is widely acknowledged that a better understanding of community social structure can provide insights into the prevalence of mental illnesses such as depression, transmittable diseases such as HIV, syphilis, and COVID, and other conditions like obesity and type-II diabetes \citep[][]{ma2007trends,tabak2012prediabetes,eisenberg2013social,perry2018egocentric}. 
Currently, available methodologies for studying the underlying social network in a community fall broadly into two
categories: (1) complete community census \citep[such as in the Framingham Heart Study,][]{mahmood2014framingham}, or (2) network sampling algorithms. 
The magnitude of modern healthcare needs makes complete censuses nearly unachievable (and generally impractical) as the relevant sampling frames are usually unknown. Consequently, network sampling techniques have proven invaluable \citep[see][for a recent survey]{raifman2022respondent}. 

Respondent-driven sampling (RDS) is a network sampling algorithm based on  participant referral which is frequently employed for surveillance in public health research \citep{heckathorn1997respondent,wejnert2011respondent,heckathorn2017network}. 
RDS begins with an initial group (usually a convenience sample) of individuals who are given a limited number of coupons and asked to recruit other members of the population of interest by giving them one of the coupons directly or by providing their contact information to the research team. Recipients of these initial coupons redeem them with study researchers; they are then compensated, interviewed, and given new coupons to recruit additional subjects.  This process continues until
a sufficient sample is generated, the study budget or duration is reached, or 
some other stopping criterion is met.

The attributes of the coupons, e.g., their number, value, call to action, expiration date, etc., can 
play a critical role in shaping the evolution of the RDS process.  However, in standard RDS, 
researchers give each participant
an identical coupon allocation which is typically based on convention rather
than characteristics of the population under study 
\citep[][]{goel2010assessing}.  This is inefficient in that coupon return
rates, population coverage, and cost may be poor relative to what could be achieved with a tailored
coupon allocation.  Despite being widely recognized, 
attempts to address this issue have been scarce.  
\citet{lunagomez2018evaluating} proposed calibrating the
number of coupons by simulating RDS under a parametric model and an 
informative prior, but they do not consider adapting coupon allocations
to information accumulating during the study.  \citet{mcfall2021optimizing} ran a two-stage 
RDS study in which an initial RDS study was used to identify 
characteristics of participants likely to be successful recruiters, then 
in a subsequent RDS study, participants with these characteristics were
given extra coupons \citep[see][for an extensive simulation study of such
two-stage designs]{vanorsdale2023adaptive}.  Two-stage designs make use
of interim data but include only a single adaptation step and require 
independent samples collected across two studies, which is often 
impractical.

We use reinforcement learning \citep[RL,][]{sutton} to adapt coupon 
allocations over time within a single RDS study in such a way that 
some study objective is optimized.  Example study objectives 
include maximizing  
information about disease prevalence and reaching as many individuals 
in the target population as possible under budget and time
constraints.  Our approach 
uses a Markov branching process 
\citep[][]{sevast1974controlled,athreya2004branching} 
as a working model to guide coupon allocation during the study.  However, 
for post-study inference, we do not assume that this model is correct. 
Instead, we develop a novel projection confidence set that provides valid finite sample coverage for a large class of functionals of the generative model, even when the true model is not identifiable under RDS.  This is non-trivial
as we are combining two procedures which are notorious for their 
inferential challenges: (i)
RDS which is complicated by differential response probabilities, homophily, and other design effects 
\citep{goel2010assessing,gile20107,tomas2011effect,lu2012sensitivity,roch2018generalized,rohe2019critical}; and (ii) adaptive
RL-based experimentation for which standard bootstrap or normality based inference
procedures can fail to provide nominal coverage 
\citep[][]{deshpande2018accurate,zhang2020inference,bibaut2021post,zhan2023policy,bibaut2024demistifying}.  To address these complications, we develop valid post adaptive-experimentation inference procedures
for $M$-estimators constructed from Markov decision processes
\citep[MDPs,][]{putterman1994markov}, a result which is of independent interest. In simulation experiments, our proposed procedure, which we term RL-RDS, significantly improves efficiency relative
to static and two-stage RDS designs.  Furthermore, the projection confidence sets deliver nominal coverage without being excessively conservative.

The contributions of this work are summarized as follows:
(1) we develop the first principled framework for adaptive RDS using RL; (2) we show that a 
Markov branching process approximation to the RDS process is useful for 
guiding online adaptation; 
(3) we prove regret bounds for our RL algorithm under the Markov branching process model; and (4) we develop valid finite sample inference
methods for adaptive-RDS without requiring identifiability.

In Section~\ref{sec:setup}, we review RDS.  In Section~\ref{sec:rl}, 
we introduce our branching process approximation and RL-based adaptive
coupon selection.  In Section~\ref{sec:regretbounds}, we present 
regret bounds for the branching process approximation.  
In Section~\ref{sec:geninf}, we introduce our inference approach and prove that confidence regions constructed by this procedure achieve nominal coverage in finite samples and concentrate asymptotically.
Lastly, in Section \ref{sec:sim}, we present 
a suite of simulation experiments comparing RL-RDS, static, and two-stage designs.  

\section{Setup and Notation}
\label{sec:setup}
An RDS study recruits participants in epochs or waves.  The initial epoch, 
which comprises individuals $\mathcal{E}_0 \subseteq \mathbb{N}$, is typically collected as a convenience sample. Generally, we use $i$ to index interim study participants organized in epochs and $v$ to index interim study participants ordered by arrival time.
Each participant $i\in\mathcal{E}_0$ receives an allocation of coupons $\bA_i\in\mathcal{A}$ 
to distribute among
their social contacts. These coupons explain the study and encourage participation by offering an incentive.  The data collected in the zeroth 
epoch are thus 
$\mathcal{Z}_0 \triangleq \left\lbrace (T_i, \bX_i, Y_i, \bA_i, C_i)\,:\, i\in
\mathcal{E}_0 \right\rbrace$, where $T_i \in \mathbb{R}_+$ is the arrival 
time of participant $i$, $\bX_i\in\mathcal{X}\subseteq \mathbb{R}^p$ 
are their covariates,
$Y_i \in \mathcal{Y}\subseteq [0,1]$ is an outcome of interest, $\bA_i \in\mathcal{A}$ is their coupon allocation, and $C_i\in\mathbb{R}_+$ is the cost
of their recruitment.  The set of individuals in epoch $j \ge 1$ is  
denoted by 
$\mathcal{E}_j \subseteq
\mathbb{N}\setminus \left \{ \cup_{j'<j} \mathcal{E}_{j'} \right \}$ and their
associated data are $\mathcal{Z}_j \triangleq
\left\lbrace (R_i, T_i,\bX_i, Y_i, \bA_i, C_i)\,:\, i \in \mathcal{E}_j\right\rbrace$, where $R_i \in \mathcal{E}_{j-1}$ denotes the recruiter of individual 
$i \in \mathcal{E}_j$.  Individuals in each epoch are given coupons to recruit the next
epoch until available resources are depleted or another stopping criterion is met; for concreteness, we assume that the study terminates the
first time the total cost exceeds a fixed budget $D$.  
Figure (\ref{fig:GRGS}) illustrates the evolution of an RDS sample.

We note that the covariates $\bX_i$ and the outcome $Y_i$ may (though they
need not) be measured simultaneously.  Even if this is the case, $Y_i$ is
distinguished by its role in defining the adaptive RDS algorithm's objective.
As detailed below, we define an optimal coupon allocation strategy as one that
maximizes the cumulative sum of the outcome across the sample.  
For example, in an RDS
study of people who inject drugs, the outcome may be choosing to be tested for
HIV, and 
participant attributes might include 
demographic information,
PrEP use, history of STI testing, and attitudes and intentions related to risky
behaviors \citep[][]{risser2009relationship}. 
In an RDS study targeting colorectal cancer
screening among non-utilizers of a healthcare 
system, participant attributes might include 
demographic information, family medical
history, previous FIT or colonoscopy screening, 
and risk factors for colorectal cancer. 
The outcome might be a participant's 
screening intention.  Post-selection inference could focus on the
distribution of the covariates, the outcome, or both \citep[][]{cooks2022telehealth}.

Let $J$ denote
the total number of epochs. We use an overline to represent history so
that $\overline{\mathcal{E}}_J \triangleq \bigcup_{j\le J} \mathcal{E}_j$ 
are the complete sets of study participants, 
and $\overline{\mathcal{Z}}_J \triangleq \bigcup_{j\le J} \mathcal{Z}_j$ 
are the data that correspond to the members of $\overline{\mathcal{E}}_J$. 
For simplicity, we assume that subjects are processed 
sequentially; i.e., no two individuals arrive at exactly
the same time. Therefore, 
the data may be equivalently represented as $\mathcal{D}^{\noepochsamplesize} \triangleq \left\lbrace \left (R^v, T^v, \bX^v, Y^v, \bA^v, C^v \right )\right\rbrace_{v=1}^{\kappa}$, in which individuals are indexed by their arrival times.

Under adaptive RDS, when an individual arrives in the study, they are assigned
a coupon allocation based on accumulated information on all prior participants.  
Let $\bH^v \in\mathcal{H}^v$ denote the information available to researchers
at the time the $v^{th}$ subject
is given their coupon allocation $\bA^v \in \mathcal{A}$. 
For each $\bH^v=\bh^v$, let $\psi^v(\bh^v) \subseteq \mathcal{A}$ 
denote the set of allowable coupon combinations given $\bh^v$ 
(e.g., this set may
be restricted to ensure budget constraints are not violated). 
A deterministic allocation strategy $\pmb{\pi} = (\pi^1, \pi^2,\ldots) \in \Pi$ is
a sequence of functions such that $\pi^v:\mathcal{H}^v \rightarrow 
\mathcal{A}$ and $\pi^v(\bh^v) \in \psi^v(\bh^v)$ for all $\bh^v$. The set $\Pi$
can be restricted to
exclude allocation strategies 
that are inherently unfair or harmful or to improve the tractability of finding the optimal allocation strategy.  
We
define an optimal allocation strategy, $\pmb{\pi}^{\mathrm{opt}} \in \Pi$, as maximizing 
the expected value of the cumulative outcome across the RDS sample.

We formalize the optimal allocation strategy
within the potential outcomes framework \citep[][]{rubin,splawa1990application}.
For each $v\geq 2$, let $\bH^{v*}(\overline{\ba}^{v-1})$ denote the
potential history under the sequence of coupon allocations
$\overline{\ba}^{v-1} \triangleq (\ba^1, \ba^2,\ldots, \ba^{v-1})$. For
any deterministic allocation strategy $\pmb{\pi} \in \Pi$, the potential history
at time $v\ge 2$ is 
\begin{equation*}
\bH^{v*}(\pmb{\pi}) \triangleq \sum_{\overline{\ba}^{v-1}}
\bH^{v*}(\overline{\ba}^{v-1})\prod_{k=1}^{v-1}\mathbb{I}\left[
\pi^k\left\lbrace 
\bH^{k*}(\overline{\ba}^{k-1})
\right\rbrace = \ba^{k}
\right],
\end{equation*}
where we have defined $\bH^{1*}(\ba^0) \equiv \bH^1$ and 
$\mathbb{I}(u)$ as an indicator of the event $u$.  
Similarly, for each individual $v$,
let $Y^{v*}(\pmb{\pi})$ be the potential outcome
and
$C^{v*}(\pmb{\pi})$ the potential cost
under $\pmb{\pi}$.  
The
potential number of participants under allocation
strategy $\pmb{\pi}$ is thus 
\begin{equation*}
n^*(\pmb{\pi}) \triangleq \inf
\left\lbrace
g\,:\, \sum_{v=1}^{g}
  C^{v*}(\pmb{\pi})
 > D
\right\rbrace.
\end{equation*}
We assume that there exists $Q \in \mathbb{N}$ such that
$\sup_{\bpi \in \Pi} n^*(\bpi) \le Q$ almost surely;
in application, this 
incurs no loss in generality.  
Let $\Delta^{v*}(\pmb{\pi})$ be an indicator that
the budget has not been exceeded when individual $v$ enters the study under
$\pmb{\pi}$.
For each $n \in \mathbb{N}$,
define the history-value function of $\bpi$ at $\bh^n$ as 
\begin{equation*}
V^n(\bh^n, \pmb{\pi}) \triangleq
 \mathbb{E}\left\lbrace
\sum_{v = n}^Q\Delta^{v*}(\bpi)Y^{v*}(\bpi)
\big|\bH^n=\bh^n
\right\rbrace,
\end{equation*}
where this expectation is over the sampling process.
An optimal
allocation strategy, $\pmb{\pi}^{\mathrm{opt}}$, satisfies
$V^n(\bh^n, \pmb{\pi}^{\mathrm{opt}}) \ge V^n(\bh^n, \pmb{\pi})$ for each $n$, all feasible
$\pmb{\pi}$, and $\bh^n \in \mathcal{H}^n$. 

We identify
$\pmb{\pi}^{\mathrm{opt}}$ in terms of
the data-generating model by making a series of assumptions that are standard for sequential decision problems \citep[][]{tsiatis2019dynamic}.
Define the collection of all potential outcomes as 
\begin{equation*}
\mathcal{W} \triangleq \left\lbrace
\bH^{v*}(\overline{\ba}^{v-1}), 
Y^{v*}(\overline{\ba}^v), 
C^{v*}(\overline{\ba}^v)\,:\,
\ba^v \in \psi^v\left\lbrace
  \bH^{v*}(\overline{\ba}^{v-1})
\right\rbrace
\right\rbrace_{v\ge 1}.
\end{equation*}
  
\begin{assumption}[Strong ignorability]\label{as:1}
For all $v \in \mathbb{N}$, $\mathcal{W}\perp \bA^v | \bH^v$.
\end{assumption}
\begin{assumption}[Consistency]\label{as:2}
For all $v \in \mathbb{N}$, $\bH^v = \bH^{v*}(\overline{\bA}^{v-1})$, $Y^v  = Y^{v*}(\overline{\bA}^v)$,  
and $C^v = C^{*v}(\overline{\bA}^v)$; 
i.e., the observed histories, outcomes, and costs are equal to their 
counterfactual counterparts under the coupon allocations actually given.
\end{assumption}
\begin{assumption}[Positivity]\label{as:3}
For all $v \in \mathbb{N}$, $\bh^v \in \mathcal{H}^v$ and
$\ba \in \psi^v(\bh^v)$, there exists $\epsilon > 0$ such that
$\bbP(\bA^v = \ba|\bH^v = \bh^v) \ge \epsilon$.
\end{assumption}
\noindent

Under these assumptions, one can express the likelihood for the
counterfactual cumulative outcome under any coupon allocation
strategy in terms of the data-generating model 
\citep[][]{robins2004optimal}. However, the curvature of this likelihood may be 
zero or near zero in some regions of the parameter space even under
 simple network models  
\citep[][]{crawford2018hidden,weltzhidden}.  Consequently, estimators 
of parameters indexing the RDS process can be extremely 
volatile and online learning based on such
estimators is similarly volatile, especially in small samples.  To make online learning tractable, one 
must impose additional
structure on the model. This could be done through an informative prior,
though such priors are difficult to posit and have been shown to
exert unacceptably large influence on the operating characteristics of 
resulting estimators \citep[][]{weltzhidden}.  Instead, we posit a simple branching process 
working model that is 
parsimonious and stable when fit to the RDS data as it accumulates, yet
flexible enough to capture salient features of the RDS process for online
learning.   
As noted previously, we do not require this model be correctly specified.

\begin{figure}
    \centering
    \includegraphics[width=\linewidth]{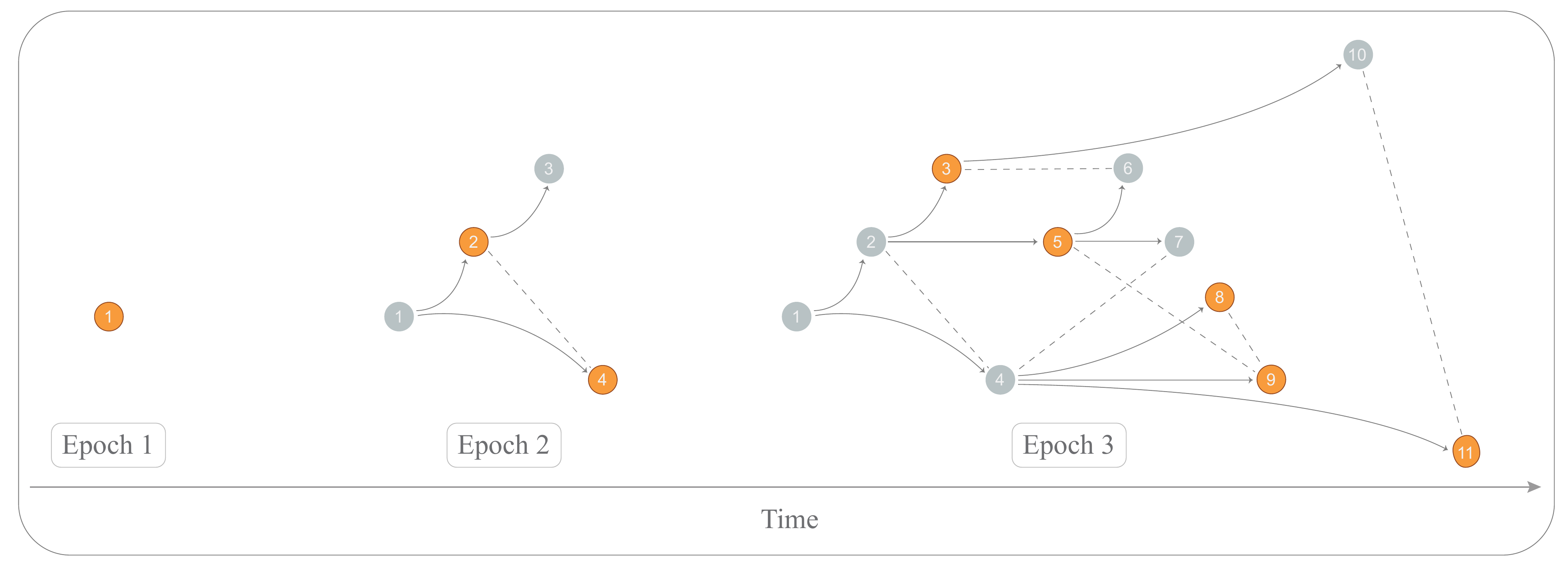}
    \caption{RDS is a complex stochastic process that samples without replacement over a social network. The observed RDS sample is composed of coupon exchanges, illustrated by arrows ($\rightarrow$). The unobserved connections between sample participants are represented as dashed lines (\protect\tikz[baseline]{\protect\draw[line width=0.5mm,densely dashed] (0,.8ex)--++(0.45,0);}). The observed data resembles a branching process.}
    \label{fig:GRGS}
\end{figure}

\section{Reinforcement Learning for RDS}
\label{sec:rl}
Most 
inferential techniques for RDS are based on simplifications 
of the RDS process \citep{heckathorn1997respondent, 
rohe2015network,crawford2018hidden} or graph model 
\citep{gile2015network, gile2011improved} because the network structure is often only weakly 
identifiable from the sample \citep{weltzhidden}.
Consequently, these estimators can be biased and 
unstable. Instead of estimating the graph dynamics directly, 
we posit a partially controllable branching process
as a working model for RDS. 
This model captures relationships between 
recruiters and recruits, and its tractable likelihood 
facilitates asymptotic regret guarantees and robust 
inference techniques.  However, we emphasize that 
our inferential approach in Section~\ref{sec:geninf} does not assume the branching
process model is correct.

Let $M^v$ be the number of recruits of participant $v$,  
$\underline{\bX}^v = (\bX_1^v,\ldots, \bX_{M^v}^v)$ the collection 
of their covariates, $\underline{\bY}^v = (Y_1^v,\ldots, Y_{M^v}^v)$ 
their response statuses, and $\underline{\bT}^v = (T_1^v,\ldots, T_{M^v}^v)$
their response times. We write $\underline{\bX}^0, \underline{\bY}^0, \bA^0$, and $\underline{\bT}^0$ to
denote the information collected in the initial sample,
treating the first epoch as recruits of a fictional recruiter
``zero." 
We posit a branching process working model that factors
as
\begin{align}   
\label{eq:naiveFactoredLH}
\begin{split}
\prod_{j=1}^{J}\prod_{v\in\mathcal{E}_j}
&f(\underline{\bY}^v|\underline{\bX}^{v},
\underline{\bT}^v, M^v, \bA^v, \bH^{v})
f(\underline{\bX}^v| \underline{\bT}^v, M^v, \bA^v, \bH^{v})  \times \\
&f(\underline{\bT}^v| M^v, \bA^v, \bH^{v})
f(M^v|\bA^v, \bH^{v})f(\bA^v|\bH^{v}),
\end{split}
\end{align} 
where we define $f(\bA^0|\bH^0) = f(\bA^0)$ and
assume that $C^v$ is a known function of 
$(\bX^v$, $\bA^v,Y^{v} , \underline{\bT}^v, \underline{\bX}^v, \underline{\bY}^v)$. 
According to this model, the recruits associated with recruiter
$v$ are not affected by the coupon allocation to individual $v'$ (or recruits of $v'$) if $v' > v$.

During the RDS process, our goal is to guide incentive allocation to maximize the expected
cumulative outcome. This outcome is chosen to characterize study effectiveness, e.g., 
an indicator that a participant consents to disease screening. 
Online sequential decision making to optimize a cumulative objective
fits naturally within the framework of reinforcement learning
\citep[RL;][]{sutton}.  
Thompson sampling is an RL algorithm that has been shown to possess
favorable theoretical properties and strong empirical performance on 
a wide range of problems \citep[][]{chapelle2011empirical,agrawal2013further,gopalan2014thompson,gopalan2015thompson,russo2016information,russo2018tutorial,laber2018optimal}.  
However, it has never
been studied as an approach to optimal allocation 
in RDS. This context
is especially challenging due to budget constraints and
complex dependence between study participants.

We consider a class of parametric models for the components of the branching process, 
specified in Equation (\ref{eq:naiveFactoredLH}), 
which we index by $\pmb{\beta} = \left(\pmb{\beta}_y, 
\pmb{\beta}_t,
\pmb{\beta}_{\bx}, \pmb{\beta}_{m}\right) \in \mathcal{B}$, where $\mathcal{B}$ is a bounded, open subset of $\mathbb{R}^k$. The contribution to the joint density from the 
$v$th participant is composed of the following components 
\begin{equation}\label{eq:param}
\begin{array}{ll}
f\left(\underline{\bY}^v|\underline{\bX}^{v}, 
\underline{\bT}^v, M^v, \bA^v, \bH^{v}\right) 
&= f\left ( \underline{\bY}^v|\underline{\bX}^{v}, 
\underline{\bT}^v, M^v, \bA^v, \bH^v; 
\pmb{\beta}_y  \right ),
\\ 
f\left(\underline{\bX}^v| \underline{\bT}^v, M^v, \bA^v, \bH^{v}\right)
&=  
f\left (
\underline{\bX}^v | \underline{\bT}^v, M^v, \bA^v, \bH^v;
\pmb{\beta}_{\bx}\right ),
\\ 
f\left(\underline{\bT}^v| M^v, \bA^v, \bH^{v}\right) 
&= 
f\left (
\underline{\bT}^v| M^v, \bA^v, \bH^v;
\pmb{\beta}_{t}
\right ), \\
f\left(M^v|\bA^v, \bH^{v}\right) 
&= 
f\left ( M^v|\vA^v, \bH^v;
\pmb{\beta}_m\right ).
\end{array}
\end{equation} 

To approximate an optimal strategy after each participant's arrival, we use the branching process model to simulate potential futures starting from the current RDS sample.
Recall that $Q \in \mathbb{N}$ is defined such that $\sup_{\bpi} n^*(\bpi) \le Q$ 
almost surely. 
For any $n \in \mathbb{N}$, $\bh^n \in \mathcal{H}^n$, $\pmb{\pi} \in \bPi$, $\bbeta \in \mathcal{B}$,
and $B \in \mathbb{N}$, we denote by  
$\mathcal{K}^B(\bh^n,\pmb{\pi}; \bbeta)$ a set of $B$ trajectories of length $Q-n$ 
simulated under $\bbeta$.
For $b \in B$, the trajectory begins with the RDS information collected so far, $\bH_{b, \bpi, \bbeta}^n = \bh^n$, and the coupons,
$\bA_{b, \bpi, \bbeta}^{\ell}$, are assigned according to $\pi^\ell(\bH^\ell_{b, \bpi, \bbeta})$
for $\ell = n, n+1,
\ldots, Q$, i.e., 
\begin{equation*}
\mathcal{K}^B(\bh^n,\pmb{\pi}; \bbeta) \triangleq 
\left\lbrace
\left(Y_{b, \bpi, \bbeta}^n, \bH_{b, \bpi, \bbeta}^n, \bA_{b, \bpi, \bbeta}^n, Y_{b, \bpi, \bbeta}^{n+1}, \bH_{b, \bpi, \bbeta}^{n+1}, \bA_{b, \bpi, \bbeta}^{n+1}, \ldots,
\bY_{b, \bpi, \bbeta}^{Q}, \bH_{b, \bpi, \bbeta}^{Q}
\right)
\right\rbrace_{b=1}^B.
\end{equation*}
 These simulated 
trajectories will be
used to approximate an optimal strategy,
$\pmb{\pi}^{\mathrm{opt}}$.
Under Assumptions~\ref{as:1}-\ref{as:2}, the history-value function reduces to
\begin{equation*}
V^n(\bh^n, \pmb{\pi}; \bbeta)  = \mathbb{E}_\bbeta\left\lbrace
\sum^{Q}_{v = n}\Delta^{v}(\bpi)Y^v (\bpi)
\big|\bH^n=\bh^n
\right\rbrace
\end{equation*}
for any $\bbeta \in \mathcal{B}$, $n \in \mathbb{N}$,
history $\bH^n=\bh^n$, and $\bpi \in \bPi$.
Under Assumption~\ref{as:3}, this expectation can be estimated by averaging cumulative reward over the simulated trajectories,
\begin{equation*}
\widehat{V}^n_{B}(\bh^{n}, \bpi; \bbeta) =
\frac{\sum_{b = 1}^B \left (
\sum_{v = n}^{Q} \Delta^v_{b, \bpi, \bbeta} Y^v_{b, \bpi, \bbeta}
\right )}{B}.
\end{equation*}
The estimated optimal policy under $\bbeta$ upon observing
$\bH^n=\bh^n$ is thus
$\widehat{\bpi}^{n}_{B}(\bh^n ; \bbeta) \in \arg \max_{\bpi \in \Pi} \widehat{V}^n_{B}(\bh^n, \bpi; \bbeta)$.

Deciding between allocations that appear to be optimal given 
current estimated parameter values and those that might 
improve parameter estimates and thereby  
lead to better decisions
in the future is a fundamental problem in reinforcement learning. This balance between information gain and optimization
is commonly known as the exploration-exploitation trade-off 
in statistics and computer science \citep[][]{berry1985bandit,sutton,csaba2010,slivkins2019introduction,lattimore2020bandit}.  
We consider a variant of Thompson sampling with clipping \citep{zhang2020inference} to ensure sufficient exploration when determining coupon allocations.

At each decision point, our method samples a coupon allocation for the newest study participant approximately proportional to the probability that it is optimal.
Let $\widehat{\bbeta}^{n}$ be an estimator 
of the parameters indexing the banching process model  based on a sample of size $n \in \mathbb{N}$ \footnote{
In application, a natural choice for this statistic 
(and the one we use for the simulations in Section~\ref{sec:sim}) is the 
maximum likelihood estimator for the parameters indexing the
branching process model; however the theory holds more generally.}, and
$\widehat{P}^{n}$ denote an estimator of the sampling distribution of  $\widehat{\bbeta}^{n}$.
Given state $\bH^n=\bh^n$ and $\bpi \in \bPi$, let
\begin{equation*}
\hat{\rho}^{n}_{B}(\bh^n, \bpi) = \int 
\mathbb{I}\left\lbrace
\bpi \in \arg\max_{\bpi \in \bPi}\widehat{V}^n_{B}(\bh^n, \bpi; \bbeta)
\right\rbrace d\widehat{P}^{n}(\bbeta)
\end{equation*}
be the estimated confidence that
$\bpi$ is the optimal policy. 
If $\hat{\rho}^{n}_{B}(\bh^n, \bpi)$ is high, it means that $\bpi$ maximizes the cumulative reward under the points in $\widehat{P}^{n}(\bbeta)$ that are most likely.
Additionally, for $ \ba^n \in \psi^n(\bh^n)$, let
\begin{equation*}
\widehat{\xi}^{n}_B(\bh^{n}, \ba^{n}) = \int_{\bPi}
\mathbb{I}\left\lbrace
 \pi^n(\bh^n) = \ba^n
\right\rbrace d\hat{\rho}^{n}_{B}(\bh^n, \bpi)
\end{equation*}
be the estimated probability that $\ba^n$ is the optimal action.
Thompson sampling with clipping at level $\epsilon \in (0,1)$
selects action $\bA^n = \ba^n$
with probability proportional to $\min  \left [1-\epsilon, \max \left \{\epsilon,\widehat{\xi}^{n}_{B} (\bh^n, \ba^n) \right \} \right ]$. 
 Lemma~\ref{lem:clip_weights} in 
 the Supplemental Materials
 shows that this clipping constraint satisfies the following assumption, which is needed for the theoretical developments in the next section.
\begin{assumption}[Bounded action selection probabilities]
    \label{as:clip}
   For all $v \in \mathbb{N}$, assume there exists $\rho_{\min}, \rho_{\max} \in \mathbb{R}^+$ such that $\rho_{\min} \leq \sqrt{1/\bbP \left (\bA^v | \bH^v \right)} \leq \rho_{\max}$ with
   probability one. 
\end{assumption}

\noindent We provide justification for the proposed variant of Thompson Sampling 
with the branching process approximation in the following section.

\subsection{Asymptotic Regret Bounds}
\label{sec:regretbounds}
In this section, we show that if the branching model is correctly specified, then 
Thompson Sampling with clipping achieves optimal regret.  This result is both novel
in its own right and justifies the use of Thompson Sampling for coupon 
allocation with our branching
process approximation.

Recall that $\mathcal{B}$ is a bounded, open subset of $\mathbb{R}^k$ for $k \in \mathbb{N}$.
We assume a parametric model indexed by fixed but unknown parameter
$\bbeta^* \in \mathcal{B}$.
We construct a weighted maximum likelihood estimator of $\bbeta^*$  
by extending the  
M-estimation approach described in \citet{zhang2021statistical} to general MDPs. 
\citeauthor{zhang2021statistical} show that the maximum likelihood estimator (MLE) constructed from adaptively sampled data in a linear contextual bandit can converge to a non-normal limit when two arms have the 
same mean reward.   A normal asymptotic limit 
is obtained by re-weighting the likelihood by a function of the propensity score 
in such a way that the asymptotic variance
is stabilized \citep[see also][]{deshpande2018accurate,hadad2021confidence, zhan2023policy,bibaut2024demistifying,zhan2024policy}.  In RL for RDS 
this re-weighting will depend on the coupon assignment distribution.  

We consider an asymptotic regime based on complete generations (epochs) of 
study participants. For $j \in \mathbb{N}$ and $i \in \epochgroup$,
we restrict the information available before the coupon allocation to participant $i$ to $\overline{\mathcal{Z}}_{j-1}$.
Additionally, we do not use coupon allocation information associated with individuals in the same epoch as the current study participant.
Define $r(i)$ to be the rank of participant $i$'s arrival
time; (i.e., if participant $i$ were the third
participant in the study, then $r(i)=3$). 
For $j \in \mathbb{N}$ and $i \in \epochgroup$, 
we assume that the coupon allocation for 
participant $i$ is based on historical information, $\bH_i$, that is an element of sigma field
\begin{equation}
    \label{eq:generationhistory}
    \sigma \left [ \left\lbrace (R^v, T^v, \bX^v, Y^v , \bA^v, C^v)\right\rbrace_{v=1}^{r(i)} \cap \left ( \overline{\mathcal{Z}}_{j-1} \setminus  \{\bA_i\}_{i \in \epochgroup} \right ) \right ].
\end{equation}
Define the complete information associated with the recruits of individual $i$ as \\ $\bD_i \triangleq \left  \{  \underline{\bY}_i, \underline{\bX}_i, \underline{\bT}_i, M_i, \bA_i \right \} \in \mathscr{D}$,
and 
define the field associated with the first $j$ generations as 
$\mathcal{F}_j \triangleq \sigma( \overline{\mathcal{Z}}_j \setminus  \{\bA_i\}_{i \in \mathcal{E}_j} )$.
In addition, under the branching process model, we assume that the individuals in  
generation $\mathcal{Z}_j$ are conditionally independent given $\epochfield$ and $\{\bA_i\}_{i \in \epochgroup}$, 
so that 
\begin{equation*}
    \bbP\left ( \left \{\bD_i \right \}_{i \in \epochgroup } \mid \{\bA_i\}_{i \in \epochgroup}, \epochfield \right ) = \prod_{i \in \epochgroup} \bbP\left ( \bD_i  \mid \bA_i, \epochfield \right ).
\end{equation*}
Define the ``complete generation" likelihood as
\begin{equation}
    \label{eq:compgenerationlike}
   \prod_{j=1}^J \bbP \left ( \mathcal{Z}_j |  \epochfield \right ) = \prod_{j=1}^J \prod_{i \in \epochgroup } \bbP \left ( \left  \{  \underline{\bY}_i, \underline{\bX}_i, \underline{\bT}_i, M_i \right \}  | \bA_i,  \epochfield \right) \bbP \left ( \bA_i |  \bH_i    \right ),
\end{equation}
where the equality follows from Equation (\ref{eq:generationhistory}) under which 
$\bbP \left ( \bA_i | \epochfield, \{\bD_k \}_{k \in \epochgroup, k \neq i} \right ) =   \bbP \left ( \bA_i | \epochfield \right ) = \bbP \left ( \bA_i | \bH_i \right )$.

In our context, the estimated optimal policy need not converge to a fixed strategy because of the 
non-stationarity of the branching process we use to model RDS. Consequently, as anticipated by 
the literature on inference after adaptive sampling, the asymptotic behavior of the 
MLE based on (\ref{eq:compgenerationlike}) is difficult to characterize 
\citep[][]{bibaut2024demistifying}.  However, we are able to obtain a parametric
rate of convergence after a suitable re-weighting of the likelihood.  
Let $\widehat{\bpi}$ be the policy followed by the reinforcement algorithm, and $\widetilde{\bpi}$ be 
the policy that samples from among the available coupon allocations with equal probability; 
i.e., $\bbP_{\widetilde{\pi}}(\pmb{a}|\bH_i=\bh_i) \triangleq 1/\left |\psi_i(\bh_i) \right |$ if
$\pmb{a} \in \psi_i(\bh_i)$ and zero otherwise, 
and 
$\bbP_{\widehat{\pi}}(\pmb{a}|\bH_i=\bh_i) \triangleq \bbP(\bA_i=\pmb{a} | \bH_i=\bh_i)$.
We use 
\begin{equation*}
    W_i \triangleq \sqrt{\frac{\bbP_{\widetilde{\pi}}(\bA_i|\bH_i)}
    {
    \bbP_{\widehat{\pi}}(\bA_i|\bH_i) }
    },
\end{equation*}
as stabilizing weights 
in the log-likelihood.  
Our complete generation M-estimator maximizes the weighted log-likelihood
\begin{align*}
    &\widehat{\bbeta}_{J} \triangleq \arg \max_{\bbeta \in \mathcal{B}} \sum_{j=1}^{J} \sum_{i \in \epochgroup }  W_i l_i(\bbeta), \\
    &\mathrm{where} \ \mathrm{for} \ j \in \mathbb{N} \ \mathrm{and} \ i \in \epochgroup, \   l_i(\bbeta) \triangleq l(\bbeta, \bD_i)  \triangleq \log \left \{ \bbP(\bD_i \mid \mathcal{F}_{j-1}; \bbeta) \right \}.
\end{align*}

Below we state the technical assumptions under which we derive regret bounds 
for the estimated optimal policy.  
We verify that these assumptions hold in the 
working model described by Equation (\ref{eq:rdsmod}) in Section~\ref{sec:examp}.
Define the number of sample participants in generation $j$ as $\kappa_j \triangleq |\mathcal{E}_j|$ and the total number of individuals recruited before generation $j$ as $\overline{\kappa}_j \triangleq |\overline{\mathcal{E}}_{j-1}|$. Additionally, note that the complete data-generating distribution is specified by the branching process parameter, $\bbeta$, and the reinforcement learning policy, $\bpi$. 
The expectation taken with respect to this distribution is denoted $\bbE_{\bbeta, \bpi}$.
\begin{assumption}[Branching asymptotics]\label{as:generation_asymptotics}
For all recruiters $v \in \mathbb{N}$, and their recruits $j \in \{1,2,\ldots, M^v\}$, there exists $\alpha > 0$ such that $T^v_j - T^v \geq \alpha$ with
probability one. In addition, the number of coupons in an allocation is bounded
above by $L_*$.
Lastly, the branching process is super-critical; i.e., there exists a random variable $\branchasymp$ such that for any $\epsilon > 0$ there exists $\delta > 0$ such that
\begin{equation*}
  \lim_{j \to \infty} \kappa_j/m^j \to \branchasymp \ \ \mathrm{a.s.}, \ \mathrm{where} \  \bbP \left ( \branchasymp \geq \delta \right ) \geq 1-\epsilon.
\end{equation*}
\end{assumption}
\begin{assumption}[Growing budget asymptotics]
    \label{as:budg}
   The budget grows over time as follows. Let $S^n$ denote the budget when individual
   $n$ is recruited.  For all $n \in \mathbb{N}$, 
   $0 < S^n - \sum_{v=1}^n C^v < C^*$ with probability one for some
   fixed $C^* \in \mathbb{R}^+$.
\end{assumption}

\begin{assumption}[Identifiability and differentiability]
    \label{as:differentiable}
      The parameter indexing the branching process is identifiable. Additionally, for all $j \in \mathbb{N}$ and $i \in \epochgroup$, the first two derivatives of $l_i(\bbeta)$ with respect to any $\bbeta \in \mathcal{B}$ exist. 
\end{assumption}
\begin{assumption}[Moment conditions]
    \label{as:moments}
    The parameter that indexes the true generative process, $\bbeta^* \in \mathcal{B}$, is in the interior of $\mathcal{B}$. For all $j \in \mathbb{N}$ and $i \in \epochgroup$, the first two moments of $l_i(\bbeta^*)$, $\dot{l}_i(\bbeta^*)$, and $\ddot{l}_i(\bbeta^*)$ conditional on $\epochfield$ are bounded almost surely.
\end{assumption}
\begin{assumption}[Lipschitz]
\label{as:lipchitz}
    There exists a real-valued function $g: \mathscr{D} \to \mathbb{R}$ such that for all $j \in \mathbb{N}$, $i \in \epochgroup$, and $\bbeta, \bbeta' \in \mathcal{B}$,
    \begin{equation*}
        |l_i(\bbeta) - l_i(\bbeta')| \leq g(\bD_i) \|\bbeta - \bbeta'\|_2,
    \end{equation*}
    where $\bbE_{\bbeta^*, \widetilde{\bpi}}\left \{ g(\bD_i)^2 | \epochfield \right \}$ is bounded almost surely.
\end{assumption}
\begin{assumption}[Well-separated maximizer]
    \label{as:wellseperated}
    For any $\epsilon > 0$, there exists $J_0 \in \mathbb{N}$ and $\delta > 0$ such that for all $J \geq J_0$,
    \begin{equation*}
        \inf_{\bbeta \in \mathcal{B}: \|\bbeta - \bbeta^* \|_2 > \epsilon} \left [ \frac{1}{\overline{\kappa}_J} \sum_{j=1}^{J} \sum_{i \in \epochgroup}  \bbE_{\widetilde{\bpi}, \bbeta^*} \left \{ l_i(\bbeta^*) - l_i(\bbeta) | \epochfield \right \} \right ] \geq \delta \ \ \mathrm{a.s.}
    \end{equation*}
\end{assumption}
Assumptions~\ref{as:generation_asymptotics} and \ref{as:budg} specify the asymptotic regime for the branching process.
Assumption~\ref{as:generation_asymptotics} implies that $J \to \infty$ and that the generation sizes are consistent with Galton-Watson processes \citep{athreya2004branching}. 
Assumption~\ref{as:budg} states that the budget grows in such a way that the
remaining budget is always bounded; this avoids trivial solutions in which maximal
resources are allocated at each time point.  
Assumptions \ref{as:differentiable}-\ref{as:wellseperated} ensure that the log-likelihood is well-behaved.
In Assumptions~\ref{as:differentiable} and \ref{as:moments}, we assume that the log-likelihood is identifiable, two times differentiable, and its components (and their derivatives) have finite second moments.. Assumption~\ref{as:lipchitz} limits the complexity of the log-likelihood function so that
the weighted log-likelihood converges uniformly. Assumption~\ref{as:wellseperated} requires that $\bbeta^*$ be a ``well-seperated" point of maximum and is a standard assumption for consistency; e.g., Theorem~5.7 of \cite{van2000asymptotic}. Note that Assumptions~\ref{as:lipchitz} and \ref{as:wellseperated} 
are unnecessary for consistency if the log-likelihood is concave.

Under Assumption~\ref{as:differentiable}, for any policy $\bpi \in \Pi$, $\bbeta\in\mathcal{B}$, $j \in \mathbb{N}$, and $i \in \epochgroup$, it follows that 
$\bbE_{\bbeta, \bpi}  \left \{ \dot{l}_i(\bbeta) \mid \bA_i, \bH_i \right \} = 0$.
Consequently, $\sum_{i \in \epochgroup} W_i\dot{l}_i(\bbeta) $ is a martingale difference with respect to the filtration $\{\mathcal{F}_j\}_{j\geq 1}$,
 \begin{align*}
    \bbE_{\bbeta, \bpi} \left \{\sum_{i \in \epochgroup} W_i\dot{l}_i(\bbeta) \mid \epochfield \right \}
    &= \bbE_{\bbeta, \bpi}  \left [ \sum_{i \in \epochgroup} W_i \bbE_{\bbeta, \bpi} \left \{\dot{l}_i(\bbeta)  \mid \bA_i, \epochfield \right \} \mid \epochfield \right ] = 0.
 \end{align*}
 We define the variance of the weighted score function conditioned on this filtration, 
 \begin{equation*}
     \eta_{J} \triangleq \sum_{j = 1}^J \bbE_{ \bbeta^*,\widehat{\bpi}} \left  \lbrace \sum_{i \in \epochgroup} W_i^2\dot{l}_i(\bbeta^*)\dot{l}_i(\bbeta^*)^\top \mid \epochfield \right \rbrace =  
     \sum_{j = 1}^J
     \bbE_{\bbeta^*, \widetilde{\bpi}} \left  \lbrace \sum_{i \in \epochgroup}  \dot{l}_i(\bbeta^*) \dot{l}_i(\bbeta^*)^\top | \epochfield \right \rbrace,
 \end{equation*} 
 where the equality follows from 
 \begin{align*}
     \bbE_{\bbeta^*, \widehat{\bpi}} \left  \lbrace W_i^2 \dot{l}_i(\bbeta^*)\dot{l}_i(\bbeta^*)^\top| \epochfield \right \rbrace &= \int W_i^2 \dot{l}_i(\bbeta^*) \dot{l}_i(\bbeta^*)^\top  f\left ( \bD_i|\bA_i, \epochfield;
\bbeta^*\right ) \bbP_{\widehat{\pi}}(\bA_i | \epochfield) d\nu\\
     &=\int \frac{\bbP_{\widetilde{\pi}}(\bA_i | \bH_i)}{\bbP_{\widehat{\pi}}(\bA_i | \bH_i) } \dot{l}_i(\bbeta^*) \dot{l}_i(\bbeta^*)^\top  f\left ( \bD_i|\bA_i, \epochfield;
\bbeta^*\right ) \bbP_{\widehat{\pi}}(\bA_i | \bH_i) d\nu \\
     &= \bbE_{\bbeta^*, \widetilde{\bpi}} \left  \lbrace \dot{l}_i(\bbeta^*) \dot{l}_i(\bbeta^*)^\top | \epochfield \right \rbrace.
 \end{align*}
 We see that re-weighting the outer product
 of the score function 
 removes  
 the dependence of the estimating equation's conditional variance on the RL algorithm 
 $\widehat{\bpi}$.
 
While $\eta_J$ is a constant in the setting considered by \cite{zhang2021statistical}, 
here it is a random variable. 
Consequently, characterizing the asymptotic behavior of $\widehat{\bbeta}_J$ requires
new theory for martingale estimating functions constructed from controllable branching
processes.  
For matrix $A$, let $\sigma_{\min}(A)$ be its minimum eigenvalue. 
In addition, let $I_k$ be the $k$ dimensional identity matrix. 
To extend asymptotic theory for $M$-estimators collected under a contextual bandit 
to the more general setting of an MDP, we make use of the following standard assumptions.  
\begin{assumption}[Martingale stabilizing variance]
    \label{as:stabalizedvariance}
    As $J \to \infty$, there exists a sequence of constant (i.e., not random) 
    positive definite matrices 
    $\{\Sigma_{J}\}_{J \geq 1}$ such that
    \begin{equation*}
        \Sigma_{J}^{-1/2} {\eta_{J}} \Sigma_{J}^{-1/2}  \overset{p}{\to} U,
    \end{equation*}
    where $U$ is a (random) positive
    definite matrix and $U \succ 0$ with probability one.
\end{assumption}
\begin{assumption}[Information accumulation]
\label{as:IA}
For some $\delta > 0$, $\lim \inf_{J \rightarrow \infty} \sigma_{\min}\left \{ \eta_{J}/\overline{\kappa}_J \right \} \geq \delta$ 
with probability one.
\end{assumption}
\begin{assumption}[Equicontinuity]
    \label{as:equicontinuity}
     There exists an $\epsilon_{\ddot{l}} > 0$ and a function $f:\mathscr{D} \to \mathbb{R}$ such that for all $j \in \mathbb{N}$, $i \in \epochgroup$ and $0 < \epsilon \leq \epsilon_{\ddot{l}}$, there exists $\delta_{\epsilon}$ such that
    \begin{equation*}
        \sup_{\bbeta\in \mathcal{B}: \|\bbeta - \bbeta^*\|_2 \leq \delta_{\epsilon}} \| \ddot{l}_i(\bbeta) - \ddot{l}_i(\bbeta^*) \|_2 \leq \epsilon  f(\bD_i) \ \ \mathrm{a.s.},
    \end{equation*}
    and  $\bbE_{\bbeta^*, \widetilde{\bpi}} \left \{ f(\bD_i) | \epochfield\right \}$ is bounded almost surely.
\end{assumption}
\noindent 
Assumption~\ref{as:stabalizedvariance} requires that the conditional variance of the log-likelihood components converges. This assumption is the crux of our extention from the contextual bandit framework of \cite{zhang2021statistical} to more general MDPs. It is also standard in asymptotic theory
for martingales, e.g., Theorem~3.2 in \cite{hall2014martingale}.
Assumption~\ref{as:IA} requires that information accumulates over time.
Assumption~\ref{as:equicontinuity} ensures the equicontinuity of the empirical Fisher information.

\begin{theorem}
\label{thm:asympNorm}
For $\delta > 0$, define the event $E_\branchasymp= \left \{ \branchasymp> \delta \right \}$, where $\branchasymp$ is defined in Assumption~\ref{as:generation_asymptotics}. 
Under Assumptions \ref{as:1}-\ref{as:equicontinuity}, 
as $J \to \infty$,
\begin{align}
    \label{eq:consistency}
    & \widehat{\bbeta}_{J} - \bbeta^* = O_p( 1/\sqrt{\overline{\kappa}_J})
\end{align}
on event $E_\branchasymp$.
\end{theorem}

We consider the cumulative reward of policies implemented over study participants as they arrive at the study, represented by data $ \mathcal{D}^{\kappa} \triangleq \left\lbrace (R^v, T^v, \bX^v, Y^v , \bA^v, C^v)\right\rbrace_{v=1}^{\kappa}$ for $\kappa \in \mathbb{N}$. Consequently, we require a Lipchitz condition for the log-likelihood of this arrival process to establish asymptotic regret bounds. For $v \in \mathbb{N}$, define the log-likelihood components of the arrival process as
\begin{align*}
    &q^v(\bbeta) \triangleq q\left \{ \bbeta,  (R^v, T^v, \bX^v, Y^v , \bA^v) \right \} \triangleq \log \left [ \bbP \left \{ (R^v, T^v, \bX^v, Y^v , \bA^v)| \bH^{v-1},\bA^{v-1}, \bbeta \right \} \right ].
\end{align*}
\begin{assumption}[Lipschitz 2]
\label{as:lipchitz2}
    There exists a real-valued function $e$ such that for all $v \in \mathbb{N}$, and $\bbeta, \bbeta' \in \mathcal{B}$,
    \begin{equation*}
        |q^v(\bbeta) - q^v(\bbeta')| \leq e\left \{  (R^v, T^v, \bX^v, Y^v , \bA^v) \right \}  \|\bbeta - \bbeta'\|_2,
    \end{equation*}
    where for any $\bpi \in \Pi$, $\bbE_{\bbeta^*, \bpi} \left [ e\left \{ (R^v, T^v, \bX^v, Y^v , \bA^v)\right \}^2 \right ]$ is bounded.
\end{assumption}
\noindent Define $t_{\min}\geq0$ as the minimum processing time for an individual and $t_{\max}>0$ as the coupon expiration time so that $t_{\min} \leq T_j^v - T^v \leq t_{\max}$ for any $v \in \mathbb{N}$ and $j \in \left \{1,2,\ldots, M^v \right \}$. For a sample of size $\kappa$, define $J_{\noepochsamplesize}$ as the last complete epoch induced by the expiration of coupons, $J_{\noepochsamplesize} \triangleq \max\{j: \forall i \in \mathcal{E}_j, \ T_i + t_{\max} < T^{\noepochsamplesize} \}$.
Additionally, define $n_J \triangleq \min \left \{ \kappa : J_{\kappa} = J \right \}$ as the index of the first individual recruited after all the coupons associated with members of epoch $J$ have expired. 
\begin{theorem}\label{thm:budg}
For any $\bhistory \in \bHistory$,  define 
    $\widehat{\bpi}_{J} \in \arg\max_{\bpi \in \Pi} \branchingvalue(\bhistory, \bpi; \widehat{\bbeta}_{J})$. For $\delta > 0$, define the event $E_\branchasymp= \left \{ \branchasymp> \delta \right \}$, where $\branchasymp$ is defined in Assumption~\ref{as:generation_asymptotics}. 
Under Assumptions \ref{as:1}-\ref{as:lipchitz2} and for a fixed $\bhistory \in \bHistory$, it follows that
\begin{equation*}
    \branchingvalue(\bhistory,\bpi^{\mathrm{opt}}; \bbeta^*) - \branchingvalue(\bhistory, \widehat{\bpi}_{J}; \bbeta^*) = O_p\left (  1/\sqrt{\overline{\kappa}_J} \right )
\end{equation*}
 as $J \to \infty$ on event $E_\branchasymp$.
\end{theorem}
\noindent The proofs of Theorems~\ref{thm:asympNorm} and \ref{thm:budg} are in 
the Supplemental Materials. Lemma~\ref{lem:generation_asymptotics_2} in the Supplementary Materials ensures that $\lim_{\kappa \to \infty} J_{\noepochsamplesize} \to \infty$ a.s. under Assumption~\ref{as:generation_asymptotics}, implying that inference based on the last complete epoch will achieve the asymptotic guaranties of Theorems~\ref{thm:asympNorm} and \ref{thm:budg}. Note that $\lim_{\kappa \to \infty} J_{\noepochsamplesize} \to \infty$ a.s. implies $n_J < \infty$ for all $J \in \mathbb{N}$.

\subsection{A Branching Process Example}
\label{sec:examp}

In this section, we provide an example of a working branching process model that might
be used in the context of RL-RDS and illustrate how the assumptions used in 
Theorems (\ref{thm:asympNorm}) and (\ref{thm:budg}) can be verified for this
model.  
Define $T_{i,l}$, $\bX_{i,l}$, $Y_{i,l}$, and $A_{i,l}$ for $l=1,\ldots, M_i$ as the arrival times, covariates, rewards, and coupon types associated with the potential recruits of recruiter $i$ respectively. 
Define $\mathbb{A} = \{-1, 1\}$ as the set of possible coupon types (these might reflect different calls to
action for example), and $\bZ_{i,l} = \left ( 1, \bX_{i,l}, \bX_{i,l} \bbI \left (A_{i,l} = -1 \right ) \right )$.

We consider a working model of the form:  
\begin{eqnarray}\label{eq:rdsmod}
\bbP(M_{i} = m_{i} | \bH_{i}, \bA_i) &=& \frac{
\lambda^{m_i}/m_i!
}{
  \sum_{\ell=0}^{|\mathbf{A}_i|} (\lambda^\ell/\ell!)  
} , m_i=0,\ldots, |\mathbf{A}_i|,   \nonumber \\[3pt]
T_{i,l} - T_i | \bH_i, \bA_i, M_i &\sim& 
\mathrm{Truncated \ Exponential}(\zeta, t_{\min}, t_{\max}),\, l=1,\ldots, M_i,  \nonumber \\[3pt] 
\bX_{i,l} | \bH_i, U_{i,l}, \bA_i, M_i &\sim& 
\mathrm{Normal}\left (\bphi_{a} + G_{a}\bX_i,\Sigma_{a} \right ),\,
l=1,\ldots, M_i, a = A_{i,l},  \nonumber \\[3pt]
Y_{i,l} | \bH_i, \bX_i, U_{i,l}, \bA_i, M_i &\sim& \mathrm{ Bernoulli} \left \{ \frac{1}{1 + \exp \left (-{\bZ_{i,l}}^\top \bbeta_y \right )} \right \} ,\, 
l=1,\ldots, M_i,
\end{eqnarray}
where $\lambda, \zeta \in \mathbb{R}$, $\{ G_{a} \}_{a \in \mathbb{A}} \in \mathbb{R}^{p \times p}$, 
$\{ \bphi_{a} \}_{a \in \mathbb{A}} \in \mathbb{R}^{p \times 1}$, $\bbeta_y \in \mathbb{R}^{\ell}$, 
and $\{\Sigma_{a} \}_{a \in \mathbb{A}} \in \mathbb{R}^{p\times p}$. 
In the context of a randomized experiment, 
Assumption 1 (strong ignorability), Assumption 3 (positivity),
and Assumption 6 (budget growth) can be ensured to hold by design. Thus, we assume them 
as a matter of course.  Assumption 2 (consistency) is also assumed (this is standard 
as consistency is sometimes considered as an axiom rather than an assumption).  
We evaluate the assumptions that are sufficient for the convergence and regret results to hold for our working model in the Supplementary Materials. We verify them under the model conditions:   
\begin{itemize}
    \item[(C1)] the same number of coupons, defined as $L \in \mathbb{N}$, are given to each participant, $t_{\min} > 0$, and 
    $\lambda$ is such that
    \begin{equation*}
      \lambda \in  \left \{ \lambda: \sum_{m_i = 1}^L m_i \frac{ \lambda^{m_i}/m_i!}{
      \sum_{\ell=0}^{L} (\lambda^\ell/\ell!)  } > 1 \right \}; 
    \end{equation*}
    \item[(C2)] for each $i \in \mathbb{N}$, $A_{i,l}$ is constant across $l \in \left \{1,2,\ldots, M_{i} \right \}$;
    \item[(C3)] the set of coupon allocations available for each participant is constant throughout the study;
    \item[(C4)] $\mathcal{X}$ is compact; 
    \item[(C5)] $\mathcal{B} \subseteq \mathbb{R}^q$ is convex and compact, and $\bbeta^* \in \mathcal{B}$ is an interior point of $\mathcal{B}$;
    \item[(C6)] 
   $\frac{1}{|\mathbb{A}|} \sum_{a \in \mathbb{A}} \log \| G_{a} \|_2 < 0,$
    recalling that $\|\cdot \|$ is the spectral norm.
\end{itemize}
Under Conditions (C1)-(C6), we show that there exists a random variable $\branchasymp$ s.t.
\begin{equation}
    \label{eq:branching_converge_2}
    \branchasymp = \lim_{j \to \infty } m^{-j} \kappa_j \quad \mathrm{a.s.} , \quad  \bbP \left ( \branchasymp \geq \delta \right ) \geq \epsilon.
\end{equation}
where $\bbE \left ( \branchasymp \right ) = 1$ \citep{athreya2004branching}. Consequently, event $E_\branchasymp$ is well defined. Condition (C6) ensures that the auto-regressive covariate process is not explosive.

\begin{theorem}\label{thm:rlbranching}
Assume that Conditions (C1)-(C6) hold as well as Assumptions 1-3 and 6. Then, under the
working model given above, the conclusions of Theorems (\ref{thm:asympNorm}) and
(\ref{thm:budg}) hold. 
\end{theorem}
\noindent
The proof of the preceding result involves deriving a new weak law of large numbers for Galton-Watson branching processes to verify Assumption~\ref{as:stabalizedvariance}. 
Theorem~\ref{thm:rlbranching} implies that for the posited model, policy-search with Thompson sampling attains favorable regret bounds as $J \to \infty$. Because Equation~\ref{eq:branching_converge_2} and Condition (C1) satisfy Assumption~\ref{as:generation_asymptotics}, we know that Lemma~\ref{lem:generation_asymptotics_2} is also satisfied under the assumptions made in Theorem~\ref{thm:rlbranching}.
In the next section, we discuss inference
given an RL-RDS sample without assuming an underlying branching model.

\section{Inference for RL-RDS}
\label{sec:geninf}

In this section, we derive valid inference for functionals of the population network model.
To account for the underlying social network and the idiosyncrasies of RDS, 
we do not assume that data-generating model is a branching process.  Instead, we 
consider a dynamic network model indexed by $\btheta^* \in \bTheta$ for which
we derive asymptotic confidence sets.  Projections of these 
regions are then used to conduct inference for functionals of the data-generating 
model, e.g., disease prevalence, rate of risky behavior,
attitudes toward public health services, etc. 

We recall that the data-generating model (i.e., $\btheta^*$) need not be identifiable
under an RDS sampling scheme.  Nevertheless, it is still possible to obtain
valid confidence intervals by inverting a test 
\citep[e.g., see][]{robins2004optimal,laber2011adaptive}.  Our test is based on the
likelihood ratio for the covariate distribution in the branching process working model.  
We reiterate that this is only used to construct a test and that we are not assuming
that this model is correct.  
Let $\ell_{\noepochsamplesize}^{\btheta} \left (\bbeta_\bx  \right )$ be the log-likelihood\footnote{Note that in this inference procedure, we use the complete sample $v \in \left \{1,2, \ldots, \noepochsamplesize \right \}$ because we are no longer restricted to epoch structured data by the branching process.} of the branching process covariate model for a collection of $\noepochsamplesize$ subjects sampled under the 
RDS process when $\btheta \in \bTheta$ is the true parameter.
Define the MLE of the working model at $\btheta$ as 
\begin{equation*}
    \widehat{\bbeta}_{\bx}^{\noepochsamplesize}(\btheta) \in \arg \max_{\bbeta_\bx \in \mathcal{B}}\ell_{\noepochsamplesize}^{\btheta} \left ( \bbeta_\bx  \right ). 
\end{equation*}
At the true parameter, $\btheta^*$, we let $\ell_{\noepochsamplesize} \triangleq \ell_{\noepochsamplesize}^{\btheta^*}$ and $\widehat{\bbeta}_{\bx}^{\noepochsamplesize} \triangleq \widehat{\bbeta}_{\bx}^{\noepochsamplesize}(\btheta^*)$.

Recall that $\mathcal{B}$ is a bounded, open subset of $\mathbb{R}^k$. Let $s:\bTheta \to \mathcal{B}$ be a fixed function. Our confidence region is based on the 
distribution of the proximity of $\widehat{\bbeta}_{\bx}^{\noepochsamplesize}(\btheta)$ to $s(\btheta)$.  In 
our application, we choose $s(\btheta)$ to be an asymptotic limit of 
$\widehat{\bbeta}_{\bx}^{\noepochsamplesize}(\btheta)$ though other choices are possible and may be more
appropriate in other contexts; hence, we let $s$ be arbitrary.  
Given $s$, define the sampling distribution of the log-likelihood ratio statistic 
at $\btheta$ as 
\begin{equation*}
     -2\left [ \ell_{\noepochsamplesize}^{\btheta} \left \{ s(\btheta) \right \} -  \ell_{\noepochsamplesize}^{\btheta} \left \{\widehat{\bbeta}_{\bx}^{\noepochsamplesize}(\btheta) \right \} \right ] \sim P_{\noepochsamplesize}^{\btheta}\left ( s \right ),
\end{equation*}
and the $1-\alpha$ quantile of $P_{\noepochsamplesize}^{\btheta}\left (s \right )$ as 
$\bgamma_{1-\alpha,\noepochsamplesize}^{\btheta} (s)$.
A confidence region for $\btheta^*$ is
\begin{equation}
    \label{eq:finiteinterval}
    \Gamma_{1-\alpha,\noepochsamplesize}(s) = \left \{ \btheta : -2\left [ \ell_{\noepochsamplesize} \left \{ s(\btheta) \right \} -  \ell_{\noepochsamplesize} \left \{\widehat{\bbeta}_{\bx}^{\noepochsamplesize} \right \} \right ] \leq  \bgamma_{1-\alpha,\noepochsamplesize}^{\btheta} \left (s \right ) \right \}.
\end{equation}
For $\btheta \in \Theta$, we sample from $P_{\noepochsamplesize}^{\btheta}\left ( s \right )$ by simulating a dataset of size $\noepochsamplesize$ under RDS at $\btheta$ and calculate $-2\left [ \ell_{\noepochsamplesize}^{\btheta} \left \{ s(\btheta) \right \} -  \ell_{\noepochsamplesize}^{\btheta} \left \{\widehat{\bbeta}_{\bx}^{\noepochsamplesize}(\btheta) \right \} \right ]$ given this data. We can approximate $\bgamma_{1-\alpha,\noepochsamplesize}^{\btheta}(s)$ to arbitrary precision by generating a large number of draws from $P_{\noepochsamplesize}^{\btheta}\left ( s \right )$. 
That  $\Gamma_{1-\alpha, \noepochsamplesize}(s)$ achieves nominal coverage is easily verified as, by
construction, we have 
\begin{equation*}
        \mathbb{P}\left \lbrace \btheta^* \in \Gamma_{1-\alpha, \noepochsamplesize}(s)  \right \rbrace = 
        \mathbb{P}\left \lbrace   -2\left [ \ell_{\noepochsamplesize} \left \{ s(\btheta^*) \right \} -  \ell_{\noepochsamplesize} \left \{\widehat{\bbeta}_{\bx}^{\noepochsamplesize} \right \} \right ] 
        \leq \bgamma_{1-\alpha, \noepochsamplesize}^{\btheta^*} \left (s\right ) \right \rbrace = 1-\alpha.
\end{equation*}
This equality holds in finite samples regardless of the function $s$.

The preceding result shows that $\Gamma_{1-\alpha, \noepochsamplesize}(s)$ achieves nominal coverage, we now
consider another key attribute of this interval, its asymptotic concentration.  
To do this, we make the following additional
assumptions.  
First, we assume that the average log-likelihood for each $\bbeta_\bx \in \mathcal{B}$ converges to a finite limit.
\begin{assumption}[Pointwise convergence]
    \label{as:exist}
    For any $\btheta \in \bTheta$ and for each $\bbeta_\bx \in \mathcal{B}$, 
    $\ell_{\noepochsamplesize}^{\btheta}(\bbeta_{\bx})/\noepochsamplesize$ converges almost surely to a finite
    limit $\overline{\ell}^{\btheta}(\bbeta_{\bx})$ as $\noepochsamplesize \rightarrow \infty$.  
    Define \\ $\bar{\bbeta}_\bx(\btheta) = 
    \arg \max_{\bbeta_\bx \in \mathcal{B}} \overline{\ell}^{\btheta}(\bbeta_\bx)$, then $\bar{\bbeta}_\bx(\btheta)$ is finite almost surely.
\end{assumption}
\noindent Under this assumption, the asymptotic limit of $\widehat{\bbeta}_{\bx}^{\noepochsamplesize}(\btheta)$ is well-defined. We 
suggest using this limit, $\overline{\bbeta}_{\bx}(\btheta)$, as the function $s(\btheta)$. Let $\Gamma_{1-\alpha, \kappa }$ denote the confidence set for $\btheta^*$ under this
choice. 
We also assume that the average log-likelihood over a compact set stays strictly concave asymptotically.
\begin{assumption}
\label{as:gen_inf_iar}
For any $\btheta \in \Theta$, the log-likelihood of the working model, ${\ell}^{\btheta}_{\noepochsamplesize}$, is concave for every $\kappa \in \mathbb{N}$. Additionally, $\dot{\ell}_{{\noepochsamplesize}}^{\btheta} \left ( \bbeta_\bx \right )/{\noepochsamplesize} = O_p(1)$ as $\kappa \to \infty$. Lastly, for any $\btheta \in \Theta$, and any compact set $\mathcal{B}' \subseteq \mathcal{B}$, there exists $\delta_{\mathcal{B}', \btheta} > 0$ such that
\begin{equation*}
    \lim_{\noepochsamplesize \to \infty} \sup_{\bbeta_{\bx} \in \mathcal{B}'} \sigma_{\min} \left \{ -\ddot{\ell}^{\btheta}_{\noepochsamplesize} \left ( \bbeta_\bx \right ) /\noepochsamplesize \right \} \geq \delta_{\mathcal{B}', \btheta} \ \ \mathrm{a.s.}
\end{equation*}
\end{assumption}
 
Define the equivalence class 
\begin{equation*}
\bTheta^* = \left \{ \btheta \in \bTheta : \overline{\bbeta}_{\bx}(\btheta) = 
\overline{\bbeta}_{\bx}(\btheta^*) \right \}    
\end{equation*}
as the set of network models that have
the same asymptotic branching process parameter as the true network model. If there is a unique $\bar{\bbeta}_{\bx}(\btheta)$ for each $\btheta \in \bTheta$, then $\bTheta^* = \left\lbrace 
\btheta^*\right\rbrace$. Otherwise, the size of $\bTheta^*$ can be thought of as the ``price" our inference approach pays for using a working model to perform inference on $\btheta^*$ instead of the true data-generating model. In Theorem~\ref{thm:misspec}, we show that Assumptions~\ref{as:exist}-\ref{as:gen_inf_iar} with Conditions (C2) and (C4) imply that $\Gamma_{1-\alpha,\noepochsamplesize}$ concentrates around $\bTheta^*$ as $\noepochsamplesize \to \infty$.
\begin{theorem}
    \label{thm:misspec}
    Under Assumptions \ref{as:1}-\ref{as:3}, \ref{as:exist}-\ref{as:gen_inf_iar}, 
    $\bar{\bbeta}_\bx(\btheta)$ is unique and
    \begin{equation*}
        \widehat{\bbeta}^{\noepochsamplesize}_{\bx}(\btheta) \overset{p}{\to} \bar{\bbeta}_\bx(\btheta)
    \end{equation*}
    as $\noepochsamplesize \to \infty$. Additionally, for any $\btheta \notin \bTheta^*$,
    \begin{equation*}
       \lim_{\noepochsamplesize \to \infty} \bbP \left ( \btheta \in \Gamma_{1-\alpha, \noepochsamplesize} \right ) \to 0.
    \end{equation*}
\end{theorem}
\noindent The proof of Theorem~\ref{thm:misspec} is in the Supplemental Materials.
If we use the branching process in Equation~\ref{eq:rdsmod} as the working model, we can substitute Assumption~\ref{as:inf_iar}, Condition (2), and Condition (4) for Assumption~\ref{as:gen_inf_iar} in Theorem~\ref{thm:misspec}.
\begin{assumption}
\label{as:inf_iar}
For any $n \in \mathbb{N}$, define fields $\mathcal{F}^n = \sigma \left [\left\lbrace R^v, T^v, \bX^v, Y^v , \bA^v, C^v \right\rbrace_{v=1}^{n} \right ]$.
For any $\btheta \in \Theta$, there exists a positive definite $\Delta$, $\delta > 0$, and $N \in \mathbb{N}$ such that $\forall \ba \in \mathcal{A}$ and $n \geq N$, 
\begin{align*}
    \bbE_{\btheta} \left \{ \bX^{R^n} \bX^{R^n} \bbI(\bA^{R^n} = \ba) \mid \mathcal{F}^{n-1} \right \} \succeq \Delta, \quad
    \bbP_{\btheta} \left ( \bA^{R^n} = \ba \mid \mathcal{F}^{n-1} \right ) \geq \delta.
\end{align*}
\end{assumption}

We note that setting $s(\btheta) = \bar{\bbeta}_{\bx}(\btheta)$ necessitates approximating $\bar{\bbeta}_\bx(\btheta)$ for an arbitrary $\btheta \in \bTheta$. Algorithm~\ref{alg:cfllr} in 
the Supplemental Materials
describes our method of approximation as well as the full procedure for constructing $\Gamma_{1-\alpha, \kappa}$ when $s(\btheta) = \bar{\bbeta}_{\bx}(\btheta)$.

\section{RL-RDS Simulations}
\label{sec:sim}

We conducted a series of simulation experiments to evaluate the operating characteristics
of RL-RDS.  To allow comparisons with the two-stage procedure proposed by 
\citet{mcfall2021optimizing} (see also \citet{vanorsdale2023adaptive}),
we consider the setting in which the goal is to recruit the largest subset of 
people in a hidden population with a given binary trait, e.g., undiagnosed HIV.     
The outcome is thus an indicator of this trait.  
We  estimate the optimal policy, $\bpi^{\mathrm{opt}}$, using RL-RDS. After collecting the sample, we 
use the confidence set derived in Section~\ref{sec:geninf} to construct projection intervals
for parameters indexing the target population's network and covariate models. 

The hidden population network is generated as follows. Given population size $N$,  each population member $v \in \{1,2,\ldots, N\}$ is assigned attributes $\bX^v \overset{\mathrm{iid}}{\sim} \mathrm{Normal}(\bmu , \Sigma)$, where 
$\bmu \in \mathbb{R}^p$, and $\Sigma \in \mathbb{R}^{p\times p}$.
Define $\mathscr{A}_N \in \{0,1\}^{N\times N}$ to be the adjacency matrix representing links between the population members.
We construct the adjacency matrix using a latent distance network model, which is known to be flexible and projective \citep{spencer2017projective}.
This model specifies the probability of a connection between individuals $i$ and $j$ as
\begin{equation*}
    \mathbb{P}\left \{\left (\mathscr{A}_N \right )_{i,j} = 1\big| \bX^i=\bx^i,
    \bX^j=\bx^j\right \}= 1/\left [ 1+\exp\left \{- \left (\rho_0 - \rho_1 \| \bx^i - \bx^j \|_2 \right ) \right \} \right ],
\end{equation*}
where $\rho_0, \rho_1 \in \mathbb{R}$.

To characterize the evolution of RDS, for each  $v \in \{1,2,\ldots, N\}$, define
\begin{equation*}
    \mathcal{N}^v = \{ j \in \{1,2,\ldots, N\} \setminus \{v\}: \left (\mathscr{A}_N \right )_{v,j}  = 1 \}
\end{equation*}
as the neighborhood of $v$. 
Label the set of individual $v$'s potential recruits 
as $\mathcal{M}_1^v$; i.e., the un-recruited members of $\mathcal{N}^v$ when $v$ is recruited. Furthermore, define $\mathcal{M}^v_2 \subseteq \mathcal{M}^v_1$ as the coupon-constrained set of potential recruits for participant $v$.
If $|\mathcal{M}_1^v| \leq |\bA^v|$, then $\mathcal{M}_2^v  \equiv \mathcal{M}_1^v$. If $|\mathcal{M}_1^v| > |\bA^v|$, then sample $|\bA^v|$ recruits from $\mathcal{M}_1^v$ according to probabilities $\mathbf{p}^v = u(\bA^v, \underline{\bX}^{v})$, where $u: \mathcal{A} \times \mathcal{X}^{|\mathcal{M}_1^v|} \to [0,1]^{|\mathcal{M}_1^v|}$, and label them $\mathcal{M}_2^v$.
This allows for the neighbor selection process to depend on the characteristics of the neighbors and the coupon allocation type. 
When individual $v$ is recruited, we assign arrival times, $T^v_j$,  to the edges between recruiter, $v$, and hidden population members, $j \in \mathcal{M}_2^v$, such that $T^v_j - T^v \sim \mathrm{Truncated \ Exponential}(\zeta, t_{\min}, t_{\max})$. At time $T^v_j$, participant $j$ enters the study (if they have not been previously recruited) and edge $\{v, j\}$ is recorded. 
The reward is independently drawn for recruit $j$ of pariticpant $v$ according to $Y^v_j | \bX^v_j, A^v_j \sim \mathrm{ Bernoulli} \left [ 1/\left \{1 + \exp \left (-{\bZ^v_j}^\top \bbeta_y \right ) \right \} \right ]$. 
We note that the complete data-generating process is parameterized by $\btheta \triangleq (\rho_0, \rho_1, \bmu, \Sigma, \bbeta_y, \zeta, u, t_{\min}, t_{\max})$.

\subsection{Policies}
We evaluate the performance of RL-RDS against a suite of alternative strategies. At each step, the researcher can choose from a finite selection of coupon types. 
The fixed allocation policies (i.e., those that give the same coupon allocation type to all participants)  represent the current standard in RDS. The train-and-implement policy (aka, explore-than-exploit) mimics the procedure used by \cite{mcfall2021optimizing}, which determines an incentive strategy using a pilot study. 
We describe each strategy below.
\begin{enumerate}
	\item \textbf{Fixed} offers a fixed coupon allocation $\ba \in \mathcal{A}$ to every study participant. If $\ba \notin  \psi^v(\bh^v)$, then pick a random coupon allocation from $\psi^v(\bh^v)$ to give to the $v^{\mathrm{th}}$ study participant.
	\item \textbf{Random} offers a random element of $\psi^v(\bh^v)$ to the $v^{\mathrm{th}}$ study participant.
         \item \textbf{Train-and-Implement} uses half of the budget for a ``pilot study," in which the the Random policy is used to assign coupon allocations. It then conducts policy search using the pilot study data to estimate the branching process working model. This estimated policy (without updating) is then used to determine coupon allocations for the remainder of the budget.
	\item \textbf{RL-RDS} uses the Random policy to assign coupon allocations to participants in a short ``warm-up" period ($50$ participants in the simulations below). Then, it performs policy search with Thompson sampling
    for the remainder of the budget.
\end{enumerate}

To conduct RL-RDS, we use the following space of policies, $\Pi$. Define $\alpha_0 \in \mathbb{R}$, $\bma{\alpha}_1 \in \mathbb{R}^p$, and $\bma{\alpha} = (\alpha_0, \bma{\alpha}_1)$. For $n \in \mathbb{N}$ and state $\bh^n \in \mathcal{H}^n$, we consider policies of the form $\bpi(\bh^n) =  \left \{ \pi^n(\bh^n), \pi^{n+1}(\bh^{n+1}), \cdots  \right \}$ such that for $v \geq n$,
\begin{equation}\label{justinsWeirdFn}
    \pi^v(\bh^v) = \pi^v(\bh^v, \bma{\alpha}) = \pi^v(\bx^v,\bma{\alpha}) = g^v \left [ \frac{1}{1+\exp \left \{ - (\alpha_0 + {\bx^v}^\top \bma{\alpha}_1) \right \} } \right ],
\end{equation}
where $g^v: (0,1) \to \phi^v(\bh^v)$ maps a continuous score (dependent on the participant's covariates) to a coupon allocation, $\ba^v \in \psi^v(\bh^v)$. 
We first draw $\widehat{\bbeta}^{n}$ from the 
generalized bootstrap estimator of the 
sampling distribution of 
the MLE \citep{chatterjee2005generalized}. We then generate synthetic data sets,
\begin{equation*}
\mathcal{K}^B(\bh^n,\pmb{\pi}; \bbeta) \triangleq 
\left\lbrace
\left(Y_{b, \bpi, \bbeta}^n, \bH_{b, \bpi, \bbeta}^n, \bA_{b, \bpi, \bbeta}^n, Y_{b, \bpi, \bbeta}^{n+1}, \bH_{b, \bpi, \bbeta}^{n+1}, \bA_{b, \bpi, \bbeta}^{n+1}, \ldots,
\bY_{b, \bpi, \bbeta}^{Q}, \bH_{b, \bpi, \bbeta}^{Q}
\right)
\right\rbrace_{b=1}^B,
\end{equation*}
and calculate
\begin{equation*}
\widehat{V}^n_{B}(\bh^{n}, \bpi; \bbeta) =
\frac{\sum_{b = 1}^B \left (
\sum_{v = n}^{Q} \Delta^v_{b, \bpi, \bbeta} Y^v_{b, \bpi, \bbeta}
\right )}{B}.
\end{equation*}
for each $\bpi \in \Pi$. To determine the coupon allocation for the current study participant, we set $\widehat{\balpha}_B^n = \arg \max_{\balpha} \widehat{V}^n_{B}(\bh^{n}, \balpha; \widehat{\bbeta}^{n} )$,
and assign $\ba^n = \pi^n(\bx^n, \widehat{\balpha}_B^n)$. Alternatively, we could use a grid approximation to $\Pi$ or a gradient descent method to approximate $\widehat{\balpha}_B^n$.  

\subsection{Results}
In the following simulations, we set the hidden population size to $N= 5,000$. The recruitment process begins with an initial sample of $25$ individuals randomly drawn from the population, $\left | \mathcal{E}_0\right | = 25$. The graph model is defined by $\rho^*_0 =  0$, $\bmu^* = (1,1,1)$, and $\Sigma^* = \mathrm{diag} \left \{ (10,10,10) \right \}$. We vary $\rho^*_1 \in \{1,2\}$, to compare simulation experiments in dense and sparse network settings respectively.
The researchers have access to three types of coupon allocations $\mathcal{A} \triangleq \{\ba_1, \ba_2, \ba_3\}$ that correspond to three coupon types $\mathbb{A} \triangleq \{a_1, a_2, a_3\}$. Each allocation has $5$ coupons. We found that limiting the number of coupons given to each study participant in the pilot study and the warm-up period of the T\&I and the RLRDS policies respectively allows us to observe the effects of the learned policies earlier in the sampling process. Consequently, we give two coupons to individuals in the pilot and warm-up period, and increase the allotment to $5$ coupons afterwards. Additionally, $C^v \equiv 1$ for all $v \in \mathbb{N}$. For recruit $j$ of participant $v$, define the reward model components as $\bZ_j^v \triangleq \left \{1, \bX^v_j, \mathbb{I}(A^v_j = a_2)\bX^v_j,  \mathbb{I}(A^v_j = a_3)\bX^v_j \right \}$ and $\bbeta_y^* \triangleq (-1,3\bk,-3\bk,-6\bk)$, where $\bk \triangleq (1,-1,-1)$. Furthermore, we define the neighbor selection probability distribution as $\mathbf{p}^v = u(\bA^v, \underline{\bX}^{v})$ such that 
\begin{align*}
    p^v_j = \begin{cases}
      & \frac{ \left \|\bX^v - \bX^v_j \right \|_2^2} {\sum_{j \in \mathcal{M}_1^v}  \left \| \bX^v - \bX^v_j \right \|_2^2 },  \ \mathrm{if} \ \bA^v = \ba_2, \\
        & 1/|\mathcal{M}_1^v|, \ \mathrm{otherwise,}
    \end{cases}
   \end{align*}
where $\mathcal{M}_1^v$ is the set of participant $v$'s potential recruits (as defined in the previous section). This implies that recruiters with the second coupon type will be more likely to recruit neighbors that are different from them.
The basic policy objective is clear: we want to ensure that coupon type 1, $\ba_1$, is used to recruit individuals with covariates that satisfy ${\bX^v}^\top \bk > 0$, and coupon type 3, $\ba_3$, is used to recruit individuals with covariates that satisfy ${\bX^v}^\top \bk < 0$. 
We also wish to recruit individuals with large $| {\bX^v}^\top \bk |$ because this will increase the probability of the desired outcome contingent on the correct coupon type being awarded. 
Consequently, it is possible that coupon type 2, $\ba_2$, is optimal for participant $v$ if $| {\bX^v}^\top \bk |$ is low because it increases the likelihood of observing high values in the future.
We define the policy space by specifying the function $g^v$ from (\ref{justinsWeirdFn}) as 
\begin{equation*}
    g^v(z) = \begin{cases}
      \ba_1, & \mathrm{if}\ z > 0.66, \\
      \ba_2, & \mathrm{if}\ 0.66 > z \geq 0.33, \\
      \ba_3, & \mathrm{if}\ 0.33 > z.
    \end{cases}
\end{equation*}
This policy space implies that correctly assigning coupons $1$ or $3$ will depend on the sign of the $\balpha$ components. The frequency of coupon 2 allocation will be determined by the magnitude of $\balpha$. This structure makes finding an optimal policy computationally feasible while maintaining sufficient difficulty to showcase the strength of RL-RDS. Lastly, we set $\zeta = 1$, $t_{\min}=0$, and $t_{\max} = 3$.

Figure~\ref{fig:cumrew} illustrates the estimated value of policies in both the sparse and dense network settings. It indicates that RL-RDS outperforms all competitor policies by a significant margin in each regime.
In the Supplemental Materials,
we test RL-RDS under two additional simulation paradigms. We vary the value/cost and number of coupons given to RDS participants in these experiments. In these contexts, we introduce another branching process working model that incorporates the incentive value in the covariate and reward models (while reducing their dimensionality). This model is described by Equation~\ref{eq:rdsmod2} in the Supplemental Materials.
Figures~\ref{fig:cumrew_price} and \ref{fig:cumrew_amount} confirm that RL-RDS outperforms all competitor policies by a significant margin under multiple graph density and coupon allocation settings.

We construct 95\% confidence regions for the full network model, $(\rho_0^*, \rho_1^*, \bmu^*, \Sigma^*)$, with a variety of inference procedures.
In these experiments, we assume that $\Sigma^*$ is diagonal and $q$ is known to reduce the computational burden of the inference methods.
We also add an additional network setting, $\rho_1^* = 0.5$ (corresponding to a higher density network). 
Table~\ref{tab:GraphCoverage} depicts coverage results for these confidence regions. 
The table indicates that our simulation-based inference (SBI) technique, which is described by Algorithm~\ref{alg:cfllr} in the Supplemental Materials,
achieves nominal coverage or greater across all graph settings.
In the sparse setting,
bootstrapping with the log-likelihood ratio (BS LLR) and approximate bayesian computation (ABC) achieve nominal coverage as well. However, as the the graph becomes denser, ABC and BS LLR's coverage decreases while SBI's coverage stays above 95\%; i.e. the comparison methods fail to adequately quantify uncertainty while our method succeeds. The details of ABC, BS LLR, and bootstrapping with the wald interval (BS WI) are described in 
the Supplemental Materials.

Table~\ref{tab:cov_dist_int} reports 95\% projection interval coverage results for each dimension of the average covariate value; i.e., $\nu_i$ for $i \in \{1,2,3\}$. 
We see that the projection intervals associated with every inference method have high coverage except for BS WI. This is unsurprising since we conduct inference calibrated to have 95\% coverage over all network parameters simultaneously. Encouragingly, our method, SBI, has a smaller interval length than contenders that achieve nominal coverage (BS LLR and ABC).

\begin{figure}
    \centering
    \includegraphics{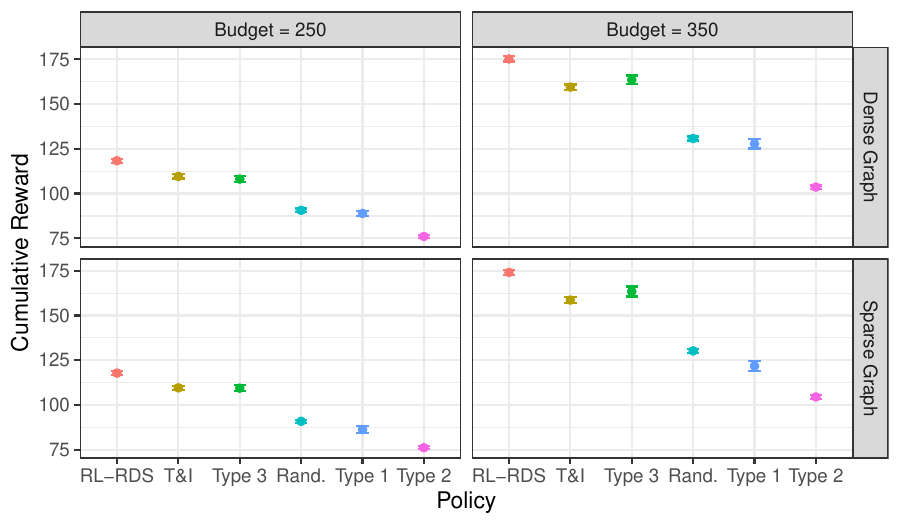}
    \caption{This figure compares the estimated cumulative reward of each policy with 90\% Monte Carlo confidence intervals over multiple sample sizes and graph densities.}
    \label{fig:cumrew}
\end{figure}

\begin{table}

\caption{\label{tab:GraphCoverage} Simultaneous Network Parameter Coverage}
\centering
\begin{tabular}[t]{rrr}
\toprule
Method & Graph Density & Coverage\\
\midrule
\textbf{SBI} & Low & 0.975\\
\cellcolor{gray!6}{\textbf{SBI}} & \cellcolor{gray!6}{Medium} & \cellcolor{gray!6}{0.962}\\
\textbf{SBI} & High & 0.979\\
\cellcolor{gray!6}{ABC} & \cellcolor{gray!6}{Low} & \cellcolor{gray!6}{0.950}\\
ABC & Medium & 0.749\\
\cellcolor{gray!6}{ABC} & \cellcolor{gray!6}{High} & \cellcolor{gray!6}{0.393}\\
BS LLR & Low & 0.945\\
\cellcolor{gray!6}{BS LLR} & \cellcolor{gray!6}{Medium} & \cellcolor{gray!6}{0.937}\\
BS LLR & High & 0.908\\
\cellcolor{gray!6}{BS WI} & \cellcolor{gray!6}{Low} & \cellcolor{gray!6}{0.526}\\
BS WI & Medium & 0.480\\
\cellcolor{gray!6}{BS WI} & \cellcolor{gray!6}{High} & \cellcolor{gray!6}{0.508}\\
\bottomrule
\end{tabular}
\begin{tablenotes}
    \item These are simultaneous coverage rates for the network parameters over $250$ simulations for three graph settings. 
\end{tablenotes}
\end{table}

\begin{table}

\caption{\label{tab:cov_dist_int}Covariate Distribution Mean Intervals}
\centering
\begin{tabular}[t]{rrrr}
\toprule
Method & Dim. & Coverage & Interval\\
\midrule
\cellcolor{gray!6}{\textbf{SBI}} & \cellcolor{gray!6}{1} & \cellcolor{gray!6}{0.99} & \cellcolor{gray!6}{2.78}\\
\textbf{SBI} & 2 & 0.99 & 2.71\\
\cellcolor{gray!6}{\textbf{SBI}} & \cellcolor{gray!6}{3} & \cellcolor{gray!6}{1.00} & \cellcolor{gray!6}{2.72}\\
ABC & 1 & 1.00 & 3.04\\
\cellcolor{gray!6}{ABC} & \cellcolor{gray!6}{2} & \cellcolor{gray!6}{1.00} & \cellcolor{gray!6}{3.23}\\
ABC & 3 & 1.00 & 3.09\\
\cellcolor{gray!6}{BS LLR} & \cellcolor{gray!6}{1} & \cellcolor{gray!6}{0.98} & \cellcolor{gray!6}{2.86}\\
BS LLR & 2 & 0.98 & 2.82\\
\cellcolor{gray!6}{BS LLR} & \cellcolor{gray!6}{3} & \cellcolor{gray!6}{0.97} & \cellcolor{gray!6}{2.84}\\
BS WI & 1 & 0.88 & 2.15\\
\cellcolor{gray!6}{BS WI} & \cellcolor{gray!6}{2} & \cellcolor{gray!6}{0.87} & \cellcolor{gray!6}{2.15}\\
BS WI & 3 & 0.89 & 2.15\\
\bottomrule
\end{tabular}
\begin{tablenotes}
    \item These are coverage rates and interval lengths averaged over graph settings for the mean of each dimension of the covariate distribution over $250$ simulations. 
\end{tablenotes}
\end{table}

\section{Discussion}
We showed that RL-driven adaptive RDS can lead to 
dramatic 
improvements in the effectiveness of RDS.  In the 
course of deriving an RL strategy, we: (i) showed that 
a branching process approximation to RDS is useful for 
guiding coupon selection without having to fully model 
the underlying system dynamics; (ii) extended 
asymptotic theory for adaptively sampled M-estimators 
to Markov Decision Processes; and (iii) developed novel 
methods for post adaptive-RDS inference that remain 
valid even when the true model is not identifiable.

This work is the first to consider fully adaptive 
RDS methods using RL.  Consequently,
there are a number of open problems and interesting
directions for future work.  One such direction
is the development of a regret
bound that explicitly depends on the quality of the
branching process approximation.  While branching 
processes
have been used extensively to model epidemics 
evolving on networks, the quality of these approximations
remains an important open question.  Another direction
for future work regards statistical efficiency.  As 
noted
in the introduction, it is possible to fold measures
of information gain into participant outcomes to improve 
power in post-study analyses.  However,
how best to do this for various network models is not clear.  
Finally, as we considered projective dynamic graph models,  
exploring whether there is a more
general class of models to which the asymptotic theory
still applies is yet another potential area for future research.

\newpage

\bibliographystyle{plainnat}
\bibliography{ref}

\begin{thebibliography}{64}
\providecommand{\natexlab}[1]{#1}
\providecommand{\url}[1]{\texttt{#1}}
\expandafter\ifx\csname urlstyle\endcsname\relax
  \providecommand{\doi}[1]{doi: #1}\else
  \providecommand{\doi}{doi: \begingroup \urlstyle{rm}\Url}\fi

\bibitem[Agrawal and Goyal(2013)]{agrawal2013further}
Shipra Agrawal and Navin Goyal.
\newblock Further optimal regret bounds for thompson sampling.
\newblock In \emph{Artificial intelligence and statistics}, pages 99--107, 2013.

\bibitem[Athreya et~al.(2004)Athreya, Ney, and Ney]{athreya2004branching}
Krishna~B Athreya, Peter~E Ney, and PE~Ney.
\newblock \emph{Branching processes}.
\newblock Courier Corporation, 2004.

\bibitem[Berry and Fristedt(1985)]{berry1985bandit}
Donald~A Berry and Bert Fristedt.
\newblock \emph{Bandit problems: sequential allocation of experiments (Monographs on statistics and applied probability)}.
\newblock Springer, 1985.

\bibitem[Bibaut and Kallus(2024)]{bibaut2024demistifying}
Aur{\'e}lien Bibaut and Nathan Kallus.
\newblock Demistifying inference after adaptive experiments.
\newblock \emph{arXiv preprint arXiv:2405.01281}, 2024.

\bibitem[Bibaut et~al.(2021)Bibaut, Dimakopoulou, Kallus, Chambaz, and van Der~Laan]{bibaut2021post}
Aur{\'e}lien Bibaut, Maria Dimakopoulou, Nathan Kallus, Antoine Chambaz, and Mark van Der~Laan.
\newblock Post-contextual-bandit inference.
\newblock \emph{Advances in neural information processing systems}, 34:\penalty0 28548--28559, 2021.

\bibitem[Brown(1986)]{brown1986fundamentals}
Lawrence~D Brown.
\newblock Fundamentals of statistical exponential families: with applications in statistical decision theory.
\newblock Ims, 1986.

\bibitem[Chapelle and Li(2011)]{chapelle2011empirical}
Olivier Chapelle and Lihong Li.
\newblock An empirical evaluation of thompson sampling.
\newblock In \emph{Advances in neural information processing systems}, pages 2249--2257, 2011.

\bibitem[Chatterjee and Bose(2005)]{chatterjee2005generalized}
Snigdhansu Chatterjee and Arup Bose.
\newblock {Generalized bootstrap for estimating equations}.
\newblock \emph{The Annals of Statistics}, 33\penalty0 (1):\penalty0 414 -- 436, 2005.
\newblock \doi{10.1214/009053604000000904}.
\newblock URL \url{https://doi.org/10.1214/009053604000000904}.

\bibitem[Cooks et~al.(2022)Cooks, Duke, Neil, Vilaro, Wilson-Howard, Modave, George, Odedina, Lok, Carek, et~al.]{cooks2022telehealth}
Eric~J Cooks, Kyle~A Duke, Jordan~M Neil, Melissa~J Vilaro, Danyell Wilson-Howard, Francois Modave, Thomas~J George, Folakemi~T Odedina, Benjamin~C Lok, Peter Carek, et~al.
\newblock Telehealth and racial disparities in colorectal cancer screening: A pilot study of how virtual clinician characteristics influence screening intentions.
\newblock \emph{Journal of clinical and translational science}, 6\penalty0 (1):\penalty0 e48, 2022.

\bibitem[Crawford et~al.(2018)Crawford, Wu, and Heimer]{crawford2018hidden}
Forrest~W Crawford, Jiacheng Wu, and Robert Heimer.
\newblock Hidden population size estimation from respondent-driven ampling: a network approach.
\newblock \emph{Journal of the American Statistical Association}, 113\penalty0 (522):\penalty0 755--766, 2018.

\bibitem[Delmas and Marsalle(2010)]{delmas2010detection}
Jean-Fran{\c{c}}ois Delmas and Laurence Marsalle.
\newblock Detection of cellular aging in a galton--watson process.
\newblock \emph{Stochastic Processes and their Applications}, 120\penalty0 (12):\penalty0 2495--2519, 2010.

\bibitem[Deshpande et~al.(2018)Deshpande, Mackey, Syrgkanis, and Taddy]{deshpande2018accurate}
Yash Deshpande, Lester Mackey, Vasilis Syrgkanis, and Matt Taddy.
\newblock Accurate inference for adaptive linear models.
\newblock In \emph{International Conference on Machine Learning}, pages 1194--1203. PMLR, 2018.

\bibitem[Eisenberg et~al.(2013)Eisenberg, Golberstein, Whitlock, and Downs]{eisenberg2013social}
Daniel Eisenberg, Ezra Golberstein, Janis~L Whitlock, and Marilyn~F Downs.
\newblock Social contagion of mental health: evidence from college roommates.
\newblock \emph{Health economics}, 22\penalty0 (8):\penalty0 965--986, 2013.

\bibitem[Fong et~al.(2007)Fong, Li, Yau, and Wong]{fong2007mixture}
Pak~Wing Fong, Wai~Keung Li, CW~Yau, and Chun~Shan Wong.
\newblock On a mixture vector autoregressive model.
\newblock \emph{Canadian Journal of Statistics}, 35\penalty0 (1):\penalty0 135--150, 2007.

\bibitem[Francq and Zako{\i}an(2001)]{francq2001stationarity}
Christian Francq and J-M Zako{\i}an.
\newblock Stationarity of multivariate markov--switching arma models.
\newblock \emph{Journal of Econometrics}, 102\penalty0 (2):\penalty0 339--364, 2001.

\bibitem[Gile(2011)]{gile2011improved}
Krista~J Gile.
\newblock Improved inference for respondent-driven sampling data with application to hiv prevalence estimation.
\newblock \emph{Journal of the American Statistical Association}, 106\penalty0 (493):\penalty0 135--146, 2011.

\bibitem[Gile and Handcock(2010)]{gile20107}
Krista~J Gile and Mark~S Handcock.
\newblock 7. respondent-driven sampling: An assessment of current methodology.
\newblock \emph{Sociological methodology}, 40\penalty0 (1):\penalty0 285--327, 2010.

\bibitem[Gile and Handcock(2015)]{gile2015network}
Krista~J Gile and Mark~S Handcock.
\newblock Network model-assisted inference from respondent-driven sampling data.
\newblock \emph{Journal of the Royal Statistical Society. Series A,(Statistics in Society)}, 178\penalty0 (3):\penalty0 619, 2015.

\bibitem[Goel and Salganik(2010)]{goel2010assessing}
Sharad Goel and Matthew~J Salganik.
\newblock Assessing respondent-driven sampling.
\newblock \emph{Proceedings of the National Academy of Sciences}, 107\penalty0 (15):\penalty0 6743--6747, 2010.

\bibitem[Gopalan and Mannor(2015)]{gopalan2015thompson}
Aditya Gopalan and Shie Mannor.
\newblock Thompson sampling for learning parameterized markov decision processes.
\newblock In \emph{Conference on Learning Theory}, pages 861--898, 2015.

\bibitem[Gopalan et~al.(2014)Gopalan, Mannor, and Mansour]{gopalan2014thompson}
Aditya Gopalan, Shie Mannor, and Yishay Mansour.
\newblock Thompson sampling for complex online problems.
\newblock In \emph{International Conference on Machine Learning}, pages 100--108, 2014.

\bibitem[Hadad et~al.(2021)Hadad, Hirshberg, Zhan, Wager, and Athey]{hadad2021confidence}
Vitor Hadad, David~A Hirshberg, Ruohan Zhan, Stefan Wager, and Susan Athey.
\newblock Confidence intervals for policy evaluation in adaptive experiments.
\newblock \emph{Proceedings of the national academy of sciences}, 118\penalty0 (15):\penalty0 e2014602118, 2021.

\bibitem[Hall and Heyde(2014)]{hall2014martingale}
Peter Hall and Christopher~C Heyde.
\newblock \emph{Martingale limit theory and its application}.
\newblock Academic press, 2014.

\bibitem[Hamilton(1989)]{hamilton1989new}
James~D Hamilton.
\newblock A new approach to the economic analysis of nonstationary time series and the business cycle.
\newblock \emph{Econometrica: Journal of the econometric society}, pages 357--384, 1989.

\bibitem[Heckathorn(1997)]{heckathorn1997respondent}
Douglas~D Heckathorn.
\newblock Respondent-driven sampling: a new approach to the study of hidden populations.
\newblock \emph{Social problems}, 44\penalty0 (2):\penalty0 174--199, 1997.

\bibitem[Heckathorn and Cameron(2017)]{heckathorn2017network}
Douglas~D Heckathorn and Christopher~J Cameron.
\newblock Network sampling: From snowball and multiplicity to respondent-driven sampling.
\newblock \emph{Annual review of sociology}, 43:\penalty0 101--119, 2017.

\bibitem[Laber and Murphy(2011)]{laber2011adaptive}
Eric~B Laber and Susan~A Murphy.
\newblock Adaptive confidence intervals for the test error in classification.
\newblock \emph{Journal of the American Statistical Association}, 106\penalty0 (495):\penalty0 904--913, 2011.

\bibitem[Laber et~al.(2018)Laber, Meyer, Reich, Pacifici, Collazo, and Drake]{laber2018optimal}
Eric~B Laber, Nick~J Meyer, Brian~J Reich, Krishna Pacifici, Jaime~A Collazo, and John~M Drake.
\newblock Optimal treatment allocations in space and time for on-line control of an emerging infectious disease.
\newblock \emph{Journal of the Royal Statistical Society: Series C (Applied Statistics)}, 67\penalty0 (4):\penalty0 743--789, 2018.

\bibitem[Lattimore and Szepesv{\'a}ri(2020)]{lattimore2020bandit}
Tor Lattimore and Csaba Szepesv{\'a}ri.
\newblock \emph{Bandit algorithms}.
\newblock Cambridge University Press, 2020.

\bibitem[Lu et~al.(2012)Lu, Bengtsson, Britton, Camitz, Kim, Thorson, and Liljeros]{lu2012sensitivity}
Xin Lu, Linus Bengtsson, Tom Britton, Martin Camitz, Beom~Jun Kim, Anna Thorson, and Fredrik Liljeros.
\newblock The sensitivity of respondent-driven sampling.
\newblock \emph{Journal of the Royal Statistical Society: Series A (Statistics in Society)}, 175\penalty0 (1):\penalty0 191--216, 2012.

\bibitem[Lunag{\'o}mez et~al.(2018)Lunag{\'o}mez, Papamichalis, Wolfe, and Airoldi]{lunagomez2018evaluating}
Sim{\'o}n Lunag{\'o}mez, Marios Papamichalis, Patrick~J Wolfe, and Edoardo~M Airoldi.
\newblock Evaluating and optimizing network sampling designs: Decision theory and information theory perspectives.
\newblock \emph{arXiv preprint arXiv:1811.07829}, 2018.

\bibitem[Ma et~al.(2007)Ma, Zhang, He, Sun, Yue, Chen, Raymond, Li, Xu, Du, et~al.]{ma2007trends}
Xiaoyan Ma, Qiyun Zhang, Xiong He, Weidong Sun, Hai Yue, Sanny Chen, H~Fisher Raymond, Yang Li, Min Xu, Hui Du, et~al.
\newblock Trends in prevalence of hiv, syphilis, hepatitis c, hepatitis b, and sexual risk behavior among men who have sex with men: results of 3 consecutive respondent-driven sampling surveys in beijing, 2004 through 2006.
\newblock \emph{JAIDS Journal of Acquired Immune Deficiency Syndromes}, 45\penalty0 (5):\penalty0 581--587, 2007.

\bibitem[Mahmood et~al.(2014)Mahmood, Levy, Vasan, and Wang]{mahmood2014framingham}
Syed~S Mahmood, Daniel Levy, Ramachandran~S Vasan, and Thomas~J Wang.
\newblock The framingham heart study and the epidemiology of cardiovascular disease: a historical perspective.
\newblock \emph{The lancet}, 383\penalty0 (9921):\penalty0 999--1008, 2014.

\bibitem[McFall et~al.(2021)McFall, Lau, Latkin, Srikrishnan, Anand, Vasudevan, Mehta, and Solomon]{mcfall2021optimizing}
Allison~M McFall, Bryan Lau, Carl Latkin, Aylur~K Srikrishnan, Santhanam Anand, Canjeevaram~K Vasudevan, Shruti~H Mehta, and Sunil~S Solomon.
\newblock Optimizing respondent-driven sampling to find undiagnosed hiv-infected people who inject drugs.
\newblock \emph{AIDS}, 35\penalty0 (3):\penalty0 485--494, 2021.

\bibitem[Newey and McFadden(1994)]{newey1994chapter}
Whitney~K Newey and Daniel McFadden.
\newblock Chapter 36 large sample estimation and hypothesis testing. volume 4 of handbook of econometrics.
\newblock \emph{Elsevier}, 12:\penalty0 2111--2245, 1994.

\bibitem[Perry et~al.(2018)Perry, Pescosolido, and Borgatti]{perry2018egocentric}
Brea~L Perry, Bernice~A Pescosolido, and Stephen~P Borgatti.
\newblock \emph{Egocentric network analysis: Foundations, methods, and models}, volume~44.
\newblock Cambridge university press, 2018.

\bibitem[Putterman(1994)]{putterman1994markov}
Martin~L Putterman.
\newblock Markov decision processes.
\newblock \emph{John Wiely and Sons, New York}, 1994.

\bibitem[Raifman et~al.(2022)Raifman, DeVost, Digitale, Chen, and Morris]{raifman2022respondent}
Sarah Raifman, Michelle~A DeVost, Jean~C Digitale, Yea-Hung Chen, and Meghan~D Morris.
\newblock Respondent-driven sampling: a sampling method for hard-to-reach populations and beyond.
\newblock \emph{Current Epidemiology Reports}, 9\penalty0 (1):\penalty0 38--47, 2022.

\bibitem[Risser et~al.(2009)Risser, Padgett, Wolverton, and Risser]{risser2009relationship}
JMH Risser, P~Padgett, M~Wolverton, and WL~Risser.
\newblock Relationship between heterosexual anal sex, injection drug use and hiv infection among black men and women.
\newblock \emph{International journal of STD \& AIDS}, 20\penalty0 (5):\penalty0 310--314, 2009.

\bibitem[Robins(2004)]{robins2004optimal}
J.M. Robins.
\newblock {Optimal structural nested models for optimal sequential decisions}.
\newblock In \emph{Proceedings of the Second Seattle Symposium in Biostatistics: Analysis of Correlated Data}, 2004.

\bibitem[Roch and Rohe(2018)]{roch2018generalized}
Sebastien Roch and Karl Rohe.
\newblock Generalized least squares can overcome the critical threshold in respondent-driven sampling.
\newblock \emph{Proceedings of the National Academy of Sciences}, 115\penalty0 (41):\penalty0 10299--10304, 2018.

\bibitem[Rohe(2015)]{rohe2015network}
Karl Rohe.
\newblock Network driven sampling; a critical threshold for design effects.
\newblock \emph{arXiv preprint arXiv:1505.05461}, 2015.

\bibitem[Rohe(2019)]{rohe2019critical}
Karl Rohe.
\newblock A critical threshold for design effects in network sampling.
\newblock \emph{Annals of Statistics}, 47\penalty0 (1):\penalty0 556--582, 2019.

\bibitem[Rubin(1978)]{rubin}
D.B. Rubin.
\newblock Bayesian inference for causal effects: The role of randomization.
\newblock \emph{The Annals of Statistics}, 6\penalty0 (1):\penalty0 34--58, 1978.

\bibitem[Russo and Van~Roy(2016)]{russo2016information}
Daniel Russo and Benjamin Van~Roy.
\newblock An information-theoretic analysis of thompson sampling.
\newblock \emph{The Journal of Machine Learning Research}, 17\penalty0 (1):\penalty0 2442--2471, 2016.

\bibitem[Russo et~al.(2018)Russo, Van~Roy, Kazerouni, Osband, and Wen]{russo2018tutorial}
Daniel~J Russo, Benjamin Van~Roy, Abbas Kazerouni, Ian Osband, and Zheng Wen.
\newblock A tutorial on thompson sampling.
\newblock \emph{Foundations and Trends{\textregistered} in Machine Learning}, 11\penalty0 (1):\penalty0 1--96, 2018.

\bibitem[Sevast’yanov and Zubkov(1974)]{sevast1974controlled}
Boris~A Sevast’yanov and Andreˇi~M Zubkov.
\newblock Controlled branching processes.
\newblock \emph{Theory of Probability \& Its Applications}, 19\penalty0 (1):\penalty0 14--24, 1974.

\bibitem[Slivkins et~al.(2019)]{slivkins2019introduction}
Aleksandrs Slivkins et~al.
\newblock Introduction to multi-armed bandits.
\newblock \emph{Foundations and Trends{\textregistered} in Machine Learning}, 12\penalty0 (1-2):\penalty0 1--286, 2019.

\bibitem[Spencer and Shalizi(2017)]{spencer2017projective}
Neil~A Spencer and Cosma~Rohilla Shalizi.
\newblock Projective, sparse, and learnable latent position network models.
\newblock \emph{arXiv preprint arXiv:1709.09702}, 2017.

\bibitem[Splawa-Neyman et~al.(1990)Splawa-Neyman, Dabrowska, Speed, et~al.]{splawa1990application}
Jerzy Splawa-Neyman, DM~Dabrowska, TP~Speed, et~al.
\newblock On the application of probability theory to agricultural experiments. essay on principles. section 9.
\newblock \emph{Statistical Science}, 5\penalty0 (4):\penalty0 465--472, 1990.

\bibitem[Stelzer(2009)]{stelzer2009markov}
Robert Stelzer.
\newblock On markov-switching arma processes—stationarity, existence of moments, and geometric ergodicity.
\newblock \emph{Econometric Theory}, 25\penalty0 (1):\penalty0 43--62, 2009.

\bibitem[Sutton and Barto(2018)]{sutton}
R.S. Sutton and A.G. Barto.
\newblock \emph{Reinforcment Learning: An Introduction}.
\newblock The MIT Press, 2018.

\bibitem[Szepesv{\'a}ri(2010)]{csaba2010}
Csaba Szepesv{\'a}ri.
\newblock Algorithms for reinforcement learning.
\newblock \emph{Synthesis Lectures on Artificial Intelligence and Machine Learning}, 4\penalty0 (1):\penalty0 1--103, 2010.

\bibitem[Tab{\'a}k et~al.(2012)Tab{\'a}k, Herder, Rathmann, Brunner, and Kivim{\"a}ki]{tabak2012prediabetes}
Adam~G Tab{\'a}k, Christian Herder, Wolfgang Rathmann, Eric~J Brunner, and Mika Kivim{\"a}ki.
\newblock Prediabetes: a high-risk state for diabetes development.
\newblock \emph{The Lancet}, 379\penalty0 (9833):\penalty0 2279--2290, 2012.

\bibitem[Tomas and Gile(2011)]{tomas2011effect}
Amber Tomas and Krista~J Gile.
\newblock The effect of differential recruitment, non-response and non-recruitment on estimators for respondent-driven sampling.
\newblock \emph{Electronic Journal of Statistics}, 5:\penalty0 899--934, 2011.

\bibitem[Tsiatis et~al.(2019)Tsiatis, Davidian, Holloway, and Laber]{tsiatis2019dynamic}
Anastasios~A Tsiatis, Marie Davidian, Shannon~T Holloway, and Eric~B Laber.
\newblock \emph{Dynamic Treatment Regimes: Statistical Methods for Precision Medicine}.
\newblock CRC press, 2019.

\bibitem[Van~der Vaart(2000)]{van2000asymptotic}
Aad~W Van~der Vaart.
\newblock \emph{Asymptotic statistics}, volume~3.
\newblock Cambridge university press, 2000.

\bibitem[VanOrsdale(2023)]{vanorsdale2023adaptive}
Josey VanOrsdale.
\newblock \emph{Adaptive Respondent Driven Sampling of Social Networks: A Simulation-Based Study Using Machine Learning}.
\newblock PhD thesis, The University of Nebraska-Lincoln, 2023.

\bibitem[Wejnert and Heckathorn(2011)]{wejnert2011respondent}
Cyprian Wejnert and Douglas Heckathorn.
\newblock Respondent-driven sampling: operational procedures, evolution of estimators, and topics for future research.
\newblock \emph{The SAGE handbook of innovation in social research methods. London: SAGE Publications, Ltd}, pages 473--97, 2011.

\bibitem[Weltz et~al.(2024)Weltz, Laber, and Volfovsky]{weltzhidden}
Justin David~Naggar Weltz, Eric Laber, and Alexander Volfovsky.
\newblock Hidden population estimation with indirect inference and auxiliary information.
\newblock In \emph{The 40th Conference on Uncertainty in Artificial Intelligence}, 2024.

\bibitem[Zhan et~al.(2023)Zhan, Ren, Athey, and Zhou]{zhan2023policy}
Ruohan Zhan, Zhimei Ren, Susan Athey, and Zhengyuan Zhou.
\newblock Policy learning with adaptively collected data.
\newblock \emph{Management Science}, 2023.

\bibitem[Zhan et~al.(2024)Zhan, Ren, Athey, and Zhou]{zhan2024policy}
Ruohan Zhan, Zhimei Ren, Susan Athey, and Zhengyuan Zhou.
\newblock Policy learning with adaptively collected data.
\newblock \emph{Management Science}, 70\penalty0 (8):\penalty0 5270--5297, 2024.

\bibitem[Zhang et~al.(2021)Zhang, Janson, and Murphy]{zhang2021statistical}
Kelly Zhang, Lucas Janson, and Susan Murphy.
\newblock Statistical inference with m-estimators on adaptively collected data.
\newblock \emph{Advances in neural information processing systems}, 34:\penalty0 7460--7471, 2021.

\bibitem[Zhang et~al.(2020)Zhang, Janson, and Murphy]{zhang2020inference}
Kelly~W Zhang, Lucas Janson, and Susan~A Murphy.
\newblock Inference for batched bandits.
\newblock \emph{arXiv preprint arXiv:2002.03217}, 2020.

\end{thebibliography}

\pagebreak
\begin{center}
\textbf{\Large Supplemental Materials}
\end{center}

\section{Glossary}
\begin{table}[ht]
    \centering
    \resizebox{12.5cm}{!}{
    \begin{tabular}{l || l}
       \small $i$ (subscript) & \small  index for interim study participants organized in epochs \\ [0.12cm]
       \small $v$ (superscript) & \small index for interim study participants ordered by arrival times \\ [0.1cm]
       \hdashline \\ [-0.4cm]
       \small $\bA_i$  &\small  coupon allocation given to $i$  \\ [0.12cm]
       \small $T_i$ &\small  arrival time of $i$ \\ [0.12cm]
       \small $\bX_i$ & \small covariate of $i$ \\ [0.12cm]
       \small $Y_i$ & \small  outcome of interest for $i$ \\ [0.12cm]
       \small $C_i$ & \small cost of recruitment for $i$ \\ [0.12cm]
       \small $R_i$ & \small recruiter of $i$ \\ [0.12cm]
       \small $M_i$ & \small number of recruits of $i$ \\ [0.12cm]
       \small $\underline{\bX}_i, \underline{\bY}_i, \underline{\bT}_i$ & \small collection of covariates, response statuses, and response times \\
       & \small associated with the recruits of participant $i$ \\ [0.12cm]
       \small $\bD_i$ & \small complete information associate with the recruits of individual $i$ \\ [0.12cm]
       \small $\bH^v$ & \small information available to researchers at the time $v$ is given their \\
       & \small coupon allocation \\ [0.1cm]
       \hdashline \\ [-0.4cm]
       \small $\mathcal{E}_j, \overline{\mathcal{E}}_j$  & \small  epoch $j$ and all epochs up to epoch $j$ respectively  \\ [0.12cm]
       \small $\kappa_j, \overline{\kappa}_j$  & \small size of epoch $j$ and size of all epochs up to epoch $j-1$ respectively  \\ [0.12cm]
       \small $\mathcal{Z}_j, \overline{\mathcal{Z}}_j$  & \small data associated with epoch $j$ and data associated with all \\
       & \small epochs up to epoch $j$ respectively \\ [0.12cm]
       \small $\mathcal{F}_j$ & \small sigma field associated with the first $j$ epochs without the coupon \\
       &\small allocation information associated with individuals in the $j^{th}$ epoch \\ [0.12cm]
       \small $\mathcal{D}^{\kappa}$ & \small data indexed by arrival time (first $\kappa$ individuals)  \\ [0.1cm]
       \hdashline \\ [-0.4cm]
       \small $\bpi$ & \small deterministic coupon allocation strategy \\ [0.12cm]
       \small $\bpi^{\mathrm{opt}}$ & \small optimal coupon allocation strategy \\ [0.12cm]
       \small $\psi^v(\bh^v)$ & \small set of allowable coupon allocations given $\bh^v$ \\ [0.12cm]
       \small $V^v(\bh^v, \pmb{\pi})$ & \small history-value function of $\bpi$ at $\bh^v$ \\ [0.12cm]
       \small $\widehat{\xi}^{v}_B(\bh^{v}, \ba^{v})$ & \small estimated probability that $\ba^v$ is the optimal action given $\bh^{v}$ \\ [0.1cm]
       \hdashline \\ [-0.4cm]
       \small $\bbeta \in \mathcal{B}$ & \small parameter that indexes the branching process \\ [0.12cm]
       \small $W_i$ & \small stabilizing weights for log-likelihood \\ [0.12cm]
       \small $l_i(\bbeta)$ & \small log-likelihood associated with $\bD_i$ conditional \\
       &\small on $\mathcal{F}_{j-1}$, where $i \in \mathcal{E}_{j-1}$\\ [0.12cm]
       \small $\widehat{\bbeta}_J$ &  \small M-estimator of weighted log-likelihood based on  complete epoch \\
       & \small information up to epoch $J$ \\ [0.12cm]
       \small $\branchasymp$ & \small asymptotic event associated with super-critical branching processes \\ [0.12cm]
       \small $q^v(\bbeta)$ & \small log-likelihood component associated with $v$ of the arrival process data, $\mathcal{D}^{\kappa}$ \\ [0.1cm]
       \hdashline \\ [-0.4cm]
       \small $\btheta \in \bTheta$ & \small parameter that indexes the dynamic network model \\ [0.12cm]
       \small $\ell_{\noepochsamplesize}^{\btheta} \left (\bbeta_\bx  \right )$ & \small log-likelihood of the branching process covariate model for a collection \\
       & \small of $\kappa$ subjects sampled under the RDS process when $\btheta$ is the true parameter \\ [0.12cm]
       \small $\bar{\bbeta}_{\bx}\left ( \btheta \right )$ & \small maximizes limit of $\ell_{\noepochsamplesize}^{\btheta} \left (\bbeta_\bx  \right )/\kappa$ \\ [0.12cm]
       \small $\Gamma_{1-\alpha, \noepochsamplesize} \left (\bar{\bbeta}_{\bx} \right )$ & \small confidence region for $\btheta^*$ based on $\bar{\bbeta}_{\bx}$
    \end{tabular}
    }
    \caption{Table of notation}
    \label{tab:gloss}
\end{table}

\newpage

\section{Inference Algorithms}
\label{app_sec:altalg}

Algorithm~\ref{alg:cfllr} generates an approximation of the confidence interval described in Section~\ref{sec:geninf}.
Recall that we need to approximate $\bar{\bbeta}_\bx(\btheta)$ for an arbitrary $\btheta \in \Theta$. We do this by first generating a large network for $\btheta$ and simulating an RDS study of size $K >> \noepochsamplesize$ over this network. Using this simulated sample, we maximize the log-likelihood of the branching process working model to obtain $\widehat{\bbeta}^K_{\bx}(\btheta)$.

\begin{algorithm}[H]
	\SetAlgoLined
	\DontPrintSemicolon
	\textbf{Input:} Set of parameter values $\Theta_S \subseteq \Theta$; the MLE of the working model, $\widehat{\bbeta}^{\noepochsamplesize}_\bx(\btheta^*)$; $K$, the sample size for the approximation of $\bar{\bbeta}_\bx(\btheta)$; and $B$, the number of simulated branching processes used to approximate the test statistic distribution\;
    \For{$\btheta \in \Theta_S$}{
	Generate a large network 
        from $\btheta$ (of size greater than $K$), and simulate an RDS study of size $K$ over this graph \;
        Use the branching process working model with this sample to find $\widehat{\bbeta}^K_{\bx}(\btheta)$\;
	Simulate $B$ branching processes of size $\noepochsamplesize$ over graphs generated according to $\btheta$\;
	For $b \in \{1,2, \ldots, B\}$,  calculate $\widehat{\bbeta}^{\noepochsamplesize}_{\bx, b}(\btheta)$\;
	Construct a set of test statistic values $\widehat{P}^{\btheta}_{\noepochsamplesize}( \bar{\bbeta}_\bx)= \left \{ -2\left [ \ell_{\noepochsamplesize}^{\btheta} \left \{ \widehat{\bbeta}^K_{\bx}(\btheta) \right \} -  \ell_{\noepochsamplesize}^{\btheta} \left \{\widehat{\bbeta}^{\noepochsamplesize}_{\bx,b}(\btheta) \right \} \right ], b \in \left \{1,2,\ldots, B \right \}  \right \}$ \;
	Determine the $1-\alpha$ upper quantile of $\widehat{P}^{\btheta}_{\noepochsamplesize}(\bar{\bbeta}_\bx)$ and label this $\widehat{\bgamma}^{\btheta}_{1-\alpha, \noepochsamplesize}(\bar{\bbeta}_\bx)$\;
	If $ -2\left [ \ell_{\noepochsamplesize}^{\btheta^*} \left \{ \widehat{\bbeta}^K_{\bx}(\btheta) \right \} -  \ell_{\noepochsamplesize}^{\btheta^*} \left \{\widehat{\bbeta}^{\noepochsamplesize}_\bx(\btheta^*) \right \} \right ] \leq \widehat{\bgamma}^{\btheta}_{1-\alpha, \noepochsamplesize}(\bar{\bbeta}_\bx)$, then include $\btheta$ in $\Gamma_{1-\alpha, \noepochsamplesize }$\;
 }
    \textbf{Output:} $\Gamma_{1-\alpha, \noepochsamplesize }$ \;
\caption{SBI Confidence Set}
\label{alg:cfllr}
\end{algorithm}

\newpage

We compare our inference approach to a series of natural alternatives.
Algorithm~\ref{alg:ABC} is inspired by approximate Bayesian computation (ABC).

\begin{algorithm}[H]
	\SetAlgoLined
	\DontPrintSemicolon
	\textbf{Input:} Set of parameter values $\Theta_S \subseteq \Theta$; $K$, the sample size for the approximation of $\bar{\bbeta}_\bx(\btheta)$; and $B$, the number of bootstrap samples\;
	
	Take $B$ bootstrapped samples using a generalized bootstrap for estimating equations as described in \cite{chatterjee2005generalized} and find the MLE for every sample, label this set 
    $\left \{\widehat{\bbeta}^{\noepochsamplesize}_{\bx, b} \right \}_{b=1}^B$\;

    For each $\btheta \in \Theta_S$, generate a large network consistent with $\btheta$ (of size greater than $K$), and simulate an RDS study of size $K$ over this graph \;
    Use the log-likelihood of the branching process working model with this sample to find $\widehat{\bbeta}^K_{\bx}(\btheta)$\;
        
	\textbf{Output:} $\Gamma_{1-\alpha, \noepochsamplesize } = \left \{\widehat{\btheta}_b = \arg \min_{\btheta \in \Theta} \| \widehat{\bbeta}^{\noepochsamplesize}_{\bx, b} - \widehat{\bbeta}^K_{\bx}(\btheta) \|^2, \ b \in \{1,2,\ldots, B\} \right \}$\;
	
	\caption{ABC Confidence Set}
	\label{alg:ABC}
\end{algorithm}

Algorithm~\ref{alg:bootstrap} describes the bootstrap with Wald interval (BS WI) and the bootstrap with likelihood ratio (BS LLR) approaches. It uses a generalized bootstrap for estimating equations as described in \cite{chatterjee2005generalized}. 

\begin{algorithm}[H]
	\SetAlgoLined
	\DontPrintSemicolon
	\textbf{Input:} Set of parameter values $\Theta_S \subseteq \Theta$; the MLE of the working model, $\widehat{\bbeta}^{\noepochsamplesize}_\bx$; $K$, the sample size for the approximation of $\bar{\bbeta}_\bx(\btheta)$; and $B$, the number of bootstrapped samples\;
	Take $B$ bootstrapped samples using a generalized bootstrap for estimating equations as described in \cite{chatterjee2005generalized}\;
    For $b \in \{1,2, \ldots, B\}$,  estimate the MLE, $\widehat{\bbeta}^{\noepochsamplesize}_{\bx, b}$, for each bootstrap from the log-likelihood $\ell_{\noepochsamplesize,b}$ \;
	Construct a set of test statistic the following manner: \\
        \uIf{LLR}{
        $\widehat{P}_{\noepochsamplesize}= \left \{ -2\left [ \ell_{\noepochsamplesize,b} \left ( \widehat{\bbeta}^{\noepochsamplesize}_\bx \right ) -  \ell_{\noepochsamplesize,b}\left \{\widehat{\bbeta}^{\noepochsamplesize}_{\bx,b} \right \} \right ], b \in \left \{1,2,\ldots, B \right \}  \right \}$ \;
        }
        \Else{
        Calculate $ \widehat{\Sigma}^{\noepochsamplesize}_B = \sum_{b=1}^B \left \{ \widehat{\bbeta}^{\noepochsamplesize}_{\bx,b} - \widehat{\bbeta}^{\noepochsamplesize}_\bx \right \} \left \{ \widehat{\bbeta}^{\noepochsamplesize}_{\bx,b} - \widehat{\bbeta}^{\noepochsamplesize}_\bx \right \}^\top / B$
        $\widehat{P}_{\noepochsamplesize} = \left \{ \left \{ \widehat{\bbeta}^{\noepochsamplesize}_{\bx,b} - \widehat{\bbeta}^{\noepochsamplesize}_\bx\right \}^\top \widehat{\Sigma}^{\noepochsamplesize}_B \left \{ \widehat{\bbeta}^{\noepochsamplesize}_{\bx,b}- \widehat{\bbeta}^{\noepochsamplesize}_\bx  \right \}   \right \}$ \; }
	Determine a $1-\alpha$ upper quantile of $\widehat{P}_{\noepochsamplesize}$ and label this $\widehat{\bgamma}_{1-\alpha, \noepochsamplesize}$\;
        \For{$\btheta \in \Theta_S$}{
         Generate a large network consistent with $\btheta$ (of size greater than $K$), and simulate an RDS study of size $K$ over this graph \;
        Use the log-likelihood of the branching process working model with this sample to find $\widehat{\bbeta}^K_{\bx}(\btheta)$\;
         \uIf{LLR}{
         If $ -2\left [ \ell_{\noepochsamplesize} \left \{ \widehat{\bbeta}^K_{\bx}(\btheta) \right \} -  \ell_{\noepochsamplesize} \left \{\widehat{\bbeta}^{\noepochsamplesize}_\bx \right \} \right ] \leq \widehat{\bgamma}_{1-\alpha, \noepochsamplesize}$, then include $\btheta$ in $\Gamma_{1-\alpha, \noepochsamplesize }$\;
        }
        \Else{
        If $\left \{ \widehat{\bbeta}^{\noepochsamplesize}_\bx- \widehat{\bbeta}^K_{\bx}(\btheta) \right \}^\top \widehat{\Sigma}^{\noepochsamplesize}_B \left \{ \widehat{\bbeta}^{\noepochsamplesize}_\bx - \widehat{\bbeta}^K_{\bx}(\btheta)  \right \} \leq \widehat{\bgamma}_{1-\alpha, \noepochsamplesize}$,
        then include $\btheta$ in $\Gamma_{1-\alpha, \noepochsamplesize }$\; }
        }
        \textbf{Output:} $\Gamma_{1-\alpha, \noepochsamplesize }$ \;
\caption{Bootstrap Confidence Set}
\label{alg:bootstrap}
\end{algorithm}

\newpage

\section{Proof of Theorem~\ref{thm:asympNorm}}
\label{app_sec:convrate}

\subsection{Supporting Martingale Limit Theory}
	Theorems~\ref{thm:SLLN}-\ref{thm:SLLNrand} establish useful martingale laws of large numbers.

\begin{thm}[Theorem 2.19 from \cite{hall2014martingale}]
    \label{thm:SLLN}
        Let $\{X_n, n \geq 1\}$ be a sequence of random variables and $\{\mathcal{F}_n, n \geq 1 \}$ be an increasing sequence of $\sigma$-fields such that $X_n$ is $\mathcal{F}^n$ measurable for each $n$. Let $X$ be a r.v. and $c$ a constant such $\bbE|X| < \infty$, $\bbE(|X|\log^+|X|) < \infty$, and $\bbP(|X_n| > x ) \leq c\bbP(|X| > x)$ for each $x \geq 0$ and $n\geq 1$ (note that if $X_n$ is bounded, this condition is satisfied). Then
        \begin{equation*}
            \frac{1}{n} \sum_{i=1}^n \left \{ X_i - \bbE(X_i | \mathcal{F}_{i-1}) \right \} \to 0
        \end{equation*}
        a.s. as $n \to \infty$. 
    \end{thm}

     \begin{thm}[Theorem 2.17 from \cite{hall2014martingale}]
     \label{thm:martconv}
    Let $\left \{S_n = \sum_{i=1}^n X_i , \mathcal{F}_n, n \geq 1 \right \}$ be a martingale, and let $1 \leq p \leq 2$. Then $S_n$ converges a.s. on the set \\ $\left \{ \sum_{i=1}^\infty \bbE \left (|X_i|^p | \mathcal{F}_{i-1} \right ) < \infty \right \}$.
     \end{thm}

     \begin{thm}[Theorem~2.18 from \cite{hall2014martingale}]
        \label{thm:SLLNrand}
         Let $\left \{S_n = \sum_{i=1}^n X_i , \mathcal{F}_n, n \geq 1 \right \}$ be a martingale and $\left \{U_n, n \geq 1 \right \}$ a non-decreasing sequence of positive random variables such that $U_n$ is $\mathcal{F}_{n-1}$-measurable for each $n$. If $1 \leq p \leq 2$, then
         \begin{equation*}
             \lim_{n \to \infty} U_n^{-1} S_n  = 0 \ a.s.
         \end{equation*}
         on the set $\left \{ \lim_{n\to \infty} U_n = \infty, \sum_{i=1}^\infty U_i^{-p} \bbE \left ( |X_i|^p | \mathcal{F}_{i-1} \right ) < \infty \right \}$.
     \end{thm}

      \noindent We also state the Toeplitz Lemma from \cite{hall2014martingale}.
        \begin{lemma}
        \label{lem:toep}
        For $i \geq 1$, if $x_i$ are real numbers, $a_i$ are positive numbers, and $b_n = \sum_{i=1}^n a_i$ diverges, then $x_n \to x$ ensures that
        \begin{equation*}
            b_n^{-1} \sum_{i=1}^n a_i x_i \to x.
        \end{equation*}
        \end{lemma}

        \noindent Lastly, we state a strong law of large numbers specific to the branching process model developed in Section~\ref{sec:rl}. Assume the historical information and epoch setup of Section~\ref{sec:rl} for the following Theorem. 

      \begin{thm}
        \label{thm:branchingSLLN}
        For function $z: \mathscr{D} \to \mathbb{R}$, assume for any $j \in \mathbb{N}$, $i \in \epochgroup$, $\bbeta \in \mathcal{B}$, and $B < \infty$ that
        \begin{equation}
            \label{eq:boundedf}
            \bbE_{\bbeta, \widetilde{\bpi}} \left \{ z^2(\bD_i) \mid \epochfield \right \} < B.
        \end{equation}
        Additionally, assume that for $j \in\mathbb{N}$, $\{\bD_i\}_{i \in \epochgroup}$ are conditionally independent given $\epochfield$.
        For $\delta > 0$, define the event $E_\branchasymp= \left \{ \branchasymp> \delta \right \}$, where $\branchasymp$ is defined in Assumption~\ref{as:generation_asymptotics}. Then, 
        \begin{equation*}
            \lim_{J \to \infty} \frac{1}{\overline{\kappa}_J} \left [ \sum_{j =1}^J \sum_{i \in \epochgroup} W_i z(\bD_i) - \bbE_{\bbeta, \widehat{\bpi}} \left \{ W_i z(\bD_i) \mid \epochfield \right \} \right ]  \to 0
        \end{equation*}
        on the event $E_\branchasymp$.
        \end{thm}
        \begin{proof}
            We first consider the sum
            \begin{equation*}
            \sum_{j =1}^J \sum_{i \in \epochgroup} W_i z(\bD_i) - \bbE_{\bbeta, \widehat{\bpi}} \left \{ W_i z(\bD_i) \mid \epochfield \right \}. 
            \end{equation*}
            This is clearly a martingale sequence with respect to fields $\left \{ \mathcal{F}_{j} \right \}_{j=1}^J$ and martingale differences 
            \begin{equation*}
               \left [ \sum_{i \in \epochgroup} W_i z(\bD_i) - \bbE_{\bbeta, \widehat{\bpi}} \left \{ \sum_{i \in \epochgroup} W_i z(\bD_i) \mid \epochfield \right \} \right ]_{j = 1}^J.
            \end{equation*}
            Consequently, we can invoke Theorem~\ref{thm:SLLNrand}, implying that
            \begin{equation*}
                \lim_{J \to \infty} \frac{1}{\overline{\kappa}_J} \left [ \sum_{j =1}^J \sum_{i \in \epochgroup} W_i z(\bD_i) - \bbE_{\bbeta, \widehat{\bpi}} \left \{ W_i z(\bD_i) \mid \epochfield \right \} \right ] = 0 
            \end{equation*}
            on the set 
            \begin{align*}
                &\Bigg \{ \lim_{J \to \infty} \overline{\kappa}_J \to \infty, \\
                &\lim_{J \to \infty} \sum_{j=1}^J \frac{1}{\overline{\kappa}^2_j} \bbE_{\bbeta, \widehat{\bpi}} \left ( \left [  \sum_{i \in \epochgroup} W_i z(\bD_i) - \bbE_{\bbeta, \widehat{\bpi}} \left \{ W_i z(\bD_i) | \epochfield \right \}  \right ]^2 \mid \epochfield \right ) < \infty \Bigg \}. 
            \end{align*}
            We examine
            \begin{align*}
                &\sum_{j=1}^J\frac{1}{\overline{\kappa}_j^2} \bbE_{\bbeta, \widehat{\bpi}} \left ( \left [  \sum_{i \in \epochgroup} W_i z(\bD_i) - \bbE_{\bbeta, \widehat{\bpi}} \left \{ W_i z(\bD_i) \mid \epochfield \right \}  \right ]^2 \mid \epochfield\right ) \overset{(a)}{=} \\
                &\sum_{j=1}^J\frac{1}{\overline{\kappa}_j^2} \sum_{i \in \epochgroup}  \bbE_{\bbeta, \widehat{\bpi}} \left ( \left [ W_i z(\bD_i) - \bbE_{\bbeta, \widehat{\bpi}} \left \{ W_i z(\bD_i) \mid \epochfield \right \}  \right ]^2 \mid \epochfield\right ) = \\
                &\sum_{j=1}^J\frac{1}{\overline{\kappa}_j^2} \sum_{i \in \epochgroup} \Big (  \bbE_{\bbeta, \widehat{\bpi}} \left \{W^2_i z^2(\bD_i) \mid \epochfield\right \} - \\
                &\hspace{1cm}\bbE_{\bbeta, \widehat{\bpi}} \left [ \bbE_{\bbeta, \widehat{\bpi}} \left \{ W_i z(\bD_i) \mid \epochfield \right \}^2 \mid \epochfield  \right ] \Big ) \leq \\
                &\sum_{j=1}^J\frac{1}{\overline{\kappa}_j^2} \sum_{i \in \epochgroup} \bbE_{\bbeta, \widehat{\bpi}} \left (W^2_i z^2(\bD_i) \mid \epochfield\right ) = \\ 
                &\sum_{j=1}^J\frac{1}{\overline{\kappa}_j^2} \sum_{i \in \epochgroup} \int W_i^2 z^2(\bD_i) f\left ( \bD_i|\vA_i, \epochfield \right ) \bbP_{\widehat{\pi}}(\bA_i | \epochfield ) d\nu  \overset{(b)}{=} \\
                &\sum_{j=1}^J\frac{1}{\overline{\kappa}_j^2} \sum_{i \in \epochgroup} \int \frac{\bbP_{\widetilde{\pi}}(\bA_i | \bH_i)}{\bbP_{\widehat{\pi}}(\bA_i | \bH_i) } z^2(\bD_i)  f\left ( \bD_i|\vA_i, \epochfield \right ) \bbP_{\widehat{\pi}}(\bA_i | \bH_i) d\nu = \\
                &\sum_{j=1}^J\frac{1}{\overline{\kappa}_j^2} \sum_{i \in \epochgroup} \bbE_{\bbeta, \widetilde{\bpi}} \left (z^2(\bD_i) \mid \epochfield\right ) \overset{(c)}{\leq} \\
                &\sum_{j=1}^J \frac{\kappa_j}{\overline{\kappa}_j^2} B \overset{(d)}{\leq} \\
                &\sum_{j=1}^J \frac{1}{\kappa_j} B. 
            \end{align*}
            Equality $(a)$ follows from the fact that for $i>g$,
            \begin{align*}
                &\bbE_{\bbeta, \widehat{\bpi}} \left (  \left [ W_i z(\bD_i) - \bbE_{\bbeta, \widehat{\bpi}} \left \{ W_i z(\bD_i) \mid \epochfield \right \}  \right ]\left [ W_g z(\bD_g) - \bbE_{\bbeta, \widehat{\bpi}} \left \{ W_g z(\bD_g) \mid \epochfield \right \}  \right ] \mid \epochfield\right ) \\
                &=\bbE_{\bbeta, \widehat{\bpi}} \Bigg (  \left [ \bbE_{\bbeta, \widehat{\bpi}}  \left \{ W_i z(\bD_i) \mid  \epochfield, \bD_g\right \}  - \bbE_{\bbeta, \widehat{\bpi}} \left \{ W_i z(\bD_i) \mid \epochfield \right \}  \right ] \times \\
                &\hspace{2cm}\left [ W_g z(\bD_g) - \bbE_{\bbeta, \widehat{\bpi}} \left \{ W_g z(\bD_g) \mid \epochfield \right \}  \right ] \mid \epochfield \Bigg ) \\
                &\overset{(i)}{=} \bbE_{\bbeta, \widehat{\bpi}} \Bigg (  \left [ \bbE_{\bbeta, \widehat{\bpi}}  \left \{ W_i z(\bD_i) \mid  \epochfield\right \}  - \bbE_{\bbeta, \widehat{\bpi}} \left \{ W_i z(\bD_i) \mid \epochfield \right \}  \right ] \times \\
                &\hspace{2cm}\left [ W_g z(\bD_g) - \bbE_{\bbeta, \widehat{\bpi}} \left \{ W_g z(\bD_g) \mid \epochfield \right \}  \right ] \mid \epochfield \Bigg ) \\
                &= 0.
            \end{align*}
            Equality $(i)$ follows because
            \begin{equation*}
               \bbE_{\bbeta, \widehat{\bpi}}  \left \{ W_i z(\bD_i) \mid  \epochfield, \bD_g\right \} =  \bbE_{\bbeta, \widehat{\bpi}}  \left \{ W_i z(\bD_i) \mid  \epochfield \right \}.
            \end{equation*}
            Equality $(b)$ follows from the fact that $\bbP_{\widehat{\pi}}(\bA_i | \epochfield ) = \bbP_{\widehat{\pi}}(\bA_i | \bH_i)$ and the definition of $W_i$.
            Inequality $(c)$ follows from Equation~\ref{eq:boundedf}.
             Additionally, we know
            \begin{align*}
                \sum_{j=1}^J \frac{1}{\kappa_j} B &= \sum_{j=1}^J \frac{1}{m^{j-1}} \frac{m^{j-1}}{\kappa_j} B \\
                &= \left ( \sum_{j=1}^J \frac{1}{m^{j-1}} \right ) \left ( \sum_{j=1}^J \frac{1}{m^{j-1}} \right )^{-1} \sum_{j=1}^J \frac{1}{m^{j-1}} \frac{m^{j-1}}{\kappa_j} B \\ &\overset{\mathrm{a.s.}}{\to} 
                \left ( \sum_{j=1}^J \frac{1}{m^{j-1}} \right ) \branchasymp^{-1}B
            \end{align*}
            by Lemma~\ref{lem:toep} and Assumption~\ref{as:generation_asymptotics}.
            Because we assume event $E_\branchasymp$, we know that $\branchasymp^{-1} < \frac{1}{\delta}$. Consequently,
            \begin{equation*}
                \sum_{j=1}^J \frac{1}{\kappa_j} B \overset{\mathrm{a.s.}}{\to} 
               \frac{1}{1-1/m} \branchasymp^{-1}B < \frac{B}{\delta(1-1/m)} < \infty.
            \end{equation*}
            
        \end{proof}
        
\subsection{Consistency}
	First, we show asymptotic consistency of $\widehat{\bbeta}_{J}$. We follow the proof of \cite{zhang2021statistical}, making slight changes to account for a more general context. We define $e_i(\bbeta) \triangleq  W_i l_i(\bbeta)$. By the definition of $\widehat{\bbeta}_J$,
	
	\begin{equation}
		\label{eq:esteq1}
		\sum_{j = 1}^J \sum_{i \in \epochgroup} e_i  (\widehat{\bbeta}_{J}) = \sup_{\bbeta \in \mathcal{B}}  \sum_{j = 1}^J \sum_{i \in \epochgroup} e_i \left (\bbeta \right ) \geq \sum_{j = 1}^J \sum_{i \in \epochgroup} e_i \left (\bbeta^* \right ).
	\end{equation}
	Note that $\|\widehat{\bbeta}_J - \bbeta^*\| > \epsilon > 0$ implies that
	\begin{equation}
		\label{eq:esteq2}
		\sup_{\bbeta \in \mathcal{B}: \|\bbeta - \bbeta^*\| > \epsilon} \sum_{j = 1}^J \sum_{i \in \epochgroup} e_i \left (\bbeta \right ) = \sup_{\bbeta \in \mathcal{B}} \sum_{j = 1}^J \sum_{i \in \epochgroup} e_i \left (\bbeta \right ). 
	\end{equation}
        Define $\bbP_{\bbeta^*, \widehat{\bpi}}$ as the probability distribution associated with the complete data-generating distribution, where $\bbeta^*$ is the generative parameter associated with the branching process and $\widehat{\bpi}$ is the policy followed during data collection. Equations~\ref{eq:esteq1} and \ref{eq:esteq2} imply the following inequality,
	\begin{align*}
		\bbP_{\bbeta^*, \widehat{\bpi}}\left ( \| \widehat{\bbeta}_{J}- \bbeta^* \| > \epsilon \right ) &\leq \bbP_{\bbeta^*, \widehat{\bpi}}\left \{ \sup_{\bbeta \in \mathcal{B}: \|\bbeta - \bbeta^*\| > \epsilon} \sum_{j = 1}^J \sum_{i \in \epochgroup}e_i \left (\bbeta \right ) \geq  \sum_{j = 1}^J \sum_{i \in \epochgroup}e_i \left (\bbeta^* \right ) \right \} \\
		&\hspace{-2.5cm}=\bbP_{\bbeta^*, \widehat{\bpi}}\left [ \sup_{\bbeta \in \mathcal{B}: \|\bbeta - \bbeta^*\| > \epsilon} \frac{1}{\overline{\kappa}_J} \sum_{j = 1}^J \sum_{i \in \epochgroup}W_il_i\left (\bbeta \right ) - \frac{1}{\overline{\kappa}_J} \sum_{j = 1}^J \sum_{i \in \epochgroup} W_il_i\left (\bbeta^* \right ) \geq 0  \right ]\\
		&\hspace{-4cm} = \bbP_{\bbeta^*, \widehat{\bpi}} \Bigg [ \sup_{\bbeta \in \mathcal{B}: \|\bbeta - \bbeta^*\| > \epsilon} \Bigg \{ \Bigg ( \frac{1}{\overline{\kappa}_J} \sum_{j = 1}^J \sum_{i \in \epochgroup}\Bigg [  W_il_i \left (\bbeta \right ) -  \bbE_{\bbeta^*, \widehat{\bpi}} \left \{  W_il_i \left (\bbeta \right ) | \epochfield\right \}  + \\
        &\hspace{-3.5cm} \bbE_{\bbeta^*, \widehat{\bpi}} \left \{  W_il_i \left (\bbeta \right ) | \epochfield\right \} \Bigg ] \Bigg ) \\
		&\hspace{-3.5cm}-\frac{1}{\overline{\kappa}_J} \sum_{j = 1}^J \sum_{i \in \epochgroup}\left [  W_il_i \left (\bbeta^* \right ) -  \bbE_{\bbeta^*, \widehat{\bpi}} \left \{  W_il_i \left (\bbeta^* \right ) | \epochfield\right \}  + \bbE_{\bbeta^*, \widehat{\bpi}} \left \{  W_il_i \left (\bbeta^* \right ) | \epochfield\right \} \right ] \geq 0 \Bigg \} \Bigg ].
	\end{align*}
	By the triangle inequality,
	\begin{align}
		\begin{split}
			\label{eq:consist}
			\leq \bbP_{\bbeta^*, \widehat{\bpi}} \Bigg \{& \underbrace{\sup_{\bbeta \in \mathcal{B}: \|\bbeta - \bbeta^*\| > \epsilon} \left ( \frac{1}{\overline{\kappa}_J} \sum_{j = 1}^J \sum_{i \in \epochgroup}\left [  W_il_i \left (\bbeta \right ) -  \bbE_{\bbeta^*, \widehat{\bpi}} \left \{  W_il_i \left (\bbeta \right ) | \epochfield\right \} \right ] \right )}_{\mathrm{(a)}} \\
			&+ \underbrace{\sup_{\bbeta \in \mathcal{B}: \|\bbeta - \bbeta^*\| > \epsilon} \left \{ \frac{1}{\overline{\kappa}_J} \sum_{j = 1}^J \sum_{i \in \epochgroup} \left (\bbE_{\bbeta^*, \widehat{\bpi}} \left [  W_i \left \{ l_i \left (\bbeta \right ) -  l_i \left (\bbeta^* \right ) \right \} | \epochfield\right ] \right ) \right \}}_{\mathrm{(b)}} \\
			&\hspace{0.2cm}- \underbrace{\frac{1}{\overline{\kappa}_J} \sum_{j = 1}^J \sum_{i \in \epochgroup}\left [  W_il_i \left (\bbeta^* \right ) -  \bbE_{\bbeta^*, \widehat{\bpi}} \left \{  W_il_i \left (\bbeta^* \right ) | \epochfield\right \}\right ]}_{\mathrm{(c)}} \geq 0 \Bigg \}.
		\end{split}
	\end{align}
	We first analyze quantity (c), which is a martingale by construction. By Assumption~\ref{as:moments} (the moments condition), we know that there exists $B_{l^2} \in \mathbb{R}^+$ such that \\ $\bbE_{\bbeta^*, \widetilde{\bpi}} \left \{ l_i^2(\bbeta^*) | \epochfield\right \}  \leq B_{l^2}$. Consequently,
	by Theorem~\ref{thm:branchingSLLN}, $\lim_{J \to \infty}|\mathrm{(c)}| \to 0$ almost surely. We use a uniform martingale strong law of large numbers, Lemma~\ref{lem:USLLN}, to prove that $\lim_{ J \to \infty}|\mathrm{(a)}| \to 0$ almost surely.
	
	Therefore, it is sufficient for consistency to show that there exists $\delta' > 0$ such that
	\begin{equation*}
		\lim_{J \to \infty} \sup_{\bbeta \in \mathcal{B}: \|\bbeta - \bbeta^*\| > \epsilon} \left \{ \frac{1}{\overline{\kappa}_J} \sum_{j = 1}^J \sum_{i \in \epochgroup} \left (\bbE_{\bbeta^*, \widehat{\bpi}} \left [  W_i \left \{ l_i \left (\bbeta \right ) -  l_i \left (\bbeta^* \right ) \right \} | \epochfield\right ] \right ) \right \} \leq -\delta' \ \mathrm{a.s.}
	\end{equation*}
	By the law of iterated expectations,
	\begin{align*}
		&\sup_{\bbeta \in \mathcal{B}: \|\bbeta - \bbeta^*\| > \epsilon} \left \{ \frac{1}{\overline{\kappa}_J} \sum_{j = 1}^J \sum_{i \in \epochgroup} \left (\bbE_{\bbeta^*, \widehat{\bpi}} \left [  W_i \left \{ l_i \left (\bbeta \right ) -  l_i \left (\bbeta^* \right ) \right \} \mid \epochfield \right ] \right ) \right \} = \\
		&\sup_{\bbeta \in \mathcal{B}: \|\bbeta - \bbeta^*\| > \epsilon} \Bigg \{ \frac{1}{\overline{\kappa}_J} \sum_{j = 1}^J \sum_{i \in \epochgroup} \Big (\int_{\mathcal{A}_i} \bbP \left ( \bA_i = \ba \mid \epochfield \right ) \times \\
        &\hspace{2.5cm}W_i \bbE_{\bbeta^*} \left \{ l_i \left (\bbeta \right ) -  l_i \left (\bbeta^* \right ) \mid \bA_i = \ba, \epochfield\right \} d\ba \Big ) \Bigg \}. \ \ (*)
	\end{align*}
	First, by the restriction of information explained in Section~\ref{sec:rl}, we know that \\ $\bbP \left ( \bA_i = \ba \mid \epochfield \right ) = \bbP_{\widehat{\bpi}} \left ( \ba \mid H_i \right )$. For $\rho_{\min}, \rho_{\max} \in \mathbb{R}$, we know that $\rho_{\min} \leq W_i \leq \rho_{\max}$ by Assumption~\ref{as:clip}.
	Because $  \bbE_{\bbeta^*} \left \{ l_i \left (\bbeta \right ) -  l_i \left (\bbeta^* \right ) | \bA_i = \ba_i, \epochfield\right \}\leq 0$ with probability $1$, we find that 
	\begin{align*}
		& (*) \  \leq \sup_{\bbeta \in \mathcal{B}: \|\bbeta - \bbeta^*\| > \epsilon} \Bigg \{ \frac{1}{\rho_{\max} \overline{\kappa}_J} \sum_{j = 1}^J \sum_{i \in \epochgroup} \Bigg (\int_{\mathcal{A}_i} \bbP_{\widehat{\bpi}}(\ba \mid \bH_i) \times \\
        &\hspace{7.2cm}W_i^2 \bbE_{\bbeta^*} \left \{ l_i \left (\bbeta \right ) -  l_i \left (\bbeta^* \right ) | \bA_i = \ba, \epochfield \right \} d\ba \Bigg ) \Bigg \} \\
		&= \sup_{\bbeta \in \mathcal{B}: \|\bbeta - \bbeta^*\| > \epsilon} \Bigg \{ \frac{1}{\rho_{\max} \overline{\kappa}_J} \sum_{j = 1}^J \sum_{i \in \epochgroup} \Bigg ( \int_{\mathcal{A}_i} \bbP_{\widetilde{\bpi}}(\ba \mid \bH_i) \times \\ &\hspace{7.2cm}\bbE_{\bbeta^*} \left \{ l_i \left (\bbeta \right ) -  l_i \left (\bbeta^* \right ) | \bA_i = \ba, \epochfield\right \} d\ba  \Bigg ) \Bigg \} \\
		&= \sup_{\bbeta \in \mathcal{B}: \|\bbeta - \bbeta^*\| > \epsilon} \left [ \frac{1}{\rho_{\max} \overline{\kappa}_J} \sum_{j = 1}^J \sum_{i \in \epochgroup} \bbE_{\bbeta^*, \widetilde{\bpi}} \left \{ l_i \left (\bbeta \right ) -  l_i \left (\bbeta^* \right ) | \epochfield\right \}  \right ].
	\end{align*}
    From Assumption~\ref{as:wellseperated}, it follows that
		\begin{equation*}
			\frac{1}{\rho_{\max}} \lim_{J \to \infty} \sup_{\bbeta \in \mathcal{B}: \|\bbeta - \bbeta^*\| > \epsilon} \left [ \frac{1}{\overline{\kappa}_J}\sum_{j = 1}^J \sum_{i \in \epochgroup}\bbE_{\bbeta^*, \widetilde{\bpi}} \left \{ l_i \left (\bbeta \right ) -  l_i \left (\bbeta^* \right ) | \epochfield\right \} \right ] \leq -\delta \frac{1}{\rho_{\max}} \ \ \mathrm{a.s.}
	\end{equation*}
	Thus, Equation~\ref{eq:consist} holds with $\delta' \triangleq \delta \frac{1}{\rho_{\max}}$.

\subsection{Convergence Rate}
    First, we will prove that
    \begin{equation}
    \label{eq:almostrate}
        \Sigma_{J}^{1/2}  \left ( \widehat{\bbeta}_{J} - \bbeta^* \right ) = O_p \left ( 1 \right ).
    \end{equation}
    Define the weighted log-likelihood at $\bbeta \in \mathcal{B}$ as
    \begin{equation*}
      \mathcal{M}_{J}(\bbeta) \triangleq  \sum_{j = 1}^J \sum_{i \in \epochgroup}W_i l_i(\bbeta).
    \end{equation*}
    We begin with a Taylor series expansion of $\dot{\mathcal{M}_{J}}$ between $\widehat{\bbeta}_{J}$ and $\bbeta^*$,
	\begin{equation*}
		0 = \dot{\mathcal{M}_{J}}(\bbeta^*) + \ddot{\mathcal{M}}_{J}(\bar{\bbeta}^{J})(\widehat{\bbeta}_{J} - \bbeta^*),
	\end{equation*}
	where $\bar{\bbeta}^{J}$ is between $\widehat{\bbeta}_{J}$ and $\bbeta^*$. Multiplying and dividing by $\Sigma_{J}^{-1/2}$, we find that
	\begin{align*}
		-\Sigma_{J}^{-1/2}\dot{\mathcal{M}}_{J}(\bbeta^*) &=  \Sigma_{J}^{-1/2} \left \{ \ddot{\mathcal{M}}_{J}(\bar{\bbeta}^{J}) \right \} \Sigma_{J}^{-1/2} \Sigma_{J}^{1/2} \left (\widehat{\bbeta}_{J} - \bbeta^* \right ).
	\end{align*}
	If we show that
	\begin{equation}
		\label{eq:asconv}
		\Sigma_{J}^{-1/2}\dot{\mathcal{M}}_{J}(\bbeta^*) = O_p(1),
	\end{equation}
    and 
    \begin{equation}
        \label{eq:scale}
        \sigma_{\min} \left (\Sigma_{J}^{-1/2} \ddot{\mathcal{M}}_{J}(\bar{\bbeta}^{J}) \Sigma_{J}^{-1/2} \right )^{-1}  = O_p \left (1 \right ),
    \end{equation}
    then Equation~\ref{eq:almostrate} follows.
	Consequently, we will divide the proof of Equation~\ref{eq:almostrate} into proofs of Equations~\ref{eq:asconv} and \ref{eq:scale}. This proof only relies on Assumptions~\ref{as:1}-\ref{as:moments} and \ref{as:stabalizedvariance}.
	
	\subsubsection{Proof of Equation~\ref{eq:asconv}}
    \label{app_sec:asconv}
 
	  We first prove that
	\begin{equation*}
	\Sigma_{J}^{-1/2}\dot{\mathcal{M}}_{J}(\bbeta^*) = O_p(1).
    \end{equation*}
	Note that
	\begin{equation*}
		\Sigma_{J}^{-1/2}\dot{\mathcal{M}}_{J}(\bbeta^*) = \Sigma_{J}^{-1/2} \left \{ \sum_{j = 1}^J \sum_{i \in \epochgroup}W_i \dot{l}_i(\bbeta^*) \right \}.
	\end{equation*}
	Define $\bc_g  \in \mathbb{R}^k$ as a standard basis vector with $1$ in the $g^{th}$ position and $0$ elsewhere. We know that for any $g \in \{1,2,\ldots, k \}$,
	\begin{equation*}
		\left \{ \bc_g^\top \Sigma_{J}^{-1/2} \sum_{i \in \epochgroup} W_i \dot{l}_i(\bbeta^*)   \right \}_{j=1}^{J}
	\end{equation*}
	is a martingale difference sequence with respect to fields $\left \{ \mathcal{F}_j \right \}_{j=1}^{J}$. For any $j \in \mathbb{N}$ and $\bpi \in \Pi$,
	\begin{align*}
		\bbE_{\bbeta^*, \bpi} \left \{ \bc_g^\top \sum_{i \in \epochgroup} W_i \dot{l}_i(\bbeta^*) | \epochfield\right \} &=	\bc_g^\top \sum_{i \in \epochgroup} \bbE_{\bbeta^*, \bpi} \left \{  W_i \dot{l}_i(\bbeta^*) | \epochfield \right \} \\
		&= \bbE_{\bbeta^*, \bpi} \left [ \bc_g^\top \sum_{i \in \epochgroup}  W_i \bbE_{\bbeta^*}  \left \{ \dot{l}_i(\bbeta^*) | \epochfield, \bA_i \right \} \mid \epochfield \right ] \\
		&= 0.
	\end{align*}
	We now apply Theorem~\ref{thm:martconv} (Theorem~2.15 from \cite{hall2014martingale}) to this martingale difference sequence. By this theorem,
    \begin{equation*}
        \sum_{j = 1}^J  \bbE_{\bbeta^*, \widehat{\bpi}} \left [  \left \{ \bc_g^\top \Sigma_{J}^{-1/2} \sum_{i \in \epochgroup} W_i \dot{l}_i(\bbeta^*) \right \}^2  \mid \epochfield\right ] \overset{p}{\to} \bc^\top U \bc.
    \end{equation*}
    is sufficient for 
    \begin{equation*}
        \sum_{j=1}^J \bc_g^\top \Sigma_{J}^{-1/2} \sum_{i \in \epochgroup} W_i \dot{l}_i(\bbeta^*) 
    \end{equation*}
    to converge almost surely. We find that 
    \begin{align*}
        &\sum_{j = 1}^J \bbE_{\bbeta^*, \widehat{\bpi}} \left [  \left \{ \bc_g^\top \Sigma_{J}^{-1/2} \sum_{i \in \epochgroup}  W_i \dot{l}_i(\bbeta^*) \right \}^2  | \epochfield\right ] \\
        &=\sum_{j = 1}^J \left [  \left \{ \bc_g^\top \Sigma_{J}^{-1/2}\bbE_{\bbeta^*, \widehat{\bpi}} \left \{ \left ( \sum_{i \in \epochgroup} W_i \dot{l}_i(\bbeta^*) \right ) \left ( \sum_{i \in \epochgroup} W_i \dot{l}_i(\bbeta^*)^\top \right ) \mid \epochfield \right \}
        \Sigma_{J}^{-1/2} \bc_g \right \}\right ] \\
        &\overset{(a)}{=} \bc_g^\top \Sigma_{J}^{-1/2} \left \{ \sum_{j = 1}^J \bbE_{\bbeta^*, \widehat{\bpi}}  \left [ \sum_{i \in \epochgroup}  \left \{ W_i^2 \dot{l}(\bbeta^*) \dot{l}(\bbeta^*)^\top  \right \} | \epochfield\right ] \right \}\Sigma_{J}^{-1/2} \bc_g \\
        &\overset{(b)}{=} \bc_g^\top \Sigma_{J}^{-1/2} \left \{ \sum_{j = 1}^J \sum_{i \in \epochgroup} \bbE_{\bbeta^*, \widehat{\bpi}}  \left  \lbrace \frac{\bbP_{\widetilde{\pi}}(\bA_i)}{\bbP_{\widehat{\pi}}(\bA_i) }\dot{l}_i(\bbeta^*) \dot{l}_i(\bbeta^*)^\top | \epochfield\right \rbrace \right \}\Sigma_{J}^{-1/2} \bc_g \\
        &= \bc_g^\top \Sigma_{J}^{-1/2} \Bigg \{ \sum_{j = 1}^J \sum_{i \in \epochgroup} \int \frac{\bbP_{\widetilde{\pi}}(\bA_i  | \bH_i )}{\bbP_{\widehat{\pi}}(\bA_i | \bH_i ) } \dot{l}_i(\bbeta^*) \dot{l}_i(\bbeta^*)^\top \times \\
        &\hspace{2.5cm} f\left ( \bD_i|\vA_i, \epochfield; \bbeta^*\right ) \bbP_{\widehat{\pi}}(\bA_i | \epochfield) d\nu \Bigg \}\Sigma_{J}^{-1/2} \bc_g \\
        &\overset{(c)}{=} \bc_g^\top \Sigma_{J}^{-1/2} \Bigg \{ \sum_{j = 1}^J \sum_{i \in \epochgroup} \int \frac{\bbP_{\widetilde{\pi}}(\bA_i  | \bH_i )}{\bbP_{\widehat{\pi}}(\bA_i | \bH_i ) } \dot{l}_i(\bbeta^*) \dot{l}_i(\bbeta^*)^\top \times \\
        &\hspace{2.5cm}f\left ( \bD_i|\vA_i, \epochfield; \bbeta^*\right ) \bbP_{\widehat{\pi}}(\bA_i | \bH_i) d\nu \Bigg \}\Sigma_{J}^{-1/2} \bc_g \\
        &= \bc_g^\top \Sigma_{J}^{-1/2} \left \{ \sum_{j = 1}^J \sum_{i \in \epochgroup} \int \dot{l}_i(\bbeta^*) \dot{l}_i(\bbeta^*)^\top  f\left ( \bD_i|\vA_i, \epochfield; \bbeta^*\right ) \bbP_{\widetilde{\pi}}(\bA_i | \bH_i) d\nu \right \}\Sigma_{J}^{-1/2} \bc_g \\
        &= \bc_g^\top \Sigma_{J}^{-1/2} \left [ \sum_{j = 1}^J \sum_{i \in \epochgroup} \bbE_{\bbeta^*, \widetilde{\bpi}}  \left  \lbrace \dot{l}_i(\bbeta^*) \dot{l}_i(\bbeta^*)^\top | \epochfield\right \rbrace \right ]\Sigma_{J}^{-1/2} \bc_g \\
        &\overset{(d)}{=} \bc_g^\top \Sigma_{J}^{-1/2} \eta_{J}\Sigma_{J}^{-1/2} \bc_g.
    \end{align*}
        Equality $(a)$ follows from the law of iterated expectations applied to the cross terms of the sum (as explained in the proof of Theorem~\ref{thm:branchingSLLN}); equality $(b)$ follows from the definition of $W_i$; equality $(c)$ follows from the definition of historical information, $\bH_i$; and equality $(d)$ follows from the definition of $\eta_J$.
	By Assumption~\ref{as:stabalizedvariance},
	\begin{equation*}
		\bc_g^\top \Sigma_{J}^{-1/2} \eta_{J}\Sigma_{J}^{-1/2} \bc_g \overset{p}{\to} \bc_g^\top U \bc_g.
	\end{equation*}
    By Theorem~\ref{thm:martconv}, this implies that
    \begin{equation*}
        \bc_g^\top \Sigma_{J}^{-1/2} \sum_{i \in \epochgroup} W_i \dot{l}_i(\bbeta^*)
    \end{equation*}
    converges a.s. to a random variable, which we label $V_g$ for $g \in \left \{1,2,\ldots, k \right \}$. Labeling $\bV = (V_1, V_2, \ldots, V_k)$, we observe that
    \begin{equation*}
     \lim_{J \to \infty} \Sigma_{J}^{-1/2} \sum_{i \in \epochgroup} W_i \dot{l}_i(\bbeta^*) \overset{p}{\to} \bV.
    \end{equation*}
    This verifies equation~\ref{eq:asconv}.

	\subsubsection{Proof of Equation~\ref{eq:scale}}
    First, note that the differentiability and moment conditions of the log-likelihood (Assumptions~\ref{as:differentiable} and \ref{as:moments}) imply that
    \begin{equation*}
    \eta_J = \sum_{j=1}^J \sum_{i \in \mathcal{E}_{j-1}} \bbE_{\bbeta^*, \widetilde{\bpi}} \left \{ \dot{l}_i(\bbeta^*) \dot{l}_i(\bbeta^*)^\top \mid \epochfield \right \} = \sum_{j=1}^J \sum_{i \in \mathcal{E}_{j-1}} \bbE_{\bbeta^*, \widetilde{\bpi}} \left \{ -\ddot{l}_i(\bbeta^*) \mid \epochfield \right \}.
    \end{equation*}
    Define
    \begin{equation*}
        \alpha_J \triangleq \sum_{j = 1}^J -\bbE_{\bbeta^*, \widehat{\bpi}} \left \{  \sum_{i \in \epochgroup} W_i \ddot{l}_i(\bbeta^*) \mid \epochfield \right \}.
    \end{equation*}
    By Lemmas~\ref{lem:mineigen} and \ref{lem:submult} and the fact that $\rho_{\min} \leq  W_i \leq \rho_{\max} $,
    \begin{align*}
        \sigma_{\min} \left ( \Sigma_{J}^{-1/2} \left \{  -\ddot{\mathcal{M}}_{J}(\bar{\bbeta}^{J}) \right \} \Sigma_{J}^{-1/2} \right )
        & \geq \sigma_{\min} \left ( \Sigma_{J}^{-1/2} \left \{- \ddot{\mathcal{M}}_{J}(\bar{\bbeta}^{J}) +  \ddot{\mathcal{M}}_{J}(\bbeta^*)\right \}  \Sigma_{J}^{-1/2} \right ) + \\
        &\hspace{0.5cm}\sigma_{\min} \left ( \Sigma_{J}^{-1/2} \left \{ -\ddot{\mathcal{M}}_{J}(\bbeta^*) - \alpha_J \right \}  \Sigma_{J}^{-1/2} \right )  + \\
        &\hspace{0.5cm}  \sigma_{\min} \left ( \Sigma_{J}^{-1/2} \alpha_J \Sigma_{J}^{-1/2} \right ) \\
        & \geq \sigma_{\min}  \left  \{ \left (\Sigma_{J}/\overline{\kappa}_J \right )^{-1/2} \right \} \sigma_{\min} \left ( \frac{\ddot{\mathcal{M}}_{J}(\bar{\bbeta}^{J}) -  \ddot{\mathcal{M}}_{J}(\bbeta^*)}{\overline{\kappa}_J} \right ) \times \\
        &\hspace{0.5cm} \sigma_{\min} \left  \{ \left (\Sigma_{J}/\overline{\kappa}_J \right )^{-1/2} \right \}  + \\
        &\hspace{0.5cm} \sigma_{\min} \left \{ \left (\Sigma_{J}/\overline{\kappa}_J \right )^{-1/2} \right \} \sigma_{\min} \left \{  \frac{-\ddot{\mathcal{M}}_{J}(\bbeta^*) - \alpha_J}{\overline{\kappa}_J} \right \} \times \\
        &\hspace{0.5cm} \sigma_{\min} \left \{ \left (\Sigma_{J}/\overline{\kappa}_J \right )^{-1/2} \right  \} + \\
        &\hspace{0.5cm}   \sigma_{\min} \left ( \Sigma_{J}^{-1/2} \alpha_J \Sigma_{J}^{-1/2} \right ) \\
         & \overset{(a)}{\geq}  \sigma_{\min} \left \{ \left (\Sigma_{J}/\overline{\kappa}_J \right )^{-1/2} \right \} \sigma_{\min} \left \{ \frac{\ddot{\mathcal{M}}_{J}(\bar{\bbeta}^{J}) -  \ddot{\mathcal{M}}_{J}(\bbeta^*)}{\overline{\kappa}_J} \right  \} \times \\
         & \hspace{0.5cm} \sigma_{\min} \left \{ \left (\Sigma_{J}/\overline{\kappa}_J \right )^{-1/2} \right \} + \\
        &\hspace{0.5cm} \sigma_{\min} \left \{ \left (\Sigma_{J}/\overline{\kappa}_J \right )^{-1/2} \right \} \sigma_{\min} \left \{ \frac{-\ddot{\mathcal{M}}_{J}(\bbeta^*) - \alpha_J}{\overline{\kappa}_J} \right \} \times \\
        & \hspace{0.5cm} \sigma_{\min} \left \{ \left (\Sigma_{J}/\overline{\kappa}_J \right )^{-1/2} \right \} + \\
        &\hspace{0.5cm}  \frac{1}{\rho_{\max}} \sigma_{\min} \left ( \Sigma_{J}^{-1/2} \eta_J  \Sigma_{J}^{-1/2} \right ).
    \end{align*}
    By Assumption~\ref{as:stabalizedvariance}, $\sigma_{\min}\left (\Sigma_{J}^{-1/2} \eta_J \Sigma_{J}^{-1/2} \right )^{-1} = O_p(1)$ as $J \to \infty$. Inequality (a) will be verified by Property (4) below. In sum, it is sufficient to establish four properties:
    \begin{enumerate}
        \item $\sigma_{\min} \left \{ \left (\Sigma_{J}/\overline{\kappa}_J \right )^{-1/2} \right \} \leq \left \| \left (\Sigma_{J}/\overline{\kappa}_J \right )^{-1/2} \right \|_2  = O_p(1)$
        \item $ \left \| \left \{ \ddot{\mathcal{M}}_{J}(\bbeta^*) - \ddot{\mathcal{M}}_{J}(\bar{\bbeta}^{J}) \right \}/\overline{\kappa}_J \right \|_2  = o_p(1)$.
        \item $ \left  \| \left \{-\ddot{\mathcal{M}}_{J}(\bbeta^*) - \alpha_J \right \}/\overline{\kappa}_J \right \|_2= o_p(1)$
        \item $\alpha_J \succeq \frac{1}{\rho_{\max}} \eta_J$
    \end{enumerate}
    Properties (1)-(3) ensure that
    \begin{align*}
        &\sigma_{\min} \left \{ \left (\Sigma_{J}/\overline{\kappa}_J \right )^{-1/2} \right \} \sigma_{\min} \left \{ \frac{\ddot{\mathcal{M}}_{J}(\bar{\bbeta}^{J}) -  \ddot{\mathcal{M}}_{J}(\bbeta^*)}{\overline{\kappa}_J} \right  \} \sigma_{\min} \left \{ \left (\Sigma_{J}/\overline{\kappa}_J \right )^{-1/2} \right \} = o_p(1), \ \mathrm{and} \\
        &\sigma_{\min} \left \{ \left (\Sigma_{J}/\overline{\kappa}_J \right )^{-1/2} \right \} \sigma_{\min} \left \{ \frac{-\ddot{\mathcal{M}}_{J}(\bbeta^*) - \alpha_J}{\overline{\kappa}_J} \right \} \sigma_{\min} \left \{ \left (\Sigma_{J}/\overline{\kappa}_J \right )^{-1/2} \right \} = o_p(1).
    \end{align*}
    This allows us to conclude that
    \begin{align*}
        \sigma_{\min} \left ( \Sigma_{J}^{-1/2} \left \{  -\ddot{\mathcal{M}}_{J}(\bar{\bbeta}^{J}) \right \} \Sigma_{J}^{-1/2} \right )^{-1} &\leq \left [ \frac{1}{\rho_{\max}}\sigma_{\min} \left \{ \Sigma_{J}^{-1/2} \eta_J  \Sigma_{J}^{-1/2} \right \} + o_p(1) + o_p(1) \right ]^{-1} \\
        &= O_p(1).
    \end{align*}
    
    \textbf{We begin with Property (1).} By Assumption~\ref{as:IA}, there exists $\epsilon > 0$ such that
	\begin{align*}
		\label{eq:info}
		& \sigma_{\min}(\eta_{J}/\overline{\kappa}_J)  \geq \epsilon + o_p(1)  \Rightarrow \\
        &\overline{\kappa}_J \|\eta_{J}^{-1}\|_2 \leq 1/\epsilon + o_p(1).
	\end{align*}
        By Assumption~\ref{as:stabalizedvariance}, $ \left \| \Sigma_{J}^{-1/2} \eta_J \Sigma_{J}^{-1/2} \right \|_2 = O_p(1)$.
	Consequently, by the sub-multiplicativity of the spectral norm,
	\begin{align*}
		\overline{\kappa}_J \left \| \Sigma_{J}^{-1} \right \|_2 &= \overline{\kappa}_J  \left \| \Sigma_{J}^{-1}\eta_J \eta^{-1}_J \right \|_2  \\ 
        &\leq\overline{\kappa}_J  \left \| \Sigma_{J}^{-1}\eta_J \right \|_2 \left  \| \eta^{-1}_J \right \|_2  \\ 
        &\leq  O_p(1) \left \{\frac{1}{\epsilon} + o_p(1) \right \}
	\end{align*}
    as $J \to \infty$.
    Note that the last equality follows because the eigenvalues of $\Sigma_{J}^{-1} \eta_{J}$ are the same as $\Sigma_{J}^{-1/2} \eta_{J} \Sigma_{J}^{-1/2}$. 
    
	\noindent \textbf{We now prove Property (2).} 
    We analyze
	\begin{align*}
		\left \| \frac{ \ddot{\mathcal{M}}_{J}(\bbeta^*) - \ddot{\mathcal{M}}_{J}(\bar{\bbeta}^{J})}{\overline{\kappa}_J} \right \|_2 &= \left \| \frac{ \sum_{j = 1}^J \sum_{i \in \epochgroup} W_i \left \{\ddot{l}_i(\bbeta^*) - \ddot{l}_i(\bar{\bbeta}^{J}) \right \}}{\overline{\kappa}_J} \right \|_2 \\
		&\leq \frac{ \sum_{j = 1}^J \sum_{i \in \epochgroup} W_i \|\ddot{l}_i(\bbeta^*) - \ddot{l}_i(\bar{\bbeta}^{J})\|_2}{\overline{\kappa}_J},
	\end{align*}
	where the last inequality follows from the triangle inequality. We wish to show that for any $\epsilon >0$,
	\begin{equation}
		\label{eq:equicontinuity}
		\lim_{J \to \infty} \bbP\left (\frac{ \sum_{j = 1}^J \sum_{i \in \epochgroup} W_i \|\ddot{l}_i(\bbeta^*) - \ddot{l}_i(\bar{\bbeta}^{J})\|_2}{\overline{\kappa}_J} > \epsilon \right ) = 0.
	\end{equation}
	Recall that Assumption~\ref{as:equicontinuity} implies there exists an $\epsilon_{\ddot{l}} > 0$ and a function $f:\mathscr{D} \to \mathbb{R}$ such that for all 
    $0 < \epsilon^* \leq \epsilon_{\ddot{l}}$, there exists $\delta_{\epsilon^*}$ such that
		\begin{equation*}
			\sup_{j \in \mathbb{N}, i \in \mathcal{E}_{j-1}, \ \bbeta\in \mathcal{B}: \|\bbeta - \bbeta^*\|_2 \leq \delta_{\epsilon^*}} \| \ddot{l}_i(\bbeta) - \ddot{l}_i(\bbeta^*) \|_2 \leq  f(\bD_i) \epsilon^* \ \ \mathrm{a.s.},
		\end{equation*}
	and  $\bbE_{\bbeta^*, \widetilde{\bpi}} \left \{ f^2(\bD_i) | \epochfield\right \}$ is bounded almost surely.
	Because $\bar{\bbeta}^{J}$ is consistent for $\bbeta^*$, we know that for any $\delta_{\epsilon^*}$, $\lim_{J \to \infty} \|\bar{\bbeta}^{J} - \bbeta^* \| \leq \delta_{\epsilon^*}$ a.s. Consequently,
	\begin{align*}
		&\lim_{J \to \infty} \bbP\left (\frac{ \sum_{j = 1}^J \sum_{i \in \epochgroup} W_i \|\ddot{l}(\bbeta^*) - \ddot{l}(\bar{\bbeta}^{J})\|_2}{\overline{\kappa}_J} > \epsilon \right ) \\
        &\leq \lim_{J \to \infty} \bbP\left ( \epsilon^* \frac{ \sum_{j = 1}^J \sum_{i \in \epochgroup} W_i f(\bD_i)}{\overline{\kappa}_J} > \epsilon \right ).
	\end{align*}
	We now analyze the quantity
	\begin{align*}
		\frac{1}{\overline{\kappa}_J} \sum_{j = 1}^J \sum_{i \in \epochgroup} W_i f(\bD_i)  &= \frac{1}{\overline{\kappa}_J} \sum_{j = 1}^J \sum_{i \in \epochgroup} W_i f(\bD_i) - \frac{1}{\overline{\kappa}_J} \sum_{j = 1}^J \sum_{i \in \epochgroup}\bbE_{\bbeta^*, \widehat{\bpi}} \left \{ W_i f(\bD_i) | \epochfield\right \} \\
		&+ \frac{1}{\overline{\kappa}_J} \sum_{j = 1}^J \sum_{i \in \epochgroup}\bbE_{\bbeta^*, \widehat{\bpi}} \left \{ W_i f(\bD_i) | \epochfield\right \}.
	\end{align*}
	We know that
	\begin{equation}
		\label{eq:sllnequicont}
		\frac{1}{\overline{\kappa}_J} \sum_{j = 1}^J \sum_{i \in \epochgroup} W_i f(\bD_i) - \frac{1}{\overline{\kappa}_J} \sum_{j = 1}^J \sum_{i \in \epochgroup}\bbE_{\bbeta^*, \widehat{\bpi}} \left \{ W_i f(\bD_i) | \epochfield\right \} = o_p(1)
	\end{equation}
	by Theorem~\ref{thm:branchingSLLN} because Assumption~\ref{as:equicontinuity} states that $\bbE_{\bbeta^*, \widetilde{\bpi}} \left \{ f^2(\bD_i) | \epochfield\right \}$ is bounded almost surely.
	Therefore,
	\begin{equation*}
		\frac{1}{\overline{\kappa}_J} \sum_{j = 1}^J \sum_{i \in \epochgroup} W_i f(\bD_i) = o_p(1) + \frac{1}{\overline{\kappa}_J} \sum_{j = 1}^J \sum_{i \in \epochgroup}\bbE_{\bbeta^*, \widehat{\bpi}} \left \{ W_i f(\bD_i) | \epochfield\right \}.
	\end{equation*}
	By Assumption~\ref{as:clip}, we know that $\rho_{\min} \leq W_i \leq \rho_{\max}$. Because $f(\bD_i)$ is a positive function, we can conclude that
	\begin{align*}
		\frac{1}{\overline{\kappa}_J} \sum_{j = 1}^J \sum_{i \in \epochgroup} W_i f(\bD_i) &\leq o_p(1) + \frac{1}{\rho_{\min}\overline{\kappa}_J} \sum_{j = 1}^J \sum_{i \in \epochgroup} \bbE_{\bbeta^*, \widehat{\bpi}} \left \{ W_i^2 f(\bD_i) | \epochfield\right \} \\
		&\leq o_p(1) + B_f/\rho_{\min}.
	\end{align*}
	Defining $\epsilon^* \triangleq \epsilon \rho_{\min}/(2B_f)$, we find that
	\begin{align*}
		&\lim_{J \to \infty} \bbP\left (\frac{ \sum_{j = 1}^J \sum_{i \in \epochgroup} W_i \|\ddot{l}(\bbeta^*) - \ddot{l}(\bar{\bbeta}^{J})\|_2}{\overline{\kappa}_J} > \epsilon \right ) \\
        &\leq 
		\lim_{J \to \infty} \bbP\left (\epsilon \frac{ \sum_{j = 1}^J \sum_{i \in \epochgroup} W_i f(\bD_i)}{\overline{\kappa}_J} > \epsilon \right ) \\
		&\leq \lim_{J \to \infty} \bbP\Big ( \epsilon \rho_{\min}/(2B_f) \\
        &\hspace{1.5cm}\left \{o_p(1) + B_f/\rho_{\min} \right \} > \epsilon \Big ) \\
		&\leq \lim_{J \to \infty} \bbP\left ( o_p(1)  > \epsilon/2 \right ) \to 0.
	\end{align*}
	Consequently, Property (2) is satisfied. 
	
	\noindent \textbf{We now prove Property (3).}
	We verify the following condition,
    \begin{equation*}
        \lim_{J \to \infty} \left \| \frac{\sum_{j = 1}^J \sum_{i \in \epochgroup} -W_i \ddot{l}(\bbeta^*) - \sum_{j = 1}^J \bbE_{\bbeta^*, \widehat{\bpi}} \left \{  \sum_{i \in \epochgroup} -W_i \ddot{l}(\bbeta^*) \mid \epochfield \right \} }{\overline{\kappa}_J} \right \|_2 = 0 \ \ \mathrm{a.s.}
    \end{equation*}
    By Assumption~\ref{as:moments} (the moments assumption) and Assumption~\ref{as:clip}, for any $p,q \in [k]$, there exists $B_{\ddot{l}^2}$ such that for any $j \in \mathbb{N}$ and $i \in \epochgroup$,
	\begin{align*}
		\bbE_{\bbeta^*, \widetilde{\bpi}} \left \{[\ddot{l}_i(\bbeta^*)]_{p,q}^2 | \epochfield\right \} &\leq B_{\ddot{l}^2} \ \ \mathrm{a.s.}
	\end{align*}
	Consequently, we know that 
	\begin{equation*}
		\left \| \frac{ \sum_{j = 1}^J \sum_{i \in \epochgroup} -W_i \ddot{l}(\bbeta^*) - \alpha_{J}}{\overline{\kappa}_J} \right \|  = o_p(1)
	\end{equation*}
    by Theorem~\ref{thm:branchingSLLN} (component-wise).

    \noindent \textbf{Lastly, we prove Property (4).} First, we find that
    \begin{align*}
        \alpha_J &\triangleq -\sum_{j = 1}^J \sum_{i \in \epochgroup} \bbE_{\bbeta^*, \widehat{\bpi}} \left \{ W_i \ddot{l}(\bbeta^*) \mid \epochfield \right \} \\
        &= \sum_{j = 1}^J \sum_{i \in \epochgroup}  \bbE_{\bbeta^*, \widehat{\bpi}} \left [  W_i \bbE_{\bbeta^*} \left \{ -\ddot{l}(\bbeta^*) \mid \bA_i, \epochfield \right \} \mid \epochfield \right ].
    \end{align*}
    We know $\bbE_{\bbeta^*} \left \{ -\ddot{l}(\bbeta^*) \mid \bA_i, \epochfield \right \} \succeq 0$ (because this is equivalent to the variance of the score function), and therefore we can express
    \begin{align*}
        &\sum_{j = 1}^J \sum_{i \in \epochgroup}  \bbE_{\bbeta^*, \widehat{\bpi}} \left [  W_i \bbE_{\bbeta^*} \left \{ -\ddot{l}(\bbeta^*) \mid \bA_i, \epochfield \right \} \mid \epochfield \right ] \\
        &\succeq \frac{1}{\rho_{\max}} \sum_{j = 1}^J \sum_{i \in \epochgroup}  \bbE_{\bbeta^*, \widehat{\bpi}} \left [  W^2_i \bbE_{\bbeta^*} \left \{ -\ddot{l}(\bbeta^*) \mid \bA_i, \epochfield \right \} \mid \epochfield \right ] \\
        &= \frac{1}{\rho_{\max}} \sum_{j = 1}^J \sum_{i \in \epochgroup}  \bbE_{\bbeta^*, \widetilde{\bpi}} \left [ \bbE_{\bbeta^*} \left \{ -\ddot{l}(\bbeta^*) \mid \bA_i, \epochfield \right \} \mid \epochfield \right ] \\
        &= -\frac{1}{\rho_{\max}} \sum_{j = 1}^J \sum_{i \in \epochgroup}  \bbE_{\bbeta^*, \widetilde{\bpi}} \left \{ \ddot{l}(\bbeta^*) \mid \epochfield \right \} \\
       & = \frac{1}{\rho_{\max}} \eta_J.
    \end{align*}
    This satisfies Property (4).
    We have verified Equation~\ref{eq:scale}.

    \noindent\rule{16cm}{0.4pt}

    \noindent Equation~\ref{eq:almostrate},
    \begin{equation*}
        \Sigma_{J}^{1/2}  \left ( \widehat{\bbeta}_{J} - \bbeta^* \right ) = O_p \left ( 1 \right ),
    \end{equation*}
    follows from Equations~\ref{eq:asconv} and \ref{eq:scale}.
    Equation~\ref{eq:almostrate} is used to conclude that
    \begin{equation*}
        \left ( \widehat{\bbeta}_{J} - \bbeta^* \right ) = O_p \left ( 1/\sqrt{\overline{\kappa}_{J} } \right )
    \end{equation*}
    by showing
    \begin{align*}
          \left \| \sqrt{\overline{\kappa}_{J} } \left ( \widehat{\bbeta}_{J} - \bbeta^* \right ) \right \|_2 &= \left \| \overline{\kappa}_{J}^{1/2} \Sigma_{J}^{-1/2} \Sigma_{J}^{1/2}   \left ( \widehat{\bbeta}_{J} - \bbeta^* \right ) \right \|_2 \\
          &\leq \left \| \overline{\kappa}_{J}^{1/2} \Sigma_{J}^{-1/2} \right \|_2 \left \| \Sigma_{J}^{1/2}   \left ( \widehat{\bbeta}_{J} - \bbeta^* \right ) \right \|_2 \\
          &= O_p(1).
    \end{align*}
    The last equality follows from Property (1) and Equation~\ref{eq:almostrate}.

	\begin{lemma}
		\label{lem:USLLN}
		Let $e_i(\bbeta) \triangleq W_il_i(\bbeta)$. Under Assumptions~\ref{as:1}-\ref{as:equicontinuity},
		\begin{equation*}
			\sup_{\bbeta \in \mathcal{B}} \left (\frac{1}{\overline{\kappa}_J}\sum_{j=1}^J \sum_{i \in \epochgroup} \left [ e_i(\bbeta) -  \bbE_{\bbeta^*, \widehat{\bpi}} \left \{ e_i(\bbeta) | \epochfield\right \} \right ] \right ) = o_p(1).
		\end{equation*} 
	\end{lemma}
	Lemma~\ref{lem:USLLN} is a martingale uniform law of large numbers, and its proof mirrors Lemma~2 of \cite{zhang2021statistical}. 
	\begin{proof}
		Like \cite{zhang2021statistical}, we begin by establishing a finite bracketing number for the log-likelihood function based on Assumption~\ref{as:lipchitz} and the fact that $\mathcal{B}$ is bounded.
		\paragraph{Finite Bracketing Number.} Let $\delta > 0$. We construct a set $B_\delta$ that is made up of pairs of functions $(b,u)$ and satisfies the following criteria.
		\begin{enumerate}
			\item We denote $l_i(\bbeta)$ as $l(\bbeta, \bD_i)$ to emphasize its dependence on data $\bD_i \in \mathscr{D}$. 
                We show that for any $\bbeta \in \mathcal{B}$, we can find functions $(b,u) \in B_\delta$ such that $b(\bD_i) \leq l(\bbeta, \bD_i) \leq u(\bD_i)$ and 
			\begin{equation*}
				\sup_{j \in \mathbb{N}, i \in \mathcal{E}_{j-1}} \bbE_{\bbeta^*, \widetilde{\bpi}} \left \{ u(\bD_i) - b(\bD_i) | \epochfield \right \} \leq \delta \ \ \mathrm{a.s.}
			\end{equation*}
			\item There are a finite number of pairs in $B_\delta$, $|B_\delta| < \infty$.
			\item For any $(b,u) \in B_\delta$, there exists $m_{g} < \infty$ that does not depend on $\delta$ such that
			\begin{align*}
				&\sup_{j \in \mathbb{N}, i \in \mathcal{E}_{j-1}} \bbE_{\bbeta^*, \widetilde{\bpi}} \left \{ b( \bD_i)^2 | \epochfield  \right \} \leq m_{g} \ \ \mathrm{a.s}, \ \mathrm{and} \\
				&\sup_{j \in \mathbb{N}, i \in \mathcal{E}_{j-1}} \bbE_{\bbeta^*, \widetilde{\bpi}} \left \{ u(\bD_i)^2 | \epochfield \right \} \leq m_{g} \ \ \mathrm{a.s}.
			\end{align*}
			
		\end{enumerate}
		We now construct $B_\delta$. Create a grid over $\mathcal{B}$ with a meshwidth of $\lambda > 0$, and let the points in this grid be the set $ G_{\lambda} \subseteq \mathcal{B}$. By construction, this implies that for any $\bbeta \in \mathcal{B}$, we can find $\bbeta' \in G_{\lambda}$ such that $\|\bbeta' - \bbeta \| \leq \lambda$.
		
		By Assumption~\ref{as:lipchitz}, we know that for any $\bbeta, \bbeta' \in \mathcal{B}$, $j \in \mathbb{N}$, and $i \in \mathcal{E}_{j-1}$, there exists function $g$ and constant $m_g < \infty$ such that $|l(\bbeta, \bD_i) - l(\bbeta', \bD_i)| \leq g(\bD_i) \|\bbeta - \bbeta'\|_2$ and
		\begin{equation}
			\label{eq:boundg}
			\bbE_{\bbeta^*, \widetilde{\bpi}} \left \{ g(\bD_i)^2 | \epochfield\right \} \leq m_g \ \ \mathrm{a.s.}
		\end{equation}
		We specify $B_\delta = \left \{(l(\bbeta, \bD_i) - \lambda g(\bD_i),  l(\bbeta, \bD_i) + \lambda g(\bD_i) ): \bbeta \in G_\lambda \right \}$ and show that this set satisfies properties (1)-(3) for a certain value of $\lambda$. Note that because $\mathcal{B}$ is bounded, property (2) is satisfied because the number of points in $G_{\lambda}$ is finite.
		To show that (1) holds for $B_\delta$, recall that for any $\bbeta \in \mathcal{B}$, we can find $\bbeta' \in G_\lambda$ such that $\|\bbeta - \bbeta'\| \leq \lambda$. Define $\bbeta_{\lambda}'(\bbeta) \triangleq \min_{\bbeta' \in G_{\lambda}} \left \| \bbeta - \bbeta' \right \|$.
        By Assumption~\ref{as:lipchitz}, this implies that 
		\begin{equation*}
			\left |l(\bbeta, \bD_i) - l\left \{ \bbeta_{\lambda}'(\bbeta), \bD_i \right \} \right | \leq g(\bD_i) \| \bbeta - \bbeta_{\lambda}'(\bbeta) \|_2 \leq g(\bD_i) \lambda.
		\end{equation*}
		Therefore,
		\begin{equation*}
			l \left \{ \bbeta_{\lambda}'(\bbeta), \bD_i \right \} - \lambda g(\bD_i) \leq l(\bbeta, \bD_i) \leq l \left \{ \bbeta_{\lambda}'(\bbeta), \bD_i \right \} + \lambda g(\bD_i),
		\end{equation*}
		implying that there exists $(b,u) \in B_\delta$ such that  $b(\bD_i) \leq l(\bbeta, \bD_i) \leq u(\bD_i)$ for all $\bD_i \in \mathcal{D}$.
		Additionally, note that for any  $j \in \mathbb{N}$ and $i \in \mathcal{E}_{j-1}$,
		\begin{align*}
			&\bbE_{\bbeta^*, \widetilde{\bpi}} \left [ \left |l \left \{ \bbeta_{\lambda}'(\bbeta), \bD_i \right \} + \lambda g(\bD_i) - \left [l \left \{ \bbeta_{\lambda}'(\bbeta), \bD_i \right \} - \lambda g(\bD_i) \right ] \right | \Big | \epochfield\right ] \\
			&= 2 \lambda \bbE_{\bbeta^*, \widetilde{\bpi}} \left ( |g(\bD_i)| |  \epochfield\right ) \leq 2\lambda \sqrt{m_g} \ \ \mathrm{a.s.}
		\end{align*}
		The inequality above holds by Jensen's inequality,
        \begin{equation*}
            \bbE_{\bbeta^*, \widetilde{\bpi}}  \left \{ |g(\bD_i)| | \epochfield\right \} \leq \sqrt{\bbE_{\bbeta^*, \widetilde{\bpi}} \left \{ g(\bD_i)^2 | \epochfield\right \}} \ \mathrm{a.s.,}
        \end{equation*}
         and Equation~\ref{eq:boundg}. We conclude that (1) holds for our choice of $B_\delta$ by letting the meshwidth $\lambda \triangleq \delta/(2\sqrt{m_g})$.
		
		Lastly, we show that (3) holds. Note that for any $\bbeta \in \mathcal{B}$,
		\begin{align*}
			&\sup_{j \in \mathbb{N}, i \in \mathcal{E}_{j-1}} \bbE_{\bbeta^*, \widetilde{\bpi}} \left [ \left \{l(\bbeta, \bD_i) + \lambda g(\bD_i) \right \}^2 | \epochfield \right ] \\
			&\leq 2 \sup_{j \in \mathbb{N}, i \in \mathcal{E}_{j-1}} \bbE_{\bbeta^*, \widetilde{\bpi}} \left \{ l(\bbeta, \bD_i)^2 | \epochfield \right \} + 
			2 \lambda^2 \sup_{j \in \mathbb{N}, i \in \mathcal{E}_{j-1}} \bbE_{\bbeta^*, \widetilde{\bpi}} \left \{  g(\bD_i)^2 | \epochfield\right \} \ \ \mathrm{a.s.}
		\end{align*}
		by the triangle inequality and the fact that for any $p,q \in \mathbb{R}$,
		\begin{align*}
			(p+q)^2 &= p^2 + 2pq + q^2 \\
			&= 2p^2 + 2q^2 -p^2+2pq-q^2 \\
			&= 2p^2 + 2q^2 - (p-q)^2 \\
			&\leq 2p^2 + 2q^2.
		\end{align*}
		By the same logic,
		\begin{align*}
			\sup_{j \in \mathbb{N}, i \in \mathcal{E}_{j-1}} \bbE_{\bbeta^*, \widetilde{\bpi}} \left [ \left \{l(\bbeta, \bD_i) - \lambda g(\bD_i) \right \}^2 | \epochfield\right ] &\leq
			2 \sup_{j \in \mathbb{N}, i \in \mathcal{E}_{j-1}} \bbE_{\bbeta^*, \widetilde{\bpi}} \left \{ l(\bbeta, \bD_i)^2 | \epochfield\right \} \\
			&+ 2 \lambda^2  \sup_{j \in \mathbb{N}, i \in \mathcal{E}_{j-1}} \bbE_{\bbeta^*, \widetilde{\bpi}} \left \{ g(\bD_i)^2 | \epochfield \right \} \ \ \mathrm{a.s.}
		\end{align*}
		Because $l(\bD_i, \bbeta) = l(\bD_i, \bbeta) - l(\bD_i, \bbeta^*) + l(\bD_i, \bbeta^*)$, we upper bound
		\begin{align*}
			&2 \sup_{j \in \mathbb{N}, i \in \mathcal{E}_{j-1}} \bbE_{\bbeta^*, \widetilde{\bpi}} \left \{ l(\bbeta, \bD_i)^2 | \epochfield\right \} + 
			2 \lambda^2 \sup_{j \in \mathbb{N}, i \in \mathcal{E}_{j-1}} \bbE_{\bbeta^*, \widetilde{\bpi}} \left \{  g(\bD_i)^2 | \epochfield \right \} \\
			&\leq 4 \sup_{j \in \mathbb{N}, i \in \mathcal{E}_{j-1}} \bbE_{\bbeta^*, \widetilde{\bpi}} \left \{ l(\bbeta^*, \bD_i)^2 | \epochfield\right \} +
			4 \sup_{j \in \mathbb{N}, i \in \mathcal{E}_{j-1}} \bbE_{\bbeta^*, \widetilde{\bpi}} \left [ \left \{ l(\bbeta, \bD_i) - l(\bbeta^*, \bD_i) \right \}^2 | \epochfield \right ] \\
			&\hspace{5.8cm}+ 2 \lambda^2 \sup_{j \in \mathbb{N}, i \in \mathcal{E}_{j-1}} \bbE_{\bbeta^*, \widetilde{\bpi}} \left \{ g(\bD_i)^2 | \epochfield\right \} \ \ \mathrm{a.s.}
		\end{align*}
		Note that $\sup_{ j \in \mathbb{N}, i \in \mathcal{E}_{j-1}} \bbE_{\bbeta^*, \widetilde{\bpi}} \left \{  g(\bD_i)^2 |\epochfield\right \} \leq m_g$ a.s., and $\bbE_{\bbeta^*, \widetilde{\bpi}} \left \{ l(\bbeta^*, \bD_i)^2 |\epochfield\right \}$ is bounded by Assumption~\ref{as:moments}. 
		
		By Assumption~\ref{as:lipchitz}, for any $\bbeta \in \mathcal{B}$, 
		\begin{align*}
			&|l(\bbeta, \bD_i) - l(\bbeta^*, \bD_i)| \leq g(\bD_i)\|\bbeta - \bbeta^* \|_2 \Rightarrow \\
			& \left \{l(\bbeta, \bD_i) - l(\bbeta^*, \bD_i) \right \}^2 \leq g(\bD_i)^2 \|\bbeta - \bbeta^* \|_2^2 \Rightarrow \\
			&\bbE_{\bbeta^*, \widetilde{\bpi}} \left [ \left \{l(\bbeta, \bD_i) - l(\bbeta^*, \bD_i) \right \}^2 | \epochfield\right ] \leq \bbE_{\bbeta^*, \widetilde{\bpi}} \left \{ g(\bD_i)^2 | \epochfield\right \}\|\bbeta - \bbeta^* \|_2^2 \ \ \mathrm{a.s.}
		\end{align*}
		$\|\bbeta - \bbeta^* \|_2^2$ is bounded because  $\mathcal{B}$ is bounded, and $\sup_{ j \in \mathbb{N}, i \in \mathcal{E}_{j-1}} \bbE_{\bbeta^*, \widetilde{\bpi}} \left \{  g(\bD_i)^2 | \epochfield\right \} \leq m_g$ a.s. We conclude that property (3) holds.
		We now use $B_\delta$ in the main argument for the proof of uniform convergence.
		
		\paragraph{Main Argument:} We now show that for any $\epsilon> 0$, 
		\begin{equation}
			\label{eq:unsllngoal}
			\bbP_{\bbeta^*, \widehat{\bpi}} \left \{ \sup_{\bbeta \in \mathcal{B}} \left (\frac{1}{\overline{\kappa}_J} \sum_{j=1}^J \sum_{i \in \epochgroup} \left [ e_i(\bbeta) -  \bbE_{\bbeta^*, \widehat{\bpi}} \left \{ e_i(\bbeta) | \epochfield\right \} \right ] \right ) > \epsilon \right \} = o(1).
		\end{equation}
		Let $\delta > 0$ (we will choose $\delta$ later). Let $B_\delta$ be the set of pairs of functions as constructed earlier. Note that by property (1) of $B_\delta$, we get the following upper bound,
		\begin{align*}
			&\sup_{\bbeta \in \mathcal{B}} \left (\frac{1}{\overline{\kappa}_J} \sum_{j=1}^J \sum_{i \in \epochgroup} \left [ W_i l_i(\bbeta) -  \bbE_{\bbeta^*, \widehat{\bpi}} \left \{ W_i l_i(\bbeta) | \epochfield\right \} \right ] \right ) \\
			&\leq 
			\max_{(b, u) \in B_\delta} \left (\frac{1}{\overline{\kappa}_J} \sum_{j=1}^J \sum_{i \in \epochgroup} \left [ W_i u(\bD_i) -  \bbE_{\bbeta^*, \widehat{\bpi}} \left \{ W_i b(\bD_i) | \epochfield\right \} \right ] \right ) \ \ \mathrm{a.s.} \ \ (*)
		\end{align*}
		By subtracting and adding $\bbE_{\bbeta^*, \widehat{\bpi}}\left \{ W_i u(\bD_i) | \epochfield\right \}$, and using the triangle inequality, we find that
		\begin{align*}
			(*) \ \leq &\max_{(b, u) \in B_\delta} \left [\frac{1}{\overline{\kappa}_J} \sum_{j=1}^J \sum_{i \in \epochgroup}  \bbE_{\bbeta^*, \widehat{\bpi}} \left [W_i \left \{ u(\bD_i) -   b(\bD_i) \right \} | \epochfield\right ] \right ] + \\
			&\max_{(b, u) \in B_\delta} \left (\frac{1}{\overline{\kappa}_J} \sum_{j=1}^J \sum_{i \in \epochgroup} \left [ W_i u(\bD_i) -  \bbE_{\bbeta^*, \widehat{\bpi}} \left \{ W_i u(\bD_i) | \epochfield\right \} \right ] \right ) \ \ \mathrm{a.s.}
		\end{align*}
		By Assumption~\ref{as:clip} and the fact that $u(\bD_i) - b(\bD_i) \geq 0$,
		\begin{align*}
			\bbE_{\bbeta^*, \widehat{\bpi}} \left [W_i \left \{u(\bD_i) -   b(\bD_i) | \epochfield\right \} \right] &\leq \frac{1}{\rho_{\min}} \bbE_{\bbeta^*, \widehat{\bpi}} \left [W_i^2 \left \{u(\bD_i) -   b(\bD_i) | \epochfield\right \} \right] \\
			&= \frac{1}{\rho_{\min}} \bbE_{\bbeta^*, \widetilde{\bpi}} \left \{u(\bD_i) -   b(\bD_i) | \epochfield\right \}  \\
			&\leq \frac{1}{\rho_{\min}} \delta.
		\end{align*}
		The last inequality holds by property (1) of $B_\delta$. Because $\max_{i \in  [n]} a_i \leq \sum_{i=1}^n |a_i|$,
		\begin{align*}
			&\sup_{\bbeta \in \mathcal{B}} \left (\frac{1}{\overline{\kappa}_J} \sum_{j=1}^J \sum_{i \in \epochgroup} \left [ W_i l_i(\bbeta) -  \bbE_{\bbeta^*, \widehat{\bpi}} \left \{ W_i l_i(\bbeta) | \epochfield\right \} \right ] \right ) \\
			&\leq \frac{1}{\rho_{\min}} \delta + \sum_{(b, u) \in B_\delta} \left | \frac{1}{\overline{\kappa}_J} \sum_{j=1}^J \sum_{i \in \epochgroup} \left [ W_i u(\bD_i) -  \bbE_{\bbeta^*, \widehat{\bpi}} \left \{ W_i u(\bD_i) | \epochfield\right \} \right ] \right |.
		\end{align*}
		By property (3) of $B_\delta$, we know that 
        \begin{equation*}
            \bbE_{\bbeta^*, \widehat{\bpi}} \left \{ W_i^2 u(\bD_i)^2 | \epochfield\right \} = \bbE_{\bbeta^*, \widetilde{\bpi}} \left \{ u(\bD_i)^2 | \epochfield\right \} \leq m_g.
        \end{equation*}
		Therefore, by Theorem~\ref{thm:branchingSLLN}, for any $(b, u) \in B_\delta$,
		\begin{equation*}
			\left |\frac{1}{\overline{\kappa}_J} \sum_{j=1}^J \sum_{i \in \epochgroup} \left [ W_i u(\bD_i) -  \bbE_{\bbeta^*, \widehat{\bpi}} \left \{ W_i u(\bD_i) | \epochfield\right \} \right ] \right | = o_p(1).
		\end{equation*}
		Because $|B_\delta| < \infty$ by Property (2), the convergence holds for all $(b, u) \in B_\delta$ simultaneously, so
		\begin{equation*}
			(*) \leq \frac{1}{\rho_{\min}} \delta + o_p(1).
		\end{equation*}
		Equation~\ref{eq:unsllngoal} is satisfied by choosing $\delta = \rho_{\min} \epsilon/2$.
	\end{proof}
	
	\section{Proof of Theorem~\ref{thm:budg}}
	\label{app_sec:regret}
	
	We establish asymptotic regret bounds for RL-RDS by proving Theorem~\ref{thm:budg}.
    The following consequence of Markov's inequality will be useful.
	
	\begin{lemma}
		\label{lem:op}
		Let $\{X_n\}_{n\geq1}$ be a sequence of random variables where $X_n \in \mathcal{X} \subset \mathbb{R}^p$ and $||\cdot||$ is an arbitrary norm on $\mathbb{R}^p$. If $\mathbb{E}||X_n|| = O(1) $ for all $n \geq 1$, then $\left \{ X_n \right \}_{n \geq 1} = O_p(1)$.
	\end{lemma}
	
	\begin{proof}
		If $\mathbb{E}||X_n|| = O(1)$ for all $n \geq 1$, $\exists M>0$ such that $\forall n \geq 1$,
		$$\mathbb{E}||X_n|| \leq M.$$
		We need to show that given $\epsilon$, $\exists V_\epsilon$ such that $P(||X_n|| \geq V_\epsilon) \le \epsilon$ for all $n \geq 1$. We know that
		$$\mathbb{P}(||X_n|| \geq  V_\epsilon) \leq \frac{\bbE||X_n||}{V_\epsilon} \leq \frac{M}{V_\epsilon} $$
		by the Markov Inequality. Choosing $V_\epsilon = M/\epsilon$ gives the desired result.
	\end{proof}

	In Assumption~\ref{as:budg}, we define $S^n$ as the budget left when individual $n$ is recruited, where $S^n$ satisfies $0 < S^n - \sum_{v=1}^n C^v < C^*$ for $C^* \in \mathbb{R}^+$ and each $n \in \mathbb{N}$.  We define an upper bound on the
	potential number of additional recruits under allocation
	strategy $\pmb{\pi}$ at any state $\bhistory$ for $J \in \mathbb{N}$ as 
	\begin{equation*}
		n^*(\pmb{\pi}, \bh^k) \triangleq \inf
		\left\lbrace
		n\,:\, \sum_{v=k}^{n}
		C^{v*}(\pmb{\pi}, \bh^k)
		> C^*
		\right\rbrace.
	\end{equation*}
	We assume that there exists $S^* \in \mathbb{N}$ such that for any $\bhistory \in \bHistory$,
	$\sup_{\bpi \in \Pi} n^*(\bpi, \bhistory) \le S^*$ a.s.

	\begin{proof}[Proof of Theorem \ref{thm:budg}]
		\begin{align*}
			\branchingvalue(\bhistory, \bpi; \bbeta^*) &= \branchingvalue(\bhistory,\bpi; \widehat{\bbeta}_{J}) + \branchingvalue(\bhistory,\bpi; \bbeta^*) - \branchingvalue(\bhistory,\bpi; \widehat{\bbeta}_{J}) \\
			&\leq  \branchingvalue(\bhistory, \widehat{\bpi}_J; \widehat{\bbeta}_{J}) + \branchingvalue(\bhistory,\bpi; \bbeta^*) - \branchingvalue(\bhistory, \bpi; \widehat{\bbeta}_{J}) \\
			& = \branchingvalue(\bhistory, \widehat{\bpi}_J; {\bbeta^*}) + \branchingvalue(\bhistory, \widehat{\bpi}_J; \widehat{\bbeta}_{J}) - \branchingvalue(\bhistory, \widehat{\bpi}_J; \bbeta^*) + \\
                & \hspace{0.6cm} \branchingvalue(\bhistory,\bpi; \bbeta^*) - \branchingvalue(\bhistory, \bpi; \widehat{\bbeta}_{J}) \\
			&\leq \branchingvalue(\bhistory, \widehat{\bpi}_J; \bbeta^*) + 2\sup_{\bpi \in \Pi} \left |\branchingvalue(\bhistory,\bpi; \bbeta^*) - \branchingvalue(\bhistory, \bpi; \widehat{\bbeta}_{J}) \right |,
		\end{align*}
		where the first inequality follows from the definition of $\widehat{\bpi}_J$. The upper bound above does not depend on $\bpi$, and so it holds for $\bpi^{\mathrm{opt}}$,
		\begin{equation*}
			\branchingvalue(\bhistory, \bpi^{\mathrm{opt}}; \bbeta^*) -  \branchingvalue(\bhistory, \widehat{\bpi}_J; \bbeta^*)  \leq  2\sup_{\bpi \in \Pi} |\branchingvalue(\bhistory,\bpi; \bbeta^*) - \branchingvalue(\bhistory, \bpi; \widehat{\bbeta}_{J})|.
		\end{equation*}
		We know that $\branchingvalue(\bhistory, \bpi^\mathrm{opt}; \bbeta^*) -  \branchingvalue(\bhistory, \widehat{\bpi}_J; \bbeta^*) \geq 0$ because $\bpi^\mathrm{opt} \in \arg \max_{\bpi \in \Pi} \branchingvalue(\bhistory, \bpi; \bbeta^*) $, so 
		\begin{equation*}
			|\branchingvalue(\bhistory, \bpi^\mathrm{opt}; \bbeta^*) -  \branchingvalue(\bhistory, \widehat{\bpi}_J; \bbeta^* )|  \leq  2\sup_{\bpi \in \Pi} |\branchingvalue(\bhistory,\bpi; \bbeta^*) - \branchingvalue(\bhistory, \bpi; \widehat{\bbeta}_{J})|.
		\end{equation*}
		We now show that $\sup_{\bpi \in \Pi} |\branchingvalue(\bhistory,\bpi; \bbeta^*) - \branchingvalue(\bhistory, \bpi; \widehat{\bbeta}_{J})| = O_p(1/\sqrt{\overline{\kappa}_J})$, which proves the desired result. We expand $\branchingvalue(\bhistory, \bpi; \bbeta^*) -  \branchingvalue(\bhistory, \bpi; \widehat{\bbeta}_{J})$ as follows. Recall that $\Delta^{v}$ is an indicator that the budget has not been exceeded when individual $v$ enters the study. Note that $\exists S^*$ such that we can express
		\begin{align}
			\begin{split}
				&\branchingvalue(\bhistory, \bpi; \bbeta^*) -  \branchingvalue(\bhistory, \bpi; \widehat{\bbeta}_{J}) = \\
				&\int\left ( \sum_{k=1}^{S^*} \Delta^k y^k \right ) \left \{ \prod_{m=1}^{S^*} f(\bh^{m+1} | \bh^{m}, \ba^m; \bbeta^*) \pi(\ba^m | \bh^m) \right \} d\lambda(\bh^{S^*}) - \\ &\int\left ( \sum_{k=1}^{S^*} \Delta^k y^k \right ) \left \{ \prod_{m=1}^{S^*} f(\bh^{m+1} | \bh^{m}, \ba^m; \widehat{\bbeta}_{J}) \pi(\ba^m | \bh^m) \right \} d\lambda(\bh^{S^*}),
				\label{eq:Int}
			\end{split}  
		\end{align}
		where $\bh^1 = \bhistory$.
		Refactoring this expression and applying the definition of the log-likelihood results in
		\begin{align*}
			&\branchingvalue(\bhistory, \bpi; \bbeta^*) -  \branchingvalue(\bhistory, \bpi; \widehat{\bbeta}_{J}) = \\
			& \int\left ( \sum_{k=1}^{S^*} \Delta^k y^k \right ) \left \{ 1- \frac{\prod_{m=1}^{S^*} f(\bh^{m+1} | \bh^{m}, \ba^m; \widehat{\bbeta}_{J}) \pi(\ba^m | \bh^m) } {\prod_{m=1}^{S^*} f(\bh^{m+1} | \bh^{m}, \ba^m; \bbeta^*) \pi(\ba^m | \bh^m) } \right \} \times \\
			&\prod_{m=1}^{S^*} f(\bh^{m+1} | \bh^{m}, \ba^m; \bbeta^*) \pi(\ba^m | \bh^m)  d\lambda(\bh^{S^*}) = \\
			&\int\left ( \sum_{k=1}^{S^*} \Delta^k y^k \right ) \left [ 1- \exp \left \{ \ell_{S^*} (\widehat{\bbeta}_{J}) -  \ell_{S^*}(\bbeta^*) \right \} \right ] \times \\
                &\hspace{0.5cm}\left \{ \prod_{m=1}^{S^*} f(\bh^{m+1} | \bh^{m}, \ba^m; \bbeta^*) \pi(\ba^m | \bh^m) \right \} d\lambda(\bh^{S^*}).
		\end{align*}
		By Assumption~\ref{as:lipchitz2}, we know that 
		\begin{align*}
			\sqrt{\overline{\kappa}_J}\left |\ell_{S^*} (\widehat{\bbeta}_{J}) -  \ell_{S^*}(\bbeta^*) \right |
			& \leq \sqrt{\overline{\kappa}_J} \sum_{v=1}^{S^*} |q^v(\widehat{\bbeta}_{J}) - q^v(\bbeta^*) | \\
			&\leq \left [ \sum_{v=1}^{S^*} |e \left \{ (R^v, T^v, \bX^v, Y^v , \bA^v) \right \} | \right ] \left \| \sqrt{\overline{\kappa}_J}  \left (\widehat{\bbeta}_{J} - \bbeta^* \right ) \right \|_2.
		\end{align*}
		We know $\sqrt{\overline{\kappa}_J}(\widehat{\bbeta}_{J} -  \bbeta^*) = O_p(1)$ by Theorem~\ref{thm:asympNorm}. Labeling $e^v \triangleq e \left \{ (R^v, T^v, \bX^v, Y^v , \bA^v) \right \} $ and noting that $e^v$ is strictly positive, we also know that for any $\bpi \in \Pi$, 
        \begin{equation}
            \label{eq:o1bound}
            \bbE_{\bbeta^*, \bpi}  \left (| e^v | \right ) = O(1)
        \end{equation}
        by Jensen's Inequality. 
		By Lemma~\ref{lem:op}, this implies that $\sum_{v=1}^{S^*} \left | e^v \right | = O_p(1)$. We conclude that 
		\begin{equation*}
			\sqrt{\overline{\kappa}_J}\left |\ell_{{S^*}} (\widehat{\bbeta}_{J}) -  \ell_{{S^*}}(\bbeta^*) \right | = O_p(1) \Rightarrow 
			\left |\ell_{{S^*}} (\widehat{\bbeta}_{J}) -  \ell_{{S^*}}(\bbeta^*) \right | = O_p(1/\sqrt{\overline{\kappa}_J}).
		\end{equation*}
            By Taylor expansion, 
		\begin{equation}
                \label{OpLim}
			1- \exp \left \{ \ell_{S^*} (\widehat{\bbeta}_{J}) -  \ell_{S^*}(\bbeta^*) \right \} = \ell_{S^*} (\widehat{\bbeta}_{J}) -  \ell_{S^*}(\bbeta^*) + o_p\left ( 1/\sqrt{\overline{\kappa}_J} \right ).
		\end{equation}
		Because $Y^{v}  \in [0,1]$ for all $v \in \mathbb{N}$,
		\begin{equation*}
			\left |\sum_{k=1}^{S^*} Y^{k} \Delta^k \right | = \sum_{k=1}^{S^*} |Y^{k} \Delta^k| \leq S^*,
		\end{equation*}
		and
		\begin{align}
			\begin{split}
				|(\ref{eq:Int})|\leq  & S^* \times \left |\ell_{S^*} (\widehat{\bbeta}_{J}) -  \ell_{S^*}(\bbeta^*) + o_p\left (  1/\sqrt{\overline{\kappa}_J}   \right ) \right| \\ & \times \left \{ \prod_{m=1}^{S^*} f(\bh^{m+1} | \bh^{m}, \ba^m; \bbeta^*) \pi(\ba^m | \bh^m) \right \} d\lambda(\bh^{S^*}).
				\label{RewLim}
			\end{split}
		\end{align}
		Leveraging equations (\ref{eq:o1bound}), (\ref{OpLim}), and (\ref{RewLim}) as well as the Cauchy-Schwartz Inequality, we know that as $J \to \infty$
		\begin{align*}
			|(\ref{eq:Int})|\leq  & \ {S^*} \times \frac{1}{\sqrt{\overline{\kappa}_J}} \left \|\sqrt{\overline{\kappa}_J} \left \{ \widehat{\bbeta}_{J} - \bbeta^* \right \}\right \|_2  \left \{ \sum_{v=1}^{S^*} | e^v | \right \}  \\ & \times \left \{ \prod_{m=1}^{{S^*}} f(\bh^{m+1} | \bh^{m}, \ba^m; \bbeta^*) \pi(\ba^m | \bh^m) \right \} d\lambda(\bh^{{S^*}}) + \\ 
            & o_p\left ( 1/\sqrt{\overline{\kappa}_J} \right ) \times \frac{1}{\sqrt{\overline{\kappa}_J}} \left \|\sqrt{\overline{\kappa}_J} \left \{ \widehat{\bbeta}_{J} - \bbeta^* \right \}\right \|_2  \left \{ \sum_{v=1}^{S^*} | e^v | \right \} \\
            & \times \left \{ \prod_{m=1}^{{S^*}} f(\bh^{m+1} | \bh^{m}, \ba^m; \bbeta^*) \pi(\ba^m | \bh^m) \right \} d\lambda(\bh^{{S^*}}) \\ 
            \leq & {S^*} \times \frac{1}{\sqrt{\overline{\kappa}_J}} \left \|\sqrt{\overline{\kappa}_J} \left \{ \widehat{\bbeta}_{J} - \bbeta^* \right \}\right \|_{2}  \bbE_{\bbeta^*, \widetilde{\bpi}}  \left \{ \sum_{v=1}^{S^*} | e^v | \right \} +  \\
            & o_p\left ( 1/\sqrt{\overline{\kappa}_J} \right )  \times \frac{1}{\sqrt{\overline{\kappa}_J}} \left \|\sqrt{\overline{\kappa}_J} \left \{ \widehat{\bbeta}_{J} - \bbeta^* \right \}\right \|_{2}  \bbE_{\bbeta^*, \widetilde{\bpi}}  \left \{ \sum_{v=1}^{S^*} | e^v | \right \}.
		\end{align*}
		This follows because $\sqrt{\overline{\kappa}_J}(\widehat{\bbeta}_{J} -  \bbeta^*) = O_p(1)$, and $\bbE_{\bbeta^*, \widetilde{\bpi}}  \left ( \sum_{v=1}^{S^*} |e^v| \right ) = O(1)$. \\ Therefore, $|\branchingvalue(\bhistory, \bpi; \bbeta^*) -  \branchingvalue(\bhistory, \bpi; \widehat{\bbeta}_{J})| = O_p(1/\sqrt{\overline{\kappa}_J})$ as $J \to \infty$. This proves the result.
	\end{proof}

\section{Epoch Asymptotics}

As mentioned in Section~\ref{sec:regretbounds} in the main text,  we observe the branching process data in order of arrival time, $\mathcal{D}^{\noepochsamplesize} \triangleq \left\lbrace \left (R^v, T^v, \bX^v, Y^v, \bA^v, C^v \right )\right\rbrace_{v=1}^{\kappa}$. Define $J_{\noepochsamplesize}$ as the last complete epoch induced by coupon expiration, $J_{\kappa} \triangleq \max\{j: \forall i \in \mathcal{E}_j, \ T_i + t_{\max} < T^{\noepochsamplesize} \}$. Lemma~\ref{lem:generation_asymptotics_2} ensures that $\lim_{\kappa \to \infty} J_{\noepochsamplesize} \to \infty$ a.s. under Assumption~\ref{as:generation_asymptotics}, implying that inference based on the last complete epoch will achieve the asymptotic guarantees of Theorem~\ref{thm:rlbranching}.

\begin{lemma}
    \label{lem:generation_asymptotics_2}
    Assume that the working model specified in Equation~\ref{eq:rdsmod} is the true generative model. Under Assumption~\ref{as:generation_asymptotics},
    \begin{equation*}
        J_{\noepochsamplesize} \to \infty \ \mathrm{a.s.}
    \end{equation*}
\end{lemma}
\begin{proof}
First, we show that 
    \begin{equation*}
        \lim_{ \noepochsamplesize \to \infty} T^{\noepochsamplesize} = \infty.
    \end{equation*}
    Define $V_S$ as the seed sample set. Under Assumption~\ref{as:generation_asymptotics}, the number of coupons in a coupon allocation is upper bounded by $L_*$. It is sufficient to show that for any $g \in \mathbb{N}$, $T^{|V_S|L_*^{g-1} + 1} \geq gt_{\min}$.
    
    We employ a proof by induction. First, we show the base case. Suppose $g=1$. We know that $T^{|V_S| + 1} \geq t_{\min}$ by Assumption~\ref{as:generation_asymptotics}. We now assume that for any $g \in \mathbb{N}$, 
    \begin{equation*}
        T^{|V_S|*L_*^{g-1} + 1} \geq gt_{\min}.
    \end{equation*}
    Note that $t_{\min} > 0$.
    Consequently, an upper bound on the number of active coupons in the branching process at time $T^{|V_S|*L_*^{g-1}}$ is
    \begin{equation*}
        |V_S|L_*^{g-1}*L_*= |V_S|L_*^{g}.
    \end{equation*}
    Note that any study participant recruited by a new member of the sample (an individual recruited after the first $|V_S|L_*^{g-1}$ study participants) will have a recruitment time greater than $(g+1)t_{\min}$; i.e., for any $\noepochsamplesize$ such that $R^{\noepochsamplesize} > |V_S|L_*^{g-1}$, $T^{\noepochsamplesize} \geq (g+1)t_{\min}$.
    Because $|V_S|L_*^{g} - |V_S|L_*^{g-1}$ is an upper bound on the number of active coupons at time $T^{|V_S|L_*^{g-1}}$, we know that
    $ R^{|V_S|L_*^{g}}  \geq |V_S|L_*^{g-1}$. Consequently,
    \begin{equation*}
        T^{|V_S|L_*^{g} + 1} > (g+1)t_{\min}.
    \end{equation*}
    We conclude that
    \begin{equation*}
        \lim_{ \noepochsamplesize \to \infty} T^{\noepochsamplesize} = \infty
    \end{equation*}
    by induction. Now we establish that 
    \begin{equation*}
        \lim_{\noepochsamplesize \to \infty} J_{\noepochsamplesize} = \infty.
    \end{equation*}
  Recall that $\overline{\kappa}_J \triangleq |\overline{\mathcal{E}}_J|$. 
  Additionally, recall that $|\bA^v| \leq L_*$ for all $v \in N$. 
  Consequently, for all $j \in \mathbb{N}$, $ \left | \overline{\mathcal{E}}_j \right | \leq L_*^j < \infty$. 
  This implies that $\max_{i \in \overline{\mathcal{E}}_j} T_i \leq t_{\max} L_*^j$. Consequently, for any $j \in \mathbb{N}$, 
  \begin{equation*}
      1 = P \left (\lim_{\noepochsamplesize \to \infty} \left \{ t_{\max}L_*^j + t_{\max} < T^{\noepochsamplesize} \right \} \right )  \leq P \left (\lim_{\noepochsamplesize \to \infty} \left \{ \forall i \in \overline{\mathcal{E}}_j, \ T_i + t_{\max} < T^{\noepochsamplesize} \right \} \right ).
  \end{equation*}
   We conclude that $\lim_{\noepochsamplesize \to \infty} J_{\noepochsamplesize} = \infty$. 
\end{proof}

\section{Proof of Theorem~\ref{thm:rlbranching}}

In this section, we prove Theorem~\ref{thm:rlbranching} for the branching model described in Equation~\ref{eq:rdsmod}. 
Under Conditions (C1)-(C6), we verify Assumptions~\ref{as:clip}, \ref{as:generation_asymptotics}, and \ref{as:differentiable}-\ref{as:equicontinuity}. This will allow us to invoke Theorems~\ref{thm:asympNorm} and \ref{thm:budg} for the branching model described by Equation~\ref{eq:rdsmod}.

\paragraph{General notation for this section.} 
For matrix $A \in \mathbb{R}^{n \times n}$, $\sigma_{\min}(A)$ and $\sigma_{\max}(A)$ are the minimum and maximum singular values of $A$ respectively. Define $\mathrm{vec}(A)$ as the vectorization of $A$. For a symmetric matrix $X \in \mathbb{R}^{d\times d}$, define $\mathrm{vech}(X)$ as the vectorization of the lower-triangular elements of $X$. 
Label $I_d \in \mathbb{R}^{d \times d}$ as the $d$-dimensional identity matrix.

\subsection{The Branching Models Derivatives}
\label{app_sec:likelihood}

Before we begin verifying assumptions, it will be useful to calculate the derivatives of the log-likelihood of the branching process specified in Equation~\ref{eq:rdsmod}. First, we define the data 
\begin{equation*}
    \bD_i =\left  (\underline{\bY}_i, \underline{\bX}_i, \underline{\bT}_i, M_i, \bA_i \right ).
\end{equation*}
For example, $\bX_i$ is the covariate vector of recruiter $i$, and $\underline{\bX}_i = \left ( \bX_{i,1}, \bX_{i,2}, \ldots,\bX_{i,M_i} \right )$ are the recruits of recruiter $i$.
The complete
likelihood for the branching process is
\begin{align}
    \label{eq:rdsmod_likelihood_first}
    \begin{split}
	&\mathcal{L}_{J}( \bbeta) \triangleq \\
	& \mathcal{L}_{J}(\bbeta_y)  \mathcal{L}_{J}( \left \{\bphi_{\ba}, G_{a}, \Sigma_{a} \right \}_{a \in \mathbb{A}}) \mathcal{L}_{J}(\zeta , \lambda) \triangleq \\
	&\prod_{j =1}^J  \prod_{i \in \epochgroup} \prod_{l=1}^{M_i} \left [ \frac{1}{1 + \exp(-{\bZ_{i,l}}^\top\bbeta_y)} \right ]^{Y_{i,l}} \left [ \frac{1}{1 + \exp({\bZ_{i,l}}^\top \bbeta_y)} \right ]^{1-Y_{i,l}} \times \\
	&\prod_{j =1}^J  \prod_{i \in \epochgroup} \prod_{l=1}^{M_i} (2\pi)^{-p/2} |\sum_{a \in \mathbb{A}}\Sigma_{a} \mathbb{I}\left (A_{i,l} = a \right )|^{-1/2} \times \\
        &\hspace{0.5cm} \exp \left [ -\frac{1}{2} \sum_{a \in \mathbb{A}} \left \{  (\bX_{i,l} - \bphi_{a} - G_{a}\bX_i)^\top \Sigma^{-1}_{a}(\bX_{i,l} - \bphi_{a} - G_{a}\bX_i) \right \} \mathbb{I}(A_{i,l} = a) \right ] \\
	&\prod_{j =1}^J  \prod_{i \in \epochgroup} M_i! \left [ \frac{ \prod_{l=1}^{M_i} \zeta e^{-\zeta U_{i,l}} } {[e^{-\zeta t_{\min}}-e^{-\zeta t_{\max}}]^{M_i} } \right ] \frac{
		\lambda^{M_i}/M_i!
	}{
		\sum_{\ell=k}^{|\mathbf{A}^v|} (\lambda^\ell/\ell!)
	}.
    \end{split}
\end{align}
Lastly, we rewrite the logarithm of the part of the likelihood that involves $\bbeta_t = \zeta$ and $\bbeta_m = \lambda$ as
\begin{align}
\label{eq:time}
    \begin{split}
    \ell_{J}(\zeta, \lambda) \triangleq &\sum_{j = 1}^J \sum_{i \in \epochgroup} \sum_{l=1}^{M_i}  \left \{  \log ( \zeta ) - \zeta U_{i,l} - \log \left (e^{-\zeta t_{\min}}-e^{-\zeta t_{\max}}\right ) \right \} + \\
    &\sum_{j = 1}^J \sum_{i \in \epochgroup}\left [ M_i \log (\lambda)  - \log \{ \sum_{\ell=k}^{|\mathbf{A}_i|} (\lambda^\ell/\ell!)  \} \right ].
    \end{split}
\end{align}

We will need the hessian of the log-likelihood of this branching process for the proofs that follow. Recall that $\mathbb{A}$ is the set of possible coupon types. We note that each coupon allocation is a set of identical coupons, implying that the sets $\mathcal{A}$ and $\mathbb{A}$ have a one-to-one correspondence. Consequently, for $\ba \in \mathcal{A}$, there exists $a \in \mathbb{A}$ such that
\begin{equation*}
    \bbI (\bA_i = \ba) = \bbI (A_{i,l} = a)
\end{equation*}
for every $l \in \left \{1,2, \ldots, M_i \right \}$. Consequently, we can represent the complete branching process likelihood as
\begin{align}
    \label{eq:rdsmod_likelihood}
    \begin{split}
	&\mathcal{L}_{J}( \bbeta) \triangleq \\
	& \mathcal{L}_{J}(\bbeta_y)  \mathcal{L}_{J}( \left \{\bphi_{\ba}, G_{\ba}, \Sigma_{\ba} \right \}_{\ba \in \mathcal{A}}) \mathcal{L}_{J}(\zeta , \lambda) \triangleq \\
	&\prod_{j =1}^J  \prod_{i \in \epochgroup} \prod_{l=1}^{M_i} \left [ \frac{1}{1 + \exp(-{\bZ_{i,l}}^\top\bbeta_y)} \right ]^{Y_{i,l}} \left [ \frac{1}{1 + \exp({\bZ_{i,l}}^\top \bbeta_y)} \right ]^{1-Y_{i,l}} \times \\
	&\prod_{j =1}^J  \prod_{i \in \epochgroup} \prod_{l=1}^{M_i} (2\pi)^{-p/2} |\sum_{\ba \in \mathbb{A}}\Sigma_{\ba} \mathbb{I}(\bA_{i} = \ba)|^{-1/2} \times \\
        &\hspace{0.5cm} \exp \left [ -\frac{1}{2} \sum_{\ba \in \mathbb{A}} \left \{  (\bX_{i,l} - \bphi_{\ba} - G_{\ba}\bX_i)^\top \Sigma^{-1}_{\ba}(\bX_{i,l} - \bphi_{\ba} - G_{\ba}\bX_i) \right \} \mathbb{I}(\bA_{i} = \ba) \right ] \\
	&\prod_{j =1}^J  \prod_{i \in \epochgroup} M_i! \left [ \frac{ \prod_{l=1}^{M_i} \zeta e^{-\zeta U_{i,l}} } {[e^{-\zeta t_{\min}}-e^{-\zeta t_{\max}}]^{M_i} } \right ] \frac{
		\lambda^{M_i}/M_i!
	}{
		\sum_{\ell=k}^{|\mathbf{A}^v|} (\lambda^\ell/\ell!)
	}.
    \end{split}
\end{align}
We will use this likelihood for the proofs that follow.

\paragraph{The Hessian for the log-likelihood of the covariate model.}
Define $\bX_i^* = \left (1, \bX_i \right )$ and $G^{\dagger}_{\ba}= \left (\phi_{\ba}^\top, G_{\ba}^\top \right )^\top$. The likelihood of the covariate model parameter, \\ $\bbeta_\bx = \left \{G^{\dagger}_{\ba}, \Sigma_{\ba} \right \}_{\ba \in \mathcal{A}}$, is
\begin{align*}
    \mathcal{L}_{J} \left ( \left \{G^{\dagger}_{\ba}, \Sigma_{\ba} \right \}_{\ba \in \mathcal{A}} \right ) &\triangleq \prod_{j =1}^J  \prod_{i \in \epochgroup} \prod_{l=1}^{M_i} (2\pi)^{-p/2} \left |\sum_{\ba \in \mathcal{A}}\Sigma_{\ba} \mathbb{I}(c = \ba) \right |^{-1/2} \times \\
    &\hspace{0.5cm} \exp \left [ -\frac{1}{2} \sum_{\ba \in \mathcal{A}} \left \{  (\bX_{i,l} -  G^{\dagger}_{\ba}\bX^*_i)^\top \Sigma^{-1}_{\ba}(\bX_{i,l} - G^{\dagger}_{\ba}\bX^*_i) \right \} \mathbb{I}(\bA_i = \ba) \right ] \\
    &= \prod_{j =1}^J  \prod_{i \in \epochgroup} \prod_{l=1}^{M_i}  (2\pi)^{-p/2} \left |\sum_{\ba \in \mathcal{A}}\Sigma_{\ba} \mathbb{I}(\bA_i = \ba) \right |^{-1/2} \times \\
    &\hspace{0.5cm} \exp \Bigg [ -\frac{1}{2} \sum_{\ba \in \mathcal{A}} \Big \{  {\bX_{i,l}}^\top \Sigma^{-1}_{\ba} \bX_{i,l} - 2 {\bX^*_i}^\top {G^{\dagger}_{\ba}}^\top  \Sigma^{-1}_{\ba}\bX_{i,l}  + \\
    &\hspace{1.5cm} {\bX^*_i}^\top {G^{\dagger}_{\ba}}^\top  \Sigma^{-1}_{\ba} G^{\dagger}_{\ba} \bX^*_i  \Big \} \mathbb{I}(\bA_i = \ba) \Bigg ].
\end{align*}
The equality above simply follows from distributing.
For $\ba \in \mathcal{A}$, we reparameterize $\Omega_{\ba} = - \frac{1}{2} \Sigma_{\ba}^{-1}$ and $\Gamma_{\ba} = \Sigma_{\ba}^{-1}G^{\dagger}_{\ba}$.
\begin{align*}
    \mathcal{L}_{J} \left ( \left \{\Gamma_{\ba}, \Omega_{\ba} \right \}_{\ba \in \mathcal{A}} \right ) &= \prod_{j =1}^J  \prod_{i \in \epochgroup} \prod_{l=1}^{M_i}  (2\pi)^{-p/2} |\sum_{\ba \in \mathcal{A}}-2\Omega_{\ba} \mathbb{I}(\bA_i = \ba)|^{1/2} \times \\
    &\hspace{0.5cm} \exp \Bigg [ \sum_{\ba \in \mathcal{A}} \Big \{  {\bX_{i,l}}^\top \Omega_{\ba} \bX_{i,l} + {\bX^*_i}^\top \Gamma_{\ba}^\top\bX_{i,l}  + \\
    &\hspace{1.5cm} \frac{1}{4}{\bX^*_i}^\top \Gamma_{\ba}^\top {\Omega_{\ba}}^{-1} \Gamma_{\ba} \bX^*_i  \Big \} \mathbb{I}(\bA_i = \ba) \Bigg ] .
\end{align*}
The reparameterized log-likelihood is
\begin{align*}
    \mathcal{\ell}_{J} \left ( \left \{\Gamma_{\ba}, \Omega_{\ba} \right \}_{\ba \in \mathcal{A}} \right ) &\triangleq \sum_{j = 1}^J \sum_{i \in \epochgroup} \sum_{l=1}^{M_i}  \sum_{\ba \in \mathcal{A}} \Bigg [ (-p/2)\log(2\pi) + \frac{1}{2}\log|-2\Omega_{\ba}| + \\
    &\hspace{2cm} {\bX_{i,l}}^\top \Omega_{\ba} \bX_{i,l} + {\bX^*_i}^\top \Gamma_{\ba}^\top\bX_{i,l}  + \\
    &\hspace{2cm} \frac{1}{4}{\bX^*_i}^\top \Gamma_{\ba}^\top {\Omega_{\ba}}^{-1} \Gamma_{\ba} \bX^*_i \Bigg ] \mathbb{I}(\bA_i = \ba)\\
    &= \sum_{j = 1}^J \sum_{i \in \epochgroup} \sum_{l=1}^{M_i} \sum_{\ba \in \mathcal{A}} \Bigg [ (-p/2)\log(2\pi) + \frac{1}{2}\log|-2\Omega_{\ba}| + \\
    &\hspace{2cm} \mathrm{tr} \left ( \Omega_{\ba} \bX_{i,l} {\bX_{i,l}}^\top \right ) + \mathrm{tr} \left ( \Gamma_{\ba}^\top\bX_{i,l} {\bX^*_i}^\top \right ) + \\
    &\hspace{2cm} \mathrm{tr} \left (\frac{1}{4} \Gamma_{\ba}^\top {\Omega_{\ba}}^{-1} \Gamma_{\ba} \bX^*_i {\bX^*_i}^\top \right )  \Bigg ] \mathbb{I}(\bA_i = \ba),
\end{align*}
where the second equality follows from rearranging terms and using the properties of the trace operator.
For $\ba \in \mathcal{A}$, define $n_{\ba} \triangleq \sum_{j = 1}^J \sum_{i \in \epochgroup} M_i \mathbb{I}(\bA_i = \ba)$ and $\bV^J_{\ba}\in \mathbb{R}^{(p+1)\times (p+1)}$ such that $\bV^J_{\ba} \triangleq \sum_{j = 1}^J \sum_{i \in \epochgroup} M_i \bX^*_i {\bX^*_i}^\top  \mathbb{I}(\bA_i = \ba) $. We apply the differential operator two times and find
\begin{align*}
    \bd^2\mathcal{\ell}_{J} \left ( \left \{\Gamma_{\ba}, \Omega_{\ba} \right \}_{\ba \in \mathcal{A}} \right ) &= \sum_{j = 1}^J \sum_{i \in \epochgroup} \sum_{l=1}^{M_i} \sum_{\ba \in \mathcal{A}} -\bd^2 \Bigg \{ -\frac{1}{2}\log|-2\Omega_{\ba}| - \\
      &\hspace{2.7cm} \frac{1}{4} \mathrm{tr} \left ( \Gamma_{\ba}^\top {\Omega_{\ba}}^{-1} \Gamma_{\ba} \bX^*_i {\bX^*_i}^\top \right ) \Bigg \} \mathbb{I}(\bA_i = \ba) \\
      &= -\bd \Bigg \{ -\frac{n_{\ba}}{2}\mathrm{tr} \left (\Omega_{\ba}^{-1} \bd \Omega_{\ba} \right ) - \frac{1}{4} \mathrm{tr} \left (2 \Gamma_{\ba}^\top \Omega_{\ba}^{-1} \bd \Gamma_{\ba} \bV^J_{\ba}\right ) + \\
      &\hspace{1.5cm} \frac{1}{4} \mathrm{tr} \left ( \Gamma_{\ba}^\top {\Omega_{\ba}}^{-1} \bd \Omega_{\ba} {\Omega_{\ba}}^{-1} \Gamma_{\ba} \bV^J_{\ba}\right ) \Bigg \}  \\
      &= -\Bigg \{ \frac{n_{\ba}}{2}\mathrm{tr} \left (\Omega_{\ba}^{-1} \bd \Omega_{\ba} \Omega_{\ba}^{-1} \bd \Omega_{\ba} \right ) - \\
      &\hspace{1.2cm}\frac{1}{2} \mathrm{tr} \left ( \Gamma_{\ba}^\top \Omega_{\ba}^{-1} \bd \Omega_{\ba} \Omega_{\ba}^{-1} \bd \Omega_{\ba} \Omega_{\ba}^{-1} \Gamma_{\ba} \bV^J_{\ba}\right ) - \\ 
      &\hspace{1.2cm}\frac{1}{2} \mathrm{tr} \left ( \bd \Gamma_{\ba}^\top \Omega_{\ba}^{-1} \bd \Gamma_{\ba} \bV^J_{\ba}\right ) + \\
      &\hspace{1.2cm} \mathrm{tr} \left ( \Gamma_{\ba}^\top \Omega_{\ba}^{-1} \bd \Omega_{\ba} \Omega_{\ba}^{-1} \bd \Gamma_{\ba} \bV^J_{\ba}\right ) +  \\
      &\hspace{1.2cm} \frac{1}{4} \mathrm{tr} \left ( \Gamma_{\ba}^\top {\Omega_{\ba}}^{-1} \bd \Omega_{\ba} {\Omega_{\ba}}^{-1} \Gamma_{\ba} \bV^J_{\ba}\right ) \Bigg \}.
\end{align*}
We observe that for $\ba, \ba' \in \mathcal{A}$ and $\ba \neq \ba'$,  
\begin{equation*}
    \frac{\partial \mathcal{\ell}_{J} \left ( \left \{\Gamma_{\ba}, \Omega_{\ba} \right \}_{\ba \in \mathcal{A}} \right )} { \partial \left (  \mathrm{vec}(\Gamma_{\ba}),  \mathrm{vec}(\Omega_{\ba}) \right )\partial \left (  \mathrm{vec}(\Gamma_{\ba'}),  \mathrm{vec}(\Omega_{\ba'}) \right )^{\top}} = [0]_{(p+1)^2 \times p^2}. 
\end{equation*}
We express
\begin{align*}
    &\ddot{\ell}_{J} \left ( \left \{\Gamma_{\ba}, \Omega_{\ba} \right \}_{\ba \in \mathcal{A}} \right )= \\
    &-\mathrm{diag} \left [\left \{\begin{pmatrix}
        \bV^J_{\ba}\otimes \Sigma_{\ba} & 2\left ( \bV^J_{\ba}{G^*_{\ba}}^\top  \otimes \Sigma_{\ba} \right ) \\
        2 \left ( G^*_{\ba}\bV^J_{\ba}\otimes \Sigma_{\ba} \right ) & 4 \left ( G^*_{\ba} \bV^J_{\ba}{G^*_{\ba}}^\top  \otimes \Sigma_{\ba} \right ) + 2 n_{\ba} \left ( \Sigma_{\ba} \otimes \Sigma_{\ba} \right )
    \end{pmatrix} \right \}_{\ba \in \mathcal{A}}
    \right ] \\
    &= -\mathrm{diag} \left [\left \{ \begin{pmatrix}
        \bV^J_{\ba}& 2 \bV^J_{\ba}{G^*_{\ba}}^\top \\
        2 G^*_{\ba}\bV^J_{\ba} & 4 G^*_{\ba} \bV^J_{\ba}{G^*_{\ba}}^\top   + 2 n_{\ba}\Sigma_{\ba}
    \end{pmatrix} \otimes \Sigma_{\ba} \right \}_{\ba \in \mathcal{A}}
    \right ].
\end{align*}
The second equality follows from properties of Kronecker products.
\paragraph{The Hessian for the log-likelihood of the reward model.}
The estimating equation for the reward model parameter, $\bbeta_y$, is 
\begin{equation*}
	\frac{\partial \ell_{J} (\bbeta_y)}{\partial \bbeta_y} = \sum_{j = 1}^J \sum_{i \in \epochgroup} \sum_{l=1}^{M_i}  \left [ Y_{i,l}  \bZ_{i,l} -  \left \{\frac{1}{1+\exp(-{\bZ_{i,l}}^\top \bbeta_y)} * \bZ_{i,l} \right \} \right ].
\end{equation*}
The derivative of the estimating equation (the hessian of the log-likelihood) is
\begin{align*}
    \frac{\partial^2 \ell_{J} (\bbeta_y)} {\partial\bbeta_y \partial \bbeta_y^\top} = - \sum_{j = 1}^J \sum_{i \in \epochgroup} \sum_{l=1}^{M_i} \bZ_{i,l} {\bZ_{i,l}}^\top \left \{ \frac{1}{1 + \exp(-{\bZ_{i,l}}^\top\bbeta_y)} \right \} \left \{\frac{1}{1 + \exp \left ({\bZ_{i,l}}^\top \bbeta_y \right )} \right \}.
\end{align*}

\paragraph{The Hessian for the log-likelihood of the arrival model.}
Appealing to Equation~\ref{eq:time}, the estimating equation for the arrival model parameter, $\bbeta_t = \zeta$, is 
\begin{align*}
	&\frac{\partial \ell_{J}(\zeta)}{\partial \zeta} = \sum_{j = 1}^J \sum_{i \in \epochgroup} \sum_{l=1}^{M_i}   \frac{1}{\zeta} - U_{i,l} - \frac{t_{\min}e^{-\zeta t_{\min}} -  t_{\max}e^{-\zeta t_{\max}}}{e^{-\zeta t_{\min}}-e^{-\zeta t_{\max}}}.
\end{align*}
The derivative of the estimating equation (the hessian of the log-likelihood) is
\begin{align*}
    &\frac{\partial^2 \ell_{J}(\zeta)}{\partial\zeta ^2}  =  \sum_{j = 1}^J \sum_{i \in \epochgroup}  -M_i\left ( \frac{1}{\zeta^2}  + \frac{t_{\min}e^{-\zeta t_{\min}} - t_{\max}^2 e^{-\zeta t_{\max}}}{ [e^{-\zeta t_{\min}} -e^{-\zeta t_{\max}}]^2 } \right ).
\end{align*}

\paragraph{The Hessian for the log-likelihood of the family model.}
We reparameterize the log-likelihood in Equation~\ref{eq:time}, where $\tau = \log(\lambda)$:
\begin{equation*}
     \ell_{J}(\tau) = \sum_{j = 1}^J \sum_{i \in \epochgroup}\left [ M_i \tau - \log \{ \sum_{\ell=k}^{|\mathbf{A}_i|} (e^{\tau \ell}/\ell!)  \} \right ].
\end{equation*}
The score function is
\begin{equation*}
    \frac{\partial \ell_{J}(\tau)}{\partial \tau} = \sum_{j = 1}^J \sum_{i \in \epochgroup}\left [ M_i - \tau\frac{\sum_{\ell=k}^{|\mathbf{A}_i|} \ell e^{\tau \ell}/\ell!)}{\sum_{\ell=k}^{|\mathbf{A}_i|} e^{\tau \ell}/\ell!}   \right ].
\end{equation*}
This makes the Hessian of the log-likelihood
\begin{align*}
    &\frac{\partial^2 \ell_{J}(\tau)}{\partial \tau^2} = \sum_{j = 1}^J \sum_{i \in \epochgroup} - \tau^2 \frac{ \left \{ \sum_{\ell=k}^{|\mathbf{A}_i|} e^{\tau \ell}/\ell! \right \} \left \{  \sum_{\ell=k}^{|\mathbf{A}_i|} \ell^2 e^{\tau \ell}/\ell! \right \} -  \left \{  \sum_{\ell=k}^{|\mathbf{A}_i|} \ell e^{\tau \ell}/\ell! \right \}^2 }{\left \{ \sum_{\ell=k}^{|\mathbf{A}_i|} e^{\tau \ell}/\ell! \right \}^2}.
\end{align*}

\subsection{Verification of Assumptions~\ref{as:clip}-\ref{as:generation_asymptotics} and \ref{as:differentiable}-\ref{as:moments}}
\label{app_sec:branchinginference}

\subsubsection{Assumption~\ref{as:clip}}

\textbf{Assumption~\ref{as:clip}} is satisfied by Thompson sampling with a clipping constraint. We show this in Lemma~\ref{lem:clip_weights}.
\begin{lemma}
    \label{lem:clip_weights}
    For any $j \in \mathbb{N}$ and $i \in \epochgroup$, assume $\bA_i$ is assigned by RL-RDS with a clipping constraint at level $\epsilon_1 \in (0,1)$ (as specified in Section~\ref{sec:rl}). Additionally, assume Conditions (C1)-(C6). There exists $\epsilon_{\min}, \epsilon_{\max} > 0$ such that for all $\ba \in \mathcal{A}$,
    \begin{equation*}
        \epsilon_{\min} \leq \bbP(\bA_i = \ba | \bH_i ) \leq \epsilon_{\max}\ \ \mathrm{w.p. \ 1}.
    \end{equation*} 
    This implies that for $\rho_{\min} = \sqrt{\epsilon_{\min}}$ and $\rho_{\max} = \sqrt{\epsilon_{\max}}$,
    \begin{equation*}
        \rho_{\min} \leq \bW_i \leq \rho_{\max} \ \ \mathrm{w.p. \ 1}.
    \end{equation*}
    \begin{proof}
        Note that by Condition (C3), the set of possible coupon allocations is constant for the duration of the study.
        For any $j \in \mathbb{N}$ and $i \in \epochgroup$, Thompson sampling with clipping at level $\epsilon_1 \in (0,1)$
        selects allocation $\ba \in \mathcal{A}$
        with probability
        \begin{equation*}
            \bbP \left (  \bA_i = \ba \right ) = \frac{p_i^{\ba}}{ \sum_{\ba' \in \mathcal{A}} p_i^{\ba'}},
        \end{equation*}
        where $p^{\ba}_i \triangleq \min  \left [1-\epsilon_1, \max \left \{\epsilon_1,\widehat{\xi}^{i}_{B} (\bH_i, \ba) \right \} \right ]$ (refer to Section~\ref{sec:rl} for the definition of $\widehat{\xi}^{i}_{B}$). This implies that
        \begin{equation*}
            \bbP \left (  \bA_i = \ba \right ) = \frac{p_i^{\ba}}{ \sum_{\ba' \in \mathcal{A}} p_i^{\ba'}} \geq \frac{\epsilon_1}{(1 - \epsilon_1) |\mathcal{A}|}.
        \end{equation*}
        Consequently, for any $\ba \in \mathcal{A}$,
        \begin{equation*}
            \frac{ 1- \epsilon_1}{1 + (|\mathcal{A}| -2)\epsilon_1 } \geq \bbP \left (  \bA_i = \ba \right ). 
        \end{equation*}
            The inequality follows from the fact that $ 1-\epsilon_1 \geq p_i^{\ba} \geq \epsilon_1$.
        We define
        \begin{equation*}
            \epsilon_{\min} \triangleq \epsilon_1/ \left \{ (1-\epsilon_1) |\mathcal{A}| \right \}, \ \ \epsilon_{\max} \triangleq \frac{ 1- \epsilon_1}{1 + (|\mathcal{A}| -2)\epsilon_1 }.
        \end{equation*}
    \end{proof}
\end{lemma}

\subsubsection{Assumption~\ref{as:generation_asymptotics}}

\textbf{Assumption~\ref{as:generation_asymptotics}} is satisfied by Equation~\ref{eq:branching_converge_2} (in Section~\ref{sec:regretbounds} of the main text) and Condition~(C1). See the verification of Assumption~\ref{as:stabalizedvariance} for a further discussion of Equation~\ref{eq:branching_converge_2}.

\subsubsection{Assumptions~\ref{as:differentiable} and \ref{as:moments}}

\textbf{Assumption~\ref{as:differentiable}} is verified by the derivatives calculated in Section~\ref{app_sec:likelihood} and the fact that each component of this branching model is a full rank exponential family. Consequently, the natural parameters are identifiable \citep{brown1986fundamentals}. \textbf{Assumption~\ref{as:moments}} follows from the score and Hessian functions calculated in Section~\ref{app_sec:likelihood}. Observe that the likelihood and Hessian components are continuous and finite in the data for a given parameter value. Consequently, the fact that the data are bounded implies that for any $j \in \mathbb{N}$ and $i \in \epochgroup$, $\dot{l}_i(\bbeta)$ and  $\ddot{l}_i(\bbeta)$ are bounded (by the extreme value theorem). This implies that the conditional expectations of these functions are bounded and $\textbf{Assumption~\ref{as:moments}}$ is satisfied.

\subsection{Verification of Assumptions~\ref{as:lipchitz} and \ref{as:stabalizedvariance}-\ref{as:lipchitz2}}
We now verify Assumptions~\ref{as:lipchitz} and \ref{as:stabalizedvariance}-\ref{as:lipchitz2}.

\subsubsection{Note on Assumption~\ref{as:wellseperated}}
Note that Assumption~\ref{as:wellseperated} is only used in the proof of consistency in Section~\ref{app_sec:convrate}. It is unnecessary because the log-likelihood of the branching process specified in Equation~\ref{eq:rdsmod} is concave. We will provide another proof of consistency in Section~\ref{app_sec:branch_consistency} that does not use Assumption~\ref{as:wellseperated}.

\subsubsection{Supporting Lemmas}

Lemmas~\ref{lem:mineigen} and \ref{lem:submult} are useful properties of positive semi-definite matrices.
	
	\begin{lemma}
		\label{lem:mineigen}
		Define positive semi-definite matrices $A,B \in \mathbb{R}^{d\times d}$. It follows that
		\begin{equation*}
			\sigma_{\min}(A + B) \geq \sigma_{\min}(A) + \sigma_{\min}(B).
		\end{equation*}
	\end{lemma}
	\begin{proof}
		\begin{align*}
			\sigma_{\min}(A + B) &= \min_{x\neq 0} \frac{x^\top(A + B) x}{x^\top x} \\
			&= \min_{x\neq 0} \left ( \frac{x^\top A x}{x^\top x} + \frac{x^\top B x}{x^\top x} \right ) \\
			&\geq \min_{x\neq 0} \frac{x^\top A x}{x^\top x} + \min_{x\neq 0}  \frac{x^\top B x}{x^\top x} \\
			&= \sigma_{\min}(A) + \sigma_{\min}(B).
		\end{align*}
        The first line follows from the definition of the minimum singular value for positive semi-definite matrices, and the second from distributing. For any functions $f$ and $g$, $\min_{x \in \mathcal{X}} \left \{f(x) + g(x) \right \} \geq \min_{x \in \mathcal{X}} f(x) +  \min_{x \in \mathcal{X}} g(x)$. Consequently, line 3 is true. Line 4 applies the definiton of the minimum singular value again.
	\end{proof}
	
	\begin{lemma}
		\label{lem:submult}
		Let $A,B \in \mathbb{R}^{d \times d}$ be positive semi-definite matrices. It follows that
		\begin{equation*}
			\sigma_{\min}(A)\sigma_{\min}(B) \leq \sigma_{\min}(AB).
		\end{equation*}
		
		\begin{proof}
			First, we assume that martices $A$ and $B$ are invertible.
                By the sub-multiplicativity of the spectral norm, we know that
			\begin{equation*}
				\sigma_{\max}(AB) \leq \sigma_{\max}(A)\sigma_{\max}(B).
			\end{equation*}
			Consequently,
			\begin{align*}
				\sigma_{\max}\{ (AB)^{-1} \} \leq \sigma_{\max}(A^{-1})\sigma_{\max}(B^{-1}) \Rightarrow \\
				\sigma_{\min}( AB )^{-1} \leq \sigma_{\min}(A)^{-1}\sigma_{\min}(B)^{-1} \Rightarrow \\
				\sigma_{\min}(A)\sigma_{\min}(B) \leq \sigma_{\min}(AB).
			\end{align*}
			The first line applies the sub-multiplicativity of the spectral norm. For any matrix $A$, $\sigma_{\max}(A^{-1}) = \sigma_{\min}(A)^{-1}$. Consequently, line 2 is true. 

            If matrices $A$ and $B$ are not invertible, then $\sigma_{\min}(A) = 0 $, $\sigma_{\min}(B) = 0$, and $\sigma_{\min}(AB) = 0 $. Therefore, the result follows trivially.
		\end{proof}
	\end{lemma}
	
	\begin{lemma}\label{lem:blckpd}
		Let $M \in \mathbb{R}^{d \times d}$ be a matrix such that
		\begin{align*}
			M = \begin{pmatrix}
				A & B \\
				B^\top & C
			\end{pmatrix}
		\end{align*}
		where $A \in \mathbb{R}^{d_1 \times d_1}$ and $C \in \mathbb{R}^{d_2 \times d_2}$ are symmetric and $d_1 + d_2 = d$.
		$M$ is positive definite if and only if $A$ is positive definite and $C - B^\top A^{-1} B$ is positive definite. Additionally, 
		\begin{equation*}
			\sigma_{\min}(M) \geq \sigma_{\min}(A) \sigma_{\min}(C - B^\top A^{-1} B).
		\end{equation*}
	\end{lemma}
	\begin{proof}
		Define $I_{d_1} \in \mathbb{R}^{d_1 \times d_1}$ and $I_{d_2} \in \mathbb{R}^{d_2 \times d_2}$ as the $d_1$ and $d_2$-dimensional identity matrices respectively.
		We express $M$ as
		\begin{align*}
			M &= \begin{pmatrix}
				A & B \\
				B^\top & C
			\end{pmatrix} \\
			&= \begin{pmatrix}
				I_{d_1} & 0\\
				B^\top A^{-1} & I_{d_2}
			\end{pmatrix}
			\begin{pmatrix}
				A & 0 \\
				0 & C - B^\top A^{-1} B
			\end{pmatrix}
			\begin{pmatrix}
				I_{d_1} & A^{-1}B \\
				0 & I_{d_2}
			\end{pmatrix}.
		\end{align*}
		Because
		\begin{align*}
			\begin{pmatrix}
				I_{d_1} & 0\\
				B^\top A^{-1} & I_{d_2}
			\end{pmatrix}
		\end{align*}
		and its transpose are invertible, we know that $M$ is positive definite if and only if
		\begin{align*}
			Q= \begin{pmatrix}
				A & 0 \\
				0 & C - B^\top A^{-1} B
			   \end{pmatrix}
		\end{align*}
		is positive definite. The matrix $Q$ is block diagonal, so it is positive definite if and only if
		$A$ and $C - B^\top A^{-1}B$ are positive definite. Additionally, by Lemma~\ref{lem:submult},
		\begin{align*}
			\sigma_{\min} (M) &\geq \sigma_{\min} \left \{ \begin{pmatrix} I_{d_1} & 0\\
				B^\top A^{-1} & I_{d_2}
			\end{pmatrix} \right \}
			\sigma_{\min} \left \{ 
			\begin{pmatrix}
				A & 0 \\
				0 & C - B^\top A^{-1} B
			\end{pmatrix}
			\right \} \times \\
			&\hspace{0.5cm}\sigma_{\min} \left \{
			\begin{pmatrix}
				I_{d_1} & A^{-1}B \\
				0 & I_{d_2}
			\end{pmatrix}
			\right \} \\
			&\geq \sigma_{\min}(A) \sigma_{\min} \left ( C - B^\top A^{-1} B \right ) \sigma_{\min} \left ( I_{d_1} \right )^2 \sigma_{\min} \left ( I_{d_2} \right )^2 \\
                &= \sigma_{\min}(A) \sigma_{\min} \left ( C - B^\top A^{-1} B  \right ) 
		\end{align*}
	\end{proof}

\subsubsection{Assumptions~\ref{as:lipchitz} and \ref{as:lipchitz2}}

We verify \textbf{Assumption~\ref{as:lipchitz}} first. We write $\dot{l}_i(\bbeta^*) = \dot{l}(\bbeta^*, \bD_i) $ as a function $\dot{l}:\mathcal{B}\times \mathscr{D} \to \mathbb{R}$. From Section~\ref{app_sec:branchinginference}, we know that $\dot{l}$ is continuous over both $\mathcal{B}$ and $\mathscr{D}$. Because $\mathcal{X}$ is compact (and the other data types are bounded by definition), we know that $\mathscr{D}$ is compact. By the extreme value theorem, we can define $\gamma < \infty$,
\begin{equation*}
    \max_{\bd \in \mathscr{D}} \max_{\bbeta \in \mathcal{B}} \left  \| \dot{l}(\bbeta, \bd) \right \|_2 \leq \gamma.
\end{equation*}
$\mathcal{B}$ is a convex subset of $\mathbb{R}^q$. Consequently, by the mean value theorem, we know that for all $j \in \mathbb{N}$, $i \in \epochgroup$, and $\bbeta, \bbeta' \in \mathcal{B}$,
\begin{equation*}
    \left |l_i(\bbeta) - l_i(\bbeta') \right | \leq |\dot{l}_i(\bar{\bbeta})| \|\bbeta - \bbeta'\|_2,
\end{equation*}
for some $\bar{\bbeta} =  t\bbeta + (1-t)\bbeta'$ where $t \in [0,1]$. Because $|\dot{l}_i(\bar{\bbeta})| = |\dot{l}(\bar{\bbeta}, \bD_i)| \leq \gamma$, we know that
\begin{equation*}
    \left  |l_i(\bbeta) - l_i(\bbeta') \right | \leq \gamma \|\bbeta - \bbeta'\|_2
\end{equation*}
and \textbf{Assumptions~\ref{as:lipchitz}} is satisfied.

The likelihood of the arrival process implied by data $\mathcal{D}^{\noepochsamplesize} = \left \{ (R^v, T^v, \bX^v, Y^v , \bA^v) \right \}_{v=1}^{\noepochsamplesize}$ is simply a censored (or integrated) version of the complete generation likelihood of Section~\ref{app_sec:branchinginference}. 
Consequently, write $\dot{q}^v(\bbeta) = \dot{q}(\bbeta, \bD^v)$ (the components of this likelihood as defined in Section~\ref{sec:rl}) as a function $\dot{q}:\mathcal{B}\times \mathscr{D} \to \mathbb{R}$, we know that $\dot{q}$ is continuous over both $\mathcal{B}$ and $\mathscr{D}$. Therefore, \textbf{Assumption~\ref{as:lipchitz2}} is satisfied by the logic presented in the previous paragraph.

\subsubsection{Assumption~\ref{as:stabalizedvariance}}
In this section, we verify \textbf{Assumption~\ref{as:stabalizedvariance}}. To do this, we introduce the branching process in a manner consistent with \cite{delmas2010detection}. Most of the following setup is taken directly from their paper; we reproduce it here for the reader's convenience. 
We modify the background and theory where necessary to generalize their proofs to more than two possible recruits (or cells in their case).

Let $\mathbb{G}_0 = \emptyset$, $\bbG_j  = \{1,  \ldots, L \}^j$ for $j \in \mathbb{N}$, $\mathbb{T}_r = \cup_{0\leq j \leq r} \mathbb{G}_j$. The set $\bbG_j$ contains all possible recruits in the $j$-th generation. Note that for participant $i \in \bbG_j$, $j = |i|$ denotes the generation of $i$. 
Additionally, for $i,j \in \bbT$, we can represent the concatenation of their positions as $ij = (i, j)$.
For recruit $i \in \mathbb{T}$, we still denote $\bX_i \in \mathbb{R}^p$ as an individual's covariates as specified in the model described by Equation~\ref{eq:rdsmod}. We define the transition kernel implied by this covariate process over the branching tree as follows. 
Let $(E, \mathcal{E})$ be a measurable space, and let $P$ be a probability kernel on $E \times \mathcal{E}^L$ with values in $[0,1]$. This implies that $P(\cdot, A)$ is measurable for all $A \in \mathcal{E}^L$ and $P(\bx, \cdot )$ is a probability measure on $(E^L, \mathcal{E}^L)$. Additionally, for any real-valued function $g$ defined on $E^L$ and $\bq = (\bq_1, \bq_2, \ldots, \bq_L)$, we set
\begin{equation*}
    Pg(\bx) = \int_{E^L} g(\bq) P(\bx,  d\bq).
\end{equation*}

We call a stochastic process indexed by $\bbT$, $\left \{ \bX_i, i \in \bbT \right \}$, a fixed branching Markov chain on a measurable space $(E, \mathcal{E} )$ with initial distribution $\nu$ and probability kernel $P$ if:
\begin{itemize}
    \item $\bX_{\emptyset}$ is distributed as $\nu$
    \item It satisfies a Markov property: for any measureable real-valued bounded functions $(g_i, i \in \bbT)$ defined on $E^L$, we have that for all $j \geq 0$
    \begin{equation*}
        \bbE \left [ \prod_{i \in \bbG_j} g_i \left (  \bX_{i1}, \bX_{i2}, \ldots, \bX_{iL} \right ) | \sigma(\bX_k, k \in \bbT_j ) \right ] = \prod_ { i \in \bbG_j} Pg_i \left (\bX_i \right ). 
    \end{equation*}
\end{itemize}
We consider the metric measurable space $\left ( S, \mathcal{S}\right )$, and add a cemetery point to $S$, $\delta$. Let $\overline{S} = \mathbb{R}^p \cup \{\delta\}$ and $\overline{\mathcal{S}}$ be the $\sigma$-field generated by $\mathcal{S}$ and $\{\delta\}$. Let $P^*$ be a probability kernel defined on $\overline{S} \times \overline{\mathcal{S}}^L$ such that
\begin{equation}
    \label{eq:absorb}
    P^*\left (\delta, \left  \{(\delta, \delta, \ldots, \delta ) \right \}  \right ) = 1.
\end{equation}
Note that this condition means that $\delta$ is an absorbing state.

The covariate process specified in Equation~\ref{eq:rdsmod}, $\{ \bX_i, i \in \bbT^* \}$ with $\bbT^* = \left \{ i \in \bbT: \bX_i \neq \delta \right \}$ is a fixed branching Markov chain on $(\overline{S}, \overline{\mathcal{S}})$ with $P^*$ satisfying Equation~\ref{eq:absorb}. Note that the recruit covariate distributions are independently and identically distributed. In other words, for $V \triangleq (V_1 \times V_2 \times \cdots \times V_L) \in \mathcal{S}^L$ and an induced kernel $P_1$,
\begin{equation*}
    \int_{V} P(\bx, \bq)  d\bq = \int_{V^1} \int_{V^2} \cdots \int_{V^L} P_1(\bx, d\bq_1 )P_1(\bx, d\bq_2) \cdots P_1(\bx, d\bq_L) d\bq_1 d\bq_2\cdots d\bq_L.
\end{equation*}
Additionally, the family size distribution is only dependent on the total number of recruits. This fact, along with the i.i.d. property described above, implies that, for $\overline{V} = ( \overline{V}^1, \overline{V}^2, \ldots, \overline{V}^L) \in \overline{S}^L$,
\begin{align*}
    \int_{\overline{V}} P^*(\bx, \bq)  d\bq &= \int_{V^1} \int_{V^2} \cdots \int_{V^L} P^*(\bx, \bq)  d\bq \\
    &= \int_{V^{z(1)}} \int_{V^{z(2)}} \cdots \int_{V^{z(L)}} P^*(\bx, \bq) d\bq 
\end{align*}
for any permutation $z$ of $\{1,2, \ldots, L \}$. We will call this property the ``exchangeability" of $P^*$.
Lastly, the process is ``spatially homogeneous;" i.e., for $\overline{V} \in \left \{ S, \delta \right \}^L$ and $\bx_1, \bx_2 \in S$,
\begin{equation*}
    P^* \left (\bx_1, \overline{V} \right ) = P^* \left (\bx_2, \overline{V} \right).
\end{equation*}
In other words, the family size distribution does not depend on the covariate value of the parent. 
By Condition (C1), the Galton-Watson (GW) tree implied by this process is super-critical.

We now cite some well known results from the branching process literature \citep{athreya2004branching}. For any subset $O \subset \bbT$, let
\begin{equation*}
    O^* = O \cap \bbT^* = \left \{ o \in O, \bX_k \neq \delta \right \}
\end{equation*}
be the subset of $O$ that are ``realized" recruits. Labeling $Z_j \triangleq | \bbG_j^* |$, $\left \{ Z_j, j \in \mathbb{N} \right \}$ is a GW process. For $j \geq 0$, we know
\begin{equation*}
    \bbE \left ( \left | \bbG_j^* \right | \right ) = m^j.
\end{equation*}
And, for $r \geq 0$, we know
\begin{equation*}
     \bbE \left ( \left | \bbT^*_r \right |  \right ) = \sum_{j=0}^r \bbE \left ( \left | \bbG_j^* \right | \right ) = \sum_{j = 0}^r m^j = \frac{m^{r+1} - 1}{m-1}.
\end{equation*}
Additionally, there exists a random variable $\branchasymp$ s.t.
\begin{equation}
    \label{eq:branching_converge}
    \branchasymp = \lim_{j \to \infty } m^{-j} \left | \bbG_j^* \right | \quad \mathrm{a.s. \ and \ in} \ L^2
\end{equation}
where $\bbE \left ( \branchasymp \right ) = 1$.

We now define a series of useful sub-probability kernels. The first is on $S \times \mathcal{S}^2$  such that 
\begin{equation*}
   P^*_2 = P^* \left \{ \cdot, \left ( \cdot \cap S^2 \right ) \times \overline{S}^{L-2} \right \}.
\end{equation*}
Note that by the exchangeability of $P^*$, for any $i,j \in \left \{0,1,\ldots, L-2 \right \}$ such that $ i + j \leq L-2$, 
\begin{equation*}
    P^*_2 = P^* \left \{ \cdot, \left ( \cdot \cap S^2 \right ) \times \overline{S}^{L-2} \right \} = P^* \left \{ \cdot, \overline{S}^i \times  \left ( \cdot \cap S \right ) \times \overline{S}^j \times \left ( \cdot \cap S \right ) \times \overline{S}^{L-2 -i - j}   \right \}.
\end{equation*}
The second sub-probability kernel is defined on  $S \times \mathcal{S}$ in the following manner,
\begin{equation*}
    P^*_1 = P^* \left \{ \cdot, \left ( \cdot \cap S \right ) \times \overline{S}^{L-1} \right \}.
\end{equation*}
By the same exchangeability property, we know that for any $0 \leq i \leq L-1$, 
\begin{equation*}
    P^*_1 = P^* \left \{ \cdot, \left ( \cdot \cap S \right ) \times \overline{S}^{L-1} \right \} = P^* \left \{ \cdot, \overline{S}^{i} \times \left ( \cdot \cap S \right ) \times \overline{S}^{L-1-i} \right \}.
\end{equation*}

We introduce an auxiliary Markov chain. Let $\left \{ \bK_j, j \in \mathbb{N} \right \} $ be a Markov chain on $S$ with $\bK_0$ distributed as $\bX^{\emptyset}$ and transition kernel
\begin{equation*}
    Q = \frac{L}{m}P_1^*.
\end{equation*}
The distribution of $\bK_j$ corresponds to the distribution of $\bX_I$ conditional on $\{ I \in \bbT^* \}$, where $I$ is chosen at random from $\bbG_j$. 

We explore the nature of transition kernel $Q$ as it will be useful in theory later. We define $P_1$ as a probability kernel on $S \times \mathcal{S}$ that represents the marginal of $P$ (note the use of $P$ instead of $P^*$ here). Consequently, it represents the dynamics of a single recruiter, recruit covariate process
(with the guarantee that the recruit exists). Formally, $P_1(\cdot, V)$ is measurable for $V \in \mathcal{S}$, where
\begin{equation*}
    P_1(\cdot, V) = P(\cdot, V \times S^{L-1}),
\end{equation*}
and $P_1(\bx, \cdot)$ is a probability measurable on $\left (S, \mathcal{S} \right )$ such that
\begin{equation*}
    P_1(\bx, \cdot) = P(\bx, \cdot \times S^{L-1}).
\end{equation*}
We can also define $P_1$ as the sub-probability kernel $P_1^*$ re-normalized over $S$. We see that for $\bx \in S$,
\begin{equation*}
    P_1(\bx, \cdot) = P^*_1(\bx, \cdot | \cdot \in S) =  \frac{P_1^*(\bx, \cdot )}{P_1^*(\bx, S)},
\end{equation*}
or in other words $P_1 = P_1^*/P_1^*(\bx, S)$. We now show that 
\begin{equation*}
    Q = \frac{L}{m} P_1^*(\bx, S) = P_1^*/P_1^*(\bx, S) = P_1,
\end{equation*}
where $Q$ is defined as the transition kernel for the auxiliary Markov chain $(\bK_j, j \in \mathbb{N})$ above.
\begin{lemma}
    Let $Q$ be defined as above, then
    \begin{equation*}
        Q = \frac{L}{m} P_1^*= P^*_1/P_1^*(\bx, S) = P_1
    \end{equation*}
\end{lemma}
\begin{proof}
    To prove the following statement, we need to show that
    \begin{equation*}
        \frac{m}{L} = P_1^*(\bx, S).
    \end{equation*}
    Due to the spatial homogeneity of the Markov process, we can define $p_\gamma$ for \\ $\gamma \in \left  \{0,1,\ldots, L \right \}$ as the probability that $\gamma$ people are recruited (regardless of the recruiter's covariates). We evaluate $P_1^*(\bx, S)$ in terms of these probabilities (keeping in mind that the recruits are identically distributed),
    \begin{align*}
        P_1^*(\bx, S) &= \frac{p_1}{L!/\left \{ (L-1)!1! \right \}} + \frac{(L-1)!/\left \{ (L-2)!1! \right \} p_2 }{L!/\left \{ (L-2)!2! \right \}} + \\
        &\hspace{0.5cm}\frac{ (L-1)!/\left \{ (L-3)!2! \right \} p_3}{L!/\left \{ (L-3)!3! \right \}} + \cdots + p_L \\
        &= \sum_{i=1}^L \frac{ip_i}{L} \\
        &= \frac{m}{L}.
    \end{align*}
\end{proof}

If $(E, \mathcal{E})$ is a metric measurable space, then define $\mathcal{B}_b(E)$ (resp. $\mathcal{B}_+(E)$) to be the set of bounded (resp. non-negative) real-valued measurable functions on $E$. The set $\mathcal{C}_b(E)$ (resp. $C_+(E)$) denotes the set of bounded (resp. non-negative) real-valued continuous functions defined on $E$. For a finite measure $\lambda$ on $\left (E, \mathcal{E} \right )$ and $f \in \mathcal{B}_b(E) \cup \mathcal{B}_+(E)$ we shall write
\begin{equation*}
    \langle \lambda,f \rangle = \int f(\bx) d\lambda (\bx).
\end{equation*}
Additionally, we write $\bbE_{\bx}$ when $\bX_{\emptyset} = \bx$. We end this preamble with the definition of ergodicity for the Markov chain $\left \{ \bK_j, j \in \mathbb{N} \right \} $.
\begin{definition}
    The Markov chain $\left \{ \bK_j, j \in \mathbb{N} \right \}$ is ergodic if there exists a probability measure $\mu$ on $ \left (S, \mathcal{S} \right )$ such that for all $f\in \mathcal{C}_b (S)$ and all $\bx \in S$,
    \begin{equation*}
        \lim_{j \to \infty} \bbE_{\bx} \left \{ f(\bK_j)  \right \} = \langle \mu, f \rangle.
    \end{equation*}
\end{definition}
Before tackling the strong law of large numbers for the branching process covariate model, we reproduce a helpful Lemma from \cite{delmas2010detection} with slight adaptations.
\begin{lemma}
\label{lem:spine}
    For $f \in \mathcal{B}_b(S) \cup \mathcal{B}_+(S)$, 
    \begin{align*}
        \bbE \left \{ f(\bK_j) \right \} &\overset{(a)}{=} m^{-j} \sum_{i \in \bbG_j} \bbE \left \{ f(\bX_i) \mathbb{I} \left (i \in \bbT^* \right ) \right \} \overset{(b)}{=} \frac{ \sum_{i \in \bbG_j} f(\bX_i) \mathbb{I} \left (i \in \bbT^* \right ) }{ \sum_{i \in \bbG_j} \bbP \left (i \in \bbT^* \right )} \\
        &= \bbE \left \{ f(\bX_I) \mid I \in \bbT^* \right \}.
    \end{align*}
    where $I$ is a uniform random variable on $\bbG_j$ independent of $\bX$.
\end{lemma}
\begin{proof}
    We consider equality (a). Recall that $\bK_0$ has distribution $\nu$. For $i \in \bbG_j$, we know
    \begin{equation*}
         \bbE \left \{ f(\bX_i) \mathbb{I} \left (i \in \bbT^* \right ) \right \} = \bbE \left \{ f(\bX_i) \mathbb{I} \left (\bX_i \neq \delta \right ) \right \} = \langle \nu, (P^*_1)^j f \rangle, 
    \end{equation*}
    following from Equation~\ref{eq:absorb} and the definition of $P^*_1$.
    Consequently,
    \begin{align*}
        \sum_{i \in \bbG_j} \bbE \left \{ f(\bX_i) \mathbb{I} \left (i \in \bbT^* \right ) \right \} &= \sum_{i \in \{1,2,\ldots, L \}^j } \langle \nu, (P^*_1)^j f \rangle = \langle \nu, (L P^*_1)^j f \rangle = m^j \langle \nu, Q^j f \rangle \\
        &= m^j \bbE \left \{ f(\bK_j) \right \}.
    \end{align*}
    This gives the first equality. Then take $f=1$ in the previous equality to get $m^j = \sum_{i \in \bbG_j} \bbP \left  ( i \in \bbT^* \right ) $ and equality (b). To prove Equality (c), we show by the law of total probability that
    \begin{align*}
        \bbE \left \{ f(\bX_I) | I \in \bbT^* \right \} &= \frac{ f(\bX_I) \mathbb{I} \left (I \in \bbT^* \right ) }{\bbP \left (I \in \bbT^* \right )} \\
        &= \frac{ \sum_{i \in \bbG_j} f(\bX_i) \mathbb{I} \left (i \in \bbT^* \right ) \bbP(I = i) }{ \sum_{i \in \bbG_j} \bbP \left (i \in \bbT^* \right ) \bbP(I = i)} \\
        &= \frac{ \sum_{i \in \bbG_j} f(\bX_i) \mathbb{I} \left (i \in \bbT^* \right ) }{ \sum_{i \in \bbG_j} \bbP \left (i \in \bbT^* \right )} .
    \end{align*}
\end{proof}

We recall that $\nu$ denotes the distribution of $\bX_{\emptyset}$. Any function $f$ defined on $S$ is extended to $\overline{S}$ by $f(\delta) = 0$. Let $F$ be a vector subspace of $\mathcal{B}(S)$ s.t.
\begin{enumerate}
    \item $F$ contains the constants
    \item $F^2 \triangleq \left \{f^2: f \in F \right \} \subset F$
    \item $F \otimes F \subset L^1(P(\bx, \cdot ))$ for all $x \in S$ and $P \left (f_0 \otimes f_1  \right ) \in F$ for all $f_0, f_1 \in F$
    \item For $\delta \in [0,1]$, $F \subset L^1 ( P^*_1(\bx, \cdot))$ for all $\bx \in S$ and $P_1^*(f) \in F$ for all $f \in F$
    \item There exists a probability measure $\mu$ on $(S, \mathcal{S})$ such that $F \subset L^1(\mu)$ and \\ $\lim_{j \to \infty} \bbE_{\bx} \left \{ f(\bK_j) \right \} = \langle \mu, f \rangle$ for all $\bx \in S$ and $f \in F$
    \item For all $f \in F$, there exists $g \in F$ such that for all $j \in \mathbb{N}$, $|Q^j f | \leq g$.
    \item $F \subset L^1(\nu)$
\end{enumerate}
By convention, a function defined on $\overline{S}$ is said to belong to $F$ if its restriction to $S$ belongs to $F$. 
\textbf{Note that if $\bK_j$ is ergodic and $\bx \to P^*g$ is continuous on $S$ for all $g \in \mathcal{C}_b(\overline{S}^2)$ then the set $F=\mathcal{C}_b(S)$ fulfills Properties~(1)-(7). Additionally, the stipulation that $\bx \to P^*g$ is continuous on $S$ for all $g \in \mathcal{C}_b(\overline{S}^2)$ is satisfied if $\bK_j$ has a continuous density (because integrals of continuous functions are continuous).} 

\noindent\rule{16cm}{0.4pt}

\paragraph{The Covariate Process of Equation~\ref{eq:rdsmod}.}

First, note that any continuous $f$ is automatically bounded since $\mathcal{X}$ is bounded. Additionally, $P^*g$ is an integral over a continuous density, so we know that $\bx \to P^*g$ is continuous on $S$ for all $g \in \mathcal{C}_b(\overline{S}^2)$. Consequently, in order for $F$ to satisfy Properties (1)-(7) above for the covariate process specified in Equation~\ref{eq:rdsmod}, we simply require that $F$ is the space of continuous functions and $\bK_j$ is Ergodic.

\textbf{Assume that the probability distributions (and expectations) of this section are with respect to the true branching parameters, $\bbeta^*$, and the uniform stabilizing policy, $\widetilde{\bpi}$.}
We now examine the covariate process specified in Equation~\ref{eq:rdsmod} under $\widetilde{\bpi}$, in which the coupon allocations are assigned uniformly at random. 
Define $\mathscr{R}^p$ as the Borel $\sigma$-algebra generated by $\mathbb{R}^p$.
For $\bx \in \mathbb{R}^p$ and $B \subset \mathscr{R}^p$, we know that
\begin{equation*}
    P_1(\bx, B) =  Q(\bx, B) = \frac{1}{|\mathcal{A}|}  \sum_{\ba \in \mathcal{A}} \mathrm{Normal}\left (\bphi_{\ba} + G_{\ba}\bx,\Sigma_{\ba} \right )
\end{equation*}
Consequently, the auxiliary Markov chain $(\bK_j, j \in \mathbb{N})$ is
\begin{equation*}
    f(\bK_{j} \mid \bK_{j-1}  = \bk_{j-1}) = \frac{1}{|\mathcal{A}|}  \sum_{\ba \in \mathcal{A}} \mathrm{Normal}\left (\bphi_{\ba} + G_{\ba}\bk_{j-1},\Sigma_{\ba} \right ).
\end{equation*}
This is a Markov-Switching Autoregressive Model \citep{hamilton1989new, francq2001stationarity, stelzer2009markov}. According to \cite{fong2007mixture}, Condition (C6), which states that
\begin{equation*}
   \frac{1}{|\mathcal{A}|} \sum_{\ba \in \mathcal{A}} \log \| G_{\ba} \|_2 < 0,
\end{equation*}
implies that the $(\bK_j, j \in \mathbb{N})$ is Ergodic. 

We label the stationary distribution for $\bK_j$ as $\mu$ such that for all $\bx \in S$,
\begin{equation*}
    \lim_{j \to \infty} \bbE_{\bx} \left \{ f(\bK_j)  \right \} = \langle \mu, f \rangle.
\end{equation*}

We lower bound the  variance (with the law of total variation) of the stationary distribution, $\mu$,
\begin{equation*}
    \mathrm{Var}_{\mu}(\bK_j) \succeq \mathrm{Var}_{\mu}\left \{ \bbE \left ( \bK_j \mid \bK_{j-1} \right ) \right \} \succeq \frac{1}{\left | \mathcal{A} \right  | } \sum_{\ba \in \mathcal{A}}\Sigma_{\ba}.
\end{equation*}
Define $\Sigma_{\bk} \triangleq \frac{1}{\left | \mathcal{A} \right  | } \sum_{\ba \in \mathcal{A}}\Sigma_{\ba}$.

We now prove a Weak Law of Large numbers over the function class $F$. We again follow \cite{delmas2010detection} with slight additions. Because $(\bK_j, j \in \mathbb{N})$ is Ergodic and $\bK_j$ has a continuous density for all $j \in \mathbb{N}$, we only need $F = \mathcal{C}_b(S)$ to satisfy Properties (1)-(7).
\begin{theorem}
    \label{thm:BranchingWLLN1}
    Let $\{\bX_i, i \in \bbT^* \}$ be the covariate process specified in Equation~\ref{eq:rdsmod}. 
     Let $F = \mathcal{C}_b(S)$, which implies it satisfies Properties (1)-(7).
    Then,
    \begin{equation*}
       \left \{  \frac{1}{m^n}  \sum_{i \in \bbG_j^*} f(\bX_i),  j \in \mathbb{N} \right \} \to \langle \mu, f \rangle \branchasymp
    \end{equation*}
    in $L^2$ as $j \to \infty$, where $\branchasymp$ is defined by Equation~\ref{eq:branching_converge}.
\end{theorem}
\begin{proof}
    We first assume that $\langle \mu, f \rangle = 0$. We have
    \begin{align*}
        \left \| \sum_{i \in \bbG_j^*} f(\bX_i) \right \|_{2}^2 &= \bbE \left [ \left (  \sum_{i \in \bbG_j^*} f(\bX_i) \bbI(i \in \bbT^*) \right )^2 \right ] \\
        &=  \sum_{i \in \bbG_j^*} \bbE \left \{ f^2(\bX_i) \bbI(i \in \bbT^*) \right \} + B_j \\
        &= m^j \bbE \left \{ f^2(\bK_j) \right \} + B_j,
    \end{align*}
    with $B_j = \sum_{(p,q) \in \bbG^2_j, i \neq j} \bbE \left [ f(\bX_p) f(\bX_q) \bbI \left \{ (p,q) \in {\bbT^*}^2\right \} \right ]$. We used Lemma~\ref{lem:spine} for the last equality. For $p, q \in \bbG_j^*$, define $p \wedge q$ as the most recent common ancestor of $p$ and $q$. We compute $B_q$ by decomposing the sum according to $k = p \wedge q$: $B_j = \sum_{r=0}^{j-1} \sum_{k \in \bbG_r} C_k$ where
    \begin{equation*}
        C_k = \sum_{(p,q) \in \bbG_j^2. p \wedge q = k} \bbE \left [ f(\bX_p) f(\bX_q) \bbI\left \{ (p,q) \in {\bbT^*}^2 \right \} \right ].
    \end{equation*}
    If $|k| = q-1$, we get
    \begin{align*}
        C_k &= \sum_{(p,q) \in \bbG_1^2, p \wedge q = \emptyset} \bbE \left ( \bbE_{\bX_k} \left [ f \left (\bX_{kp} \right ) f \left (\bX_{kq} \right ) \bbI \left \{ (kp, kq) \in {\bbT^*}^2 \right \}    \right ] \bbI( k \in \bbT^*)  \right ) \\
        &= L(L-1)\bbE \left \{ P_2^*(f \otimes f) (\bX_k) \bbI(k \in \bbT^*) \right \}
    \end{align*}
    If $|k| < q-1$, we set $r = |k|$ and derive
    \begin{align*}
        C_k &= \sum_{(w,z) \in G_1^2, w \wedge z = \emptyset} \sum_{(p,q) \in \bbG^2_{j-r-1}} \bbE \Bigg [ \bbE_{\bX_{kw}} \left \{ f\left (\bX_{kwp} \right ) \bbI \left (kwp \in \bbT^* \right )  \right \} \times \\
        & \hspace{5.5cm} \bbE_{\bX_{kz}} \left \{ f(\bX_{kzq} ) \bbI(kzq \in \bbT^* ) \right \} \bbI(kw \in \bbT^*, kz \in \bbT^*) \Bigg ] \\
        &= \sum_{(w,z) \in \bbG_1^2, w \wedge z = \emptyset} \bbE \Bigg ( \left [  \sum_{p \in \bbG_{j-r-1}} \bbE_{\bX_{kw}} \left \{ f\left (\bX_{kwp} \right ) \bbI \left (kwp \in \bbT^* \right )  \right \} \right ] \times \\
        & \hspace{3.5cm}\left [ \sum_{q \in \bbG_{j-r-1}}  \bbE_{\bX_{kz}} \left \{ f(\bX_{kzq} ) \bbI(kzq \in \bbT^* ) \right \} \right ] \bbI(kw \in \bbT^*, kz \in \bbT^*) \Bigg ) \\
         &\overset{(a)}{=} m^{2(j-r-1)} \sum_{(w,z) \in \bbG_1^2, w \wedge z = \emptyset} \bbE \Bigg [ \bbE_{\bX_{kw}} \left \{ f\left (\bK_{j-r-1} \right )  \right \} \times \\
        & \hspace{5.5cm}  \bbE_{\bX_{kz}} \left \{ f(\bK_{j-r-1} ) \right \} \bbI(kw \in \bbT^*, kz \in \bbT^*) \Bigg ] \\ 
         &\overset{(b)}{=} m^{2(j-r-1)} \sum_{(w,z) \in \bbG_1^2, w \wedge z = \emptyset} \bbE \left [ Q^{j-r-1} f\left (\bX_{kw}\right )  Q^{j-r-1} f\left (\bX_{kz}\right )  \bbI(kw \in \bbT^*, kz \in \bbT^*) \right ]  \\
         &= m^{2(j-r-1)} \sum_{(w,z) \in \bbG_1^2, w \wedge z = \emptyset} \bbE \left [ \left ( Q^{j-r-1} f \otimes Q^{j-r-1} f \right ) \left (\bX_{kw}, \bX_{kz} \right ) \bbI(kw \in \bbT^*, kz \in \bbT^*) \right ] \\
         &\overset{(c)}{=} m^{2(j-r-1)} \sum_{(w,z) \in \bbG_1^2, w \wedge z = \emptyset} \bbE \left [ P^*_2 \left ( Q^{j-r-1} f \otimes Q^{j-r-1} f \right ) \left (\bX_{k} \right ) \bbI(k \in \bbT^*) \right ] \\
         &= L(L-1)m^{2(j-r-1)} \bbE \left [ P^*_2 \left ( Q^{j-r-1} f \otimes Q^{j-r-1} f \right ) \left (\bX_{k} \right ) \bbI(k \in \bbT^*) \right ]
    \end{align*}
    Equality (a) follows from Lemma~\ref{lem:spine} and (b) through the definition of $Q$. Equality (c) follows from the definition of $P^*_2$. Consequently, by Lemma~\ref{lem:spine},
    \begin{align*}
       B_j &= L(L-1) \sum_{r=0}^{j-1} m^{2(j-r-1)} \sum_{k \in \bbG_r}  \bbE \left [ P^*_2 \left ( Q^{j-r-1} f \otimes Q^{j-r-1} f \right ) \left (\bX_{k} \right ) \bbI(k \in \bbT^*) \right ] \\
       &\overset{(a)}{=} L(L-1) \sum_{r=0}^{j-1} m^{2(j-r-1)}m^r\langle \nu, Q^rP_2^* \left ( Q^{j-r-1} f \otimes Q^{j-r-1} f \right )\rangle \\
       &= L(L-1) \sum_{r=0}^{j-1} m^{2j-r-2}\langle \nu, Q^rP_2^* \left ( Q^{j-r-1} f \otimes Q^{j-r-1} f \right )\rangle,
    \end{align*}
    where equality (a) follows from Lemma~\ref{lem:spine}. Therefore, we find that
    \begin{align*}
         \left \| m^{-j} \sum_{i \in \bbG_j^*} f(\bX_i) \right \|_{2}^2 &= m^{-2j} \left \| \sum_{i \in \bbG_j^*} f(\bX_i) \right \|_{2}^2 \\
         &= m^{-j} \bbE \left \{ f^2(\bK_j) \right \} + \\
         &\hspace{0.5cm}L(L-1)m^{-2}\sum_{r=0}^{q-1} m^{-r} \langle \nu, Q^rP^* \left ( Q^{j-r-1} f \otimes Q^{j-r-1} f \right )\rangle
    \end{align*}
    Because $f \in F$, properties (2), (4), and (6) imply that $E\left \{ f^2(\bK_j) \right \} < \infty$ for any $j \in \mathbb{N}$ and
    \begin{equation*}
        \lim_{j \to \infty} m^{-j}\bbE \left \{ f^2(\bK_j) \right \} = 0.
    \end{equation*}
    Properties (3), (4), and (5) along with $\langle \mu, f \rangle = 0$ imply that $P \left (Q^{j-r-1}f \otimes Q^{j-r-1}f \right  )$ converges to $0$ as $j \to \infty$ (with $r$ fixed) and is bounded uniformly in $j>r$ by a function in $F$. Thus, properties (5) and (6) imply that $\langle \nu, Q^rP^*_2 \left ( Q^{j-r-1} f \otimes Q^{j-r-1} f \right )\rangle$ converges to $0$ as $j \to \infty$ (with $r$ fixed) and is bounded uniformly in $j>r$ by a fixed constant, say $C$. 

    For any $\epsilon >0$, we can choose $r_0$ such that $\sum_{r > r_0} m^{-r}C \leq \epsilon/2$. Additionally, choose $j_0 > r_0$ such that for $j \geq j_0$ and $r \leq r_0$, we have
    \begin{equation*}
        \left |  \langle \nu, Q^rP^*_2 \left ( Q^{j-r-1} f \otimes Q^{j-r-1} f \right )\rangle \right | \leq \epsilon/(2r_0).
    \end{equation*}
    We get that for all $j \geq j_0$
    \begin{equation*}
        \sum_{r=0}^{j-1} m^{-r} \left | \langle \nu, Q^rP^*_2 \left ( Q^{j-r-1} f \otimes Q^{j-r-1} f \right )\rangle \right | \leq \sum_{r=0}^{r_0} \epsilon/(2r_0) + \sum_{r > r_0}^{j-1} m^{-r} C \leq \epsilon/2 + \epsilon/2 = \epsilon.
    \end{equation*}
    Therefore,
    \begin{equation*}
        \lim_{j \to \infty} \sum_{r=0}^{j-1} m^{-r} \langle \nu, Q^rP^*_2 \left ( Q^{j-r-1} f \otimes Q^{j-r-1} f \right )\rangle = 0.
    \end{equation*}
    We can now conclude that if $\langle \mu, f \rangle = 0$, then 
    \begin{equation*}
        \lim_{j \to \infty} \left \| m^{-j}\sum_{i \in \bbG_j^*} f(\bX_i) \right \|_{2} = 0.
    \end{equation*}
    We consider the case when $\langle \mu, f \rangle$ is arbitrary. For any function $f \in F$, we can construct $g = f - \langle \mu, f \rangle$. We decompose as follows,
    \begin{equation*}
         m^{-j}\sum_{i \in \bbG_j^*} f(\bX_i) =  m^{-j}\sum_{i \in \bbG_j^*} g(\bX_i) + \langle \mu, f \rangle m^{-j} \left | \bbG_j^* \right |. 
    \end{equation*}
    Because $g \in F$ and $\langle \mu, g \rangle = 0$, we know that
    \begin{equation*}
       \lim_{j \to \infty} \left \| m^{-j}\sum_{i \in \bbG_j^*} g(\bX_i) \right \|_{2} = 0.
    \end{equation*}
    Because $\left ( m^{-j} \left | \bbG_j^* \right |, j \geq 1  \right )$ converges to $L^2$ (and a.s.) to $\branchasymp$, we can conclude that
    \begin{equation*}
        m^{-j}\sum_{i \in \bbG_j^*} f(\bX_i) \overset{L^2}{\to} \langle \mu, f \rangle \branchasymp
    \end{equation*}
    as $j \to \infty$.
\end{proof}

We now prove a similar result for the average over all individuals in $\bbT^*_r$, the first $r$ generations. We set $t_r \triangleq \bbE \left ( \left | \bbT_r^* \right |\right )$ and recall that
\begin{equation*}
    t_r = \frac{m^{r+1} - 1}{m-1}.
\end{equation*}
We now state a lemma that is a direct consequences of Lemma~\ref{lem:toep}.
\begin{lemma}
    \label{lem:toeplitzknockoff}
    Let $(v_r, r \in \mathbb{N})$ be a sequence of real numbers converging to $a \in \mathbb{R}_+$, and $m$ be a real number such that $m > 1$. Let
    \begin{equation*}
        w_r = \sum_{j=0}^r m^{j-r-1} v_j.
    \end{equation*}
    Then the sequence $(w_r, r \in \mathbb{N} )$ converges to $a/(m-1)$.
\end{lemma}
\begin{proof}
    First we rearrange the sum $w_r = \sum_{j=0}^r m^{j-r-1} v_j$,
    \begin{equation*}
        w_r = \sum_{j=0}^r m^{j-r-1} v_j = \sum_{j=0}^r m^{-j-1} v_{r-j} = \frac{1}{m} \left ( \sum_{j=0}^r m^{-j} \right )  \left ( \sum_{j=0}^r m^{-j} \right )^{-1}  \sum_{j=0}^r m^{-j} v_{r-j}.
    \end{equation*}
    By Lemma~\ref{lem:toep},
    \begin{equation*}
         \lim_{r \to \infty} \left ( \sum_{j=0}^r m^{-j} \right )^{-1}  \sum_{j=0}^r m^{-j} v_{r-j} = a.
    \end{equation*}
    By the formula for a geometric series,
    \begin{equation*}
        \lim_{r \to \infty}  \sum_{j=0}^r m^{-j}  = \frac{1}{1-\frac{1}{m}}.
    \end{equation*}
    Consequently,
    \begin{equation*}
        \lim_{r\to \infty} w_r =  \frac{1}{m} \left ( \sum_{j=0}^r m^{-j} \right )  \left ( \sum_{j=0}^r m^{-j} \right )^{-1}  \sum_{j=0}^r m^{-j} v_{r-j} = \frac{1}{m} \frac{1}{1-\frac{1}{m}} a = \frac{a}{m-1}.
    \end{equation*}
\end{proof}

We now state the weak law of large numbers when averaging over the complete tree $\bbT^*_r$ as $r \to \infty$.
\begin{theorem}
\label{thm:branchingtreeaverage}
    Let $\{\bX_i, i \in \bbT^* \}$ be the covariate process specified in Equation~\ref{eq:absorb}. Let $F$ satisfy (i)-(vi) and $f \in F$. 
    Then,
    \begin{equation*}
       \left \{  \frac{1}{t_r}  \sum_{i \in \bbT_r^*} f(\bX_i),  r \in \mathbb{N} \right \} \to \langle \mu, f \rangle \branchasymp
    \end{equation*}
    in $L^2$ as $r \to \infty$, where $\branchasymp$ is defined by Equation~\ref{eq:branching_converge}.
\end{theorem}
\begin{proof}
    Recalling that
    \begin{equation*}
    t_r = \frac{m^{r+1} - 1}{m-1}
    \end{equation*}
    we know that
    \begin{align*}
        \left \| \frac{1}{t_r} \sum_{i \in \bbT_r^*} f(\bX_i) - \langle \mu, f \rangle \branchasymp \right \|_{2} &= \left \| \sum_{j=0}^r \frac{m^j}{t_r} \left ( \frac{1}{m^j} \sum_{i \in \bbG_j^*} f(\bX_i) - \langle \mu, f \rangle \branchasymp \right ) \right \|_{2} \\
        &\leq  \sum_{j=0}^r \frac{m^j}{t_r}  \left \| \frac{1}{m^j} \sum_{i \in \bbG_j^*} f(\bX_i) - \langle \mu, f \rangle \branchasymp  \right \|_{2} \\
        &= m^{r+1} \frac{m-1}{m^{r+1} - 1} \sum_{j=0}^r m^{j-r-1}  \left \| \frac{1}{m^j} \sum_{i \in \bbG_j^*} f(\bX_i) - \langle \mu, f \rangle \branchasymp \right \|_{2} \\
        &= \frac{m-1}{1 - m^{-r-1}} \sum_{j=0}^r m^{j-r-1}  \left \| \frac{1}{m^j} \sum_{i \in \bbG_j^*} f(\bX_i) - \langle \mu, f \rangle \branchasymp \right \|_{2}
    \end{align*}
    Consequently,
    \begin{equation*}
        \left \| \frac{1}{t_r} \sum_{i \in \bbT_r^*} f(\bX_i) - \langle \mu, f \rangle \branchasymp\right \|_{2} \to 0
    \end{equation*}
    as $r \to \infty$ by Theorem~\ref{thm:BranchingWLLN1} and Lemma~\ref{lem:toeplitzknockoff}.
\end{proof}

\textbf{Note that convergence in $L^2$ implies convergence in probability, so we have proven a weak law of large numbers (WLLN) for the branching covariate process.}

\noindent\rule{16cm}{0.4pt}

Now that we have proven a WLLN for the branching covariate process implied by Equation~\ref{eq:rdsmod}, we return to the task of verifying Assumption~\ref{as:stabalizedvariance}. First, recall that
\begin{equation*}
   m = \sum_{m^v = 1}^L m^v \frac{ \lambda^{m^v}/m^v!}{
  \sum_{\ell=0}^{|\mathbf{A}^v|} (\lambda^\ell/\ell!)  }
\end{equation*}

Additionally, define $\Sigma_J = t_{J-1}$.
Note that the differentiability and moment conditions of the log-likelihood imply that
\begin{equation*}
    \eta_J = \sum_{j=1}^J \sum_{i \in \mathcal{E}_{j-1}} \bbE_{\bbeta^*, \widetilde{\bpi}} \left \{ \dot{l}_i(\bbeta^*) \dot{l}_i(\bbeta^*)^\top \mid \epochfield \right \} = \sum_{j=1}^J \sum_{i \in \mathcal{E}_{j-1}} \bbE_{\bbeta^*, \widetilde{\bpi}} \left \{ -\ddot{l}_i(\bbeta^*) \mid \epochfield \right \}.
\end{equation*}
Consequently, we will use the Hessians calculated in Section~\ref{app_sec:branchinginference} in the following theory. 
We divide $\eta_J$ into components corresponding to the different models that compose the branching process, $\eta_J = \mathrm{diag} \left (\eta_J^{\bx},  \eta_J^{y}, \eta_J^{m}, \eta_J^t \right )$. 

\textbf{First, we analyze the covariate model component, $\eta_J^{\bx}$.} 
We define quantities
\begin{equation*}
     \bV_j^{\ba} \triangleq \sum_{i \in \epochgroup} M_i \bX^*_i {\bX^*_i}^\top  \mathbb{I}(\bA_i = \ba), \quad
     n_j^{\ba} \triangleq \sum_{i \in \epochgroup} M_i \mathbb{I}(\bA_i = \ba),
\end{equation*}
\begin{equation*}
   \eta_J^{\bx} \triangleq \sum_{j=1}^J  \mathrm{diag} \left [ \bbE_{\bbeta^*, \widetilde{\bpi}} \left \{ \begin{pmatrix}
         \bV_j^{\ba} & 2 \bV_j^{\ba} {G^*_{\ba}}^\top \\
        2 G^*_{\ba} \bV_j^{\ba} & 4 G^*_{\ba}  \bV_j^{\ba} {G^*_{\ba}}^\top  + 2 n_j^{\ba} \Sigma_{\ba}^*
    \end{pmatrix} \otimes \Sigma_{\ba}^* \mid \epochfield \right \}_{\ba \in \mathcal{A}}
    \right ].
\end{equation*}
We define
\begin{align*}
   Q_{j, \ba}^{\bx} \triangleq \bbE \left \{ \sum_{i \in \epochgroup} M_i \bX^*_i {\bX^*_i}^\top  \mathbb{I}(\bA_i = \ba) \mid \epochfield \right \} = \frac{1}{|\mathcal{A}|} \sum_{i \in \epochgroup} m \bX^*_i {\bX^*_i}^\top,
\end{align*}
and
\begin{align*}
   Q_{j, \ba}^{\ba} \triangleq \bbE \left \{ \sum_{i \in \epochgroup} M_i \mathbb{I}(\bA_i = \ba) \mid \epochfield \right \} = \frac{1}{|\mathcal{A}|} m \kappa_{j-1}.
\end{align*}
Recall that
\begin{equation*}
    t_J = \frac{m^{J+1} - 1}{m-1}.
\end{equation*}
Then
\begin{align*}
   & {\Sigma_J}^{-1/2} \eta_J^{\bx}  {\Sigma_J}^{-1/2} = \\
   &\sum_{j=1}^J {\Sigma_J}^{-1/2} \mathrm{diag} \left [\left \{ \begin{pmatrix}
        Q_{j, \ba}^{\bx} & 2 Q_{j, \ba}^{\bx} {G^*_{\ba}}^\top \\
        2 G^*_{\ba}Q_{j, \ba}^{\bx} & 4 G^*_{\ba} Q_{j, \ba}^{\bx} {G^*_{\ba}}^\top   + 2 Q_{j, \ba}^{\ba}\Sigma_{\ba}^*
    \end{pmatrix} \otimes \Sigma_{\ba}^* \right \}_{\ba \in \mathcal{A}}  \right ] {\Sigma_J}^{-1/2} = \\
   & \sum_{j=1}^J \mathrm{diag} \left [\left \{ \begin{pmatrix}
        Q_{j, \ba}^{\bx}& 2 Q_{j, \ba}^{\bx}{G^*_{\ba}}^\top \\
        2 G^*_{\ba}Q_{j, \ba}^{\bx} & 4 G^*_{\ba} Q_{j, \ba}^{\bx}{G^*_{\ba}}^\top   + 2 Q_{j, \ba}^{\ba}\Sigma_{\ba}^*
    \end{pmatrix} \otimes \Sigma_{\ba}^* \right \}_{\ba \in \mathcal{A}}  \right ]/t_{J-1}.
\end{align*}
We know that 
  \begin{align*}
     Q_{j, \ba}^{\bx}/t_{J-1} &=  \frac{1}{|\mathcal{A}|} \sum_{i \in \epochgroup} m \bX^*_i {\bX^*_i}^\top/t_{J-1} \\
   &=  \frac{1}{|\mathcal{A}|} m \begin{pmatrix}
				\kappa_{j-1} /t_{J-1} & \sum_{i \in \epochgroup} {\bX_i}^\top /t_{J-1}\\
				\sum_{i \in \epochgroup} {\bX_i} /t_{J-1} & \sum_{i \in \epochgroup} \bX_i {\bX_i}^\top /t_{J-1} 
			\end{pmatrix}.
  \end{align*}
    Note that $\bX_i$ and $\bX_i \bX_i^\top$ are both continuous, bounded functions of $\bX_i$ because $\mathcal{X}$ is bounded (as previously mentioned). By Theorem~\ref{thm:branchingtreeaverage}, as $J \to \infty$,
  \begin{equation*}
      \begin{pmatrix}
				\sum_{j=1}^J \kappa_{j-1} /t_{J-1} & \sum_{i \in \epochgroup} {\bX_i}^\top /t_{J-1} \\
				\sum_{i \in \epochgroup} {\bX_i} /t_{J-1} & \sum_{i \in \epochgroup} \bX_i {\bX_i}^\top /t_{J-1} 
			\end{pmatrix} \overset{p}{\to} \branchasymp C^{\bx}_*,
  \end{equation*}
   where $C_*^{\bx}$ is a constant matrix. Additionally, 
  \begin{equation*}
      Q_{j, \ba}^{\ba} = \frac{1}{|\mathcal{A}|} m \kappa_{j-1}/t_{J-1} \overset{p}{\to} \frac{m}{|\mathcal{A}|}\branchasymp.
  \end{equation*}
  Consequently, we know that
    \begin{align*}
        &{\Sigma_J}^{-1/2} \eta_J^{\bx}  {\Sigma_J}^{-1/2} = \\
        &\sum_{j=1}^J  \mathrm{diag} \left [ \bbE_{\bbeta^*, \widetilde{\bpi}} \left \{ \begin{pmatrix}
         \bV_j^{\ba} & 2 \bV_j^{\ba} {G^*_{\ba}}^\top \\
        2 G^*_{\ba} \bV_j^{\ba} & 4 G^*_{\ba}  \bV_j^{\ba} {G^*_{\ba}}^\top  + 2 n_j^{\ba} \Sigma_{\ba}^*
    \end{pmatrix} \otimes \Sigma_{\ba}^* \mid \epochfield \right \}_{\ba \in \mathcal{A}}
    \right ]/t_{J-1} \overset{p}{\to} C^{\bx}_{*}\branchasymp
    \end{align*}
    as $J\to \infty$. 

    \begin{lemma}
    \label{lem:covia}
        Under conditions (1)-(6), for any covariate parameters $\bbeta_{\bx} = \left \{ \Sigma_{\ba}, G_{\ba} \right \}_{\ba \in \mathcal{A}} \in \mathcal{B}$ and some $\delta_{\bbeta_\bx} > 0$,
        \begin{align*}
        &\sigma_{\min} \left ( \sum_{j=1}^J  \bbE_{\bbeta^*, \widetilde{\bpi}} \left [ \mathrm{diag} \left \{ \begin{pmatrix}
         \bV_j^{\ba} & 2 \bV_j^{\ba} {G_{\ba}}^\top \\
        2 G_{\ba} \bV_j^{\ba} & 4 G_{\ba}  \bV_j^{\ba} {G_{\ba}}^\top  + 2 n_j^{\ba} \Sigma_{\ba}
    \end{pmatrix} \otimes \Sigma_{\ba} \right \}_{\ba \in \mathcal{A}} /t_{J-1} \mid \epochfield \right ] \right ) \\
        &\geq \delta_{\bbeta_{\bx}} \ \ \mathrm{w.p. \ 1}
        \end{align*}
    \end{lemma}
    \begin{proof}
    By Lemma~\ref{lem:blckpd}, we know that
		\begin{align*}
		&\lim_{J \to \infty} \sigma_{\min} \left \{ \sum_{j=1}^J \left \{ \begin{pmatrix}
        Q_{j, \ba}^{\bx}& 2 Q_{j, \ba}^{\bx}{G_{\ba}}^\top \\
        2 G_{\ba}Q_{j, \ba}^{\bx} & 4 G_{\ba} Q_{j, \ba}^{\bx}{G_{\ba}}^\top   + 2 Q_{j, \ba}^{\ba}\Sigma_{\ba}
    \end{pmatrix} \otimes \Sigma_{\ba} \right \}/ t_{J-1} \right \} \geq \\
        &\lim_{J \to \infty}  \sigma_{\min} \left (  \sum_{j=1}^J  Q_{j, \ba}^{\bx}/t_{J-1} \right )  \sigma_{\min} \left (  \sum_{j=1}^J   2 Q_{j, \ba}^{\ba} \Sigma_{\ba}/t_{J-1} \right )
		\end{align*}
		because
    \begin{equation*}
        \lim_{J \to \infty} \sum_{j=1}^J \left \{ 4 G_{\ba} Q_{j, \ba}^{\bx}{G_{\ba}}^\top  + 2 Q_{j, \ba}^{\ba}\Sigma_{\ba} - 4G_{\ba}Q_{j, \ba}^{\bx}{Q_{j, \ba}^{\bx}}^{-1} Q_{j, \ba}^{\bx}{G_{\ba}}^\top \right \}/t_{J-1} =  \lim_{J \to \infty} \sum_{j=1}^J   2 Q_{j, \ba}^{\ba} \Sigma_{\ba}/t_{J-1}.
    \end{equation*}
    We know that $\lim_{J \to \infty} \sum_{j=1}^J 2Q_{j, \ba}^{\ba} \Sigma_{\ba}/t_{J-1} = 2m\branchasymp\Sigma_{\ba}/\left |\mathcal{A} \right |$ a.s., so we only need to prove that 
        \begin{equation*}
            \ \sum_{j=1}^J  Q_{j, \ba}^{\bx}/t_{J-1} \overset{p}{\to} \frac{m\branchasymp}{|\mathcal{A}|}\begin{pmatrix}
                 1 & \bbE_{\mu}\left ( \bK_i \right )^\top \\
				\bbE_{\mu}\left ( \bK_i \right )^\top & \bbE_{\mu} \left ( \bK_i \bK_i^\top \right )
        \end{pmatrix}
        \end{equation*}
        is positive definite. By a second application of Lemma~\ref{lem:blckpd},
		\begin{align*}
			\frac{m\branchasymp}{|\mathcal{A}|} \sigma_{\min}\left \{ \begin{pmatrix}
			    1 & \bbE_{\mu}\left ( \bK_i \right )^\top \\
				\bbE_{\mu}\left ( \bK_i \right ) & \bbE_{\mu} \left ( \bK_i \bK_i^\top \right )
                \end{pmatrix}\right \} & \frac{m\branchasymp}{|\mathcal{A}|} \geq \sigma_{\min} \left (  1 \right ) \times \\
                &\hspace{1.3cm}\sigma_{\min} \left \{  \bbE_{\mu} \left ( \bK_i \bK_i^\top \right ) - \bbE_{\mu}\left ( \bK_i \right ) \bbE_{\mu}\left ( \bK_i \right )^\top \right \} \\ 
                &\hspace{1cm}=  \frac{m\branchasymp}{|\mathcal{A}|} \sigma_{\min} \left \{ \mathrm{Var}_{\mu} \left ( \bK_i \right ) \right \} \\
                &\hspace{1cm}\succeq \frac{m\branchasymp}{|\mathcal{A}|} \sigma_{\min}\left (\Sigma_{\bk} \right ).
		\end{align*}
    Note that there exists $\delta_*$ such that $\branchasymp>\delta_*$ by assumption. Using the properties of kronecker products, we know that defining
    \begin{equation*}
        \epsilon_{\ba} \triangleq 2\frac{m^2\delta_*^2}{|\mathcal{A}|^2} \sigma_{\min} \left ( \Sigma_{\bk} \right )\sigma_{\min} \left ( \Sigma_{\ba} \right )^2,
    \end{equation*}
    and
    \begin{equation*}
        \delta_{\bbeta_{\bx}} \triangleq \prod_{\ba \in \mathcal{A}} \epsilon_{\ba}
    \end{equation*}
    completes the proof.
    \end{proof}
    
    By Lemma~\ref{lem:covia}, 
    $\sigma_{\min} \left (\branchasymp C^{\bx}_{*} \right ) \geq \delta_{\bbeta_{\bx}}$.
   Consequently, Assumption~\ref{as:stabalizedvariance} is satisfied for the covariate model. 
    
\textbf{Next, we analyze the reward model component, $\eta_J^{y}$.} 

\begin{align*}
    \frac{\partial^2 \ell_{J} (\bbeta_y)} {\partial\bbeta_y \partial \bbeta_y^\top} = -\sum_{j =1}^J \sum_{i\in \mathcal{E}_j} & \bZ_i {\bZ_i}^\top \left \{ \frac{1}{1 + \exp(-{\bZ_i}^\top\bbeta_y)} \right \} \left \{\frac{1}{1 + \exp \left ({\bZ_i}^\top \bbeta_y \right )} \right \}.
\end{align*}
Remember that $\bZ_i = \sum_{\ba \in \mathcal{A}} h_{\ba}(\bX_i) \bbI \left (\bA_{r^i} = \ba \right )$, where $h_{\ba}$ is a continuous function in $\bX_i$ that depends on $\ba \in \mathcal{A}$. Consequently,
\begin{align*}
     \frac{\partial^2 \ell_{J} (\bbeta_y)} {\partial\bbeta_y \partial \bbeta_y^\top} =& -\sum_{j =1}^J \sum_{i\in \mathcal{E}_j} \sum_{\ba \in \mathcal{A}} h_{\ba}(\bX_i) h_{\ba}(\bX_i)^\top \left \{ \frac{1}{1 + \exp(-h_{\ba}(\bX_i)^\top\bbeta_y)} \right \} \times 
     \\ &\left \{\frac{1}{1 + \exp \left (h_{\ba}(\bX_i)^\top \bbeta_y \right )} \right \} \bbI \left (\bA_{r^i} = \ba \right ).
\end{align*}
We now express
\begin{align*}
   \eta^y_J &=  \sum_{j =1}^J \bbE \Bigg \{ \sum_{i\in \mathcal{E}_j} \sum_{\ba \in \mathcal{A}} h_{\ba}(\bX_i) h_{\ba}(\bX_i)^\top \times \\
   &\hspace{2cm} \left \{ \frac{1}{1 + \exp(-h_{\ba}(\bX_i)^\top\bbeta^*_y)} \right \} \left \{\frac{1}{1 + \exp \left (h_{\ba}(\bX_i)^\top \bbeta^*_y \right )} \right \} \times \\
   &\hspace{2cm}  \bbI \left (\bA_{r^i} = \ba \right )\mid \epochfield \Bigg \} \\
   &\sum_{j =1}^J \bbE \Bigg \{ \sum_{i\in \mathcal{E}_j} \sum_{\ba \in \mathcal{A}} \sum_{l \in \mathbb{G}_1} \mathbb{I}(il \in \bbT^*) h_{\ba}(\bX_i) h_{\ba}(\bX_i)^\top \times \\
   &\hspace{2cm}\left \{ \frac{1}{1 + \exp(-h_{\ba}(\bX_i)^\top\bbeta^*_y)} \right \} \left \{\frac{1}{1 + \exp \left (h_{\ba}(\bX_i)^\top \bbeta^*_y \right )} \right \} \times \\
   &\hspace{2cm}  \bbI \left (\bA_{r^i} = \ba \right )\mid \epochfield \Bigg \} \\
   &=  \sum_{j =1}^J\sum_{i\in \mathcal{E}_{j-1}} \sum_{l \in \mathbb{G}_1} \bbE \Bigg \{ \mathbb{I}(il \in \bbT^*)  \bbE \Bigg \{ \sum_{\ba \in \mathcal{A}} h_{\ba}(\bX_{il}) h_{\ba}(\bX_{il})^\top \left \{ \frac{1}{1 + \exp(-h_{\ba}(\bX_{il})^\top\bbeta^*_y)} \right \} \times \\
   &\hspace{2cm} \left \{\frac{1}{1 + \exp \left (h_{\ba}(\bX_{il})^\top \bbeta^*_y \right )} \right \} \bbI \left (\bA_{i} = \ba \right ) \mid \bX_i, \ il \in \bbT^*  \Bigg \} \mid \bX_i \Bigg \}\\
   &=  \sum_{j =1}^J\sum_{i\in \mathcal{E}_{j-1}} \sum_{l \in \mathbb{G}_1} \bbE \Bigg \{ \mathbb{I}(il \in \bbT^*)  \frac{1}{|\mathcal{A}|} \sum_{\ba' \in \mathcal{A}} \sum_{\ba \in \mathcal{A}} \bbE \Bigg \{h_{\ba}(\bX_{il}) h_{\ba}(\bX_{il})^\top \\
   &\hspace{2cm} \left \{ \frac{1}{1 + \exp(-h_{\ba}(\bX_{il})^\top\bbeta^*_y)} \right \} \times \\
   &\hspace{2cm} \left \{\frac{1}{1 + \exp \left (h_{\ba}(\bX_{il})^\top \bbeta^*_y \right )} \right \} \bbI \left (\bA_{i} = \ba \right ) \mid \bA_i = \ba', \bX_i, \ il \in \bbT^*  \Bigg \} \mid \bX_i \Bigg \}\\
   &=  \sum_{j =1}^J\sum_{i\in \mathcal{E}_{j-1}} \sum_{l \in \mathbb{G}_1} \bbE \Bigg \{ \mathbb{I}(il \in \bbT^*)  \frac{1}{|\mathcal{A}|} \sum_{\ba \in \mathcal{A}} \bbE \Bigg \{  h_{\ba}(\bX_{il}) h_{\ba}(\bX_{il})^\top \\
   &\hspace{2cm} \left \{ \frac{1}{1 + \exp(-h_{\ba}(\bX_{il})^\top\bbeta^*_y)} \right \} \times \\
  &\hspace{2cm} \left \{\frac{1}{1 + \exp \left (h_{\ba}(\bX_{il})^\top \bbeta^*_y \right )} \right \} \mid \bA_i = \ba, \bX_i, \ il \in \bbT^* \Bigg \} \mid \bX_i \Bigg \}
\end{align*}
We label
\begin{align*}
    &f(\bX_{il}) =   \frac{1}{|\mathcal{A}|} \sum_{\ba \in \mathcal{A}} \bbE \Bigg \{  h_{\ba}(\bX_{il}) h_{\ba}(\bX_{il})^\top \left \{ \frac{1}{1 + \exp(-h_{\ba}(\bX_{il})^\top\bbeta^*_y)} \right \} \times \\
   &\hspace{2cm} \left \{\frac{1}{1 + \exp \left (h_{\ba}(\bX_{il})^\top \bbeta^*_y \right )} \right \} \mid \bA_i = \ba, \bX_i \Bigg \},
\end{align*}
where $f$ is continuous because $h_{\ba}$ is continuous (and it is a composition of continuous functions). Consequently,
\begin{align*}
    \eta^y_J&= \sum_{j =1}^J \sum_{i\in \mathcal{E}_{j-1}} \sum_{l \in \mathbb{G}_1} \bbE \left \{ \mathbb{I}(il \in \bbT^*) f(\bX_{il})\mid  \bX_i   \right \} \\
    &= \sum_{j =1}^J \sum_{i\in \mathcal{E}_{j-1}} L P_1^* \left \{f(\bX_{i}) \right \}.
\end{align*}
Lastly, we can define $f^*(\bX_{i}) = L P_1^* \left \{f(\bX_{i}) \right \}$ and
\begin{align*}
    \eta_J^y = \sum_{j =1}^J \sum_{i\in \mathcal{E}_{j-1}} f^*(\bX_{i}).
\end{align*}
Note that $f^*$ is continuous and bounded since compositions and integrals of continuous functions are continuous.
It is now clear that
\begin{align*}
    \eta_J^y/t_{J-1} = \sum_{j =1}^J \sum_{i\in \mathcal{E}_{j-1}} f^*(\bX_{i})/t_{J-1}  \overset{p}{\to} \branchasymp C_*^y,
\end{align*}
where $C_*^y$ is a constant matrix. 
\begin{lemma}
    \label{lem:rewia}
    Under conditions (1)-(6), for any $\bbeta_y \in \mathcal{B}$, there exists $\delta_{\bbeta_y}$ such that
    \begin{align*}
    &\lim_{J \to \infty} \sigma_{\min} \Bigg ( \sum_{j =1}^J \bbE \Bigg \{ \sum_{i\in \mathcal{E}_j} \sum_{\ba \in \mathcal{A}} h_{\ba}(\bX_i) h_{\ba}(\bX_i)^\top \left \{ \frac{1}{1 + \exp(-h_{\ba}(\bX_i)^\top\bbeta_y)} \right \} \times \\
    &\hspace{2cm} \left \{\frac{1}{1 + \exp \left (h_{\ba}(\bX_i)^\top \bbeta_y \right )} \right \} \times \\
    &\hspace{2cm}  \bbI \left (\bA_{i} = \ba \right )\mid \epochfield \Bigg \}/t_{J-1} \Bigg ) \geq \delta_{\bbeta_y} \ \ \mathrm{w.p. \ 1}
    \end{align*}
    
\end{lemma}

\begin{proof}
Label 
    \begin{align*}
    \sigma_y &= \lim_{J \to \infty} \sigma_{\min} \Bigg ( \sum_{j =1}^J \bbE \Bigg \{ \sum_{i\in \mathcal{E}_j} \sum_{\ba \in \mathcal{A}} h_{\ba}(\bX_i) h_{\ba}(\bX_i)^\top \left \{ \frac{1}{1 + \exp(-h_{\ba}(\bX_i)^\top\bbeta_y)} \right \} \times \\
    &\hspace{2cm} \left \{\frac{1}{1 + \exp \left (h_{\ba}(\bX_i)^\top \bbeta_y \right )} \right \} \times \\
    &\hspace{2cm}  \bbI \left (\bA_{i} = \ba \right )\mid \epochfield \Bigg \} \Bigg ).
    \end{align*}
We start with analyzing the quantity
\begin{align*}
    &\sum_{l \in \mathbb{G}_1} \bbE \Bigg \{ \mathbb{I}(il \in \bbT^*)  \frac{1}{|\mathcal{A}|} \sum_{\ba \in \mathcal{A}} \bbE \Bigg \{  h_{\ba}(\bX_{il}) h_{\ba}(\bX_{il})^\top \left \{ \frac{1}{1 + \exp(-h_{\ba}(\bX_{il})^\top\bbeta_y)} \right \} \times \\
  &\hspace{2cm} \left \{\frac{1}{1 + \exp \left (h_{\ba}(\bX_{il})^\top \bbeta_y \right )} \right \} \mid il \in \bbT^*, \bA_i = \ba, \bX_i \Bigg \} \mid \bX_i \Bigg \},
\end{align*}
Because $\mathcal{X}$ is compact and $h_{\ba}$ is continuous, we know that
\begin{equation*}
    c_{\exp} \triangleq \min_{\bx \in \mathcal{X}} \min_{\ba \in \mathcal{A}}  \left \{ \frac{1}{1 + \exp(-h_{\ba}(\bx)^\top\bbeta_y)} \right \}  \left \{\frac{1}{1 + \exp \left (h_{\ba}(\bx)^\top \bbeta_y \right )} \right \},
\end{equation*}
where $c_{\exp} > 0$ by the extreme value theorem. Consequently,
\begin{align*}
  &\succ c_{\exp} \sum_{l \in \mathbb{G}_1} \bbE \Bigg \{ \mathbb{I}(il \in \bbT^*)  \frac{1}{|\mathcal{A}|} \sum_{\ba \in \mathcal{A}} \bbE \Bigg \{  h_{\ba}(\bX_{il}) h_{\ba}(\bX_{il})^\top \mid \bA_i = \ba, \bX_i, il \in \bbT^* \Bigg \} \mid \bX_i \Bigg \}.
\end{align*}
To progress we need to be more explicit about $h_{\ba}$. In our model $Z_i$ is an interaction between $X_i$ and an indicator of the coupon type given to pariticipant $i$. We demonstrate the positive definite property for a paradigm where we have two coupon types $\mathcal{A} =\{ \ba, \ba' \}$, making $\bZ_i = \left (1, \bX_i, \bX_i \mathbb{I}(\bA_i = \ba) \right )$. However, this logic extends to a larger number of coupon types. We find that
\begin{align*}
    &\frac{1}{|\mathcal{A}|} \sum_{\ba \in \mathcal{A}} \bbE \Bigg \{  \bZ_{il} \bZ_{il}^\top \mid il \in \bbT^*, \bA_i = \ba, \bX_i \Bigg \} \\
    &= \frac{1}{|\mathcal{A}|} \sum_{\ba \in \mathcal{A}} \bbE \Bigg \{  \bZ_{il} \bZ_{il}^\top \mid \bA_i = \ba, \bX_i \Bigg \} \\
    &= \mathbb{E} \left \{ \begin{pmatrix}
        1 & \bX_{il}^\top & 0 \\
        \bX_{il}^\top & \bX_{il}\bX_{il}^\top & 0 \\
        0 & 0 & 0
    \end{pmatrix} | \bX_i, \bA_i = \ba' \right \}  \\
    &\hspace{1cm}+ \mathbb{E} \left \{ \begin{pmatrix}
        1 & \bX_{il}^\top &  \bX_{il}^\top \\
        \bX_{il}^\top & \bX_{il}\bX_{il}^\top &    \bX_{il} \bX_{il}^\top \\
        {\bX_{il}}^\top &   \bX_{il}\bX_{il}^\top &   \bX_{il} {\bX_{il}}^\top
    \end{pmatrix} | \bX_i, \bA_i = \ba \right \}.
\end{align*}
We label $\mu_{\ba, \bf{x}} = \mathbb{E}  \left \{ \bX_{il} | \bX_i = \bf{x}, \bA_i = \ba \right \} $ and $\Delta_{\ba} = \bbE \left \{ \bX_{il}\bX_{il}^\top | \bX_i = \bx, \bA_i = \ba \right \}$. Note that $\Delta_{\ba}$ does not depend on $\bX_i$. We find that
\begin{align*}
    \Delta_0 \triangleq \frac{1}{|\mathcal{A}|} \sum_{\ba \in \mathcal{A}} \bbE \Bigg \{  \bZ_{il} \bZ_{il}^\top \mid \bA_i = \ba, \bX_i = \bf{x} \Bigg \}  &= \begin{pmatrix}
        1 & \mu_{\ba', \bf{x}}^\top  & 0 \\
        \mu_{\ba', \bf{x}}  & \Delta_{\ba'} &  0 \\
        0 & 0 & 0
    \end{pmatrix} + \begin{pmatrix}
        1 & \mu_{\ba, \bf{x}}^\top &  \mu_{\ba, \bf{x}}^\top \\
        \mu_{\ba, \bf{x}} & \Delta_{\ba}  &    \Delta_{\ba} \\
        \mu_{\ba, \bf{x}}  &   \Delta_{\ba} &   \Delta_{\ba}
    \end{pmatrix} \\
    &= \begin{pmatrix}
        2 & \mu_{\ba', \bf{x}}^\top + \mu_{\ba, \bf{x}}^\top   & \mu_{\ba, \bf{x}}^\top \\
        \mu_{\ba', \bf{x}} + \mu_{\ba, \bf{x}}  & \Delta_{\ba'} + \Delta_{\ba}  &  \Delta_{\ba} \\
        \mu_{\ba, \bf{x}}  & \Delta_{\ba} &  \Delta_{\ba}
    \end{pmatrix}.
\end{align*}
By Lemma~\ref{lem:blckpd}, we know that $\sigma_{\min} \left ( \Delta_0\right ) \geq 2\sigma_{\min}(\Delta_1)$, where
\begin{align*}
    \Delta_1 &\triangleq \begin{pmatrix}
        \Sigma_{\ba} + \Sigma_{\ba'} + \frac{(\mu_{\ba, \bf{x}} - \mu_{\ba', \bf{x}})(\mu_{\ba, \bf{x}} - \mu_{\ba', \bf{x}})^\top }{2} &
        \Sigma_{\ba} + \frac{(\mu_{\ba, \bf{x}}- \mu_{\ba', \bf{x}}){\mu_{\ba, \bf{x}}}^\top}{2} \\
        \Sigma_{\ba} + \frac{\mu_{\ba, \bf{x}}(\mu_{\ba, \bf{x}}- \mu_{\ba', \bf{x}})^\top}{2} & \Sigma_{\ba} + \frac{\mu_{\ba, \bf{x}}{\mu_{\ba, \bf{x}}}^\top}{2}
    \end{pmatrix} \\
    &= \begin{pmatrix}
        \Sigma_{\ba} + \Sigma_{\ba'} &
        \Sigma_{\ba} \\
        \Sigma_{\ba}  & \Sigma_{\ba} 
    \end{pmatrix} + \begin{pmatrix}
        \frac{(\mu_{\ba, \bf{x}} - \mu_{\ba', \bf{x}})(\mu_{\ba, \bf{x}} - \mu_{\ba', \bf{x}})^\top }{2} &
        \frac{(\mu_{\ba, \bf{x}}- \mu_{\ba', \bf{x}}){\mu_{\ba, \bf{x}}}^\top}{2} \\
        \frac{\mu_{\ba, \bf{x}}(\mu_{\ba, \bf{x}}- \mu_{\ba', \bf{x}})^\top}{2} & \frac{\mu_{\ba, \bf{x}}{\mu_{\ba, \bf{x}}}^\top}{2}
    \end{pmatrix} \\
    &\succeq \begin{pmatrix}
        \Sigma_{\ba} + \Sigma_{\ba'} &
        \Sigma_{\ba} \\
        \Sigma_{\ba}  & \Sigma_{\ba} 
        \end{pmatrix}.
\end{align*}
We label 
\begin{equation*}
    \Delta_2 \triangleq \begin{pmatrix}
        \Sigma_{\ba} + \Sigma_{\ba'} &
        \Sigma_{\ba} \\
        \Sigma_{\ba}  & \Sigma_{\ba} 
    \end{pmatrix}.
\end{equation*}
By Lemma~\ref{lem:blckpd} again, we find that
\begin{align*}
    \sigma_{\min} \left \{ \Delta_2 \right \} & \geq \sigma_{\min}(\Sigma_{\ba})\sigma_{\min} (\Sigma_{\ba} + \Sigma_{\ba'} - \Sigma_{\ba} {\Sigma_{\ba}}^{-1} \Sigma_{\ba}) \\
    &= \sigma_{\min}(\Sigma_{\ba})\sigma_{\min} (\Sigma_{\ba'}) > 0.
\end{align*}
We can conclude that
\begin{align*}
    &c_{\exp} \sum_{l \in \mathbb{G}_1} \bbE \Bigg \{ \mathbb{I}(il \in \bbT^*)  \frac{1}{|\mathcal{A}|} \sum_{\ba \in \mathcal{A}} \bbE \Bigg \{  h_{\ba}(\bX_{il}) h_{\ba}(\bX_{il})^\top \mid \bA_i = \ba, \bX_i \Bigg \} \mid \bX_i \Bigg \} \\
    & \succ c_{\exp} \sum_{l \in \mathbb{G}_1} \bbE \Bigg \{ \mathbb{I}(il \in \bbT^*)  \Delta_2 \mid \bX_i \Bigg \}  \\
    & \succ c_{\exp}m  \Delta_2. 
\end{align*}
Therefore,
\begin{align*}
    \sigma_y &\geq \lim_{J \to \infty} \sum_{j =1}^J \sum_{i\in \mathcal{E}_{j-1}} \frac{1}{t_{J-1}} c_{\exp}m \sigma_{\min} \left ( \Delta_2 \right ) \\
    &\geq \lim_{J \to \infty} \frac{\overline{\kappa}_J}{t_{J-1}} c_{\exp}m \sigma_{\min} \left ( \Delta_2 \right ) \\
    &= \branchasymp c_{\exp}m \sigma_{\min} \left ( \Delta_2 \right ).
\end{align*}
 Note that there exists $\delta_*$ such that $\branchasymp>\delta_*$ by assumption. Defining
 \begin{equation*}
     \delta_{\bbeta_y} \triangleq \delta^* c_{\exp}m \sigma_{\min} \left ( \Delta_2 \right )
 \end{equation*}
  completes the proof.
\end{proof}

We conclude that $\sigma_{\min}\left ( \branchasymp C_*^y \right ) \geq  \delta_{\bbeta^*_y} $ and we can conclude \textbf{Assumption~\ref{as:stabalizedvariance}} is satisfied for the reward model component of the branching process.

\textbf{Next, we analyze the time model component, $\eta_J^{t}$.}  
The Hessian for this part of the model is
\begin{align*}
    &\frac{\partial^2 \ell_{J}(\zeta)}{\partial\zeta ^2}  =  \sum_{j = 1}^J \sum_{i \in \epochgroup}  -M_i \left ( \frac{1}{\zeta^2}  + \frac{ae^{-\zeta a} - b^2 e^{-\zeta b}}{ [e^{-\zeta a} -e^{-\zeta b}]^2 } \right ).
\end{align*}
We get that
\begin{align*}
    \lim_{J \to \infty } \eta_J^t/t_{J-1} &= \lim_{J \to \infty }  \overline{\kappa}_J/t_{J-1} \left ( \frac{1}{{\zeta^*}^2}  + \frac{ae^{-{\zeta^*} a} - b^2 e^{-{\zeta^*} b}}{ [e^{-{\zeta^*} a} -e^{-{\zeta^*} b}]^2 } \right )m \\
    &= \branchasymp \left ( \frac{1}{{\zeta^*}^2}  + \frac{ae^{-{\zeta^*} a} - b^2 e^{-{\zeta^*} b}}{ [e^{-{\zeta^*} a} -e^{-{\zeta^*} b}]^2 } \right )m  \\
    &> 0.
\end{align*}
 Where the last inequality follows because there exists $\delta_*$ such that $\branchasymp>\delta_*$ by assumption. We can conclude \textbf{Assumption~\ref{as:stabalizedvariance}} is satisfied for the time model component of the branching process.

\textbf{Lastly, we analyze the family size model component, $\eta_J^{m}$.}
\begin{align*}
    &\frac{\partial^2 \ell_{J}(\tau)}{\partial \tau^2} = \sum_{j = 1}^J \sum_{i \in \epochgroup} - \tau^2 \frac{ \left \{ \sum_{\ell=k}^{|\mathbf{A}_i|} e^{\tau \ell}/\ell! \right \} \left \{  \sum_{\ell=k}^{|\mathbf{A}_i|} \ell^2 e^{\tau \ell}/\ell! \right \} -  \left \{  \sum_{\ell=k}^{|\mathbf{A}_i|} \ell e^{\tau \ell}/\ell! \right \}^2 }{\left \{ \sum_{\ell=k}^{|\mathbf{A}_i|} e^{\tau \ell}/\ell! \right \}^2}.
\end{align*}
We find
\begin{align*}
    &\lim_{J \to \infty}  \eta_J/t_{J-1} \\
    &= \lim_{J \to \infty} \sum_{j = 1}^J \sum_{i \in \epochgroup}  {\tau^*}^2 \frac{ \left \{ \sum_{\ell=k}^{|\mathbf{A}_i|} e^{{\tau^*} \ell}/\ell! \right \} \left \{  \sum_{\ell=k}^{|\mathbf{A}_i|} \ell^2 e^{{\tau^*} \ell}/\ell! \right \} -  \left \{  \sum_{\ell=k}^{|\mathbf{A}_i|} \ell e^{{\tau^*} \ell}/\ell! \right \}^2 }{\left \{ \sum_{\ell=k}^{|\mathbf{A}_i|} e^{{\tau^*} \ell}/\ell! \right \}^2}/t_{J-1} \\
    &=  \lim_{J \to \infty} \frac{\overline{\kappa}_J}{t_{J-1}} {\tau^*}^2 \frac{ \left \{ \sum_{\ell=k}^{L} e^{{\tau^*} \ell}/\ell! \right \} \left \{  \sum_{\ell=k}^{L} \ell^2 e^{{\tau^*} \ell}/\ell! \right \} -  \left \{  \sum_{\ell=k}^{L} \ell e^{{\tau^*} \ell}/\ell! \right \}^2 }{\left \{ \sum_{\ell=k}^{L} e^{{\tau^*} \ell}/\ell! \right \}^2} \\
    &=  \branchasymp {\tau^*}^2 \frac{ \left \{ \sum_{\ell=k}^{L} e^{{\tau^*} \ell}/\ell! \right \} \left \{  \sum_{\ell=k}^{L} \ell^2 e^{{\tau^*} \ell}/\ell! \right \} -  \left \{  \sum_{\ell=k}^{L} \ell e^{{\tau^*} \ell}/\ell! \right \}^2 }{\left \{ \sum_{\ell=k}^{L} e^{{\tau^*} \ell}/\ell! \right \}^2}
\end{align*}
 Note that there exists $\delta_*$ such that $\branchasymp>\delta_*$ by assumption. We can conclude \textbf{Assumption~\ref{as:stabalizedvariance}} is satisfied for the family model component of the branching process.

In conclusion, we have confirmed that 
\begin{equation*}
    \lim_{J \to \infty} \eta_J/t_{J-1} = \branchasymp\mathbb{C}
\end{equation*}
where $\branchasymp\mathbb{C} \succ 0$ w.p. 1 on the event $E_\branchasymp$. We establish one more lemma that will be useful in the consistency proof of Section~\ref{app_sec:branch_consistency}.

\begin{lemma}
     \label{lem:iatimefam}
    Recall in the model described in Equation~\ref{eq:rdsmod}, $\bbeta_t \triangleq \zeta$ and $\bbeta_m \triangleq \tau$. Under conditions (1)-(6), for any $\zeta, \tau \in \mathcal{B}$, there exists $\delta_{\bbeta_t}, \delta_{\bbeta_m} > 0$,
    \begin{align*}
       &\lim_{J \to \infty} \sum_{j = 1}^J \bbE \left [  \sum_{i \in \epochgroup} M_i \left ( \frac{1}{\zeta^2}  + \frac{ae^{-\zeta a} - b^2 e^{-\zeta b}}{ [e^{-\zeta a} -e^{-\zeta b}]^2 } \right ) \mid \epochfield \right ]/t_{J-1} \geq \delta_{\bbeta_t} \ \mathrm{a.s.} \\
       &\lim_{J \to \infty} \sum_{j = 1}^J \bbE \left \{  \sum_{i \in \epochgroup}  {\tau}^2 \frac{ \left \{ \sum_{\ell=k}^{|\mathbf{A}_i|} e^{{\tau} \ell}/\ell! \right \} \left \{  \sum_{\ell=k}^{|\mathbf{A}_i|} \ell^2 e^{{\tau} \ell}/\ell! \right \} -  \left \{  \sum_{\ell=k}^{|\mathbf{A}_i|} \ell e^{{\tau} \ell}/\ell! \right \}^2 }{\left \{ \sum_{\ell=k}^{|\mathbf{A}_i|} e^{{\tau} \ell}/\ell! \right \}^2} \mid \epochfield \right \}/t_{J-1} \\
       &\hspace{1cm}\geq \delta_{\bbeta_m} \ \mathrm{a.s.}
    \end{align*}
    on event $E_\branchasymp$.
\end{lemma}
\begin{proof}
    For any $\zeta \in \mathcal{B}$, 
    \begin{align*}
       &\lim_{J \to \infty} \sum_{j = 1}^J \bbE \left [  \sum_{i \in \epochgroup} M_i \left ( \frac{1}{\zeta^2}  + \frac{ae^{-\zeta a} - b^2 e^{-\zeta b}}{ [e^{-\zeta a} -e^{-\zeta b}]^2 } \right ) \mid \epochfield \right ]/t_{J-1} = \\
       &\lim_{J \to \infty} \left ( \overline{\kappa}_J/t_{J-1} \right ) m\left ( \frac{1}{\zeta^2}  + \frac{ae^{-\zeta a} - b^2 e^{-\zeta b}}{ [e^{-\zeta a} -e^{-\zeta b}]^2 } \right ) = \\
       & \branchasymp m\left ( \frac{1}{\zeta^2}  + \frac{ae^{-\zeta a} - b^2 e^{-\zeta b}}{ [e^{-\zeta a} -e^{-\zeta b}]^2 } \right ) \geq \\
       & \delta^* m\left ( \frac{1}{\zeta^2}  + \frac{ae^{-\zeta a} - b^2 e^{-\zeta b}}{ [e^{-\zeta a} -e^{-\zeta b}]^2 } \right ) \ \mathrm{a.s.}
    \end{align*}
    where the last inequality follows from the fact that there exists $\delta_*$ such that $\branchasymp>\delta_*$ by assumption. Defining 
    \begin{equation*}
        \delta_{\bbeta_t} \triangleq \delta^* m\left ( \frac{1}{\zeta^2}  + \frac{ae^{-\zeta a} - b^2 e^{-\zeta b}}{ [e^{-\zeta a} -e^{-\zeta b}]^2 } \right ) 
    \end{equation*}
    completes the first part of the lemma.

    For any $\tau \in \mathcal{B}$, we find that
    \begin{align*}
        &\lim_{J \to \infty} \sum_{j = 1}^J \bbE \left \{  \sum_{i \in \epochgroup}  {\tau}^2 \frac{ \left \{ \sum_{\ell=k}^{|\mathbf{A}_i|} e^{{\tau} \ell}/\ell! \right \} \left \{  \sum_{\ell=k}^{|\mathbf{A}_i|} \ell^2 e^{{\tau} \ell}/\ell! \right \} -  \left \{  \sum_{\ell=k}^{|\mathbf{A}_i|} \ell e^{{\tau} \ell}/\ell! \right \}^2 }{\left \{ \sum_{\ell=k}^{|\mathbf{A}_i|} e^{{\tau} \ell}/\ell! \right \}^2} \mid \epochfield \right \}/t_{J-1} \geq \\
        &\lim_{J \to \infty} \frac{\overline{\kappa}_J}{\tau_{J-1}} {\tau}^2 \frac{ \left \{ \sum_{\ell=k}^{|\mathbf{A}_i|} e^{{\tau} \ell}/\ell! \right \} \left \{  \sum_{\ell=k}^{|\mathbf{A}_i|} \ell^2 e^{{\tau} \ell}/\ell! \right \} -  \left \{  \sum_{\ell=k}^{|\mathbf{A}_i|} \ell e^{{\tau} \ell}/\ell! \right \}^2 }{\left \{ \sum_{\ell=k}^{|\mathbf{A}_i|} e^{{\tau} \ell}/\ell! \right \}^2} \geq \\
        &\branchasymp {\tau}^2 \frac{ \left \{ \sum_{\ell=k}^{|\mathbf{A}_i|} e^{{\tau} \ell}/\ell! \right \} \left \{  \sum_{\ell=k}^{|\mathbf{A}_i|} \ell^2 e^{{\tau} \ell}/\ell! \right \} -  \left \{  \sum_{\ell=k}^{|\mathbf{A}_i|} \ell e^{{\tau} \ell}/\ell! \right \}^2 }{\left \{ \sum_{\ell=k}^{|\mathbf{A}_i|} e^{{\tau} \ell}/\ell! \right \}^2} \geq \\
        &\delta^* {\tau}^2 \frac{ \left \{ \sum_{\ell=k}^{|\mathbf{A}_i|} e^{{\tau} \ell}/\ell! \right \} \left \{  \sum_{\ell=k}^{|\mathbf{A}_i|} \ell^2 e^{{\tau} \ell}/\ell! \right \} -  \left \{  \sum_{\ell=k}^{|\mathbf{A}_i|} \ell e^{{\tau} \ell}/\ell! \right \}^2 }{\left \{ \sum_{\ell=k}^{|\mathbf{A}_i|} e^{{\tau} \ell}/\ell! \right \}^2} \ \mathrm{a.s.}
    \end{align*}
    where the last inequality follows from the fact that there exists $\delta_*$ such that $\branchasymp>\delta_*$ by assumption. Defining 
    \begin{equation*}
        \delta_{\bbeta_m} \triangleq \delta^* {\tau}^2 \frac{ \left \{ \sum_{\ell=k}^{|\mathbf{A}_i|} e^{{\tau} \ell}/\ell! \right \} \left \{  \sum_{\ell=k}^{|\mathbf{A}_i|} \ell^2 e^{{\tau} \ell}/\ell! \right \} -  \left \{  \sum_{\ell=k}^{|\mathbf{A}_i|} \ell e^{{\tau} \ell}/\ell! \right \}^2 }{\left \{ \sum_{\ell=k}^{|\mathbf{A}_i|} e^{{\tau} \ell}/\ell! \right \}^2} 
    \end{equation*}
    completes the proof of the lemma.
    
\end{proof}

\subsubsection{Assumption~\ref{as:IA}}
\textbf{Assumption~\ref{as:IA}} follows from the fact that $\Sigma_J = t_{J-1}$, verification of Assumption~\ref{as:stabalizedvariance} and Lemmas~\ref{lem:covia}-\ref{lem:iatimefam},
\begin{equation*}
    \lim_{J \to \infty} \eta_J/t_{J-1} =\branchasymp \mathbb{C}
\end{equation*}
where $\mathbb{C} \succ 0$ w.p. 1. Consequently,
\begin{equation*}
     \frac{\eta_J}{\overline{\kappa}_J} =  \frac{\eta_J}{t_{J-1}} \frac{t_{J-1}}{\overline{\kappa}_J} \overset{p}{\to} \mathbb{C} \branchasymp \branchasymp^{-1} = \mathbb{C} \succ 0
\end{equation*}
by the continuous mapping theorem. We conclude that \textbf{Assumption~\ref{as:IA}} is satisfied.

\subsubsection{Assumption~\ref{as:equicontinuity}}

We prove \textbf{Assumption~\ref{as:equicontinuity}}.
Set any $\epsilon_{\ddot{l}} > 0 $. We conceptualize $\ddot{l}_i(\bbeta^*)$ as a function $\ddot{\lambda}:\mathcal{B}\times \mathscr{D} \to \mathbb{R}^{k \times k}$. From Section~\ref{app_sec:branchinginference}, we know that $\ddot{\lambda}$ is continuous over both $\mathcal{B}$ and $\mathscr{D}$.  Because $\mathcal{X}$ is compact, we know that $\mathscr{D}$ is compact. Consequently, $\ddot{\lambda}$ is uniformly continuous in $\mathcal{B} \times \mathscr{D}$ (since $\mathcal{B}$ is compact too). Therefore, for any $0<\epsilon \leq \epsilon_{\ddot{l}}$, we can find $\delta_{\epsilon}$ such that for any $\bbeta, \bbeta'$ and $\bd, \bd'$ such 
\begin{equation*}
    \sup_{\|\bbeta - \bbeta'\|_2 + \|\textbf{d} - \textbf{d}'\|_2 \leq \delta_{\epsilon}}  \left  \| \ddot{\lambda}(\bbeta, \bd) - \ddot{\lambda}(\bbeta', \bd') \right \|_2 \leq \epsilon.
\end{equation*}
Since $\{ \bbeta, \bd, \bd': \|\bbeta - \bbeta^*\|_2 \leq  \delta_{\epsilon},  \|\bd - \bd'\|_2  = 0 \} \subseteq \{ \bbeta,\bbeta', \bd, \bd: \|\bbeta - \bbeta'\|_2 + \|\bd - \bd'\|_2 \leq \delta_{\epsilon} \}$, we know that
\begin{equation*}
    \sup_{\bbeta : \| \bbeta - \bbeta^* \|_2 \leq \delta_\epsilon} \sup_{\textbf{d} \in \mathscr{D}} \left  \| \ddot{\lambda}(\bbeta, \bd) - \ddot{\lambda}(\bbeta^*, \bd) \right \|_2 \leq \epsilon.
\end{equation*}
\textbf{Assumption~\ref{as:equicontinuity}} is satisfied.

\subsection{Consistency Proof}
\label{app_sec:branch_consistency}

Note that we do not need Assumptions~\ref{as:lipchitz} and \ref{as:wellseperated} in the proof of consistency for the branching process described in Equation~\ref{eq:rdsmod}. 
We begin with a finite Taylor series expansion,
	\begin{equation*}
		0 = \dot{\mathcal{M}_{J}}(\bbeta^*) + \ddot{\mathcal{M}}_{J}(\bar{\bbeta}^{J})(\widehat{\bbeta}_{J} - \bbeta^*),
	\end{equation*}
	where $\bar{\bbeta}^{J}$ is between $\widehat{\bbeta}_{J}$ and $\bbeta^*$.
    By the concavity of exponential families \citep{brown1986fundamentals} -- this can also be observed in the proofs of Lemmas~\ref{lem:covia}-\ref{lem:iatimefam} -- we know that for all $J \in \mathbb{N}$ and $\bbeta \in \mathcal{B}$, $\ddot{\mathcal{M}}_{J}(\bbeta)/t_{J-1}$ is invertible. Consequently, 
    \begin{align*}
        &-\dot{\mathcal{M}_{J}}(\bbeta^*)/t_{J-1} = \left \{ \ddot{\mathcal{M}}_{J}(\bar{\bbeta}^{J})/t_{J-1} \right \} (\widehat{\bbeta}_{J} - \bbeta^*) \Rightarrow \\
        &-\left \{ \ddot{\mathcal{M}}_{J}(\bar{\bbeta}^{J})/t_{J-1} \right \}^{-1} \dot{\mathcal{M}_{J}}(\bbeta^*)/t_{J-1} = \widehat{\bbeta}_{J} - \bbeta^*.
    \end{align*}
    Under Assumptions~\ref{as:1}-\ref{as:moments} and Assumption~\ref{as:stabalizedvariance}, we know that
	\begin{equation*}
	    \dot{\mathcal{M}}_{J}(\bbeta^*)/t_{J-1}^{1/2} = O_p(1)
	\end{equation*}
    by Section~\ref{app_sec:asconv}.
    Consequently, 
    \begin{equation*}
     \dot{\mathcal{M}}_{J}(\bbeta^*)/t_{J-1} = o_p(1).
    \end{equation*}
    Therefore, we only need to show
    \begin{equation*}
    \left \{- \ddot{\mathcal{M}}_{J}(\bar{\bbeta}^{J})/t_{J-1} \right \}^{-1} = O_p(1)
    \end{equation*}
    to prove that 
    \begin{equation*}
        \widehat{\bbeta}_{J} - \bbeta^* = o_p(1).
    \end{equation*}
     To do this, we show that for any $\bbeta \in \mathcal{B}$,
    \begin{equation}
        \label{eq:consistentgoal}
        \sigma_{\min} \left \{ -\ddot{\mathcal{M}}_{J}(\bbeta)/t_{J-1} \right \} \geq \epsilon^*_{\bbeta}/\rho_{\max} + o_p(1).
    \end{equation}
    This implies that
    \begin{equation*}
        \left \| \left \{-\ddot{\mathcal{M}}_{J}(\bbeta)/t_{J-1} \right \}^{-1} \right \|_2 \leq \rho_{\max}/\epsilon^*_{\bbeta} + o_p(1).
    \end{equation*}
    To show Equation~\ref{eq:consistentgoal}, we first define
    \begin{align*}
        &\nu_J(\bbeta) = -\sum_{j=1}^J \bbE_{\bbeta^*, \widetilde{\bpi}} \left \{ \sum_{i \in \mathcal{E}_{j-1}} \ddot{l}(\bbeta) \ \Bigg | \ \epochfield \right \}, \\
        &\alpha_J(\bbeta) = -\sum_{j=1}^J \bbE_{\bbeta^*, \widehat{\bpi}} \left \{ \sum_{i \in \mathcal{E}_{j-1}} W_i\ddot{l}(\bbeta) \ \Bigg | \ \epochfield \right \}
    \end{align*}
    By Lemmas~\ref{lem:covia}-\ref{lem:iatimefam}, we know that (on the event $E_\branchasymp$) for any $\bbeta \in \mathcal{B}$ and \\ $\epsilon^*_{\bbeta} = \min \left \{ \delta_{\bbeta_m^*}, \delta_{\bbeta_t^*}, \delta_{\bbeta_{\bx}^*}, \delta_{\bbeta_y^*} \right \}$ (where $\delta_{\bbeta_m^*}, \delta_{\bbeta_t^*}, \delta_{\bbeta_{\bx}^*}, \delta_{\bbeta_y^*}$ are defined in Lemmas~\ref{lem:covia}-\ref{lem:iatimefam}),
    \begin{equation}
        \label{eq:invertibility}
        \lim_{J \to \infty} \nu_J(\bbeta) /t_{J-1} \geq \epsilon^*_{\bbeta}.
    \end{equation}
    By Lemmas~\ref{lem:mineigen} and \ref{lem:submult},
    \begin{align*}
        \lim_{J \to \infty} \sigma_{\min} \left \{ \ddot{\mathcal{M}}_{J}(\bbeta)/t_{J-1} \right \} &\geq \frac{1}{\rho_{\max}}  \lim_{J \to \infty} \sigma_{\min} \left \{ \left \{ -\sum_{j = 1}^J \sum_{i \in \epochgroup} W_i \ddot{l}_i(\bbeta) \right \}\alpha_J(\bbeta)^{-1} \nu_J(\bbeta)/t_{J-1} \right \} \\
        &\hspace{-4cm} \geq  \frac{1}{\rho_{\max}} \lim_{J \to \infty} \sigma_{\min} \left \{ \left \{- \sum_{j = 1}^J \sum_{i \in \epochgroup} W_i \ddot{l}_i(\bbeta) \right \}\alpha_J(\bbeta)^{-1}  \right \}  \sigma_{\min} \left \{ \nu_J(\bbeta)/t_{J-1} \right \} \\
        &\hspace{-4cm} = \frac{1}{\rho_{\max}} \lim_{J \to \infty} \sigma_{\min} \left \{ \left \{ -\sum_{j = 1}^J \sum_{i \in \epochgroup} W_i \ddot{l}_i(\bbeta) \right \}\alpha_J(\bbeta)^{-1} - I + I  \right \}  \sigma_{\min} \left \{ \nu_J(\bbeta)/t_{J-1} \right \} \\ 
        &\hspace{-4cm} \geq \frac{1}{\rho_{\max}} \lim_{J \to \infty} \left ( \sigma_{\min} \left \{ \left \{ -\sum_{j = 1}^J \sum_{i \in \epochgroup} W_i \ddot{l}_i(\bbeta) \right \}\alpha_J(\bbeta)^{-1} - I \right \}  + 1  \right )   \sigma_{\min} \left \{ \nu_J(\bbeta)/t_{J-1} \right \}.
    \end{align*}
    The first inequality follows from the inequality
    \begin{equation}
        \label{eq:upperalpha}
        \alpha_J \succ \frac{1}{\rho_{\max}} \nu_J.
    \end{equation}
    Equation~\ref{eq:upperalpha} follows from the fact that ${l}(\bbeta)$ is concave (because each component of the branching process is a member of a full exponential family -- this can also be observed in the proofs of Lemmas~\ref{lem:covia}-\ref{lem:iatimefam}), $-\ddot{l}(\bbeta) \succeq 0$, so for every $\bbeta \in \mathcal{B}$,
    \begin{align*}
     \alpha_J &\triangleq \sum_{j = 1}^J \sum_{i \in \epochgroup} \bbE_{\bbeta^*, \widehat{\bpi}} \left \{ W_i (-\ddot{l}(\bbeta)) \mid \epochfield \right \} \\
     &\succeq \frac{1}{\rho_{\max}} \sum_{j = 1}^J \sum_{i \in \epochgroup} \bbE_{\bbeta^*, \widehat{\bpi}} \left \{ W^2_i (-\ddot{l}(\bbeta))  \mid \epochfield \right \} \\
     &= \frac{1}{\rho_{\max}} \sum_{j = 1}^J \sum_{i \in \epochgroup} \bbE_{\bbeta^*, \widetilde{\bpi}} \left \{ (-\ddot{l}(\bbeta))  \mid \epochfield \right \}.
    \end{align*}
    Additionally, from Property (3) of Section~\ref{app_sec:convrate}, we know that
    \begin{equation}
        \label{eq:secderconv}
        \lim_{J \to \infty} \left \| \frac{\sum_{j = 1}^J \sum_{i \in \epochgroup} W_i \ddot{l}_i(\bbeta) - \alpha_J(\bbeta) }{\overline{\kappa}_J} \right \|_2 = 0 \ \ \mathrm{a.s.}
    \end{equation}
    We can extend this to 
    \begin{align*}
         &  \left \|  \left \{-\sum_{j = 1}^J \sum_{i \in \epochgroup} W_i \ddot{l}(\bbeta) \right \} \alpha_J(\bbeta)^{-1} - I \right \|_2 \\
          & = \left \|  \left \{ \left \{-\sum_{j = 1}^J \sum_{i \in \epochgroup} W_i \ddot{l}(\bbeta) \right \}  - \alpha_J(\bbeta) \right \} \alpha_J(\bbeta)^{-1}\right \|_2 \\
           & \leq \left \|  \left \{ \left \{-\sum_{j = 1}^J \sum_{i \in \epochgroup} W_i \ddot{l}(\bbeta) \right \}  - \alpha_J(\bbeta) \right \} \right \|_2 \left \| \alpha_J(\bbeta)^{-1}\right \|_2 \\
        & \leq \rho_{\max} \left \| \frac{ \left \{  \left \{-\sum_{j = 1}^J \sum_{i \in \epochgroup} W_i \ddot{l}(\bbeta) \right \} - \alpha_J(\bbeta) \right \} }{\overline{\kappa}_J} \right \|_2 \frac{\overline{\kappa}_J}{t_{J-1}} \left \| \nu_J(\bbeta)^{-1}\right \|_2 t_{J-1} \\
        & = o_p(1) \left \{ O_p(1) + o_p(1) \right \} \left \{ \frac{1}{\epsilon^*_{\bbeta}} + o_p(1) \right \} \\
        &= o_p(1).
    \end{align*}
    The third to last equality follows from Equation~\ref{eq:upperalpha}. The second to last equality follows from Equation~\ref{eq:invertibility}, Equation~\ref{eq:secderconv}, and the fact that
    \begin{equation*}
        \lim_{J \to \infty} \frac{\overline{\kappa}_J}{t_{J-1}} = \branchasymp \ \mathrm{a.s.}
    \end{equation*}
    Consequently, we know that
    \begin{equation*}
        \sigma_{\min} \left \{ -\ddot{\mathcal{M}}_{J}(\bbeta)/t_{J-1} \right \} \geq \epsilon_{\bbeta}^*/\rho_{\max} + o_p(1).
    \end{equation*}
    This implies that 
    \begin{equation*}
        t_{J-1}  \left \| \ddot{\mathcal{M}}_{J}(\bbeta)^{-1} \right \|_2 \leq \rho_{\max}/\epsilon_{\bbeta}^* + o_p(1).
    \end{equation*}
    Note that $\left \| \ddot{\mathcal{M}}_{J}(\bbeta)^{-1} \right \|_2$ is a continuous function of $\bbeta$ by Section~\ref{app_sec:likelihood}. Because $\mathcal{B}$ is compact, we know that
    \begin{equation*}
        \sup_{\bbeta \in \mathcal{B}} t_{J-1}  \left \| \ddot{\mathcal{M}}_{J}(\bbeta)^{-1} \right \|_2 \leq \sup_{\bbeta \in \mathcal{B}} \rho_{\max}/\epsilon_{\bbeta}^* + o_p(1) < \infty
    \end{equation*}
    by the extreme value theorem. Consequently, we know that
    \begin{equation*}
        t_{J-1}  \left \| \ddot{\mathcal{M}}_{J}(\bar{\bbeta}_J)^{-1} \right \|_2 = O_p(1).
    \end{equation*}
    This demonstrates that $\widehat{\bbeta}_{J}$ is consistent.

    \subsection{Proof of Theorem~\ref{thm:rlbranching}}
    We have proved consistency and verified Assumptions~\ref{as:generation_asymptotics}-\ref{as:moments} and Assumptions~\ref{as:stabalizedvariance}-\ref{as:lipchitz2}. This proves Theorems~\ref{thm:asympNorm} and \ref{thm:budg}. Theorem~\ref{thm:rlbranching} follows from Theorems~\ref{thm:asympNorm} and \ref{thm:budg}.

\section{Generalized RDS Inference}

\subsection{Hessian of the Covariate Model}
\label{app_sec:hess_cov_mod}

\paragraph{The Hessian for the log-likelihood of the covariate model.} Define $\bX_*^v \triangleq \left  (1, \bX^v \right )$ and $G^{\dagger}_{\ba} = \left  (\bphi_{a}^{\top}, G_{a}^{\top} \right )$.
Define $\mathbb{A}$ to be the set of possible coupon types. 
 We note that each coupon allocation is a set of identical coupons, implying that the sets $\mathcal{A}$ and $\mathbb{A}$ have a one to one correspondence. Consequently, for $\ba \in \mathcal{A}$, there exists $a \in \mathbb{A}$ such that
\begin{equation*}
    \bbI (\bA_i = \ba) = \bbI (A_{i,l} = a)
\end{equation*}
for every $i \in \mathbb{N}$ and $l \in \left \{1,2, \ldots, M_i \right \}$. Consequently, we can represent the complete branching process likelihood as
\begin{align*}
    \mathcal{L}^{\btheta}_{\noepochsamplesize} \left ( \left \{G^{\dagger}_{\ba}, \Sigma_{\ba} \right \}_{\ba \in \mathcal{A}} \right ) &\triangleq \prod_{v =1}^{\noepochsamplesize}(2\pi)^{-p/2} \left |\sum_{\ba \in \mathcal{A}}\Sigma_{\ba} \mathbb{I}(\bA^{R^v} = \ba) \right |^{-1/2} \times \\
    &\hspace{0.5cm} \exp \left [ -\frac{1}{2} \sum_{\ba \in \mathcal{A}} \left \{  (\bX^v -  G^{\dagger}_{\ba}\bX_*^{R^v})^\top \Sigma^{-1}_{\ba}(\bX^v - G^{\dagger}_{\ba}\bX_*^{R^v}) \right \} \mathbb{I}(\bA^{R^v} = \ba) \right ] \\
    &= \prod_{v =1}^{\noepochsamplesize} (2\pi)^{-p/2} \left |\sum_{\ba \in \mathcal{A}}\Sigma_{\ba} \mathbb{I}(\bA^{R^v} = \ba) \right |^{-1/2} \times \\
    &\hspace{0.5cm} \exp \Bigg [ -\frac{1}{2} \sum_{\ba \in \mathcal{A}} \Big \{  {\bX^v}^\top \Sigma^{-1}_{\ba} \bX^v - 2 {\bX_*^{R^v}}^\top {G^{\dagger}_{\ba}}^\top \Sigma^{-1}_{\ba}\bX^v  + \\
    &\hspace{1.5cm} {\bX_*^{R^v}}^\top {G^{\dagger}_{\ba}}^\top \Sigma^{-1}_{\ba} {G^{\dagger}_{\ba}} \bX_*^{R^v}  \Big \} \mathbb{I}(\bA^{R^v} = \ba) \Bigg ].
\end{align*}
We will use this likelihood for the proofs that follow.

For $\ba \in \mathcal{A}$, we reparameterize $\Omega_{\ba} = - \frac{1}{2} \Sigma_{\ba}^{-1}$ and $\Gamma_{\ba} = \Sigma_{\ba}^{-1}{G^{\dagger}_{\ba}}$.
\begin{align*}
    \mathcal{L}^{\btheta}_{\noepochsamplesize} \left ( \left \{\Gamma_{\ba}, \Omega_{\ba} \right \}_{\ba \in \mathcal{A}} \right ) &= \prod_{v =1}^{\noepochsamplesize} (2\pi)^{-p/2} |\sum_{\ba \in \mathcal{A}}-2\Omega_{\ba} \mathbb{I}(\bA^{R^v} = \ba)|^{1/2} \times \\
    &\hspace{0.5cm} \exp \Bigg [ \sum_{\ba \in \mathcal{A}} \Big \{  {\bX^v}^\top \Omega_{\ba} \bX^v + {\bX_*^{R^v}}^\top \Gamma_{\ba}^\top\bX^v  + \\
    &\hspace{1.5cm} {\frac{1}{4}\bX_*^{R^v}}^\top \Gamma_{\ba}^\top {\Omega_{\ba}}^{-1} \Gamma_{\ba} \bX_*^{R^v}  \Big \} \mathbb{I}(\bA^{R^v} = \ba) \Bigg ] .
\end{align*}
The reparameterized log-likelihood is
\begin{align*}
    \mathcal{\ell}^{\btheta}_{{\noepochsamplesize}} \left ( \left \{\Gamma_{\ba}, \Omega_{\ba} \right \}_{\ba \in \mathcal{A}} \right ) &= \sum_{v=1}^{{\noepochsamplesize}}\sum_{\ba \in \mathcal{A}} \Bigg [ (-p/2)\log(2\pi) + \frac{1}{2}\log|-2\Omega_{\ba}| + \\
    &\hspace{2cm} {\bX^v}^\top \Omega_{\ba} \bX^v + {\bX_*^{R^v}}^\top \Gamma_{\ba}^\top\bX^v  + \\
    &\hspace{2cm} \frac{1}{4}{\bX_*^{R^v}}^\top \Gamma_{\ba}^\top {\Omega_{\ba}}^{-1} \Gamma_{\ba} \bX_*^{R^v} \Bigg ] \mathbb{I}(\bA^{R^v} = \ba)\\
    &= \sum_{v=1}^{{\noepochsamplesize}}\sum_{\ba \in \mathcal{A}} \Bigg [ (-p/2)\log(2\pi) + \frac{1}{2}\log|-2\Omega_{\ba}| + \\
    &\hspace{2cm} \mathrm{tr} \left ( \Omega_{\ba} \bX^v {\bX^v}^\top \right ) + \mathrm{tr} \left ( \Gamma_{\ba}^\top\bX^v {\bX_*^{R^v}}^\top \right ) + \\
    &\hspace{2cm} \mathrm{tr} \left (\frac{1}{4} \Gamma_{\ba}^\top {\Omega_{\ba}}^{-1} \Gamma_{\ba} \bX_*^{R^v} {\bX_*^{R^v}}^\top \right )  \Bigg ] \mathbb{I}(\bA^{R^v} = \ba),
\end{align*}
where the second equality follows from rearranging terms and using the properties of the trace operator.
For $\ba \in \mathcal{A}$, define ${\noepochsamplesize}_{\ba} = \sum_{v=1}^{{\noepochsamplesize}}\mathbb{I}(\bA^{R^v} = \ba)$ and $\bV^{{\noepochsamplesize}}_{\ba} \in \mathbb{R}^{(p+1)\times (p+1)}$ such that $\bV^{{\noepochsamplesize}}_{\ba} = \sum_{v=1}^{{\noepochsamplesize}}\bX_*^{R^v} {\bX_*^{R^v}}^\top  \mathbb{I}(\bA^{R^v} = \ba) $. We apply the differential operator two times and find
\begin{align*}
    \bd^2\mathcal{\ell}^{\btheta}_{{\noepochsamplesize}} \left ( \left \{\Gamma_{\ba}, \Omega_{\ba} \right \}_{\ba \in \mathcal{A}} \right ) &= \sum_{v=1}^{{\noepochsamplesize}}\sum_{\ba \in \mathcal{A}} -\bd^2 \Bigg \{ -\frac{1}{2}\log|-2\Omega_{\ba}| - \\
      &\hspace{2.7cm} \frac{1}{4} \mathrm{tr} \left ( \Gamma_{\ba}^\top {\Omega_{\ba}}^{-1} \Gamma_{\ba} \bX_*^{R^v} {\bX_*^{R^v}}^\top \right ) \Bigg \} \mathbb{I}(\bA^{R^v} = \ba) \\
      &= -\bd \Bigg \{ -\frac{{\noepochsamplesize}_{\ba}}{2}\mathrm{tr} \left (\Omega_{\ba}^{-1} \bd \Omega_{\ba} \right ) - \frac{1}{4} \mathrm{tr} \left (2 \Gamma_{\ba}^\top \Omega_{\ba}^{-1} \bd \Gamma_{\ba} \bV^{{\noepochsamplesize}}_{\ba} \right ) + \\
      &\hspace{1.5cm} \frac{1}{4} \mathrm{tr} \left ( \Gamma_{\ba}^\top {\Omega_{\ba}}^{-1} \bd \Omega_{\ba} {\Omega_{\ba}}^{-1} \Gamma_{\ba} \bV^{{\noepochsamplesize}}_{\ba} \right ) \Bigg \}  \\
      &= -\Bigg \{ \frac{{\noepochsamplesize}_{\ba}}{2}\mathrm{tr} \left (\Omega_{\ba}^{-1} \bd \Omega_{\ba} \Omega_{\ba}^{-1} \bd \Omega_{\ba} \right ) - \\
      &\hspace{1.2cm}\frac{1}{2} \mathrm{tr} \left ( \Gamma_{\ba}^\top \Omega_{\ba}^{-1} \bd \Omega_{\ba} \Omega_{\ba}^{-1} \bd \Omega_{\ba} \Omega_{\ba}^{-1} \Gamma_{\ba} \bV^{{\noepochsamplesize}}_{\ba} \right ) - \\ 
      &\hspace{1.2cm}\frac{1}{2} \mathrm{tr} \left ( \bd \Gamma_{\ba}^\top \Omega_{\ba}^{-1} \bd \Gamma_{\ba} \bV^{{\noepochsamplesize}}_{\ba} \right ) + \\
      &\hspace{1.2cm} \mathrm{tr} \left ( \Gamma_{\ba}^\top \Omega_{\ba}^{-1} \bd \Omega_{\ba} \Omega_{\ba}^{-1} \bd \Gamma_{\ba} \bV^{{\noepochsamplesize}}_{\ba} \right ) +  \\
      &\hspace{1.2cm} \frac{1}{4} \mathrm{tr} \left ( \Gamma_{\ba}^\top {\Omega_{\ba}}^{-1} \bd \Omega_{\ba} {\Omega_{\ba}}^{-1} \Gamma_{\ba} \bV^{{\noepochsamplesize}}_{\ba} \right ) \Bigg \}.
\end{align*}
We observe that for $\ba, \ba' \in \mathcal{A}$ and $\ba \neq \ba'$,  
\begin{equation*}
    \frac{\partial \mathcal{\ell}^{\btheta}_{{\noepochsamplesize}} \left ( \left \{\Gamma_{\ba}, \Omega_{\ba} \right \}_{\ba \in \mathcal{A}} \right )} { \partial \left (  \mathrm{vec}(\Gamma_{\ba}),  \mathrm{vec}(\Omega_{\ba}) \right )\partial \left (  \mathrm{vec}(\Gamma_{\ba'}),  \mathrm{vec}(\Omega_{\ba'}) \right )} = [0]_{(p+1)^2 \times p^2}. 
\end{equation*}
We express
\begin{align*}
    &\ddot{\ell}^{\btheta}_{{\noepochsamplesize}} \left ( \left \{\Gamma_{\ba}, \Omega_{\ba} \right \}_{\ba \in \mathcal{A}} \right ) \\
    &= -\mathrm{diag} \left [\left \{\begin{pmatrix}
        \bV^{{\noepochsamplesize}}_{\ba} \otimes \Sigma_{\ba} & 2\left ( \bV^{{\noepochsamplesize}}_{\ba} {G^{\dagger}_{\ba}}^\top \otimes \Sigma_{\ba} \right ) \\
        2 \left ( {G^{\dagger}_{\ba}}\bV^{{\noepochsamplesize}}_{\ba} \otimes \Sigma_{\ba} \right ) & 4 \left ( {G^{\dagger}_{\ba}} \bV^{{\noepochsamplesize}}_{\ba} {G^{\dagger}_{\ba}}^\top \otimes \Sigma_{\ba} \right ) + 2{\noepochsamplesize}_{\ba}\left ( \Sigma_{\ba} \otimes \Sigma_{\ba} \right )
    \end{pmatrix} \right \}_{\ba \in \mathcal{A}}
    \right ] \\
    &= -\mathrm{diag} \left [\left \{ \begin{pmatrix}
        \bV^{{\noepochsamplesize}}_{\ba} & 2 \bV^{{\noepochsamplesize}}_{\ba} {G^{\dagger}_{\ba}}^\top\\
        2 {G^{\dagger}_{\ba}}\bV^{{\noepochsamplesize}}_{\ba}  & 4 {G^{\dagger}_{\ba}} \bV^{{\noepochsamplesize}}_{\ba} {G^{\dagger}_{\ba}}^\top  + 2 {\noepochsamplesize}_{\ba}\Sigma_{\ba}
    \end{pmatrix} \otimes \Sigma_{\ba} \right \}_{\ba \in \mathcal{A}}
    \right ].
\end{align*}
The second equality follows from properties of kronecker products.

\subsection{Consistency of an M-estimator}
	
	We establish the conditions for consistency of an M-estimator under a sequence of concave estimating equations. This will be leveraged in Section~\ref{app_sec:netinf}.
	\begin{lemma}[Theorem 2.7 from \cite{newey1994chapter}]
		\label{lem:concave}
		If there is a function $Q_0(\theta)$ and sequence of functions $\{\widehat{Q}_n(\theta)\}_{n \geq  1}$  such that
		\begin{enumerate}
			\item $Q_0(\theta)$ is uniquely maximized at $\theta_0$.
			\item $\theta_0$ is an element of the interior of a convex set $\Theta$.
			\item For all $n \geq 1$, $\widehat{Q}_n(\theta)$ is concave.
			\item $\widehat{Q}_n(\theta) \overset{p}{\to} Q_0(\theta)$ for all $\theta \in \Theta$.
		\end{enumerate}
		Then, $\widehat{\theta}_n = \arg \max_\theta \widehat{Q}_n(\theta) $ exists with probability approaching one, and $\widehat{\theta}_n \overset{p}{\to} \theta_0$.
	\end{lemma}

\subsection{Generalized RDS Inference Proof}
\label{app_sec:netinf}

We prove consistency of the MLE in this section. First, we establish some helpful lemmas.

\begin{lemma}
    \label{lem:log_concave}
    Under Assumptions \ref{as:1}-\ref{as:3}, the log-likelihood for the branching process specified in Equation~\ref{eq:rdsmod} is concave.
\end{lemma}
\begin{proof}
Recall that for $\ba \in \mathcal{A}$, $\Omega_{\ba} \triangleq - \frac{1}{2} \Sigma_{\ba}^{-1}$ and $\Gamma_{\ba} \triangleq \Sigma_{\ba}^{-1}{G^{\dagger}_{\ba}}$. For $\ba \in \mathcal{A}$, define $\noepochsamplesize_{\ba} \triangleq \sum_{v=1}^{\noepochsamplesize}\mathbb{I}(\bA^{R^v} = \ba)$ and $\bV^{{\noepochsamplesize}}_\ba \in \mathbb{R}^{(p+1)\times (p+1)}$ such that $\bV^{\noepochsamplesize}_{\ba}\triangleq \sum_{v=1}^{\noepochsamplesize}\left (1, \bX^{r^v} \right) \left (1, \bX^{r^v} \right )^\top  \mathbb{I}(\bA^{R^v} = \ba) $.
From Section~\ref{app_sec:hess_cov_mod}, we know that
\begin{align*}
    \ddot{\ell}^{\btheta}_{{\noepochsamplesize}} \left ( \left \{\Gamma_{\ba}, \Omega_{\ba} \right \}_{\ba \in \mathcal{A}} \right ) &=
     -\mathrm{diag} \left [\left \{ \begin{pmatrix}
        \bV^{\noepochsamplesize}_{\ba}& 2 \bV^{\noepochsamplesize}_{\ba}{G^{\dagger}_{\ba}}^\top\\
        2 {G^{\dagger}_{\ba}}\bV^{\noepochsamplesize}_{\ba} & 4 {G^{\dagger}_{\ba}} \bV^{\noepochsamplesize}_{\ba}{G^{\dagger}_{\ba}}^\top  + 2 {\noepochsamplesize}_{\ba}\Sigma_{\ba}
    \end{pmatrix} \otimes \Sigma_{\ba} \right \}_{\ba \in \mathcal{A}}
    \right ].
\end{align*}
We begin by analyzing the quantity,
\begin{align*}
    \mathbb{M}^{{\noepochsamplesize}} \triangleq \begin{pmatrix}
        \bV^{\noepochsamplesize}_{\ba}& 2 \bV^{\noepochsamplesize}_{\ba}{G^{\dagger}_{\ba}}^\top\\
        2 {G^{\dagger}_{\ba}}\bV^{\noepochsamplesize}_{\ba} & 4 {G^{\dagger}_{\ba}} \bV^{\noepochsamplesize}_{\ba}{G^{\dagger}_{\ba}}^\top  + 2{\noepochsamplesize}_{\ba} \Sigma_{\ba}
    \end{pmatrix}.
\end{align*}
By Lemma~\ref{lem:blckpd}, we know that $\mathbb{M}^{{\noepochsamplesize}}$
is positive semi-definite since $\bV^{{\noepochsamplesize}}_{\ba}$ and $4 {G^{\dagger}_{\ba}} \bV^{\noepochsamplesize}_{\ba}{G^{\dagger}_{\ba}}^\top  + 2 {\noepochsamplesize}_{\ba}\Sigma_{\ba} - 4G\bV^{\noepochsamplesize}_{\ba}{\bV^{{\noepochsamplesize}}_{\ba}}^{-1} \bV^{\noepochsamplesize}_{\ba}{G^{\dagger}_{\ba}}^\top = 2 {\noepochsamplesize}_{\ba}\Sigma_{\ba}$ are positive semi-definite. Consequently, because $\mathbb{M}^{{\noepochsamplesize}}$ and $\Sigma_{\ba}$ are positive semi-definite regardless of $\ba \in \mathcal{A}$,
\begin{equation*}
    0 \succeq \ddot{\ell}^{\btheta}_{{\noepochsamplesize}} \left ( \left \{\Gamma_{\ba}, \Omega_{\ba} \right \}_{\ba \in \mathcal{A}} \right ).
\end{equation*}
We conclude that ${\ell}_{{\noepochsamplesize}}^{\btheta}\left ( \left \{\Gamma_{\ba}, \Omega_{\ba} \right \}_{\ba \in \mathcal{A}} \right )/{\noepochsamplesize}$ is concave for all ${\noepochsamplesize} \in \mathbb{N}$.
\end{proof}

We now show that the Hessian of the log-likelihood is negative definite almost surely under Assumption~\ref{as:inf_iar}.
\begin{lemma}
\label{lem:IAmisp}
     Under Assumptions \ref{as:1}-\ref{as:3}, \ref{as:budg}, \ref{as:inf_iar}, and Conditions (C2) and (C4),, assume that the true underlying model for RDS is indexed by $\btheta$, and the working model is the branching process specified in Equation~\ref{eq:rdsmod}.
    Define the MLE of the working model covariate distribution as $\widehat{\bbeta}_\bx^{\noepochsamplesize}(\btheta) = \arg \max_{\bbeta_\bx \in \mathcal{B}} \ell_{\noepochsamplesize}^{\btheta} (\bbeta_\bx)$. 
    For any compact set $\mathcal{B}' \subset \mathcal{B}$, there exists $\delta_{\btheta, \mathcal{B}'} > 0$ such that
     \begin{align*}
    \lim_{\noepochsamplesize \to \infty} \inf_{\bbeta_{\bx} \in \mathcal{B}'} \sigma_{\min} \left ( - \ddot{\ell}^{\btheta}_{\noepochsamplesize} \left ( \bbeta_{\bx} \right ) /\noepochsamplesize \right ) 
    &\geq \delta_{\btheta, \mathcal{B}'} \ \ \mathrm{a.s.}
    \end{align*}
\end{lemma}

\begin{proof}
Recall that for $\ba \in \mathcal{A}$, $\Omega_{\ba} \triangleq - \frac{1}{2} \Sigma_{\ba}^{-1}$ and $\Gamma_{\ba} \triangleq \Sigma_{\ba}^{-1}{G^{\dagger}_{\ba}}$. For $\ba \in \mathcal{A}$, define $\noepochsamplesize_{\ba} \triangleq \sum_{v=1}^{\noepochsamplesize}\mathbb{I}(\bA^{R^v} = \ba)$ and $\bV^{{\noepochsamplesize}}_\ba \in \mathbb{R}^{(p+1)\times (p+1)}$ such that $\bV^{\noepochsamplesize}_{\ba}\triangleq \sum_{v=1}^{\noepochsamplesize}\left (1, \bX^{r^v} \right) \left (1, \bX^{r^v} \right )^\top  \mathbb{I}(\bA^{R^v} = \ba) $.
From Section~\ref{app_sec:hess_cov_mod}, we know that
\begin{align*}
    \ddot{\ell}^{\btheta}_{{\noepochsamplesize}} \left ( \left \{\Gamma_{\ba}, \Omega_{\ba} \right \}_{\ba \in \mathcal{A}} \right ) &=
     -\mathrm{diag} \left [\left \{ \begin{pmatrix}
        \bV^{\noepochsamplesize}_{\ba}& 2 \bV^{\noepochsamplesize}_{\ba}{G^{\dagger}_{\ba}}^\top\\
        2 {G^{\dagger}_{\ba}}\bV^{\noepochsamplesize}_{\ba} & 4 {G^{\dagger}_{\ba}} \bV^{\noepochsamplesize}_{\ba}{G^{\dagger}_{\ba}}^\top  + 2 {\noepochsamplesize}_{\ba}\Sigma_{\ba}
    \end{pmatrix} \otimes \Sigma_{\ba} \right \}_{\ba \in \mathcal{A}}
    \right ].
\end{align*}
We begin by analyzing the quantity,
\begin{align*}
    \mathbb{M}^{{\noepochsamplesize}} \triangleq \begin{pmatrix}
        \bV^{\noepochsamplesize}_{\ba}& 2 \bV^{\noepochsamplesize}_{\ba}{G^{\dagger}_{\ba}}^\top\\
        2 {G^{\dagger}_{\ba}}\bV^{\noepochsamplesize}_{\ba} & 4 {G^{\dagger}_{\ba}} \bV^{\noepochsamplesize}_{\ba}{G^{\dagger}_{\ba}}^\top  + 2{\noepochsamplesize}_{\ba} \Sigma_{\ba}
    \end{pmatrix}.
\end{align*}
By Lemma~\ref{lem:blckpd}, we know that $\mathbb{M}^{{\noepochsamplesize}}$
is positive semi-definite since $\bV^{{\noepochsamplesize}}_{\ba}$ and $4 {G^{\dagger}_{\ba}} \bV^{\noepochsamplesize}_{\ba}{G^{\dagger}_{\ba}}^\top  + 2 {\noepochsamplesize}_{\ba}\Sigma_{\ba} - 4G\bV^{\noepochsamplesize}_{\ba}{\bV^{{\noepochsamplesize}}_{\ba}}^{-1} \bV^{\noepochsamplesize}_{\ba}{G^{\dagger}_{\ba}}^\top = 2 {\noepochsamplesize}_{\ba}\Sigma_{\ba}$ are positive semi-definite. Consequently, because $\mathbb{M}^{{\noepochsamplesize}}$ and $\Sigma_{\ba}$ are positive semi-definite regardless of $\ba \in \mathcal{A}$,
\begin{equation*}
    0 \succeq \ddot{\ell}^{\btheta}_{{\noepochsamplesize}} \left ( \left \{\Gamma_{\ba}, \Omega_{\ba} \right \}_{\ba \in \mathcal{A}} \right ).
\end{equation*}

We now show that $-\ddot{\ell}^{\btheta}_{{\noepochsamplesize}}$ is asymptotically \textbf{positive definite}. 
By Lemma~\ref{lem:blckpd}, 
\begin{equation*}
    \sigma_{\min}\left ( \mathbb{M}^{{\noepochsamplesize}}/{\noepochsamplesize} \right ) \geq \sigma_{\min} \left ( \bV^{{\noepochsamplesize}}_{\ba}/{\noepochsamplesize} \right )  \sigma_{\min} \left ( 2 {\noepochsamplesize}_{\ba}\Sigma_{\ba}/{\noepochsamplesize} \right ).
\end{equation*}
To characterize the convexity of $\mathcal{\ell}_{{\noepochsamplesize}}^{\btheta} \left ( \left \{{G^{\dagger}_{\ba}}, \Sigma_{\ba} \right \}_{\ba \in \mathcal{A}} \right )/{\noepochsamplesize}$ as ${\noepochsamplesize} \to \infty$, we first analyze
\begin{align*}
    \lim_{{\noepochsamplesize} \to \infty} \bV^{{\noepochsamplesize}}_{\ba}/{\noepochsamplesize} &= \lim_{{\noepochsamplesize} \to \infty} \frac{1}{{\noepochsamplesize}}\sum_{v=1}^{\noepochsamplesize}\left (1, \bX^{R^v} \right) \left (1, \bX^{R^v} \right )^\top  \mathbb{I}(\bA^{R^v} = \ba) \\ &= 
    \lim_{{\noepochsamplesize} \to \infty} \frac{1}{{\noepochsamplesize}}\sum_{v=1}^{\noepochsamplesize}\begin{pmatrix}
        1 & {\bX^{R^v}}^\top \\
        {\bX^{R^v}} & \bX^{R^v} {\bX^{R^v}}^\top 
    \end{pmatrix}\mathbb{I}(\bA^{R^v} = \ba).
\end{align*}
Define fields $\mathcal{F}^{\noepochsamplesize} = \sigma \left [\left\lbrace R^v, T^v, \bX^v, Y^v , \bA^v, C^v \right\rbrace_{v=1}^{{\noepochsamplesize}} \right ]$ and quantity
\begin{equation*}
    \Delta^v \triangleq  \bbE \left \{ (1, \bX^{R^v}) (1, {\bX^{R^v}})^\top  \mathbb{I}(\bA^{R^v} = \ba) \mid \mathcal{F}^{v-1} \right \}.
\end{equation*}
By Lemma~\ref{lem:blckpd},
\begin{align*}
    \Delta^v  \succ
    &\bbE \left \{ {\bX^{R^v}} {\bX^{R^v}}^\top  \mathbb{I}(\bA^{R^v} = \ba) \mid \mathcal{F}^{v-1} \right \}.
\end{align*}
Additionally, by Assumption~\ref{as:inf_iar}, there exists $N$ such that for ${\noepochsamplesize} \geq N$, $\Delta^{\noepochsamplesize} \succeq \Delta \succ 0$.

We know that $\forall v \in \mathbb{N}$, every entry of $\bX^{v}  {\bX^{v}}^\top $ is bounded because $\mathcal{X}$ is compact.
By Theorem~\ref{thm:SLLN},
\begin{align*}
    \lim_{{\noepochsamplesize} \to \infty } \left \{ \sum_{v=1}^{\noepochsamplesize} \left (1, \bX^{R^v} \right ) \left (1, \bX^{R^v} \right )^\top  \mathbb{I}(\bA^{R^v} = \ba) - \sum_{v =  1}^{\noepochsamplesize} \Delta^v \right \}/{\noepochsamplesize} = 0 \ \ \mathrm{a.s.}
\end{align*}
There exists $N \in \mathbb{N}$ such that
\begin{align*}
    &\lim_{{\noepochsamplesize} \to \infty} \frac{1}{{\noepochsamplesize}}\sum_{v=1}^{\noepochsamplesize} \left (1, \bX^{R^v} \right ) \left (1, \bX^{R^v} \right )^\top  \mathbb{I}(\bA^{R^v} = \ba) \\
    &= \lim_{{\noepochsamplesize} \to \infty} \left \{ \sum_{v=1}^{\noepochsamplesize} \left (1, \bX^{R^v} \right ) \left (1, \bX^{R^v} \right )^\top  \mathbb{I}(\bA^{R^v} = \ba) - \sum_{v =  1}^{\noepochsamplesize} \Delta^v \right \}/{\noepochsamplesize} \\
    &\hspace{1cm} + \lim_{{\noepochsamplesize} \to \infty} \sum_{v =  1}^{\noepochsamplesize} \Delta^v/{\noepochsamplesize} \\
    &= \lim_{{\noepochsamplesize} \to \infty} \sum_{v =  1}^{\noepochsamplesize} \Delta^v/{\noepochsamplesize} \ \ \mathrm{a.s.} \\
    &= \lim_{{\noepochsamplesize} \to \infty} \sum_{v =  N}^{\noepochsamplesize} \Delta^v/\noepochsamplesize + \sum_{v=  1}^N \Delta^v/N  \ \ \mathrm{a.s.} \\
    &\succeq \lim_{{\noepochsamplesize} \to \infty} \frac{{\noepochsamplesize}-N}{{\noepochsamplesize}} \Delta \\
    &\succeq \Delta
\end{align*}
Consequently, by the continuous mapping theorem
\begin{align*}
    \lim_{{\noepochsamplesize} \to \infty} \sigma_{\min} \left \{ \frac{1}{{\noepochsamplesize}}\sum_{v=1}^{\noepochsamplesize} \left (1, \bX^{R^v} \right ) \left (1, \bX^{R^v} \right )^\top  \mathbb{I}(\bA^{R^v} = \ba) \right \} \geq \sigma_{\min}(\Delta).
\end{align*}
Defining $\epsilon_1 \triangleq\sigma_{\min}(\Delta)  > 0$, we conclude that $\lim_{{\noepochsamplesize} \to \infty} \sigma_{\min}\left (\bV^{{\noepochsamplesize}}_{\ba}/{\noepochsamplesize} \right ) \geq \epsilon_1$ a.s. Next, we analyze $2 {\noepochsamplesize}_{\ba}\Sigma_{\ba}/{\noepochsamplesize}$. By Assumption~\ref{as:inf_iar} and Theorem~\ref{thm:SLLN}, 
\begin{align*}
    \lim_{{\noepochsamplesize} \to \infty} {\noepochsamplesize}_{\ba}/{\noepochsamplesize} &=  \lim_{{\noepochsamplesize} \to \infty} \left \{ \sum_{v=1}^{\noepochsamplesize}\mathbb{I}(\bA^{R^v} = \ba) - \bbP(\bA^{R^v} = \ba \mid \mathcal{F}^{{\noepochsamplesize}-1}) \right \}/{\noepochsamplesize} + \sum_{v=1}^{\noepochsamplesize} \bbP(\bA^{R^v} = \ba \mid \mathcal{F}^{{\noepochsamplesize}-1})/{\noepochsamplesize} \\
    &\geq \lim_{{\noepochsamplesize} \to \infty} \sum_{v=1}^{\noepochsamplesize} \bbP(\bA^{R^v} = \ba \mid \mathcal{F}^{{\noepochsamplesize}-1})/{\noepochsamplesize} \\
    &\geq \lim_{{\noepochsamplesize} \to \infty} \sum_{v=N}^{\noepochsamplesize} \bbP(\bA^{R^v} = \ba \mid \mathcal{F}^{{\noepochsamplesize}-1})/{\noepochsamplesize} \\
    &\geq \delta \ \ \mathrm{a.s.}
\end{align*}
Consequently, $\lim_{{\noepochsamplesize} \to \infty} {\noepochsamplesize}_{\ba}/{\noepochsamplesize} \geq \delta $ a.s. as ${\noepochsamplesize} \to \infty$. 

Therefore,
\begin{equation*}
    \lim_{{\noepochsamplesize} \to \infty} \sigma_{\min} \left ( 2 {\noepochsamplesize}_{\ba}\Sigma_{\ba}/{\noepochsamplesize} \right )  \geq 2\sigma_{\min}(\Sigma_{\ba}) \delta \ \ \mathrm{a.s.}
\end{equation*}
    Because the minimum eigenvalue of a matrix is a continuous function, by the extreme value theorem, we know that for some $\alpha > 0$,
    \begin{equation*}
        \min_{a \in \mathcal{A}} \min_{\Sigma_{\ba} \in \mathcal{B}'} \sigma_{\min}(\Sigma_{\ba}) \geq \alpha.
    \end{equation*}
    Defining $\epsilon_2 \triangleq 2 \alpha \delta \epsilon_1$, we find that
\begin{equation*}
    \lim_{{\noepochsamplesize} \to \infty} \sigma_{\min} \left ( \mathbb{M}^{\noepochsamplesize}/{\noepochsamplesize} \right ) \geq \epsilon_2 \ \ \mathrm{a.s.}
\end{equation*}
Defining $\epsilon_3 \triangleq \epsilon_2 \alpha $, for any $\ba \in \mathcal{A}$, we find that
\begin{align*}
    \lim_{{\noepochsamplesize}\to \infty} \sigma_{\min} \left \{ - \ddot{\ell}^{\btheta}_{{\noepochsamplesize}} \left ( \Gamma_{\ba}, \Omega_{\ba} \right ) \right \}/{\noepochsamplesize} &= \sigma_{\min}\left ( \mathbb{M}^{\noepochsamplesize} \otimes \Sigma_{\ba} \right ) \\ 
    &= \sigma_{\min}\left ( \mathbb{M}^{\noepochsamplesize} \right ) \sigma_{\min} \left (\Sigma_{\ba} \right ) \\
    &\geq   \epsilon_3 \ \ \mathrm{a.s.}
\end{align*}
    Line~2 follows because for square matrices $A$ and $B$, $\sigma_{\min}(A \otimes B) = \sigma_{\min}(A)\sigma_{\min}(B)$.
    In conclusion,
\begin{align*}
    \lim_{{\noepochsamplesize} \to \infty} \sigma_{\min} \left \{ \ddot{\ell}^{\btheta}_{{\noepochsamplesize}} \left ( \left \{\Gamma_{\ba}, \Omega_{\ba} \right \}_{\ba \in \mathcal{A}} \right ) /{\noepochsamplesize} \right \} &= \lim_{{\noepochsamplesize} \to \infty} \prod_{\ba \in \mathcal{A}} \sigma_{\min} \left \{ \ddot{\ell}^{\btheta}_{{\noepochsamplesize}} \left ( \Gamma_{\ba}, \Omega_{\ba} \right ) \right \}/{\noepochsamplesize} \\
    &\leq -\epsilon_3^{|\mathcal{A}|} \ \ \mathrm{a.s.} \\
    &< 0 \ \ \mathrm{a.s.},
\end{align*}
where the first line follows from Lemma~\ref{lem:blckpd} (where the off-diagonal blocks are zero).
We complete the proof by defining $\delta_{\btheta, \mathcal{B}'} \triangleq \epsilon_3$.

\end{proof}

We verify the last part of Assumption~\ref{as:gen_inf_iar}.

\begin{lemma}
    \label{lem:O_p_1}
    Under Assumptions \ref{as:1}-\ref{as:3}, \ref{as:budg}, \ref{as:inf_iar}, and Conditions (C2) and (C4), assume that the true underlying model for RDS is indexed by $\btheta$, and the working model is the branching process specified in Equation~\ref{eq:rdsmod}. For any $\bbeta_{\bx} \in \mathcal{B}$,
\begin{equation*}
    \dot{\ell}_{{\noepochsamplesize}}^{\btheta} \left ( \bbeta_\bx \right )/{\noepochsamplesize} = O_p(1)
\end{equation*}
as $\kappa \to \infty$.
\end{lemma} 
\begin{proof}
    Write
    \begin{equation*}
        \dot{\ell}_{{\noepochsamplesize}}^{\btheta} \left ( \bbeta_\bx \right ) = \sum_{v = 1}^{\kappa} q^v \left (\bbeta_{\bx} \right ) = \sum_{v = 1}^{\kappa} q\left (\bbeta_{\bx}, \left \{Y^v, \bX^v, T^v, R^v, \bA^v \right \} \right ).
    \end{equation*} We see from Section~\ref{app_sec:hess_cov_mod} that the function $q$ is continuous over the compact space $ \mathcal{Y} \times \mathcal{X} \times [t_{\min}, t_{\max}] \times \left \{1,2, \ldots v \right \} \times \mathcal{A}$ (which defines the data space). Consequently, for any given value of $\bbeta_{\bx}$, $q$ has a maximum over possible data values. Call this maximum $q^*(\bbeta_{\bx})$. Consequently, 
\begin{equation*}
    \dot{\ell}_{{\noepochsamplesize}}^{\btheta} \left ( \bbeta_\bx \right )/{\noepochsamplesize} = \sum_{v = 1}^{\kappa} q\left (\bbeta_{\bx}, \left \{Y^v, \bX^v, T^v, R^v, \bA^v \right \} \right ) \leq \sum_{v = 1}^{\kappa} q^*/\kappa = q^*(\bbeta_{\bx}).
\end{equation*}
The lemma follows.
\end{proof}

Lemmas~\ref{lem:log_concave}-\ref{lem:O_p_1} verify Assumption~\ref{as:gen_inf_iar} under the working model specified by Equation~\ref{eq:rdsmod}, Assumptions \ref{as:1}-\ref{as:3}, \ref{as:budg}, \ref{as:inf_iar}, and Conditions (C2) and (C4). 
We now prove Theorem~\ref{thm:misspec} under Assumption~\ref{as:gen_inf_iar} (which suffices to prove Theorem~\ref{thm:misspec} under both assumption sets).

\begin{proof}[Proof of Theorem~\ref{thm:misspec}]
\textbf{Proof of consistency.}
Because $\mathcal{B}$ is open, we know that $\bar{\bbeta}_\bx(\btheta)$ is an element of the interior of $\mathcal{B}$. Consequently, there exists $\gamma^* >0$ such that $$N_{\gamma^*} \triangleq \left \{ \bbeta_\bx : \left  \| \bbeta_\bx - \bar{\bbeta}_{\bx}(\btheta) \right \|_2 \leq \gamma^* \right \}$$ is a compact subset of $\mathcal{B}$.

Under Assumption~\ref{as:gen_inf_iar}, we know that ${\ell}_{{\noepochsamplesize}}^{\btheta}(\bbeta_\bx)/{\noepochsamplesize}$ is concave for all ${\noepochsamplesize} \in \mathbb{N}$, and that for $\gamma^* > 0$ there exists a $\delta > 0$ such that
\begin{align}
    \label{eq:strictconcavity}
    \lim_{\noepochsamplesize \to \infty} \inf_{\bbeta \in N_{\gamma^*}} \sigma_{\min} \left ( - \ddot{\ell}^{\btheta}_{\noepochsamplesize} \left ( \bbeta_{\bx} \right ) /\noepochsamplesize \right ) 
    &\geq \delta\ \ \mathrm{a.s.}
\end{align}

We now show that for $\bbeta_{\bx} \in N_{\gamma^*}$, $\ell_{{\noepochsamplesize}}^{\btheta} \left ( \bbeta_{\bx} \right )/{\noepochsamplesize}$ is strictly concave asymptotically. For any $\bbeta_\bx, \bbeta_\bx' \in N_{\gamma^*}$ and $t \in [0,1]$, we set $\bbeta_\bx^\dagger = t \bbeta_\bx + (1-t)\bbeta_\bx'$. For $\bbeta_\bx''\in \left  \{\alpha \bbeta_\bx + (1-\alpha)\bbeta_\bx^\dagger : \alpha \in [0,1] \right \} $, $\bbeta_\bx'''\in \left  \{\alpha \bbeta_\bx^\dagger + (1-\alpha )\bbeta_\bx' : \alpha \in [0,1] \right \} $, and an $\epsilon \in (0,\delta)$, there exists $N \in \mathbb{N}$ such that $\forall {\noepochsamplesize} \geq N$ (by Taylor expansion),
\begin{align*}
    \ell_{{\noepochsamplesize}}^{\btheta}\left ( \bbeta_\bx \right )/{\noepochsamplesize} &= \left \{ \ell_{{\noepochsamplesize}}^{\btheta}\left ( \bbeta_\bx^\dagger \right ) + \dot{\ell_{{\noepochsamplesize}}^{\btheta}}\left ( \bbeta_\bx^\dagger \right )^\top (\bbeta_\bx - \bbeta_\bx^\dagger ) + \frac{1}{2} (\bbeta_\bx - \bbeta_\bx^\dagger )^\top \ddot{\ell_{{\noepochsamplesize}}^{\btheta}}\left ( \bbeta_\bx'' \right ) (\bbeta_\bx - \bbeta_\bx^\dagger) \right \}/{\noepochsamplesize} \\
    &\leq \ell_{{\noepochsamplesize}}^{\btheta}\left ( \bbeta_\bx^\dagger \right )/{\noepochsamplesize} + \dot{\ell_{{\noepochsamplesize}}^{\btheta}}\left ( \bbeta_\bx^\dagger \right )^\top (\bbeta_\bx - \bbeta_\bx^\dagger)/{\noepochsamplesize} - \frac{1}{2} \| \bbeta_\bx - \bbeta_\bx^\dagger\|_2^2 \epsilon \ \ \mathrm{a.s.},
\end{align*}
and
\begin{align*}
    \ell_{{\noepochsamplesize}}^{\btheta}\left ( \bbeta_\bx' \right )/{\noepochsamplesize} &= \left \{ \ell_{{\noepochsamplesize}}^{\btheta}\left ( \bbeta_\bx^\dagger \right ) + \dot{\ell_{{\noepochsamplesize}}^{\btheta}}\left ( \bbeta_\bx^\dagger \right )^\top (\bbeta_\bx' - \bbeta_\bx^\dagger) + \frac{1}{2} (\bbeta_\bx' - \bbeta_\bx^\dagger)^\top \ddot{\ell_{{\noepochsamplesize}}^{\btheta}}\left ( \bbeta_\bx''' \right ) (\bbeta_\bx' - \bbeta_\bx^\dagger) \right \}/{\noepochsamplesize} \\
    &\leq \ell_{{\noepochsamplesize}}^{\btheta}\left ( \bbeta_\bx^\dagger \right )/{\noepochsamplesize} + \dot{\ell_{{\noepochsamplesize}}^{\btheta}}\left ( \bbeta_\bx^\dagger \right )^\top (\bbeta_\bx' - \bbeta_\bx^\dagger)/{\noepochsamplesize} - \frac{1}{2} \| \bbeta_\bx' - \bbeta_\bx^\dagger\|_2^2\epsilon \ \ \mathrm{a.s.}
\end{align*}
    Both of these inequalities follow because (1) for any $x \in \mathbb{R}^d$ and semi-positive definite matrix $A \in \mathbb{R}^{d\times d}$, we know that $x^\top A x \geq \sigma_{\min}(A)x^\top x$ and (2) we have a lower bound on $-\ddot{\ell_{{\noepochsamplesize}}^{\btheta}}\left ( \bbeta_\bx \right )/\noepochsamplesize$ for any $\bbeta_{\bx} \in N_{\gamma^*}$ by Equation~\ref{eq:strictconcavity}.

From these two expressions, we find that $\forall {\noepochsamplesize} \geq N$,
\begin{align*}
    t\ell_{{\noepochsamplesize}}^{\btheta}\left ( \bbeta_\bx \right )/{\noepochsamplesize} + (1-t)\ell_{{\noepochsamplesize}}^{\btheta}\left ( \bbeta_\bx' \right )/{\noepochsamplesize} &\leq \ell_{{\noepochsamplesize}}^{\btheta}\left ( \bbeta_\bx^\dagger \right )/{\noepochsamplesize} - \\
        &\hspace{0.5cm}\frac{1}{2}\left \{ t \| \bbeta_\bx - \bbeta_\bx^\dagger\|_2^2 +  (1-t) \| \bbeta_\bx' - \bbeta_\bx^\dagger\|_2^2 \right \} \epsilon\ \ \mathrm{a.s.} \\
    &= \ell_{{\noepochsamplesize}}^{\btheta}\left ( \bbeta_\bx^\dagger \right )/{\noepochsamplesize}  - \frac{1}{2}\Big \{ t(1-t)^2 \| \bbeta_\bx - \bbeta_\bx'\|_2^2 + \\
        &\hspace{0.5cm} (1-t)t^2 \| \bbeta_\bx - \bbeta_\bx'\|_2^2 \Big \}\epsilon \ \ \mathrm{a.s.} \\
    &= \ell_{{\noepochsamplesize}}^{\btheta}\left ( \bbeta_\bx^\dagger \right ) /{\noepochsamplesize} - \frac{1}{2} t(1-t)\| \bbeta_\bx - \bbeta_\bx'\|_2^2 \epsilon \ \ \mathrm{a.s.}
\end{align*}
    Lines 1, 2, and 3 follow from distributing, combining like terms, and refactoring (as well as the inequalities stated above). By point-wise convergence, this implies that
\begin{align*}
    t\overline{\ell}^{\btheta}\left ( \bbeta_\bx \right ) + (1-t)\overline{\ell}^{\btheta}\left ( \bbeta_\bx' \right ) \leq \overline{\ell}^{\btheta}\left ( \bbeta_\bx^\dagger \right )  - \frac{1}{2}\left \{ t(1-t)\| \bbeta_\bx' - \bbeta_\bx\|_2^2 \right \} \epsilon \ \ \mathrm{a.s.}
\end{align*}
Consequently, we know that $\overline{\ell}^{\btheta}$ is strictly concave a.s. over the compact set $N_{\gamma^*}$; i.e., for any $\bbeta_\bx, \bbeta_\bx' \in N_{\gamma^*}$ and $t \in (0,1)$,
\begin{align*}
    t\overline{\ell}^{\btheta}\left ( \bbeta_\bx \right ) + (1-t)\overline{\ell}^{\btheta}\left ( \bbeta_\bx' \right ) < \overline{\ell}^{\btheta}\left ( t\bbeta_\bx + (1-t)\bbeta_\bx' \right ).
\end{align*}
Lastly, we show that strict concavity over $N_{\gamma^*}$ implies that $\bar{\bbeta}_\bx(\btheta) \in \arg \max_{\bbeta_\bx \in \mathcal{B}} \overline{\ell}^{\btheta} (\bbeta_\bx)$ is unique. For a proof by contradiction, assume that $\bbeta_\bx \in \mathcal{B}$ such that $\bbeta_\bx \neq \bar{\bbeta}_\bx(\btheta)$ and $\bbeta_\bx \in \arg \max_{\bbeta_\bx \in \mathcal{B}} \overline{\ell}^{\btheta} (\bbeta_\bx)$ as well. By concavity, for any $t \in (0,1)$,
\begin{align*}
        \overline{\ell}^{\btheta} \left \{ \bar{\bbeta}_\bx(\btheta) \right \} &= t\overline{\ell}^{\btheta} \left \{ \bar{\bbeta}_\bx(\btheta) \right \} + (1-t)\overline{\ell}^{\btheta}(\bbeta_\bx)\\
        &\leq \overline{\ell}^{\btheta} \left \{ t\bar{\bbeta}_\bx(\btheta) + (1-t)\bbeta_\bx \right \}.
\end{align*}
Because $\bar{\bbeta}_\bx(\btheta)$ and $\bbeta_\bx$ are maxima, this implies that for any $t \in (0,1)$,
\begin{equation*}
    \overline{\ell}^{\btheta} \left \{ \bar{\bbeta}_\bx(\btheta) \right \} = \overline{\ell}^{\btheta} \left ( \bbeta_{\bx} \right )
    = \overline{\ell}^{\btheta} \left \{ t\bar{\bbeta}_\bx(\btheta) + (1-t)\bbeta_\bx \right \}.
\end{equation*}
For any $\bbeta_\bx \notin N_{\gamma^*}$, we find that for $t \geq 1- \gamma^*/ \|\bar{\bbeta}_\bx(\btheta) - \bbeta_\bx \|_2$,
\begin{align*}
    \| \bar{\bbeta}_\bx(\btheta) - \left ( t\bar{\bbeta}_\bx(\btheta) + (1-t)\bbeta_\bx \right ) \|_2 &=  \| (1-t)  \bar{\bbeta}_\bx(\btheta) - (1-t)\bbeta_\bx \|_2 \\
    &= (1-t) \|  \bar{\bbeta}_\bx(\btheta) - \bbeta_\bx \|_2 \\
    &\leq \gamma^*.
\end{align*}
We now pick any $t^* =  \max \left [1- \gamma^*/ \left \{ 2\|\bar{\bbeta}_\bx(\btheta) - \bbeta_\bx \|_2 \right \} , 0 \right ]$ and label $\bbeta_{\bx}' = t^*\bar{\bbeta}_\bx(\btheta) + (1-t^*)\bbeta_\bx$. We know that $\bbeta_{\bx}' \in N_{\gamma^*}$, and $\bbeta_{\bx}' \in \arg \max_{\bbeta_\bx \in \mathcal{B}} \overline{\ell}^{\btheta} (\bbeta_\bx)$. Additionally,
\begin{align}
    \begin{split}
        \label{eq:contradiction}
        \overline{\ell}^{\btheta} \left \{ \bar{\bbeta}_\bx(\btheta) \right \} &= \frac{1}{2}\overline{\ell}^{\btheta} \left \{ \bar{\bbeta}_\bx(\btheta) \right \} + \frac{1}{2}\overline{\ell}^{\btheta}(\bbeta_\bx')\\
        &< \overline{\ell}^{\btheta} \left \{ \frac{1}{2}\bar{\bbeta}_\bx(\btheta) + \frac{1}{2}\bbeta_\bx' \right \},
    \end{split}
\end{align}
where the last inequality is strict because $\overline{\ell}^{\btheta}$ is strictly concave over $N_{\gamma^*}$.
Equation~\ref{eq:contradiction} is a \textbf{contradiction} because $(\bar{\bbeta}_\bx(\btheta) + \bbeta_\bx')/2$ has a higher log-likelihood than $\bar{\bbeta}_\bx(\btheta)$. Consequently, $\bar{\bbeta}_\bx(\btheta)$ must be unique.

We now know that $\overline{\ell}^{\btheta}$ has a unique maximizer and $\ell_{{\noepochsamplesize}}^{\btheta}(\bbeta_\bx)/{\noepochsamplesize}$ is concave for all ${\noepochsamplesize} \in \mathbb{N}$. 
By Lemma~\ref{lem:concave}, $\widehat{\bbeta}_\bx^{{\noepochsamplesize}}(\btheta) \overset{P}{\to} \bar{\bbeta}_\bx(\btheta)$.

\paragraph{Ridge regression penalty.} We also prove consistency under a ridge regression penalty. The log-likelihood with a ridge regression component parameterized by $\alpha \in \mathbb{R}^+$ is
\begin{align*}
    \mathcal{\ell}^{\btheta, \alpha}_{{\noepochsamplesize}} \left ( \left \{\Gamma_{\ba}, \Omega_{\ba} \right \}_{\ba \in \mathcal{A}} \right ) &= \mathcal{\ell}^{\btheta}_{{\noepochsamplesize}} \left ( \left \{\Gamma_{\ba}, \Omega_{\ba} \right \}_{\ba \in \mathcal{A}} \right ) - \sum_{\ba \in \mathcal{A}}\frac{\alpha}{2}\mathrm{tr} \left (\Gamma_{\ba}^\top {\Omega_{\ba}}^{-1} \Gamma_{\ba} \right ).
\end{align*} 
Define $\widehat{\bbeta}_\bx^{{\noepochsamplesize}}(\btheta, \alpha) = \arg \max_{\bbeta_\bx \in \mathcal{B}} \mathcal{\ell}^{\btheta, \alpha}_{{\noepochsamplesize}}\left (\bbeta_\bx \right )$. We want to prove that $\widehat{\bbeta}_\bx^{{\noepochsamplesize}}(\btheta, \alpha)  \overset{p}{\to} \bar{\bbeta}_\bx(\btheta)$ as ${\noepochsamplesize} \to \infty$.
Since $-\Gamma_{\ba}^\top {\Omega_{\ba}}^{-1} \Gamma_{\ba} $ is negative semi-definite, we know that $- \sum_{\ba \in \mathcal{A}}\frac{\alpha}{2}\mathrm{tr} \left (\Gamma_{\ba}^\top {\Omega_{\ba}}^{-1} \Gamma_{\ba} \right )$ is concave. We conclude that $\ell^\alpha_{{\noepochsamplesize}}$ is concave for all ${\noepochsamplesize} \in \mathbb{N}$ since it is the sum of two concave functions. Additionally, since
\begin{equation*}
    \lim_{{\noepochsamplesize} \to \infty} \sum_{\ba \in \mathcal{A}}\frac{\alpha}{2}\mathrm{tr} \left (\Gamma_{\ba}^\top {\Omega_{\ba}}^{-1} \Gamma_{\ba} \right )/{\noepochsamplesize} = 0,
\end{equation*}
we know that for any $\bbeta_\bx \in \mathcal{B}$,
\begin{align*}
    \lim_{{\noepochsamplesize} \to \infty} \mathcal{\ell}^\alpha_{{\noepochsamplesize}} \left ( \bbeta_\bx \right )/{\noepochsamplesize} = \lim_{{\noepochsamplesize} \to \infty} \left [\mathcal{\ell}_{{\noepochsamplesize}} \left ( \bbeta_\bx \right )/{\noepochsamplesize} - \sum_{\ba \in \mathcal{A}}\frac{\alpha}{2}\mathrm{tr} \left ({\Gamma_{\ba}}^\top {\Omega_{\ba}}^{-1} \Gamma_{\ba} \right )/{\noepochsamplesize} \right ] = \overline{\ell}^{\btheta} (\bbeta_\bx).
\end{align*}

Using Lemma~\ref{lem:concave}, we conclude that $\widehat{\bbeta}_\bx^{{\noepochsamplesize}}(\btheta,\alpha) \overset{p}{\to} \bar{\bbeta}_\bx(\btheta)$.

\noindent\rule{16cm}{0.4pt}

\textbf{Concentration.}
We now prove the concentration statement of Theorem~\ref{thm:misspec}: for any $\btheta \notin \Theta^*$,
\begin{equation*}
   \lim_{{\noepochsamplesize} \to \infty} \bbP \left \{ \btheta \in \Gamma_{1-\alpha, \noepochsamplesize }\right \} \to 0.
\end{equation*}
Recall that $\Theta^* = \left \{\btheta \in \Theta : \bar{\bbeta}(\btheta) = \bar{\bbeta}(\btheta^*) \right \}$. For a given $\btheta \in \Theta$, we define the sampling distribution of the log-likelihood ratio statistic as $P_{\noepochsamplesize}^{\btheta}\left (\bar{\bbeta}_{\bx} \right )$ such that
\begin{equation*}
     -2\left [ \ell_{{\noepochsamplesize}}^{\btheta} \left \{ \bar{\bbeta}_\bx(\btheta) \right \} -  \ell_{{\noepochsamplesize}}^{\btheta} \left \{\widehat{\bbeta}^{{\noepochsamplesize}}_\bx(\btheta) \right \} \right ] \sim P_{\noepochsamplesize}^{\btheta}\left (\bar{\bbeta}_{\bx} \right ).
\end{equation*}
Furthermore, define the $1-\alpha$ quantile of $P_{\noepochsamplesize}^{\btheta}\left (\bar{\bbeta}_{\bx} \right )$ as $\bgamma_{1-\alpha, \noepochsamplesize}^{\btheta}\left (\bar{\bbeta}_{\bx} \right )$ and the confidence set in question as
\begin{equation*}
    \Gamma_{1-\alpha, \noepochsamplesize } = \left \{ \btheta : -2\left [ \ell_{{\noepochsamplesize}} \left \{ \bar{\bbeta}_\bx(\btheta) \right \} -  \ell_{{\noepochsamplesize}} \left \{\widehat{\bbeta}^{{\noepochsamplesize}}_\bx\right \} \right ] \leq \bgamma_{1-\alpha, \noepochsamplesize}^{\btheta}\left (\bar{\bbeta}_{\bx} \right ) \right \}.
\end{equation*}

We know that $\widehat{\bbeta}^{{\noepochsamplesize}}_\bx$ is consistent for $\bar{\bbeta}_\bx(\btheta^*)$. Therefore, we know that for any $\epsilon_0 > 0$ such that $\left \{ \bbeta_{\bx}: \left \| \bbeta_{\bx} -\bar{\bbeta}_{\bx}\left (\btheta^* \right )  \right \|_2 \leq \epsilon_0  \right \} \subset \mathcal{B}$, there exists $N_1 \in \mathbb{N}$ such that for all ${\noepochsamplesize} \geq N_1$, $\left \| \widehat{\bbeta}^{{\noepochsamplesize}}_\bx- \bar{\bbeta}_{\bx}\left (\btheta^* \right ) \right \| \leq \epsilon_0$ almost surely. 

Define the compact set $\mathcal{B}_c^{\btheta} = \bar{\bbeta}_{\bx}\left (\btheta \right ) \cup \left \{ \bbeta_{\bx}: \left \| \bbeta_{\bx} -\bar{\bbeta}_{\bx}\left (\btheta^* \right )  \right \|_2 \leq \epsilon_0  \right \}$. Additionally, for any $ \bbeta_{\bx}^p, \bbeta_{\bx}^g \in \mathcal{B}_c^{\btheta}$, define $\mathcal{B}_*^{\btheta}= \{ \bbeta_{\bx} \in \left  \{\alpha \bbeta_{\bx}^p + (1-\alpha)\bbeta_{\bx}^g : \alpha \in [0,1] \right \}$. Note that $\mathcal{B}_*^{\btheta}$ is compact because it is the image of a continuous function on a compact set in Euclidean space.

Under Assumption~\ref{as:gen_inf_iar}, we know that there exists an $N_2 \in \mathbb{N}$ such that for an $\epsilon_{\btheta} > 0$ and any ${\noepochsamplesize} \geq N_2$
\begin{align*}
	\inf_{\bbeta_{\bx} \in \mathcal{B}_*^{\btheta}} \sigma_{\min} \left \{ - \ddot{\ell}_{{\noepochsamplesize}} \left ( \bbeta_\bx \right ) /{\noepochsamplesize} \right \}
	&\geq \epsilon_{\btheta} \ \ \mathrm{a.s.}
\end{align*}

Consequently, for all ${\noepochsamplesize} \geq \max\{N_1, N_2\}$, any $\btheta \in \Theta$, and a \newline $\bbeta_\bx' \in \left  \{\alpha \widehat{\bbeta}^{{\noepochsamplesize}}_\bx+ (1-\alpha)\bar{\bbeta}_\bx(\btheta) : \alpha \in [0,1] \right \}$,
\begin{align*}
     &-2\left [ \ell_{{\noepochsamplesize}} \left \{ \bar{\bbeta}_\bx(\btheta) \right \} -  \ell_{{\noepochsamplesize}} \left \{\widehat{\bbeta}^{{\noepochsamplesize}}_\bx\right \} \right ]/{\noepochsamplesize} \\
     &= 2\left [ \ell_{{\noepochsamplesize}} \left \{\widehat{\bbeta}^{{\noepochsamplesize}}_\bx\right \} - \ell_{{\noepochsamplesize}} \left \{ \bar{\bbeta}_\bx(\btheta) \right \} \right ]/{\noepochsamplesize}\\
      &= 2\Big [ \ell_{{\noepochsamplesize}} \left \{  \widehat{\bbeta}^{{\noepochsamplesize}}_\bx\right \} -\ell_{{\noepochsamplesize}} \left \{ \widehat{\bbeta}^{{\noepochsamplesize}}_\bx\right \} - \dot{\ell}_{{\noepochsamplesize}} \left \{  \widehat{\bbeta}^{{\noepochsamplesize}}_\bx\right \}^\top\left \{ \bar{\bbeta}_\bx(\btheta) -  \widehat{\bbeta}^{{\noepochsamplesize}}_\bx\right \} - \\
      &\hspace{0.7cm} \frac{1}{2}\left \{  \bar{\bbeta}_\bx(\btheta)  - \widehat{\bbeta}^{{\noepochsamplesize}}_\bx\right \}^\top \ddot{\ell}_{{\noepochsamplesize}} \left \{ \bbeta_\bx' \right \} \left \{  \bar{\bbeta}_\bx(\btheta)  - \widehat{\bbeta}^{{\noepochsamplesize}}_\bx\right \}  \Big ]/{\noepochsamplesize} \\
      &=\left \{  \bar{\bbeta}_\bx(\btheta)  - \widehat{\bbeta}^{{\noepochsamplesize}}_\bx\right \}^\top \left \{ -\ddot{\ell}_{{\noepochsamplesize}} \left (\bbeta_\bx' \right )/{\noepochsamplesize} \right \} \left \{  \bar{\bbeta}_\bx(\btheta)  - \widehat{\bbeta}^{{\noepochsamplesize}}_\bx\right \} \\
      &\geq \epsilon_{\btheta} \left \| \bar{\bbeta}_\bx(\btheta)  - \widehat{\bbeta}^{{\noepochsamplesize}}_\bx\right \|_2^2 \ \ \mathrm{a.s.}
\end{align*}
Line~3 follows from an exact Taylor expansion. Line~4 follows from the fact that \\ $\dot{\ell}_{{\noepochsamplesize}} \left \{  \widehat{\bbeta}^{{\noepochsamplesize}}_\bx\right \} = 0$. Line~5 follows from the fact that $\bbeta_\bx' \in \mathcal{B}_*^{\btheta}$.

For $\bbeta_\bx'' \in \left  \{\alpha \widehat{\bbeta}^{{\noepochsamplesize}}_\bx(\btheta) + (1-\alpha)\bar{\bbeta}_\bx(\btheta) : \alpha \in [0,1] \right \}$, we now upper bound
\begin{align*}
    &0 \leq -2\left [ \ell_{{\noepochsamplesize}}^{\btheta} \left \{ \bar{\bbeta}_\bx(\btheta) \right \} -  \ell_{{\noepochsamplesize}}^{\btheta} \left \{\widehat{\bbeta}^{{\noepochsamplesize}}_\bx(\btheta) \right \} \right ]/{\noepochsamplesize} \\
    &= 2\left [\ell_{{\noepochsamplesize}}^{\btheta} \left \{\widehat{\bbeta}^{{\noepochsamplesize}}_\bx(\btheta) \right \} - \ell_{{\noepochsamplesize}}^{\btheta} \left \{ \bar{\bbeta}_\bx(\btheta) \right \} \right]/{\noepochsamplesize} \\
    &= 2\Big [ \ell_{{\noepochsamplesize}}^{\btheta} \left \{  \bar{\bbeta}_\bx(\btheta) \right \} - \ell_{{\noepochsamplesize}}^{\btheta} \left \{ \bar{\bbeta}_\bx(\btheta) \right \} + \dot{\ell}_{{\noepochsamplesize}}^{\btheta} \left \{  \bar{\bbeta}_\bx(\btheta) \right \}^\top\left \{\widehat{\bbeta}^{{\noepochsamplesize}}_\bx(\btheta) - \bar{\bbeta}_\bx(\btheta) \right \} + \\
    &\hspace{0.7cm} \frac{1}{2}\left \{  \bar{\bbeta}_\bx(\btheta)  - \widehat{\bbeta}^{{\noepochsamplesize}}_\bx(\btheta)\right \}^\top \ddot{\ell}_{{\noepochsamplesize}}^{\btheta} \left \{ \bbeta_\bx'' \right \} \left \{  \bar{\bbeta}_\bx(\btheta)  - \widehat{\bbeta}^{{\noepochsamplesize}}_\bx(\btheta)\right \}  \Big ]/{\noepochsamplesize} \\
    &= 2 \left [\dot{\ell}_{{\noepochsamplesize}}^{\btheta} \left \{  \bar{\bbeta}_\bx(\btheta) \right \} /{\noepochsamplesize} \right ]^\top\left \{\widehat{\bbeta}^{{\noepochsamplesize}}_\bx(\btheta) - \bar{\bbeta}_\bx(\btheta)  \right \} - \\
    &\hspace{0.7cm} \left \{  \bar{\bbeta}_\bx(\btheta)  -\widehat{\bbeta}^{{\noepochsamplesize}}_\bx(\btheta) \right \}^\top \left [ - \ddot{\ell}_{{\noepochsamplesize}}^{\btheta} \left \{ \bbeta_\bx'' \right \}/{\noepochsamplesize} \right ] \left \{  \bar{\bbeta}_\bx(\btheta)  - \widehat{\bbeta}^{{\noepochsamplesize}}_\bx(\btheta)\right \} \\
    &\leq 2 \left [\dot{\ell}_{{\noepochsamplesize}}^{\btheta} \left \{  \bar{\bbeta}_\bx(\btheta) \right \} /{\noepochsamplesize} \right ]^\top\left \{\widehat{\bbeta}^{{\noepochsamplesize}}_\bx(\btheta) - \bar{\bbeta}_\bx(\btheta)\right \} \ \ \mathrm{a.s.}
\end{align*}
Lines~1, 2, and 3 follow from simple algebra. Line~4 follows from the fact that we know from the same logic as Lemma~\ref{lem:IAmisp} that $\ell_{{\noepochsamplesize}}^{\btheta}$ is concave for any $\btheta \in \Theta$. 

By Assumption~\ref{as:gen_inf_iar}, we know that $\dot{\ell}_{{\noepochsamplesize}}^{\btheta} \left \{  \bar{\bbeta}_\bx(\btheta) \right \}/{\noepochsamplesize} = O_p(1)$ as $\kappa \to \infty$.
Because
$\lim_{{\noepochsamplesize} \to \infty} \left \| \bar{\bbeta}_\bx(\btheta)  -\widehat{\bbeta}^{{\noepochsamplesize}}_\bx(\btheta) \right \|_2^2 = 0 $ a.s., we know that
\begin{align*}
   \lim_{{\noepochsamplesize} \to \infty} \left ( 2\left [\dot{\ell}_{{\noepochsamplesize}}^{\btheta} \left \{  \bar{\bbeta}_\bx(\btheta) \right \} /{\noepochsamplesize} \right ]^\top\left \{\bar{\bbeta}_\bx(\btheta) - \widehat{\bbeta}^{{\noepochsamplesize}}_\bx(\btheta) \right \} \right ) = 0. \ \ \mathrm{a.s.}
\end{align*}
Therefore, we know that $\lim_{{\noepochsamplesize} \to \infty} \bgamma_{1-\alpha, \noepochsamplesize}^{\btheta}\left (\bar{\bbeta}_{\bx} \right )/{\noepochsamplesize} = 0$ a.s. We find that
\begin{align*}
    \Gamma_{1-\alpha, \noepochsamplesize } &\triangleq \left \{ \btheta : -2\left [ \ell_{{\noepochsamplesize}} \left \{ \bar{\bbeta}_\bx(\btheta) \right \} -  \ell_{{\noepochsamplesize}} \left \{\widehat{\bbeta}^{{\noepochsamplesize}}_\bx\right \} \right ] \leq \bgamma_{1-\alpha, \noepochsamplesize}^{\btheta}\left (\bar{\bbeta}_{\bx} \right ) \right \} \\
    &= \left \{ \btheta : -2\left [ \ell_{{\noepochsamplesize}} \left \{ \bar{\bbeta}_\bx(\btheta) \right \} -  \ell_{{\noepochsamplesize}} \left \{\widehat{\bbeta}^{{\noepochsamplesize}}_\bx\right \} \right ]/{\noepochsamplesize} \leq \bgamma_{1-\alpha, \noepochsamplesize}^{\btheta}\left (\bar{\bbeta}_{\bx} \right )/{\noepochsamplesize} \right \} \\
    &\subseteq \left \{ \btheta : \epsilon_{\btheta} \left \| \bar{\bbeta}_\bx(\btheta)  - \widehat{\bbeta}^{{\noepochsamplesize}}_\bx\right \|_2^2 \leq \bgamma_{1-\alpha, \noepochsamplesize}^{\btheta}\left (\bar{\bbeta}_{\bx} \right )/{\noepochsamplesize} \right \} \ \mathrm{a.s.}
\end{align*}
Line~2 follows by dividing both sides of the inequality by ${\noepochsamplesize}$. Line~3 follows because for any $\btheta \in \Theta$,
\begin{align*}
     &-2\left [ \ell_{{\noepochsamplesize}} \left \{ \bar{\bbeta}_\bx(\btheta) \right \} -  \ell_{{\noepochsamplesize}} \left \{\widehat{\bbeta}^{{\noepochsamplesize}}_\bx\right \} \right ]/{\noepochsamplesize} \geq \epsilon_{\btheta} \left \| \bar{\bbeta}_\bx(\btheta)  - \widehat{\bbeta}^{{\noepochsamplesize}}_\bx\right \|_2^2 \ \ \mathrm{a.s.}
\end{align*}
as ${\noepochsamplesize} \to \infty$. Because $\lim_{{\noepochsamplesize} \to \infty} \widehat{\bbeta}^{{\noepochsamplesize}}_\bx= \bar{\bbeta}_\bx(\btheta^*)$, we know that if $\bar{\bbeta}_\bx(\btheta^*) \neq \bar{\bbeta}_\bx(\btheta)$,
\begin{equation*}
    \epsilon_{\btheta} \left \| \bar{\bbeta}_\bx(\btheta)  - \widehat{\bbeta}^{{\noepochsamplesize}}_\bx\right \|_2^2 \to \epsilon_{\btheta} \left \| \bar{\bbeta}_\bx(\btheta)  - \bar{\bbeta}_\bx(\btheta^*) \right \|_2^2 > 0.
\end{equation*}
Additionally, $\bgamma_{1-\alpha, \noepochsamplesize}^{\btheta}\left (\bar{\bbeta}_{\bx} \right )/{\noepochsamplesize} \to 0$.
Consequently, for any $\btheta \notin \Theta^*$,
\begin{align*}
   \lim_{{\noepochsamplesize} \to \infty} \bbP \left ( \btheta \in \Gamma_{1-\alpha, \noepochsamplesize } \right ) \leq
   \lim_{{\noepochsamplesize} \to \infty} \bbP \left (  \epsilon_{\btheta} \left \| \bar{\bbeta}_\bx(\btheta)  - \widehat{\bbeta}^{{\noepochsamplesize}}_\bx\right \|_2^2 \leq \bgamma_{1-\alpha, \noepochsamplesize}^{\btheta}\left (\bar{\bbeta}_{\bx} \right )/\noepochsamplesize \right ) 
   \to 0.
\end{align*}
\end{proof}

\section{Branching Model Paradigm for RL-RDS}

In this section, we conduct a series of simulation experiments to evaluate the operating characteristics
of RL-RDS when the branching process is the true generative model.  To allow comparisons with the two-stage procedure proposed by 
\citet{mcfall2021optimizing} (see also \citet{vanorsdale2023adaptive}),
we consider the setting in which the goal is to recruit the largest subset of 
people in a hidden population with a given binary trait, e.g., undiagnosed HIV.     
The outcome is thus an indicator of this trait.  
We estimate the optimal policy, $\bpi^{\mathrm{opt}}$, using RL-RDS. In this paradigm, we assume that the branching process described in Section~\ref{sec:examp} is the true generative model. We restate this model here for convenience.

Recall that $T_{i,l}$, $\bX_{i,l}$, $Y_{i,l}$, and $A_{i,l}$ for $l=1,\ldots, M_i$ are the arrival times, covariates, rewards, and coupon types associated with the potential recruits of recruiter $i$ respectively. 
Treating $A_{i,l}$ as a factor, let $\bZ_{i,l} \in\mathcal{Z}\subseteq \mathbb{R}^{\ell}$ be a row in the model matrix of a model that includes a main effect, $\bX_{i,l}$ and its interaction with $A_{i,l}$.
For example, if $\mathcal{A} = \{-1, 1\}$, then $\bZ_{i,l} = \left ( 1, \bX_{i,l}, \bX_{i,l} \bbI \left (A_{i,l} = -1 \right ) \right )$. Define $\mathbb{A}$ as the set of possible coupon types.

We consider a working model of the form (same as Equation~\ref{eq:rdsmod}):  
\begin{eqnarray*}
\bbP(M_{i} = m_{i} | \bH_{i}, \bA_i) &=& \frac{
\lambda^{m_i}/m_i!
}{
  \sum_{\ell=0}^{|\mathbf{A}_i|} (\lambda^\ell/\ell!)  
} , m_i=0,\ldots, |\mathbf{A}_i|,   \nonumber \\[3pt]
T_{i,l} - T_i | \bH_i, \bA_i, M_i &\sim& 
\mathrm{Truncated \ Exponential}(\zeta, t_{\min}, t_{\max}),\, l=1,\ldots, M_i,  \nonumber \\[3pt] 
\bX_{i,l} | \bH_i, U_{i,l}, \bA_i, M_i &\sim& 
\mathrm{Normal}\left (\bphi_{a} + G_{a}\bX_i,\Sigma_{a} \right ),\,
l=1,\ldots, M_i, a = a_{i,l}  \nonumber \\[3pt]
Y_{i,l} | \bH_i, \bX_i, U_{i,l}, \bA_i, M_i &\sim& \mathrm{ Bernoulli} \left \{ \frac{1}{1 + \exp \left (-{\bZ_{i,l}}^\top \bbeta_y \right )} \right \} ,\, 
j=1,\ldots, M_i,
\end{eqnarray*}
where $\lambda, \zeta \in \mathbb{R}$, $\{ G_{a} \}_{a \in \mathbb{A}} \in \mathbb{R}^{p \times p}$, 
$\{ \bphi_{a} \}_{a \in \mathbb{A}} \in \mathbb{R}^{p \times 1}$, $\bbeta_y \in \mathbb{R}^{\ell}$, 
and $\{\Sigma_{a} \}_{a \in \mathbb{A}} \in \mathbb{R}^{p\times p}$.
We draw the covariates of the initial sample, $\mathcal{E}_0$, from a multivariate normal distribution, $\mathrm{Normal} \left (  \mu , \Sigma \right )$.

\subsection{Policies}
We evaluate the performance of RL-RDS against a suite of alternative strategies. At each step, the researcher can choose from a finite selection of coupon types. 
The fixed allocation policies (i.e., those that give the same coupon allocation type to all participants)  represent the current standard in RDS. The train and implement policy mimics the procedure used by \cite{mcfall2021optimizing}, which determines an incentive strategy using a pilot study. 
We describe each policy below.
\begin{enumerate}
	\item \textbf{Fixed} offers a fixed coupon allocation $\ba \in \mathcal{A}$ to every study participant. If $\ba \notin  \psi^v(\bh^v)$, then pick a random coupon allocation from $\psi^v(\bh^v)$ to give to the $v^{\mathrm{th}}$ study participant.
	\item \textbf{Random} offers a random element of $\psi^v(\bh^v)$ to the $v^{\mathrm{th}}$ study participant.
         \item \textbf{Train and Implement} uses half of the budget for a ``pilot study," in which the the Random policy is used to assign coupon allocations. It then conducts policy search using the pilot study data to estimate the branching process working model. This estimated policy (without updating) is used to determine coupon allocations for the remainder of the budget.
	\item \textbf{RL-RDS} uses the Random policy to assign coupon allocations to participants in a short ``warm-up" period ($50$ participants in the simulations below). Then, it performs policy search with Thompson sampling as outlined in Section~\ref{sec:rl} for the remainder of the budget.
\end{enumerate}

To conduct RL-RDS, we establish a reasonable space of policies, $\Pi$. Define $\alpha_0 \in \mathbb{R}$, $\bma{\alpha}_1 \in \mathbb{R}^p$, and $\bma{\alpha} = (\alpha_0, \bma{\alpha}_1)$. For $n \in \mathbb{N}$ and state $\bh^n \in \mathcal{H}^n$, we consider policies of the form $\bpi(\bh^n) =  \left \{ \pi^n(\bh^n), \pi^{n+1}(\bh^{n+1}), \cdots  \right \}$ such that for $v \geq n$,
\begin{equation*}
    \pi^v(\bh^v) = \pi^v(\bh^v, \bma{\alpha}) = \pi^v(\bx^v,\bma{\alpha}) = g^v \left [ \frac{1}{1+\exp \left \{ - (\alpha_0 + {\bx^v}^\top \bma{\alpha}_1) \right \} } \right ],
\end{equation*}
where $g^v: (0,1) \to \phi^v(\bh^v)$ maps a continuous score (dependent on the participant's covariates) to a coupon allocation, $\ba^v \in \psi^v(\bh^v)$. 
We first draw $\widehat{\bbeta}^{n}$ from the MLE sampling distribution using a generalized bootstrap for estimating equations \citep{chatterjee2005generalized}. We then generate synthetic data sets,
$$\mathcal{K}^B(\bh^n, \bpi; \widehat{\bbeta}^{n}) = 
\left\lbrace
\left(\bh_b^n, A_b^n, Y_b^n, \bh_b^{n+1}, A_b^{n+1}, Y_b^{n+1},\ldots,
\bh_b^{Q}, Y_b^{Q}
\right)
\right\rbrace_{b=1}^B,$$
and calculate
\begin{equation*}
\widehat{V}^n_{B}(\bh^{n}, \bpi; \widehat{\bbeta}^{n}) =
\frac{\sum_{b =1 }^B \left\lbrace
\sum_{v \geq n} \Delta^v_b Y^{v} _b
\right\rbrace}{B}
\end{equation*}
for each $\bpi \in \Pi$ (or each $\bpi$ in a grid approximation of $\Pi$). To determine the coupon allocation given to the latest study participant, we set $\widehat{\balpha}_B^n = \arg \max_{\balpha} \widehat{V}^n_{B}(\bh^{n}, \balpha; \widehat{\bbeta}^{n} )$,
and assign $\ba^n = \pi^n(\bx^n, \widehat{\balpha}_B^n) $. 

\subsection{Results}
In the following simulations, we set the hidden population size to $N= 5,000$. The recruitment process begins with an initial sample of $25$ individuals randomly drawn from the population, $\left | \mathcal{E}_0\right | = 25$ with $\mu^* = (1,1,1)^\top$ and $\Sigma^* = \mathrm{diag}(1/8)$. We define 
\begin{align*}
    &G_{a_1} = \begin{pmatrix}
         \phi_1 & 0 & 0 \\
         0 & \phi_1 & 0 \\
         0 & 0 & \phi_1
    \end{pmatrix}, 
    G_{a_2} = \begin{pmatrix}
         \phi_1*0.975 & 0 & 0 \\
         0 & \phi_1* 0.975 & 0 \\
         0 & 0 & \phi_1 * 0.975
    \end{pmatrix} \\
    &G_{a_3} = \begin{pmatrix}
        \phi_1*0.95 & 0 & 0 \\
        0 & \phi_1* 0.95 & 0 \\
        0 & 0 & \phi_1 * 0.95
    \end{pmatrix}, \ \bphi_{a_1} = \bphi_{a_2} = \bphi_{a_3} = \bphi_0,
\end{align*}
where we set $\bphi_0 = \left (0.05, 0.05, 0.05 \right )$ and $\phi_1 = 0.95$ in our weak correlation setting, and we set $\bphi_0 = \left (0.025, 0.025, 0.025 \right )$ and $\phi_1 = 0.975$ in our strong correlation setting.
The researchers have access to three types of coupon allocations $\{\ba_1, \ba_2, \ba_3\}$, and each allocation has $5$ coupons. We found that limiting the number of coupons given to each study participant in the pilot study and the warm-up period of the T\&I and the RLRDS policies respectively allows us to observe the effects of the learned policies earlier in the sampling process. Consequently, we give two coupons to individuals in the pilot and warm-up period, and increase the allotment to $5$ coupons afterwards. Additionally, $C^v \equiv 1$ for all $v \in \mathbb{N}$. Define the reward model components as $\bZ_j^v \triangleq \left \{1, \bX^v_j, \mathbb{I}(\bA^v_j = \ba_2)\bX^v_j,  \mathbb{I}(\bA^v_j = \ba_3)\bX^v_j \right \}$ and $\bbeta_y^* \triangleq (-1,3\bk,-3\bk,-6\bk)$, where $\bk \triangleq (1,-1,-1)$.
The basic policy objective is clear: we want to ensure that coupon type 1, $\ba_1$, is used to recruit individuals with covariates that satisfy ${\bX^v}^\top \bk > 0$, and coupon type 3, $\ba_3$, is used to recruit individuals with covariates that satisfy ${\bX^v}^\top \bk < 0$. 
We define the policy space by specifying the function $g^v$,
\begin{equation*}
    g^v(z) = \begin{cases}
      \ba_1, & \mathrm{if}\ z > 0.66, \\
      \ba_2, & \mathrm{if}\ 0.66 > z \geq 0.33, \\
      \ba_3, & \mathrm{if}\ 0.33 > z.
    \end{cases}
\end{equation*}
This policy space implies that correctly assigning coupons $1$ or $3$ will depend on the sign of the $\balpha$ components. The frequency of coupon 2 allocation will be determined by the magnitude of $\balpha$. This structure makes finding an optimal policy computationally feasible while maintaining sufficient difficulty to showcase the strength of RL-RDS. Lastly, we set $\lambda = 3$, $\zeta = 0.5$, $t_{\min}=0$, and $t_{\max} = 3$.

Figure~\ref{fig:branchingcumrew} illustrates the estimated value of policies in both the sparse and dense network settings. It indicates that RL-RDS outperforms all competitor policies by a significant margin in each regime.
Note that the train and implement (T\&I) method does not adapt its policy after the initial sample, causing the margin between RL-RDS and this strategy to increase for larger budgets. 

\begin{figure}
    \centering
    \includegraphics{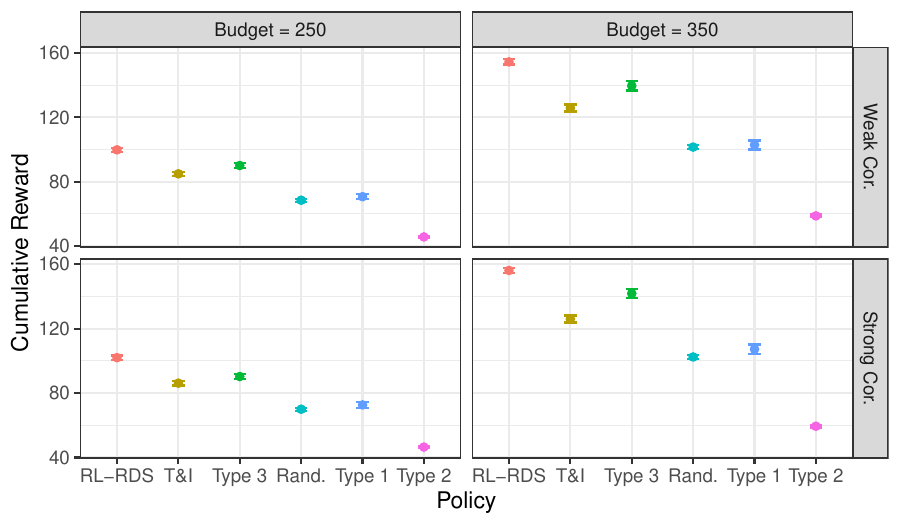}
    \caption{This figure compares the estimated cumulative reward of each policy with 90\% Monte Carlo confidence intervals over multiple sample sizes and graph densities.}
    \label{fig:branchingcumrew}
\end{figure}

\section{Alternative Graph Model Paradigms for RL-RDS}
\label{app_sec:RL-RDSalt}

The setup of this section mirrors the setup of Section~\ref{sec:sim}. 
We conduct a series of simulation experiments to evaluate the operating characteristics
of RL-RDS when the graph process of Section~\ref{sec:sim} is the true generative model.
Here, we evaluate the performance of RL-RDS under two additional simulation settings. In both of these paradigms, the coupon ``type" is held fixed.
\begin{enumerate}
    \item Researchers can vary the coupon value while the number of coupons is fixed.
    \item Researchers can vary the number of coupons while the coupon value is fixed.
\end{enumerate}

The fixed allocation policies (i.e., those that give the same coupon allocation type to all participants)  represent the current standard in RDS. The train and implement policy mimics the procedure used by \cite{mcfall2021optimizing}, which determines an incentive strategy using a pilot study. 
We describe each policy below.
\begin{enumerate}
	\item \textbf{Fixed (Min)} offers $\min_{a \in \psi^v(\bh^v)}$ to the $v^{\mathrm{th}}$ study participant.
	\item \textbf{Fixed (Half)} offers $\lfloor \max_{a \in \psi^v(\bh^v)}/2 \rfloor$ to the $v^{\mathrm{th}}$ study participant.
	\item \textbf{Fixed (Max)} offers $\max_{a \in  \psi^v(\bh^v)}$ to the $v^{\mathrm{th}}$ study participant.
	\item \textbf{Random} offers a random element of $\psi^v(\bh^v)$ to the $v^{\mathrm{th}}$ study participant.
         \item \textbf{Train and Implement} uses half of the budget for a ``pilot study," in which the the Random policy is used to assign coupon allocations. It then conducts policy search using the pilot study data to estimate the branching process working model. This estimated policy (without updating) is used to determine coupon allocations for the remainder of the budget.
	\item \textbf{RL-RDS} uses the Random policy to assign coupon allocations to participants in a short ``warm-up" period ($50$ participants in the simulations below). Then, it performs policy search with Thompson sampling as outlined in Section~\ref{sec:rl} for the remainder of the budget.
\end{enumerate}
The working model we use for inference in the experiments in this section involves the value of the coupons explicitly and reduces the dimensionality of the covariate and reward models. The possible coupons will have values between $0$ and $1$, $\mathbb{A} \subseteq [0,1]$.
 
	$T_{i,l}$,$\bX_{i,l}$, $Y_{i,l}$ and $A_{i,l}$ for $l=1,\ldots, M_i$ are the arrival times, covariates, rewards, and incentive values associated with the potential recruits of study participant $i$ respectively.
	\begin{eqnarray}\label{eq:rdsmod2}
		P(M_i=m_i) &=& \frac{
			\lambda^{m_i}/m_i!
		}{
			\sum_{\ell=0}^{|\mathbf{A}_i|} (\lambda^\ell/\ell!)  
		} , m_i=0,\ldots, |\mathbf{A}_i|,   \nonumber \\[3pt]
		U_{i,l} \triangleq T_{i,l} - T_i &\sim& 
		\mathrm{Truncated Exponential}(\zeta_0 + \zeta_1A_{i,l}, b),\, l=1,\ldots, M_i,  \nonumber \\[3pt] 
		\bX_{i,l} &\sim& 
		\mathrm{Normal}\left [\pmb{\phi}_0 + \phi_1\bX_{i,l},\left \{\Omega_0 + \left (1-A_{i,l} \right )\Omega_1 \right \}^{-1} \right ],\,
		l=1,\ldots, M_i,  \nonumber \\[3pt]
		Y_{i,l} &\sim& \mathrm{ Bernoulli} \left \{ \frac{1}{1 + \exp \left (-{\bZ_{i,l}}^\top \bbeta_y \right )} \right \} ,\, 
		l=1,\ldots, M_i,
	\end{eqnarray}
	where the incentive, $A_{i,l} \in [0,1]$, is coded so that a value of $0$ encodes the minimal incentive and $1$ encodes the maximum incentive. For simplicity (and to align with common study constraints),
	we assume that the incentive allocation strategy is such that
	$A_{i,l}$ is constant across $l$.
	Lastly, we assume that there is an upper bound on the number of coupons that can be given to a single participant; i.e., there exists $L \in \mathbb{N}$ such that $\forall i \in \mathbb{N}$, $|\bA_i | \leq L$.

\subsection{Results}
\label{app_sec:add_results}
In all following simulations, the hidden population size is $N= 5,000$. The recruitment process begins with an initial sample of $25$ individuals randomly drawn from the population, $\left | \mathcal{E}_0\right | = 25$. The graph model is defined by $\psi^*_0 =  0$ and $\psi^*_1 \in \{0.5,2\}$ to compare between simulation experiments in dense and sparse network settings respectively. 
We specify the covariate distribution with $\bmu^* = (1,1,1)$ and
$\Sigma^* = 10*I_3$. We define the arrival time distribution with
$\zeta^*_0 = 0.5$ and $\zeta^*_1 = 6$. $\zeta^*_1$ controls the relationship between the incentive offered to a potential recruit and their arrival time, which effects whether we observe them before the end of the process.
Consequently, for large $\zeta^*_1$, there is a high positive association between the incentive assigned to a recruiter and the likelihood of observing their recruits.

We switch to superscript indexing for study participants when discussing policy estimation to emphasize that the assignment of coupon allocations happens upon the arrival of a study participant. 
In the first simulation setting, we set the coupon package size to $5$ for all recruits and vary the value of the coupons. For the recruits $j \in \{1,2,\ldots, M^v \}$ of participant $v$, we define $\bZ^v_j \triangleq (1, A^v_j, \bX^v_j, A^v_j * \bX^v_j)$.
We make $\bbeta_y^* \in \mathbb{R}^{2p +2}$ sparse, setting it equal to $\bbeta_y^* = (-4,0,0,0,0,3,0,0)$. This makes the basic policy objective clear: give high incentives to study participants who are likely to recruit individuals with a particular characteristic in order to maximize cumulative utility. 
We define the possible incentive values as $\{0.1,0.2,0.3,\ldots, 1\}$ and the policy space for participant $v$ as
\begin{align*}
    &\pi^v(\bH^v) \triangleq \bA^v, \ \mathrm{where \ for} \ A_j^v \in \bA^v, \\
    &A_j^v = \frac{1}{10} \left \lfloor 10 \times \frac{1}{1+\exp \left [- (\alpha_0 + {\bx^v}^\top \bma{\alpha}_1) \right ] } + \frac{1}{2} \right \rfloor, \ \mathrm{for} \  j \in 1,\ldots, M^v.
\end{align*}
Therefore, the policy is determined by $\balpha_1 \in \mathbb{R}^p$ and $\alpha_0 \in \mathbb{R}$.

In the second simulation context, we hold the coupon value constant and allow the researcher to vary the size of $\bA_i$. 
We make $\bZ^v = (1, \bX^v)$ and set $\bbeta_y = (-1, 2,-2,-2)$. This provides a strong signal to prioritize policies that recruit individuals who exhibit the covariate pattern: $X^v_{1} >0 , X^v_{2} < 0, X^v_{3} < 0$.
In this context, we make the possible coupon package sizes $\{1,2,3, \ldots, 7\}$, and specify the policy space as
\begin{equation*}
    \pi^v(\bh^v) \triangleq |\bA^v|= \left \lfloor 7 \times \frac{1}{1+\exp \left \{- (\alpha_0 + {\bx^v}^\top \bma{\alpha}_1) \right \} } + \frac{1}{2} \right \rfloor.
\end{equation*}
Again, the policy is determined by $\balpha_1 \in \mathbb{R}^p$ and $\alpha_0 \in \mathbb{R}$.

Figures~\ref{fig:cumrew_price} and \ref{fig:cumrew_amount} contain the estimated value of policies in both simulation settings. They indicate that RL-RDS outperforms all competitor policies by a significant margin in each regime. Additionally, the effect sizes increase slightly as the network becomes sparser and the similarities between neighbors become stronger (because there is more ``signal" for the polices to leverage).
Lastly, note that the train and implement method does not adapt its policy after its ``pilot study," causing the margin between RL-RDS and this strategy to generally increase in the larger budget setting.

\begin{figure}
    \centering
    \includegraphics{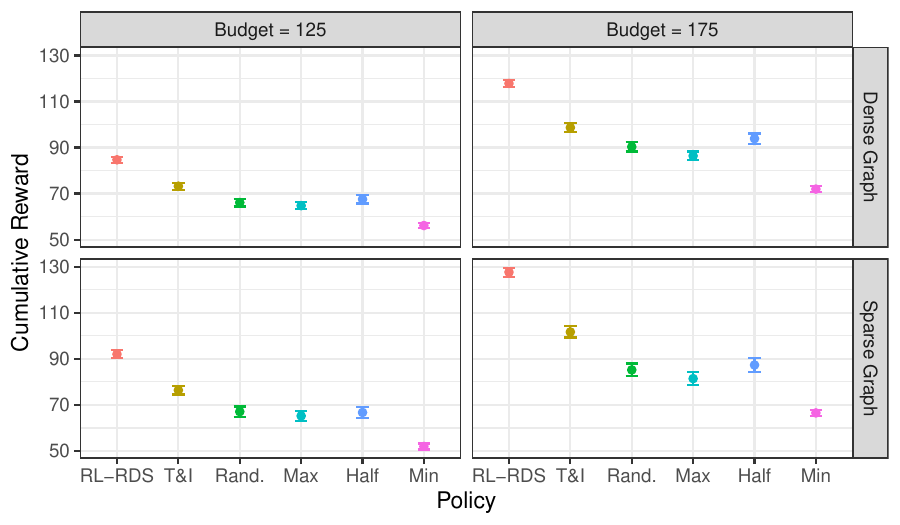}
    \caption{This figure compares the estimated cumulative reward of each policy with 90\% Monte Carlo confidence intervals in simulation setting 1.}
    \label{fig:cumrew_price}
\end{figure}

\begin{figure}
    \centering
    \includegraphics{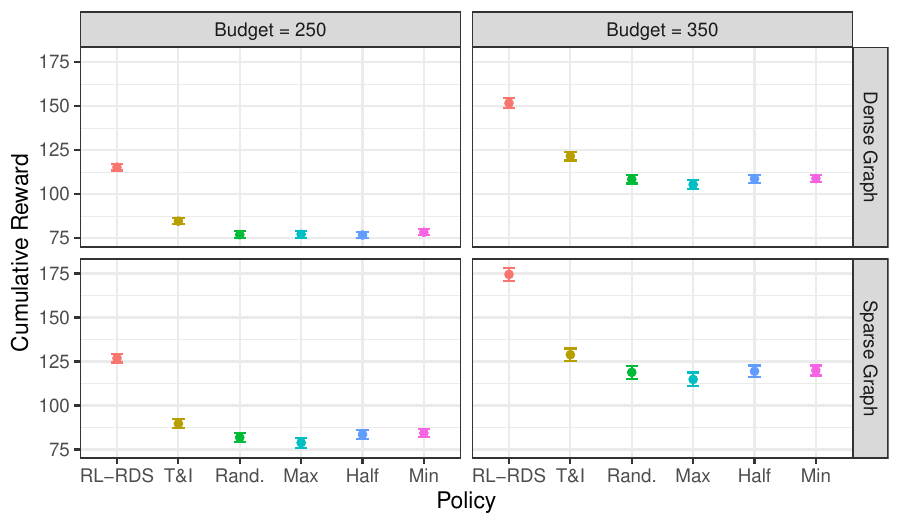}
    \caption{This figure compares the estimated cumulative reward of each policy with 90\% Monte Carlo confidence intervals in simulation setting 2.}
    \label{fig:cumrew_amount}
\end{figure}

\end{document}